\documentclass[a4paper,11pt]{article} 
\usepackage[english]{babel}
\usepackage[a4paper]{geometry}
\usepackage{amsmath}
\usepackage{amsthm}
\usepackage{amssymb}
\usepackage{color}

\usepackage{hyperref}         %%%% click to contents
\def\bR{\mathbb{R}}
\def\bT{\mathbb{T}}
\def\bN{\mathbb{N}}
\def\NN{\mathbb{N}}
\def\bZ{\mathbb{Z}}
\def\cC{\mathcal{C}}

\def\cQ{\mathcal{Q}}
\def\cD{\mathcal{D}}

\def\cR{\mathcal{R}}

\def\cV{\mathcal{V}}
\def\cO{\mathcal{O}}
\def\cF{\mathcal{F}}
\def\cG{\mathcal{G}}
\def\cL{\mathcal{L}}

\def\cN{\mathcal{N}}
\def\cE{\mathcal{E}}
\def\cK{\mathcal{K}}
\def\cH{\mathcal{H}}
\def\eps{\varepsilon}
\def\ph{\varphi}

\def\wt{\widetilde}

\def\indic{\hbox{\raise-2pt \hbox{\indbf 1}}}

\let\==\equiv

\let\0=\noindent

\def\*{{\hfill\break\null\hfill\break}}

\def\tende#1{\,\vtop{\ialign{##\crcr\rightarrowfill\crcr
             \noalign{\kern-1pt\nointerlineskip}
             \hskip3.pt${\scriptstyle #1}$\hskip3.pt\crcr}}\,}
\def\otto{\,{\kern-1.truept\leftarrow\kern-5.truept\to\kern-1.truept}\,}

\def\tr{{\rm tr}}

\newtheorem{theorem}{Theorem}[section]  % use thm for %Theorems to keep numbering consistent
\newtheorem{cor}[theorem]{Corollary}
\newtheorem{prop}[theorem]{Proposition}
\newtheorem{lemma}[theorem]{Lemma}

\numberwithin{equation}{section}

\def\be{\begin{equation}}
\def\ee{\end{equation}}

\newcommand{\hc}{\mbox{h.c.}}

\let\a=\alpha \let\b=\beta    \let\g=\gamma     \let\d=\delta     
             \let\l=\lambda
                          
\let\s=\sigma \let\t=\tau         \let\ph=\varphi   
        
        \let\L=\Lambda

\begin{document}

\title{Optimal Rate for Bose-Einstein Condensation \\ in the Gross-Pitaevskii Regime}

\author{Chiara Boccato$^1$, Christian Brennecke$^2$, Serena Cenatiempo$^3$, Benjamin Schlein$^4$ \\
\\
Institute of Science and Technology Austria, \\ 
Am Campus 1, 3400 Klosterneuburg, Austria$^1$ \\
\\
Department of Mathematics, Harvard University, \\
One Oxford Street, Cambridge MA 02138, USA$^2$ \\
\\
Gran Sasso Science Institute, \\ Viale Francesco Crispi 7, 
67100 L'Aquila, Italy$^3$\\
\\
Institute of Mathematics, University of Zurich, \\
Winterthurerstrasse 190, 8057 Zurich, Switzerland$^4$}

\maketitle

\begin{abstract}
We consider systems of bosons trapped in a box, in the Gross-Pitaevskii regime. We show that low-energy states exhibit complete Bose-Einstein condensation with an optimal bound on the number of orthogonal  excitations. This extends recent results obtained in \cite{BBCS1}, removing the assumption of small interaction potential. 
\end{abstract}

\section{Introduction}

We consider systems of $N$ bosons trapped in the three-dimensional box $\Lambda = [0;1]^3$, with periodic boundary conditions (the three dimensional torus with volume one), interacting through a repulsive potential with scattering length of the order $N^{-1}$, a scaling limit known as the 
Gross-Pitaevskii regime. The Hamilton operator is given by 
\begin{equation}\label{eq:HN} H_N = \sum_{j=1}^N -\Delta_{x_j} + \sum_{i<j}^N N^2 V (N(x_i -x_j)) \end{equation}
and acts on a dense subspace of $L^2_s (\Lambda^N)$, the Hilbert space consisting of functions in $L^2 (\L^N)$ that are invariant with respect to permutations of the $N$ particles. We assume here 
$V \in L^3 (\bR^3)$ to have compact support and to be pointwise non-negative (i.e. $V(x) \geq 0$ for almost all $x \in \bR^3$). 

\medskip

%Multiplying lengths with $N$, we observe that the Hamilton operator (\ref{eq:HN}) is unitarily equivalent to %$N^2 \wt{H}_N$, with the unscaled self-adjoint operator  
%\[ \wt{H}_N = \sum_{j=1}^N -\Delta_{x_j} + \sum_{i<j}^N V (x_i -x_j) \]
%acting on $L^2_s ([0;N]^N)$. In other words, the Gross-Pitaevskii regime describes very 
%dilute Bose gases, with density of the order $N^{-2}$. 

Instead of trapping the Bose gas into the box $\Lambda = [0;1]^3$ and imposing periodic boundary conditions, one could also confine particles through an external potential $V_\text{ext} : \bR^3 \to \bR$, with $V_\text{ext} (x) \to \infty$, as $|x| \to \infty$. In this case, the Hamilton operator would have the form 
\begin{equation}\label{eq:HNV} H^\text{trap}_N = \sum_{j=1}^N \left[ -\Delta_{x_j} + V_\text{ext} (x_j) \right] + \sum_{i<j}^N N^2 V(N(x_i -x_j)) \end{equation}
and it would act on a dense subspace of $L^2_s (\bR^{3N})$. 

Lieb-Seiringer-Yngvason proved in \cite{LSY} that the ground state energy $E^\text{trap}_N$ of (\ref{eq:HNV}) is such that, as $N \to \infty$, 
\[%\label{eq:gsV} 
\frac{E^\text{trap}_N}{N} \to \min_{\ph \in L^2 (\bR^3) : \| \ph \|_2 = 1}   \cE_\text{GP} (\ph) \]
with the Gross-Pitaevskii energy functional 
\begin{equation}\label{eq:GPfun} \cE_\text{GP} (\ph) = \int \left[ |\nabla \ph|^2 + V_\text{ext} |\ph|^2 + 4 \pi \frak{a}_0 |\ph|^4 \right] dx \end{equation}
where $\frak{a}_0$ denotes the scattering length of the unscaled interaction potential $V$. 

In \cite{LS}, Lieb-Seiringer also proved that the normalized ground state vector $\psi^\text{trap}_N$ of (\ref{eq:HNV}) exhibits complete Bose-Einstein condensation in the minimizer $\ph_\text{GP}$ of (\ref{eq:GPfun}), meaning that its reduced one-particle density matrix $\gamma_N = \tr_{2,\dots , N} |\psi^\text{trap}_N \rangle \langle \psi^\text{trap}_N |$ (normalized so that $\tr \, \gamma_N =1$) satisfies 
\begin{equation}\label{eq:BEC0} \gamma_N \to |\ph_\text{GP} \rangle \langle \ph_\text{GP}|   \end{equation}
as $N \to \infty$ (convergence holds in the trace norm topology; since the limit is a rank-one projection, all reasonable notions of convergence are equivalent). Eq. (\ref{eq:BEC0}) asserts that, in the ground state of (\ref{eq:HNV}), all bosons, up to a fraction that vanishes in the limit $N \to \infty$, occupy the same one-particle state $\ph_\text{GP}$. In \cite{LS2}, Lieb-Seiringer extended Eq. (\ref{eq:BEC0}) to reduced density matrices associated with normalized sequences of approximate ground states, ie. states with expected energy per particle converging to the minimum of (\ref{eq:GPfun}) (under the constraint $\| \ph \| = 1$). 

A new proof of the results described above has been later obtained by Nam-Rougerie-Seiringer in \cite{NRS}, making use of the quantum de Finetti theorem, first proposed in the mean-field setting by Lewin-Nam-Rougerie \cite{LNR1,LNR2}. 

The results of \cite{LSY,LS,LS2,NRS} can be translated to the Hamilton operator (\ref{eq:HN}), defined on the torus, with no external potential. They imply, first of all, that the ground state energy $E_N$ of (\ref{eq:HN}) is such that 
\begin{equation}\label{eq:en-ti} \lim_{N \to \infty} \frac{E_N}{N} = 4\pi \frak{a}_0 \, . \end{equation}
Furthermore, they imply that for any sequence of approximate ground states, ie. for any sequence $\psi_N \in L^2_s (\L^N)$ with $\| \psi_N \| = 1$ and  
\begin{equation}\label{eq:appro-gs} \lim_{N \to \infty} \frac{1}{N} \langle \psi_N , H_N \psi_N \rangle = 4 \pi \frak{a}_0 \, , \end{equation}
the reduced density matrices $\gamma_N = \tr_{2, \dots , N} |\psi_N \rangle \langle \psi_N |$ are such that 
\begin{equation}\label{eq:BEC1} \lim_{N \to \infty} \tr \, \left| \gamma_N - |\ph_0 \rangle \langle \ph_0| \right| = 0 \end{equation}
where $\ph_0 \in L^2 (\Lambda)$ is the zero momentum mode defined by $\ph_0 (x) = 1$ for all $x \in \Lambda$. Since we will make use of this result in our analysis and since, strictly speaking, the translation invariant Hamiltonian (\ref{eq:HN}) is not treated in \cite{LS2,NRS}, we shortly discuss the proof of (\ref{eq:BEC1}) (in particular, how it follows from the analysis of \cite{NRS}) in Appendix \ref{app:NRS}.

\medskip

Under the additional assumption that the interaction potential $V$ is sufficiently small, in \cite{BBCS1} we recently improved (\ref{eq:en-ti}) and (\ref{eq:BEC1}), obtaining quantitative estimates showing, on the one hand, that $E_N - 4 \pi \frak{a}_0 N$ remains bounded, uniformly in $N$, and, on the other hand, that every sequence of approximate ground states $\psi_N$ of (\ref{eq:HN}) exhibit Bose-Einstein condensation, with number of excitations bounded by the excess energy $\langle \psi_N, H_N \psi_N \rangle - 4 \pi \frak{a}_0 N$. The goal of the present paper is to extend the results of \cite{BBCS1}, removing the assumption of small interaction. 
\begin{theorem}\label{thm:main}
Let $V \in L^3 (\bR^3)$ have compact support and be pointwise non-negative. Then there exists a constant $C > 0$ such that the ground state energy $E_N$ of (\ref{eq:HN}) satisfies 
\begin{equation}\label{eq:Enbd}
|E_N - 4 \pi \frak{a}_0 N | \leq C 
\end{equation}
Furthermore, consider a sequence $\psi_N \in L^2_s (\Lambda^N)$ with $\| \psi_N \| = 1$ and such that   
\[ \langle \psi_N , H_N \psi_N \rangle  \leq 4 \pi \frak{a}_0 N + K \]
for a $K > 0$. Then the reduced density matrix $\gamma_N = \tr_{2,\dots , N} | \psi_N \rangle \psi_N |$ associated with $\psi_N$ is such that 
\begin{equation}\label{eq:BEC} 1 - \langle \ph_0 , \gamma_N \ph_0 \rangle \leq  \frac{C(K+1)}{N} \end{equation}
for all $N \in \bN$ large enough. 
\end{theorem}

{\it Remark:} Eq. (\ref{eq:BEC}) gives a bound on the number of orthogonal excitations of the Bose-Einstein condensate, for low-energy states of the Hamilton operator (\ref{eq:HN}). It implies that 
\begin{equation}\label{eq:exc-bd} \begin{split} 
\langle \psi_N, d\Gamma (1 - |\ph_0 \rangle \langle \ph_0|) \psi_N \rangle &= N - \langle \psi_N, a^* (\ph_0) a(\ph_0) \psi_N \rangle \\ &= N \left[ 1 - \langle \ph_0, \gamma_N \ph_0 \rangle \right] \leq C (K+1) \end{split} \end{equation}
and thus that, for low-energy states $\psi_N$ with finite excess energy $K$, the number of excitations of the Bose-Einstein condensate remains bounded, uniformly in $N$. Notice that the bounds (\ref{eq:BEC}), (\ref{eq:exc-bd}) remain valid and non-trivial even if $K$ grows, as $N \to \infty$, as long as $K \ll N$; in particular, they imply complete BEC for all sequences of approximate ground states $\psi_N$ satisfying  (\ref{eq:appro-gs}). 

\medskip

To prove Theorem \ref{thm:main}, we are going to introduce, in Section \ref{sec:fock}, an excitation Hamiltonian $\cL_N$, factoring out the Bose-Einstein condensate. In Section~\ref{sec:bt}, we define generalized Bogoliubov transformations that are used in Section~\ref{sec:ren} to model correlations among particles and to define a renormalized excitation Hamiltonian $\cG_{N,\ell}$; important properties of $\cG_{N,\ell}$ are collected in Prop. \ref{prop:GNell} and in Prop.~\ref{prop:GNell-loc}. A second renormalization, this time through the exponential of an operator cubic in creation and annihilation operators, is performed in Section \ref{sec:cubic}, leading to a new twice renormalized Hamiltonian $\cR_{N,\ell}$; an important bound for $\cR_{N,\ell}$ is stated in Prop. \ref{prop:RNell}. In Section \ref{sec:main}, we use the results of Prop. \ref{prop:GNell},  Prop.~\ref{prop:GNell-loc} and of Prop. \ref{prop:RNell} to show Theorem \ref{thm:main}. Section \ref{sec:GN} and Section \ref{sec:RN} are devoted to the proof of Prop. \ref{prop:GNell} and, respectively, of Prop. \ref{prop:RNell}. 

\medskip

The main novelty, with respect to the analysis in \cite{BBCS1} is the need for the second renormalization, through the exponential $S = e^A$ of a cubic operator $A$. Under the additional assumption of small potential, the analysis of $\cG_{N,\ell}$ was enough in \cite{BBCS1} to show Bose-Einstein condensation in the form (\ref{eq:BEC}).  Here, this is not the case. The point is that conjugation with a generalized Bogoliubov transformation renormalizes the quadratic terms in the excitation Hamiltonian, but it leaves the cubic term substantially unchanged. For small potentials, the cubic term can be controlled (by Cauchy-Schwarz) through the quartic interaction and through the gap in the kinetic energy. Without assumptions on the size of the potential, on the other hand, we need to conjugate with $S$, to renormalize the cubic term. After conjugation with $S$, we can apply techniques developed by Lewin-Nam-Serfaty-Solovej in \cite{LNSS} (inspired by previous work of Lieb-Solovej in \cite{LSo}) based on localization of the number of excitations. On sectors with few excitations (the cutoff will be set at $M = c N$, for a sufficiently small constant $c > 0$), the renormalized cubic term is small and it can be controlled by the gap in the kinetic energy operator. On sectors with many excitations, on the other hand, we are going to bound the energy from below, using the estimate (\ref{eq:BEC1}), due to \cite{LS2,NRS} (since on these sectors we do not have condensation, the energy per particle must be strictly larger than $4 \pi \frak{a}_0$). 

\medskip

Theorem \ref{thm:main} is the first important step that we need in \cite{BBCS3} to establish the validity of Bogoliubov theory, as proposed in \cite{B}, for the low-energy excitation spectrum of (\ref{eq:HN}). 

\medskip

{\it Acknowledgements.} We would like to thank P. T. Nam and R. Seiringer for several useful discussions and for suggesting us to use the localization techniques from \cite{LNSS}. C. Boccato has received funding from the European Research Council (ERC) under the programme Horizon 2020 (grant agreement 694227). B. Schlein gratefully acknowledges support from the NCCR SwissMAP and from the Swiss National Foundation of Science through the SNF Grant ``Dynamical and energetic properties of Bose-Einstein condensates''.

\section{The Excitation Hamiltonian}
\label{sec:fock}

The bosonic Fock space over $L^2 (\Lambda)$ is defined as 
\[ \cF = \bigoplus_{n \geq 0} L^2_s (\Lambda^{n}) = \bigoplus_{n \geq 0} L^2 (\Lambda)^{\otimes_s n} \]
where $L^2_s (\Lambda^{n})$ is the subspace of $L^2 (\Lambda^n)$ consisting of wave functions that are symmetric w.r.t. permutations. The vacuum vector in $\cF$ will be indicated with $\Omega = \{ 1, 0, \dots \} \in \cF$. 

For $g \in L^2 (\Lambda)$, the creation operator $a^* (g)$ and the annihilation operator $a(g)$ are defined by  
\[ \begin{split} 
(a^* (g) \Psi)^{(n)} (x_1, \dots , x_n) &= \frac{1}{\sqrt{n}} \sum_{j=1}^n g (x_j) \Psi^{(n-1)} (x_1, \dots , x_{j-1}, x_{j+1} , \dots , x_n) 
\\
(a (g) \Psi)^{(n)} (x_1, \dots , x_n) &= \sqrt{n+1} \int_\Lambda  \bar{g} (x) \Psi^{(n+1)} (x,x_1, \dots , x_n) \, dx \end{split} \]
Observe that $a^* (g)$ is the adjoint of $a(g)$ and that the canonical commutation relations
\begin{equation*}%\label{eq:ccr} 
[a (g), a^* (h) ] = \langle g,h \rangle , \quad [ a(g), a(h)] = [a^* (g), a^* (h) ] = 0 \end{equation*}
hold true for all $g,h \in L^2 (\Lambda)$ ($\langle g,h \rangle$ is the inner product on $L^2 (\Lambda)$). 

It will be convenient for us to work in momentum space $\Lambda^* = 2\pi \bZ^3$. For $p \in \Lambda^*$, we consider the plane wave $\ph_p  (x) = e^{-ip\cdot x}$ in $L^2 (\Lambda)$. We define the operators 
\begin{equation*}%\label{eq:ap} 
a^*_p = a^* (\ph_p), \quad \text{and } \quad  a_p = a (\ph_p) \end{equation*} 
creating and, respectively, annihilating a particle with momentum $p$. 

To exploit the non-negativity of the interaction potential $V$, it will sometimes be useful to switch to position space. To this end, we introduce operator valued distributions $\check{a}_x, \check{a}_x^*$ such that  
\begin{equation*}%\label{eq:axf}
 a(f) = \int \bar{f} (x) \,  \check{a}_x \, dx , \quad a^* (f) = \int f(x) \, \check{a}_x^* \, dx  \end{equation*}

The number of particles operator, defined on a dense subspace of $\cF$ by $(\cN\Psi)^{(n)} = n \Psi^{(n)}$, can be expressed as 
\[ \cN = \sum_{p \in \Lambda^*} a_p^* a_p  = \int \check{a}^*_x \check{a}_x \, dx \,  \]
It is then easy to check that creation and annihilation operators are bounded with respect to the square root of $\cN$, i.e.   
\begin{equation*}%\label{eq:abd} 
\| a (f) \Psi \| \leq \| f \| \| \cN^{1/2} \Psi \|, \quad \| a^* (f) \Psi \| \leq \| f \| \| (\cN+1)^{1/2} \Psi \| 
\end{equation*}
for all $f \in L^2 (\Lambda)$. 

Recall that $\ph_0 (x) = 1$ for all $x \in \Lambda$ is the zero-momentum mode in $L^2 (\Lambda)$. We define $L^2_{\perp} (\Lambda)$ as the orthogonal complement in $L^2 (\Lambda)$ of the one dimensional space spanned by $\ph_0$. The Fock space over $L^2_\perp (\Lambda)$, generated by the creation operators $a_p^*$ with $p \in \Lambda^*_+ :=  2\pi \bZ^3 \backslash \{ 0 \}$, will be denoted by  
\[ \cF_{+} = \bigoplus_{n \geq 0} L^2_{\perp} (\Lambda)^{\otimes_s n} \,  \]
On $\cF_+$, the number of particles operator will be indicated by  
\[ \cN_+ = \sum_{p \in \Lambda^*_+} a_p^* a_p \]

For $N \in \bN$, we also define the truncated Fock space  
\[ \cF_{+}^{\leq N} = \bigoplus_{n=0}^N L^2_{\perp} (\Lambda)^{\otimes_s n} \,  \]
On this Hilbert space, we are going to describe the orthogonal excitations of the 
Bose-Einstein condensate. To this end, we are going to use a unitary map $U_N : L^2_s (\Lambda^N) \to \cF_+^{\leq N}$, first introduced in \cite{LNSS}, which removes the condensate. To define $U_N$, we notice that every $\psi_N \in L^2_s (\Lambda^N)$ can be uniquely decomposed as 
\[ \psi_N = \alpha_0 \ph_0^{\otimes N} + \alpha_1 \otimes_s \ph_0^{\otimes (N-1)} + \dots + \alpha_N \]
with $\alpha_j \in L^2_\perp (\Lambda)^{\otimes_s j}$ (the symmetric tensor product of $j$ copies of the orthogonal complement $L^2_\perp (\Lambda)$ of $\ph_0$) for all $j = 0, \dots , N$. Therefore, we can put $U_N \psi_N = \{ \alpha_0, \alpha_1, \dots , \alpha_N \} \in \cF_+^{\leq N}$. We can also define $U_N$ identifying $\psi_N$ with the Fock space vector $\{ 0, 0, \dots , \psi_N, 0, \dots \}$ and using creation and annihilation operators; we find 
\[ U_N \, \psi_N = \bigoplus_{n=0}^N  (1-|\ph_0 \rangle \langle \ph_0|)^{\otimes n} \frac{a(\ph_0)^{N-n}}{\sqrt{(N-n)!}} \, \psi_N \]
for all $\psi_N \in L^2_s (\Lambda^N)$. It is then easy to check that $U_N^* : \cF_{+}^{\leq N} \to L^2_s (\Lambda^N)$ is given by 
\[ U_N^* \, \{ \alpha^{(0)}, \dots , \alpha^{(N)} \} = \sum_{n=0}^N \frac{a^* (\ph_0)^{N-n}}{\sqrt{(N-n)!}} \, \alpha^{(n)} \]
and that $U_N^* U_N = 1$, ie. $U_N$ is unitary. 

Using $U_N$, we can define the excitation Hamiltonian $\cL_N := U_N H_N U_N^*$, acting on a dense subspace of $\cF_+^{\leq N}$. To compute the operator $\cL_N$, we first write the Hamiltonian (\ref{eq:HN}) in momentum space, in terms of creation and annihilation operators. We find 
\begin{equation}\label{eq:Hmom} H_N = \sum_{p \in \Lambda^*} p^2 a_p^* a_p + \frac{1}{2N} \sum_{p,q,r \in \Lambda^*} \widehat{V} (r/N) a_{p+r}^* a_q^* a_{p} a_{q+r} 
\end{equation}
where 
\[ \widehat{V} (k) = \int_{\bR^3} V (x) e^{-i k \cdot x} dx \] 
is the Fourier transform of $V$, defined for all $k \in \bR^3$ (in fact, (\ref{eq:HN}) is the restriction of (\ref{eq:Hmom}) to the $N$-particle sector of the Fock space $\cF$). We can now determine the excitation Hamiltonian $\cL_N$ using the following rules, describing the action of the unitary operator $U_N$ on products of a creation and an annihilation operator (products of the form $a_p^* a_q$ can be thought of as operators mapping $L^2_s (\Lambda^N)$ to itself). For any $p,q \in \Lambda^*_+ = 2\pi \bZ^3 \backslash \{ 0 \}$, we find (see \cite{LNSS}):
\begin{equation}\label{eq:U-rules}
\begin{split} 
U_N \, a^*_0 a_0 \, U_N^* &= N- \cN_+  \\
U_N \, a^*_p a_0 \, U_N^* &= a^*_p \sqrt{N-\cN_+ } \\
U_N \, a^*_0 a_p \, U_N^* &= \sqrt{N-\cN_+ } \, a_p \\
U_N \, a^*_p a_q \, U_N^* &= a^*_p a_q 
\end{split} \end{equation}
We conclude that 
\begin{equation}\label{eq:cLN} \cL_N =  \cL^{(0)}_{N} + \cL^{(2)}_{N} + \cL^{(3)}_{N} + \cL^{(4)}_{N} \end{equation}
with
\begin{equation}\label{eq:cLNj} \begin{split} 
\cL_{N}^{(0)} =\;& \frac{N-1}{2N}\widehat{V} (0) (N-\cN_+ ) + \frac{\widehat{V} (0)}{2N} \cN_+  (N-\cN_+ ) \\
\cL^{(2)}_{N} =\; &\sum_{p \in \Lambda^*_+} p^2 a_p^* a_p + \sum_{p \in \Lambda_+^*} \widehat{V} (p/N) \left[ b_p^* b_p - \frac{1}{N} a_p^* a_p \right] \\ &+ \frac{1}{2} \sum_{p \in \Lambda^*_+} \widehat{V} (p/N) \left[ b_p^* b_{-p}^* + b_p b_{-p} \right] \\
\cL^{(3)}_{N} =\; &\frac{1}{\sqrt{N}} \sum_{p,q \in \Lambda_+^* : p+q \not = 0} \widehat{V} (p/N) \left[ b^*_{p+q} a^*_{-p} a_q  + a_q^* a_{-p} b_{p+q} \right] \\
\cL^{(4)}_{N} =\; & \frac{1}{2N} \sum_{\substack{p,q \in \Lambda_+^*, r \in \Lambda^*: \\ r \not = -p,-q}} \widehat{V} (r/N) a^*_{p+r} a^*_q a_p a_{q+r} 
\end{split} \end{equation}
where we introduced generalized creation and annihilation operators  
\begin{equation}\label{eq:bp-de} 
b^*_p = a^*_p \, \sqrt{\frac{N-\cN_+}{N}} , \qquad \text{and } \quad  b_p = \sqrt{\frac{N-\cN_+}{N}} \, a_p 
\end{equation}
for all $p \in \Lambda^*_+$. Observe that, by (\ref{eq:U-rules}), 
\[ U_N^* b_p^* U_N = a^*_p  \frac{a_0}{\sqrt{N}}, \qquad U_N^* b_p U_N = \frac{a_0^*}{\sqrt{N}} a_p \]
In other words, $b_p^*$ creates a particle with momentum $p \in \Lambda^*_+$ but, at the same time, it annihilates a particle from the condensate; it creates an excitation, preserving the total number of particles in the system. On states exhibiting complete Bose-Einstein condensation in the zero-momentum mode $\ph_0$, we have $a_0 , a_0^* \simeq \sqrt{N}$ and we can therefore expect that $b_p^* \simeq a_p^*$ and that $b_p \simeq a_p$. Modified creation and annihilation operators satisfy the commutation relations 
\begin{equation}\label{eq:comm-bp} \begin{split} [ b_p, b_q^* ] &= \left( 1 - \frac{\cN_+}{N} \right) \delta_{p,q} - \frac{1}{N} a_q^* a_p 
\\ [ b_p, b_q ] &= [b_p^* , b_q^*] = 0 
\end{split} \end{equation}
Furthermore, we find 
\begin{equation}\label{eq:comm2} [ b_p, a_q^* a_r ] = \delta_{pq} b_r, \qquad  [b_p^*, a_q^* a_r] = - \delta_{pr} b_q^* \end{equation}
for all $p,q,r \in \Lambda_+^*$; this implies in particular that $[b_p , \cN_+] = b_p$, $[b_p^*, \cN_+] = - b_p^*$. It is also useful to notice that the operators $b^*_p, b_p$, like the standard creation and annihilation operators $a_p^*, a_p$, can be bounded by the square root of the number of particles operators; we find
\begin{equation*}
\begin{split} 
\| b_p \xi \| &\leq \Big\| \cN_+^{1/2} \Big( \frac{N+1-\cN_+}{N} \Big)^{1/2} \xi \Big\| \leq \| \cN_+^{1/2} \xi \|  \\ 
\| b^*_p \xi \| &\leq \Big\| (\cN_+ +1)^{1/2} \Big( \frac{N-\cN_+ }{N} \Big)^{1/2} \xi \Big\| \leq  \| (\cN_+ + 1)^{1/2} \xi \| 
\end{split} \end{equation*}
for all $\xi \in \cF^{\leq N}_+$. Since $\cN_+  \leq N$ on $\cF_+^{\leq N}$, the operators $b_p^* , b_p$ are bounded, with $\| b_p \|, \| b^*_p \| \leq (N+1)^{1/2}$. 

We can also define modified operator valued distributions 
\[ \check{b}_x = \sqrt{\frac{N-\cN_+}{N}} \, \check{a}_x, \qquad \text{and } \quad  \check{b}^*_x = \check{a}^*_x \, \sqrt{\frac{N-\cN_+}{N}} \]
in position space, for $x \in \Lambda$. The commutation relations (\ref{eq:comm-bp}) take the form 
\begin{equation*}
\begin{split}  [ \check{b}_x, \check{b}_y^* ] &= \left( 1 - \frac{\cN_+}{N} \right) \delta (x-y) - \frac{1}{N} \check{a}_y^* \check{a}_x \\ 
[ \check{b}_x, \check{b}_y ] &= [ \check{b}_x^* , \check{b}_y^*] = 0 
\end{split} \end{equation*}
Moreover, (\ref{eq:comm2}) translates to 
\begin{equation*}
\begin{split}
[\check{b}_x, \check{a}_y^* \check{a}_z] &=\delta (x-y)\check{b}_z, \qquad 
[\check{b}_x^*, \check{a}_y^* \check{a}_z] = -\delta (x-z) \check{b}_y^*
\end{split} \end{equation*}
which also implies that $[ \check{b}_x, \cN_+ ] = \check{b}_x$, $[ \check{b}_x^* , \cN_+ ] = - \check{b}_x^*$. 

\section{Generalizated Bogoliubov Transformations}
\label{sec:bt}

Conjugation with $U_N$ extracts, from the original quartic interaction in (\ref{eq:Hmom}), some constant and some quadratic contributions, collected in $\cL^{(0)}_N$ and $\cL^{(2)}_N$ in (\ref{eq:cLNj}). In the Gross-Pitevskii regime, however, this is not enough; there are still large contributions to the energy hidden among cubic and quartic terms in $\cL^{(3)}_N$ and $\cL^{(4)}_N$. 

To extract the missing energy, we have to take into account the correlation structure. Since $U_N$ only removes products of the zero-energy mode $\ph_0$, correlations among particles, which play a crucial role in the Gross-Pitaevskii regime and carry an energy of order $N$, remain in the excitation vector $U_N \psi_N$. To factor out correlations, it is natural to conjugate $\cL_N$ with a Bogoliubov transformation. In fact, to make sure that the truncated Fock space $\cF_+^{\leq N}$ remains invariant, we will have to use generalized Bogoliubov transformations. Their definition and their main properties will be discussed in this section. 

For $\eta \in \ell^2 (\Lambda^*_+)$ with $\eta_{-p} = \eta_{p}$ for all $p \in \Lambda^*_+$, we define 
\begin{equation}\label{eq:defB} 
B(\eta) = \frac{1}{2} \sum_{p\in \L^*_+}  \left( \eta_p b_p^* b_{-p}^* - \bar{\eta}_p b_p b_{-p} \right) \, \end{equation}
%where $P_H=\{p \in \L^*_+ : |p| \leq \ell^{-\a}\}$. Then, we consider 
and we consider
\begin{equation}\label{eq:eBeta} 
e^{B(\eta)} = \exp \left[ \frac{1}{2} \sum_{p \in \Lambda^*_+}   \left( \eta_p b_p^* b_{-p}^* - \bar{\eta}_p  b_p b_{-p} \right) \right] 
\end{equation}
We refer to unitary operators of the form (\ref{eq:eBeta}) as generalized Bogoliubov transformations, in analogy with the standard Bogoliubov transformations 
\begin{equation}\label{eq:wteBeta} e^{\wt{B} (\eta)} = \exp \left[  \frac{1}{2} \sum_{p\in \Lambda^*_+}  \left( \eta_p a_p^* a_{-p}^* - \bar{\eta}_p a_p a_{-p} \right) \right] \end{equation}
defined by means of the standard creation and annihilation operators. In this paper, we will work with (\ref{eq:eBeta}), rather than (\ref{eq:wteBeta}), because the generalized Bogoliubov transformations, in contrast with the standard transformations, leave the truncated Fock space $\cF_+^{\leq N}$ invariant. 
The price we will have to pay is the fact that, while the action of standard Bogoliubov transformation 
on creation and annihilation operators is explicitly given by 
\begin{equation}\label{eq:act-Bog} e^{-\wt{B} (\eta)} a_p e^{\wt{B} (\eta)}  = \cosh (\eta_p) a_p + \sinh (\eta_p) a_{-p}^* \,   \end{equation}
there is no such formula describing the action of generalized Bogoliubov transformations. 

A first important tool to control the action of generalized Bogoliubov transformations is the following lemma, whose proof can be found in \cite[Lemma 3.1]{BS} (a similar result has been previously established in \cite{Sei}).
\begin{lemma}\label{lm:Ngrow}
For every $n \in \bN$ there exists a constant $C > 0$ such that, on $\cF_+^{\leq N}$, 
\begin{equation}\label{eq:bd-Beta} e^{-B(\eta)} (\cN_+ +1)^{n} e^{B(\eta)} \leq C e^{C \| \eta \|} (\cN_+ +1)^{n}  
\end{equation}
for all $\eta \in \ell^2 (\L^*)$.
\end{lemma}

Bounds of the form (\ref{eq:bd-Beta}) on the change of the number of particles operator are not enough for our purposes; we will need more precise information about the action of unitary operators of the form $e^{B(\eta)}$. To this end, we expand, for any $p \in \Lambda^*_+$,  
\[\begin{split} e^{-B(\eta)} \, b_p \, e^{B(\eta)} &= b_p + \int_0^1 ds \, \frac{d}{ds}  e^{-sB(\eta)} b_p e^{sB(\eta)} \\ &= b_p - \int_0^1 ds \, e^{-sB(\eta)} [B(\eta), b_p] e^{s B(\eta)} \\ &= b_p - [B(\eta),b_p] + \int_0^1 ds_1 \int_0^{s_1} ds_2 \, e^{-s_2 B(\eta)} [B(\eta), [B(\eta),b_p]] e^{s_2 B(\eta)} \end{split} \]
Iterating $m$ times, we find 
\begin{equation}\label{eq:BCH} \begin{split} 
e^{-B(\eta)} b_p e^{B(\eta)} = &\sum_{n=1}^{m-1} (-1)^n \frac{\text{ad}^{(n)}_{B(\eta)} (b_p)}{n!} \\ &+ \int_0^{1} ds_1 \int_0^{s_1} ds_2 \dots \int_0^{s_{m-1}} ds_m \, e^{-s_m B(\eta)} \text{ad}^{(m)}_{B(\eta)} (b_p) e^{s_m B(\eta)} \end{split} \end{equation}
where we recursively defined \[ \text{ad}_{B(\eta)}^{(0)} (A) = A \quad \text{and } \quad \text{ad}^{(n)}_{B(\eta)} (A) = [B(\eta), \text{ad}^{(n-1)}_{B(\eta)} (A) ]  \]
We are going to expand the nested commutators $\text{ad}_{B(\eta)}^{(n)} (b_p)$ and   
$\text{ad}_{B(\eta)}^{(n)} (b^*_p)$. To this end, we need to introduce some additional notation. 
We follow here \cite{BS,BBCS1,BBCS2}. For $f_1, \dots , f_n \in \ell_2 (\Lambda^*_+)$, $\sharp = (\sharp_1, \dots , \sharp_n), \flat = (\flat_0, \dots , \flat_{n-1}) \in \{ \cdot, * \}^n$, we set 
\begin{equation}\label{eq:Pi2}
\begin{split}  
\Pi^{(2)}_{\sharp, \flat} &(f_1, \dots , f_n) \\ &= \sum_{p_1, \dots , p_n \in \Lambda^*}  b^{\flat_0}_{\alpha_0 p_1} a_{\beta_1 p_1}^{\sharp_1} a_{\alpha_1 p_2}^{\flat_1} a_{\beta_2 p_2}^{\sharp_2} a_{\alpha_2 p_3}^{\flat_2} \dots  a_{\beta_{n-1} p_{n-1}}^{\sharp_{n-1}} a_{\alpha_{n-1} p_n}^{\flat_{n-1}} b^{\sharp_n}_{\beta_n p_n} \, \prod_{\ell=1}^n f_\ell (p_\ell)  \end{split} \end{equation}
where, for $\ell=0,1, \dots , n$, we define $\alpha_\ell = 1$ if $\flat_\ell = *$, $\alpha_\ell =    -1$ if $\flat_\ell = \cdot$, $\beta_\ell = 1$ if $\sharp_\ell = \cdot$ and $\beta_\ell = -1$ if $\sharp_\ell = *$. In (\ref{eq:Pi2}), we require that, for every $j=1,\dots, n-1$, we have either $\sharp_j = \cdot$ and $\flat_j = *$ or $\sharp_j = *$ and $\flat_j = \cdot$ (so that the product $a_{\beta_\ell p_\ell}^{\sharp_\ell} a_{\alpha_\ell p_{\ell+1}}^{\flat_\ell}$ always preserves {} the number of particles, for all $\ell =1, \dots , n-1$). With this assumption, we find that the operator $\Pi^{(2)}_{\sharp,\flat} (f_1, \dots , f_n)$ maps $\cF^{\leq N}_+$ into itself. If, for some $\ell=1, \dots , n$, $\flat_{\ell-1} = \cdot$ and $\sharp_\ell = *$ (i.e. if the product $a_{\alpha_{\ell-1} p_\ell}^{\flat_{\ell-1}} a_{\beta_\ell p_\ell}^{\sharp_\ell}$ for $\ell=2,\dots , n$, or the product $b_{\alpha_0 p_1}^{\flat_0} a_{\beta_1 p_1}^{\sharp_1}$ for $\ell=1$, is not normally ordered) we require additionally that $f_\ell  \in \ell^1 (\Lambda^*_+)$. In position space, the same operator can be written as 
\begin{equation}\label{eq:Pi2-pos} \Pi^{(2)}_{\sharp, \flat} (f_1, \dots , f_n) = \int   \check{b}^{\flat_0}_{x_1} \check{a}_{y_1}^{\sharp_1} \check{a}_{x_2}^{\flat_1} \check{a}_{y_2}^{\sharp_2} \check{a}_{x_3}^{\flat_2} \dots  \check{a}_{y_{n-1}}^{\sharp_{n-1}} \check{a}_{x_n}^{\flat_{n-1}} \check{b}^{\sharp_n}_{y_n} \, \prod_{\ell=1}^n \check{f}_\ell (x_\ell - y_\ell) \, dx_\ell dy_\ell \end{equation}
An operator of the form (\ref{eq:Pi2}), (\ref{eq:Pi2-pos}) with all the properties listed above, will be called a $\Pi^{(2)}$-operator of order $n$.

For $g, f_1, \dots , f_n \in \ell_2 (\Lambda^*_+)$, $\sharp = (\sharp_1, \dots , \sharp_n)\in \{ \cdot, * \}^n$, $\flat = (\flat_0, \dots , \flat_{n}) \in \{ \cdot, * \}^{n+1}$, we also define the operator 
\begin{equation}\label{eq:Pi1}
\begin{split} \Pi^{(1)}_{\sharp,\flat} &(f_1, \dots , f_n;g) \\ &= \sum_{p_1, \dots , p_n \in \Lambda^*}  b^{\flat_0}_{\alpha_0, p_1} a_{\beta_1 p_1}^{\sharp_1} a_{\alpha_1 p_2}^{\flat_1} a_{\beta_2 p_2}^{\sharp_2} a_{\alpha_2 p_3}^{\flat_2} \dots a_{\beta_{n-1} p_{n-1}}^{\sharp_{n-1}} a_{\alpha_{n-1} p_n}^{\flat_{n-1}} a^{\sharp_n}_{\beta_n p_n} a^{\flat n} (g) \, \prod_{\ell=1}^n f_\ell (p_\ell) \end{split} \end{equation}
where $\alpha_\ell$ and $\beta_\ell$ are defined as above. Also here, we impose the condition that, for all $\ell = 1, \dots , n$, either $\sharp_\ell = \cdot$ and $\flat_\ell = *$ or $\sharp_\ell = *$ and $\flat_\ell = \cdot$. This implies that $\Pi^{(1)}_{\sharp,\flat} (f_1, \dots , f_n;g)$ maps $\cF^{\leq N}_+$ back into $\cF_+^{\leq N}$. Additionally, we assume that $f_\ell \in \ell^1 (\Lambda^*_+)$ if $\flat_{\ell-1} = \cdot$ and $\sharp_\ell = *$ for some $\ell = 1,\dots , n$ (i.e. if the pair $a_{\alpha_{\ell-1} p_\ell}^{\flat_{\ell-1}} a^{\sharp_\ell}_{\beta_\ell p_\ell}$ is not normally ordered). In position space, the same operator can be written as
\begin{equation}\label{eq:Pi1-pos} \Pi^{(1)}_{\sharp,\flat} (f_1, \dots ,f_n;g) = \int \check{b}^{\flat_0}_{x_1} \check{a}_{y_1}^{\sharp_1} \check{a}_{x_2}^{\flat_1} \check{a}_{y_2}^{\sharp_2} \check{a}_{x_3}^{\flat_2} \dots  \check{a}_{y_{n-1}}^{\sharp_{n-1}} \check{a}_{x_n}^{\flat_{n-1}} \check{a}^{\sharp_n}_{y_n} \check{a}^{\flat n} (g) \, \prod_{\ell=1}^n \check{f}_\ell (x_\ell - y_\ell) \, dx_\ell dy_\ell \end{equation}
An operator of the form (\ref{eq:Pi1}), (\ref{eq:Pi1-pos}) will be called a $\Pi^{(1)}$-operator of order $n$. Operators of the form $b(f)$, $b^* (f)$, for a $f \in \ell^2 (\Lambda^*_+)$, will be called $\Pi^{(1)}$-operators of order zero. 

The next lemma gives a detailed analysis of the nested commutators $\text{ad}^{(n)}_{B(\eta)} (b_p)$ and $\text{ad}^{(n)}_{B(\eta)} (b^*_p)$ for $n \in \bN$; the proof can be found in \cite[Lemma 2.5]{BBCS1}(it is a translation to momentum space of \cite[Lemma 3.2]{BS}). 
\begin{lemma}\label{lm:indu}
Let $\eta \in \ell^2 (\Lambda^*_+)$ be such that $\eta_p = \eta_{-p}$ for all $p \in \ell^2 (\Lambda^*)$. To simplify the notation,  assume also $\eta$ to be real-valued (as it will be in applications). Let $B(\eta)$ be defined as in (\ref{eq:defB}), $n \in \bN$ and $p \in \Lambda^*$. Then the nested commutator $\text{ad}^{(n)}_{B(\eta)} (b_p)$ can be written as the sum of exactly $2^n n!$ terms, with the following properties. 
\begin{itemize}
\item[i)] Possibly up to a sign, each term has the form
\begin{equation}\label{eq:Lambdas} \Lambda_1 \Lambda_2 \dots \Lambda_i \, N^{-k} \Pi^{(1)}_{\sharp,\flat} (\eta^{j_1}, \dots , \eta^{j_k} ; \eta^{s}_p \ph_{\alpha p}) 
\end{equation}
for some $i,k,s \in \bN$, $j_1, \dots ,j_k \in \bN \backslash \{ 0 \}$, $\sharp \in \{ \cdot, * \}^k$, $ \flat \in \{ \cdot, * \}^{k+1}$ and $\alpha \in \{ \pm 1 \}$ chosen so that $\alpha = 1$ if $\flat_k = \cdot$ and $\alpha = -1$ if $\flat_k = *$ (recall here that $\ph_p (x) = e^{-ip \cdot x}$). In (\ref{eq:Lambdas}), each operator $\Lambda_w : \cF^{\leq N} \to \cF^{\leq N}$, $w=1, \dots , i$, is either a factor $(N-\cN_+ )/N$, a factor $(N-(\cN_+ -1))/N$ or an operator of the form
\begin{equation}\label{eq:Pi2-ind} N^{-h} \Pi^{(2)}_{\sharp',\flat'} (\eta^{z_1}, \eta^{z_2},\dots , \eta^{z_h}) \end{equation}
for some $h, z_1, \dots , z_h \in \bN \backslash \{ 0 \}$, $\sharp,\flat  \in \{ \cdot , *\}^h$. 
\item[ii)] If a term of the form (\ref{eq:Lambdas}) contains $m \in \bN$ factors $(N-\cN_+ )/N$ or $(N-(\cN_+ -1))/N$ and $j \in \bN$ factors of the form (\ref{eq:Pi2-ind}) with $\Pi^{(2)}$-operators of order $h_1, \dots , h_j \in \bN \backslash \{ 0 \}$, then 
we have
\begin{equation*}%\label{eq:totalb}
 m + (h_1 + 1)+ \dots + (h_j+1) + (k+1) = n+1 \end{equation*}
\item[iii)] If a term of the form (\ref{eq:Lambdas}) contains (considering all $\Lambda$-operators and the $\Pi^{(1)}$-operator) the arguments $\eta^{i_1}, \dots , \eta^{i_m}$ and the factor $\eta^{s}_p$ for some $m, s \in \bN$, and $i_1, \dots , i_m \in \bN \backslash \{ 0 \}$, then \[ i_1 + \dots + i_m + s = n .\]
\item[iv)] There is exactly one term having of the form (\ref{eq:Lambdas}) with $k=0$ and such that all $\Lambda$-operators are factors of $(N-\cN_+ )/N$ or of $(N+1-\cN_+ )/N$. It is given by 
\begin{equation*}\label{eq:iv1} 
\left(\frac{N-\cN_+ }{N} \right)^{n/2} \left(\frac{N+1-\cN_+ }{N} \right)^{n/2} \eta^{n}_p b_p 
\end{equation*}
if $n$ is even, and by 
\begin{equation*} \label{eq:iv2} 
- \left(\frac{N-\cN_+ }{N} \right)^{(n+1)/2} \left(\frac{N+1-\cN_+ }{N} \right)^{(n-1)/2} \eta^{n}_p b^*_{-p}  \end{equation*}
if $n$ is odd.
\item[v)] If the $\Pi^{(1)}$-operator in (\ref{eq:Lambdas}) is of order $k \in \bN \backslash \{ 0 \}$, it has either the form  
\[ \sum_{p_1, \dots , p_k}  b^{\flat_0}_{\alpha_0 p_1} \prod_{i=1}^{k-1} a^{\sharp_i}_{\beta_i p_{i}} a^{\flat_i}_{\alpha_i p_{i+1}}  a^*_{-p_k} \eta^{2r}_p  a_p \prod_{i=1}^k \eta^{j_i}_{p_i}  \]
or the form 
\[\sum_{p_1, \dots , p_k} b^{\flat_0}_{\alpha_0 p_1} \prod_{i=1}^{k-1} a^{\sharp_i}_{\beta_i p_{i}} a^{\flat_i}_{\alpha_i p_{i+1}}  a_{p_k} \eta^{2r+1}_p a^*_p \prod_{i=1}^k \eta^{j_i}_{p_i}  \]
for some $r \in \bN$, $j_1, \dots , j_k \in \bN \backslash \{ 0 \}$. If it is of order $k=0$, then it is either given by $\eta^{2r}_p b_p$ or by $\eta^{2r+1}_p b_{-p}^*$, for some $r \in \bN$. 
\item[vi)] For every non-normally ordered term of the form 
\[ \begin{split} &\sum_{q \in \Lambda^*} \eta^{i}_q a_q a_q^* , \quad \sum_{q \in \Lambda^*} \, \eta^{i}_q b_q a_q^* \\  &\sum_{q \in \Lambda^*} \, \eta^{i}_q a_q b_q^*, \quad \text{or } \quad \sum_{q \in \Lambda^*} \, \eta^{i}_q b_q b_q^*  \end{split} \]
appearing either in the $\Lambda$-operators or in the $\Pi^{(1)}$-operator in (\ref{eq:Lambdas}), we have $i \geq 2$.
\end{itemize}
\end{lemma}
With Lemma \ref{lm:indu}, it follows from (\ref{eq:BCH}) that, if $\| \eta \|$ is sufficiently small, 
\begin{equation}\label{eq:conv-serie}
\begin{split} e^{-B(\eta)} b_p e^{B (\eta)} &= \sum_{n=0}^\infty \frac{(-1)^n}{n!} \text{ad}_{B(\eta)}^{(n)} (b_p) \\
e^{-B(\eta)} b^*_p e^{B (\eta)} &= \sum_{n=0}^\infty \frac{(-1)^n}{n!} \text{ad}_{B(\eta)}^{(n)} (b^*_p) \end{split} \end{equation}
where the series converge absolutely (the proof is a translation to momentum space of \cite[Lemma 3.3]{BS}).

While Lemma \ref{lm:indu} gives a complete characterization of terms appearing in the expansions (\ref{eq:conv-serie}), to localize the number of particles as we do in Prop. \ref{prop:GNell-loc}, we will need to consider double commutators of $\text{ad}_{-B(\eta)}^{(n)}(b_p)$ with a smooth function $f (\cN_+/M)$ of the number of particles operator $\cN_+$. varying on the scale $M\in\mathbb{N}\backslash\left\{0\right\}$. To this end, we will apply the following corollary, which is a simple consequence of Lemma \ref{lm:indu}.
\begin{cor}\label{lm:ff}
Let $f: \bR \to \bR$ be a real, smooth and bounded function. For $M \in \mathbb{N}\backslash\left\{0\right\}$, let $f_M = f (\cN_+/M)$.  Then, for any $n \in \bN$, $p \in \Lambda^*_+$, the double commutator $[f_M , [f_M, \text{ad}_{-B(\eta)}^{(n)}(b_p)]]$ can be written as the sum of $2^n n!$ (possibly vanishing) terms, having the form 
\[ F_{M,n} (\cN_+ ) \, \Lambda_1 \Lambda_2 \dots \Lambda_i N^{-k} \Pi^{(1)}_{\sharp, \flat} (\eta^{j_1}, \dots, \eta^{j_k} ; \eta_p^s \ph_{\alpha p})  \] 
for some $i,k,s \in \bN$, $j_1, \dots ,j_k \in \bN \backslash \{ 0 \}$, $\sharp \in \{ \cdot, * \}^k$, $ \flat \in \{ \cdot, * \}^{k+1}$ and $\alpha \in \{ \pm 1 \}$ chosen so that $\alpha = 1$ if $\flat_k = \cdot$ and $\alpha = -1$ if $\flat_k = *$, where the operators $\Lambda_1, \dots , \Lambda_i$ and $\Pi^{(1)}_{\sharp, \flat} (\eta^{j_1}, \dots, \eta^{j_k} ; \eta_p^s \ph_{\alpha p})$ satisfy all properties listed in the points i)-vi) in Lemma \ref{lm:indu} and where $F_{M,n}$ is a bounded function such that 
\begin{equation}\label{eq:FMn-bd} \| F_{M,n} (\cN_+ ) \| \leq \frac{C n^2}{M^2}  \|f' \|_\infty^2 \end{equation}
for a universal constant $C > 0$ (different terms will have different functions $F_{M,n}$, but they will all satisfy (\ref{eq:FMn-bd}) with the same constant $C > 0$).    
\end{cor}
\begin{proof}
It follows from Lemma \ref{lm:indu} that, for any $n\in\mathbb{N}$, $\text{ad}_{-B(\eta)}^{(n)}(b_p)$ can be written as the sum of $2^n n!$ terms of the form (up to a sign) 
 \begin{equation}\label{eq:Lambdas1} \Lambda_1 \Lambda_2 \dots \Lambda_i \, N^{-k} \Pi^{(1)}_{\sharp,\flat} (\eta^{j_1}, \dots , \eta^{j_k} ; \eta^{s}_p \ph_{\alpha p}) 
\end{equation}
for some $i,k,s \in \bN$, $j_1, \dots ,j_k \in \bN \backslash \{ 0 \}$, $\sharp \in \{ \cdot, * \}^k$, $ \flat \in \{ \cdot, * \}^{k+1}$ and $\alpha \in \{ \pm 1 \}$ chosen so that $\alpha = 1$ if $\flat_k = \cdot$ and $\alpha = -1$ if $\flat_k = *$. In (\ref{eq:Lambdas1}), each operator $\Lambda_w : \cF^{\leq N} \to \cF^{\leq N}$, $w=1, \dots , i$, is either a factor $(N-\cN_+ )/N$, a factor $(N-(\cN_+ -1))/N$ or an operator of the form
\begin{equation}\label{eq:Pi2-ind2} N^{-h} \Pi^{(2)}_{\sharp',\flat'} (\eta^{z_1}, \eta^{z_2},\dots , \eta^{z_h}) \end{equation}
for some $h, z_1, \dots , z_h \in \bN \backslash \{ 0 \}$, $\sharp,\flat  \in \{ \cdot , *\}^h$. The commutator of (\ref{eq:Lambdas1}) with $f_M$ is therefore given by 
\begin{equation*}%\label{eq:commLambdas1}
 \begin{split}
 [f_M,&\Lambda_1 \Lambda_2 \dots \Lambda_i \, N^{-k} \Pi^{(1)}_{\sharp,\flat} (\eta^{j_1}, \dots , \eta^{j_k} ; \eta^{s}_p \ph_{\alpha p}) ]\\
 &=\sum_{u=1}^i\Big(\prod_{t=1}^{u-1}\L_t\Big)[f_M,\L_u]\Big(\prod_{t=u+1}^{i}\L_t\Big)N^{-k} \Pi^{(1)}_{\sharp,\flat} (\eta^{j_1}, \dots , \eta^{j_k} ; \eta^{s}_p \ph_{\alpha p}) \\
 &\quad+\Lambda_1 \Lambda_2 \dots \Lambda_i \, N^{-k} [f_M,\Pi^{(1)}_{\sharp,\flat} (\eta^{j_1}, \dots , \eta^{j_k} ; \eta^{s}_p \ph_{\alpha p})]
 \end{split}
\end{equation*}
Recalling \eqref{eq:Pi2} and \eqref{eq:Pi1} and using the identities $b_p \cN_+=(\cN_++1) b_p$, $b^*_p\cN_+=(\cN_+-1)b^*_p$, we obtain that 
 \begin{equation*}%\label{eq:commL}
 \begin{split}
 \big[f_M,\L_u \big]=\Big[ f\Big(\frac{\cN_+}{M}\Big)-f\Big(\frac{\cN_++e_u}{M}\Big)\Big]\L_u\\
%   \big[f_M,\Pi^{(1)}_{\sharp,\flat}\big]=\Big(f\Big(\frac{\cN_+}{M}\Big)-f\Big(\frac{\cN_++e_s}{M}\Big)\Big)\Pi^{(1)}_{\sharp,\flat}\\
 \end{split}
\end{equation*}
with $e_u=0$ if $\L_u$ is either $(N-\cN_+)/N$ or $(N-(\cN_+ -1))/N$, while $e_u$ takes values in the set $\{-2,0,2\}$ if $\L_u$ is of the form \eqref{eq:Pi2-ind2} ($\Pi^{(2)}_{\sharp,\flat}$-operators can either create or annihilate two excitations, or it can leave the number of excitations invariant). Moreover
 \begin{equation*}%\label{eq:commPi1}
 \begin{split}
 \big[f_M,\Pi^{(1)}_{\sharp,\flat} (\eta^{j_1}, \dots , \eta^{j_k} ; \eta^{s}_p \ph_{\alpha p}) 
 \big]=\Big[ f\Big(\frac{\cN_+}{M}\Big)-f\Big(\frac{\cN_+\pm1}{M}\Big)\Big] \Pi^{(1)}_{\sharp,\flat} (\eta^{j_1}, \dots , \eta^{j_k} ; \eta^{s}_p \ph_{\alpha p}) \end{split}
\end{equation*}
because $\Pi^{(1)}_{\sharp,\flat}$ can create or annihilate only one excitation.
Therefore
 \begin{equation}\nonumber
 \begin{split}
 &[f_M,\Lambda_1 \Lambda_2 \dots \Lambda_i \, N^{-k} \Pi^{(1)}_{\sharp,\flat} (\eta^{j_1}, \dots , \eta^{j_k} ; \eta^{s}_p \ph_{\alpha p}) ]\\
 &=\sum_{u=1}^i\Big(\prod_{t=1}^{u-1}\L_t\Big)\Big[ f\Big(\frac{\cN_+}{M}\Big)-f\Big(\frac{\cN_++ e_u}{M}\Big)\Big] \L_u \Big(\prod_{r=u+1}^{i}\L_t\Big)N^{-k} \Pi^{(1)}_{\sharp,\flat} (\eta^{j_1}, \dots , \eta^{j_k} ; \eta^{s}_p \ph_{\alpha p}) \\
 &\quad+\Lambda_1 \Lambda_2 \dots \Lambda_i \, N^{-k} \Big[ f\Big(\frac{\cN_+}{M}\Big)-f\Big(\frac{\cN_+\pm1}{M}\Big)\Big] \Pi^{(1)}_{\sharp,\flat} (\eta^{j_1}, \dots , \eta^{j_k} ; \eta^{s}_p \ph_{\alpha p})
 \end{split}
\end{equation}
Hence, we have 
 \begin{equation}\nonumber
 \begin{split}
 &[f_M,\Lambda_1 \Lambda_2 \dots \Lambda_i \, N^{-k} \Pi^{(1)}_{\sharp,\flat} (\eta^{j_1}, \dots , \eta^{j_k} ; \eta^{s}_p \ph_{\alpha p}) ]\\
 &=\Big\{ \sum_{u=1}^i\Big[ f\Big(\frac{\cN_++ n_{u-1}}{M}\Big)-f\Big(\frac{\cN_++ e_u+n_{u-1}}{M}\Big)\Big] \\ &\hspace{.5cm} + \Big[ f\Big(\frac{\cN_++ n_i}{M}\Big)-f\Big(\frac{\cN_+\pm1+ n_i}{M}\Big)\Big] \Big\} \,  \Lambda_1 \Lambda_2 \dots \Lambda_i \, N^{-k} \Pi^{(1)}_{\sharp,\flat} (\eta^{j_1}, \dots , \eta^{j_k} ; \eta^{s}_p \ph_{\alpha p})
 \end{split}
\end{equation}
where $n_u=\sum_{t=1}^u e_t$. By the mean value theorem, we can find functions $\theta_1:\mathbb{N}\to(0,\pm1)$, $\theta_u:\mathbb{N}\to(0,e_u)$ such that
 \begin{equation}\nonumber
 \begin{split}
 [f_M,&\Lambda_1 \Lambda_2 \dots \Lambda_i \, N^{-k} \Pi^{(1)}_{\sharp,\flat} (\eta^{j_1}, \dots , \eta^{j_k} ; \eta^{s}_p \ph_{\alpha p}) ]\\
 &=\frac{1}{M}\left[ \sum_{u=1}^ie_uf'\Big(\frac{\cN_++ \theta_u(\cN_+)}{M}\Big)+f'\Big(\frac{\cN_++ \theta_1(\cN_+)}{M}\Big)\right] \\
 &\hspace{6cm}\times \Lambda_1 \Lambda_2 \dots \Lambda_i \, N^{-k} \Pi^{(1)}_{\sharp,\flat} (\eta^{j_1}, \dots , \eta^{j_k} ; \eta^{s}_p \ph_{\alpha p})
 \end{split}
\end{equation}
It follows that
 \begin{equation*}%\label{eq:doubleComLP}
 \begin{split}
 [f_M,[f_M,&\Lambda_1 \Lambda_2 \dots \Lambda_i \, N^{-k} \Pi^{(1)}_{\sharp,\flat} (\eta^{j_1}, \dots , \eta^{j_k} ; \eta^{s}_p \ph_{\alpha p}) ]]\\
 &= F_{M,n} (\cN_+) \Lambda_1 \Lambda_2 \dots \Lambda_i \, N^{-k} \Pi^{(1)}_{\sharp,\flat} (\eta^{j_1}, \dots , \eta^{j_k} ; \eta^{s}_p \ph_{\alpha p})
 \end{split}
\end{equation*}
with 
\begin{equation*}%\label{eq:defF}
  F_{M,n} (\cN_+) = \frac{1}{M^2} \left[ \sum_{u=1}^i e_u f'\Big(\frac{\cN_++ \theta_u(\cN_+)}{M}\Big)+f'\Big(\frac{\cN_++ \theta_1(\cN_+)}{M}\Big)\right]^2
\end{equation*}
depending on the precise form of the operator $\Lambda_1 \Lambda_2 \dots \Lambda_i \, N^{-k} \Pi^{(1)}_{\sharp,\flat} (\eta^{j_1}, \dots , \eta^{j_k} ; \eta^{s}_p \ph_{\alpha p}) $. Since $e_u \not = 0$ only if $\Lambda_u$ is a $\Pi^{(2)}$ operator, since there are at most $n$ $\Pi^{(2)}$ operators among $\Lambda_1, \dots , \Lambda_i$ and since $|e_u| \leq 2$ for all $u \in \{1, \dots , i \}$, we conclude that, for example,  
\[ \| F_{M,n} \| \leq \frac{3n^2}{M^2} \| f' \|^2_\infty  \]
\end{proof}

As explained after their definition (\ref{eq:bp-de}), the generalized creation and annihilation operators $b^*_p, b_p$ are close to the standard creation and annihilation operators on states with only few excitations, ie. with $\cN_+ \ll N$. In particular, on these states we expect the action of the generalized Bogoliubov transformation (\ref{eq:eBeta}) to be close to the action (\ref{eq:act-Bog}) of the standard Bogoliubov transformation (\ref{eq:wteBeta}). To make this statement more precise we define, under the assumption that $\| \eta \|$ is small enough,  the remainder operators 
\begin{equation} \label{eq:defD}
d_q =\sum_{m\geq 0}\frac{1}{m!} \Big [\text{ad}_{-B(\eta)}^{(m)}(b_q) - \eta_q^m b_{\alpha_m q}^{\sharp_m }  \Big],\hspace{0.5cm} d^*_q =\sum_{m\geq 0}\frac{1}{m!} \Big [\text{ad}_{-B(\eta)}^{(m)}(b^*_q) - \eta_q^m b_{\alpha_m q}^{\sharp_{m+1}}  \Big]\end{equation}
where $q \in \L^*_+$, $ (\sharp_m, \alpha_m) = (\cdot, +1)$ if $m$ is even and $(\sharp_m, \alpha_m) = (*, -1)$ if $m$ is odd. It follows then from (\ref{eq:conv-serie}) that 
\begin{equation}\label{eq:ebe} 
e^{-B(\eta)} b_q e^{B(\eta)} = \gamma_q  b_q +\s_q b^*_{-q} + d_q, \hspace{1cm} e^{-B(\eta)} b^*_q e^{B(\eta)} = \g_q b^*_q +\s_q b_{-q} + d^*_q  \end{equation} 
where we introduced the notation $\g_q = \cosh (\eta_q)$ and $\s_q = \sinh (\eta_q)$. It will also be useful to introduce remainder operators in position space. For $x \in \Lambda$, we define the operator valued distributions $\check{d}_x, \check{d}^*_x$ through
\begin{equation}\label{eq:ebex} e^{-B(\eta)} \check{b}_x e^{B(\eta)} = b ( \check{\g}_x)  +  b^* (\check{\s}_x) + \check{d}_x, \qquad 
e^{-B(\eta)} \check{b}^*_x e^{B(\eta)} = b^* ( \check{\gamma}_x)  +  b (\check{\s}_x) + \check{d}^*_x
\end{equation}
where $\check{\gamma}_x (y) = \sum_{q \in \Lambda^*} \cosh (\eta_q) e^{iq \cdot (x-y)}$ and $\check{\s}_x (y) = \sum_{q \in \Lambda^*} \sinh (\eta_q) e^{iq \cdot (x-y)}$.  

The next lemma confirms the intuition that remainder operators are small, on states with $\cN_+ \ll N$, and provides estimates that will be crucial for our analysis.  
\begin{lemma} \label{lm:dp} 
Let $\eta \in \ell^2 (\Lambda_+^*)$, $n \in \bZ$. For $p \in \L_+^*$, let $d_p$ be defined as in (\ref{eq:defD}). If $\| \eta \|$ is small enough, there exists $C > 0$ such that  
\begin{equation}\label{eq:d-bds}
\begin{split} 
\| (\cN_+ + 1)^{n/2} d_p \xi \| &\leq \frac{C}{N} \left[ |\eta_p| \| (\cN_+ + 1)^{(n+3)/2} \xi \| + \| \eta \| \| b_p (\cN_+ + 1)^{(n+2)/2} \xi \| \right], \\ 
\| (\cN_+ + 1)^{n/2} d_p^* \xi \| &\leq \frac{C}{N} \, \| \eta \| \,\| (\cN_+ +1)^{(n+3)/2} \xi \| \end{split}  \end{equation}
for all $p \in \L^*_+, \xi \in \cF_+^{\leq N}$. With $\bar{\bar{d}}_p = d_p + N^{-1} \sum_{q \in \L_+^*} \eta_q b_q^* a_{-q}^* a_p$, we also have, for $p \not \in \text{supp } \eta$, the improved bound 
\begin{equation}\label{eq:off} \| (\cN_+ + 1)^{n/2} \bar{\bar{d}}_p \xi \| \leq \frac{C}{N} \| \eta \|^2 \| a_p (\cN_+ + 1)^{(n+2)/2} \xi \| \end{equation}
In position space, with $\check{d}_x$ defined as in (\ref{eq:ebex}), we find   
 \begin{equation}\label{eq:dxy-bds} 
 \| (\cN_+ + 1)^{n/2} \check{d}_x \xi \| \leq  \frac{C }{N}\, \| \eta \| \Big[ \,\| (\cN_+ + 1)^{(n+3)/2} \xi \| +  \| b_x (\cN_+ + 1) ^{(n+2)/2}\xi \| \Big] 
\end{equation}
Furthermore, letting $\check{\bar{d}}_x = \check{d}_x  + (\cN_+ / N) b^*(\check{\eta}_x)$, we find 
\be \begin{split} \label{eq:splitdbd}
\| (\cN_+ &+ 1)^{n/2} \check{a}_y \check{\bar{d}}_x \xi \| \\ &\leq \frac{C}{N} \, \Big[ \, \|\eta \|^2  \| (\cN_+ + 1)^{(n+2)/2} \xi \|  + \| \eta \| |\check{\eta} (x-y)|  \| (\cN +1)^{(n+2)/2}  \xi \| \\
& \hspace{1cm} + \| \eta \| \| \check{a}_x (\cN_++1)^{(n+1)/2} \xi \| +  \|\eta \|^2 \|\check{a}_y (\cN_+ + 1)^{(n+3)/2} \xi \|\\
& \hspace{1cm}  + \| \eta \| \| \check{a}_x \check{a}_y (\cN +1)^{(n+2)/2}  \xi \|   \, \Big]
\end{split}\ee
and, finally, 
\begin{equation}\label{eq:ddxy}
\begin{split} 
\| (\cN_+ &+ 1)^{n/2} \check{d}_x \check{d}_y \xi \|  \\ &\leq \frac C {N^2} \Big[ \; \|\eta\|^2  \| (\cN_++ 1)^{(n+6)/2} \xi \| + \| \eta \| |\check{\eta} (x-y)|  \| (\cN_+ + 1)^{(n+4)/2}  \xi \| \\ 
&\hspace{1cm} + \|\eta \|^2 \| {a}_x (\cN_+ + 1)^{(n+5)/2} \xi \|   + \| \eta \|^2 
\|{a}_y (\cN_+ + 1)^{(n+5)/2} \xi \| \\ &\hspace{1cm}  
+ \| \eta \|^2\, \|{a}_x {a}_y (\cN_+ +  1)^{(n+4)/2} \xi \| \; \Big] 
\end{split} \end{equation}
for all $\xi \in \cF^{\leq n}_+$. 
\end{lemma}
  
\begin{proof}
%[Proof of Lemma \ref{lm:dp}]. 
To prove the first bound in (\ref{eq:d-bds}), we notice that, from (\ref{eq:defD}) and from the triangle inequality (for simplicity, we focus on $n=0$, powers of $\cN_+$ can be easily commuted through the operators $d_p$), 
\begin{equation} \label{eq:d-sum}\| d_q \xi \| \leq \sum_{m \geq 0} \frac{1}{m!} \left\| \left[ \text{ad}^{(m)}_{-B(\eta)} (b_q) - \eta_q^m b^{\sharp_m}_{\alpha_m p} \right] \xi \right\| \end{equation} 
From Lemma \ref{lm:indu}, we can bound the norm $\| [ \text{ad}^{(m)}_{-B(\eta)} (b_q) - \eta_q^m b^{\sharp_m}_{\alpha_m p} ] \xi \|$ by the sum of one term of the form
\begin{equation}\label{eq:N-term} 
\left\| \left[ \left( \frac{N- \cN_+}{N} \right)^{\frac{m+ (1-\alpha_m)/2}{2}} \left( \frac{N+1-\cN_+}{N} \right)^{\frac{m-(1-\alpha_m)/2}{2}} - 1 \right] \eta_p^m b^{\sharp_m}_{\alpha_m p} \xi \right\| \end{equation}
and of exactly $2^m m! - 1$ terms of the form
\begin{equation}\label{eq:L-term} \left\| \Lambda_1 \dots \Lambda_{i_1} N^{-k_1} \Pi^{(1)}_{\sharp,\flat} (\eta^{j_1} , \dots , \eta^{j_{k_1}} ; \eta^{\ell_1}_p \ph_{\alpha_{\ell_1} p}) \xi \right\| \end{equation}
where $i_1, k_1, \ell_1 \in \bN$, $j_1, \dots , j_{k_1} \in \bN \backslash \{ 0 \}$ and where each $\Lambda_r$-operator is either a factor $(N-\cN_+ )/N$, a factor $(N+1-\cN_+ )/N$ or a $\Pi^{(2)}$-operator of the form 
\begin{equation}\label{eq:Pi2-ex}
N^{-h} \Pi^{(2)}_{\underline{\sharp}, \underline{\flat}} (\eta^{z_1} , \dots, \eta^{z_h}) 
\end{equation}
with $h, z_1, \dots , z_h \in \bN \backslash \{ 0 \}$. Furthermore, since we are considering the term (\ref{eq:N-term}) separately, each term of the form (\ref{eq:L-term}) must have either $k_1 > 0$ or it must contain at least one $\Lambda$-operator having the form (\ref{eq:Pi2-ex}). Since (\ref{eq:N-term}) vanishes for $m=0$, it is easy to bound
\[ \begin{split} &\left\| \left[ \left( \frac{N- \cN_+}{N} \right)^{\frac{m+ (1-\alpha_m)/2}{2}} \left( \frac{N+1-\cN_+}{N} \right)^{\frac{m-(1-\alpha_m)/2}{2}} - 1 \right] \eta_p^m b^{\sharp_m}_{\alpha_m p} \xi \right\| \\ & \hspace{8cm} \leq C^m |\eta_p|^{m} N^{-1} \| (\cN_+ + 1)^{3/2} \xi \| \end{split} \]
On the other hand, distinguishing the cases $\ell_1 > 0$ and $\ell_1 = 0$, we can bound
\begin{equation}\label{eq:ell1spl}  \begin{split}  &\left\| \Lambda_1 \dots \Lambda_{i_1} N^{-k_1} \Pi^{(1)}_{\sharp,\flat} (\eta^{j_1} , \dots , \eta^{j_{k_1}} ; \eta^{\ell_1}_p \ph_{\alpha_{\ell_1} p}) \xi \right\| \\ & \hspace{1cm} \leq C^m N^{-1} \left[ \|\eta\|^{m-\ell_1} \, |\eta_p|^{\ell_1}  \d_{\ell_1>0} \| (\cN_+ + 1)^{3/2} \xi \|  + \|\eta\|^m \| b_p (\cN_+ + 1) \xi \| \right] \\
& \hspace{1cm} \leq C^m  \|\eta\|^{m-1} N^{-1} \left[ \, |\eta_p|  \d_{m>0} \| (\cN_+ + 1)^{3/2} \xi \|  + \|\eta\| \| b_p (\cN_+ + 1) \xi \| \right]
\end{split} \end{equation}
where in the last line we used $|\eta_p| \leq \| \eta \|$. Inserting the last two bounds in (\ref{eq:d-sum}) and summing over $m$ under the assumption that $\| \eta \|$ is small enough, we arrive at the first estimate (\ref{eq:d-bds}). The second estimate in (\ref{eq:d-bds}) can be proven similarly (notice that, when dealing with the second estimate in (\ref{eq:d-bds}), contributions of the form (\ref{eq:L-term}) with $\ell_1 = 0$, can only be bounded by $\| b_p^* (\cN_+ +1) \xi \| \leq \| (\cN_+ + 1)^{3/2} \xi \|$). To show (\ref{eq:off}), we notice that $\bar{\bar{d}}_p$ is exactly defined to cancel the only contribution with $m=1$ that does not vanish for $p  \not \in \text{supp } \eta$. Moreover, the assumption $\eta_p = 0$ implies that only terms 
with $\ell_1 = 0$ survive in (\ref{eq:ell1spl}). Also the bounds in (\ref{eq:dxy-bds}) and \eqref{eq:splitdbd} can be shown analogously, using \cite[Lemma 7.2]{BBCS2}. 
\end{proof}

To localize the number of particles operator in Prop. \ref{prop:GNell-loc}, we will also need to control the double commutator of the remainder operators $d_p, d_p^*$ with smooth functions $f(\cN_+/M)$ of the number of particles operator, varying on the scale $M$. To this end, we use the next corollary, which is an immediate consequence of Corollary \ref{lm:ff} and of Lemma \ref{lm:dp} (and of its proof).
\begin{cor}\label{cor:ffd}
Let $f : \bR \to \bR$ be smooth and bounded. For $M \in \bN \backslash \{ 0 \}$, let $f_M = f(\cN_+ / M)$. The bounds in (\ref{eq:d-bds}), (\ref{eq:off}), (\ref{eq:dxy-bds}), (\ref{eq:splitdbd}) and (\ref{eq:ddxy}) remain true if we replace, on the left hand side, $d_p$ by $[f_M, [f_M, d_p]]$, $\bar{\bar{d}}_p$ by $[f_M, [f_M, \bar{\bar{d}}_p]]$, $\check{d}_x$ by $[f_M, [f_M, \check{d}_x]]$, $\check{a}_y \check{\bar{d}}_x$ by $[f_M, [f_M, \check{a}_y \check{\bar{d}}_x]]$ and $\check{d}_x \check{d}_y$ by $[ f_M, [f_M, \check{d}_x \check{d}_y]]$ and, on the right hand side, the constant $C$ by $C M^{-2} \| f' \|_\infty^2$. For example, the first bound in (\ref{eq:d-bds}) becomes 
\[ \begin{split} &\left\| (\cN_+ +1)^{n/2} [ f_M, [f_M, d_p]] \xi \right\| \\ &\hspace{3cm} \leq \frac{C \| f' \|^2_\infty}{N M^2} \left[ |\eta_p|  \| (\cN_+ + 1)^{(n+3)/2} \xi \| + \| \eta \| \| b_p (\cN_+ + 1)^{(n+2)/2} \xi \| \right] \end{split} \]
\end{cor}

\section{Quadratic Renormalization} 
\label{sec:ren}

We use now a generalized Bogoliubov transformation $\exp (B(\eta))$ of the form (\ref{eq:eBeta}) to renormalize the excitation Hamiltonian. To make sure that $\exp (B(\eta))$ removes correlations that are present in low-energy states, we have to choose the coefficients $\eta \in \ell^2 (\Lambda^*_+)$ appropriately. To this end, we consider the ground state solution of the Neumann problem 
\begin{equation}\label{eq:scatl} \left[ -\Delta + \frac{1}{2} V \right] f_{\ell} = \lambda_{\ell} f_\ell \end{equation}
on the  ball $|x| \leq N\ell$ (we omit here the $N$-dependence in the notation for $f_\ell$ and for $\lambda_\ell$; notice that $\lambda_\ell$ scales as $N^{-3}$), with the normalization $f_\ell (x) = 1$ if $|x| = N \ell$. By scaling, we observe that $f_\ell (N.)$ satisfies the equation 
\[ \left[ -\Delta + \frac{ N^2}{2} V (Nx) \right] f_\ell (Nx) = N^2 \lambda_\ell f_\ell (Nx) \]
on the ball $|x| \leq \ell$. We choose $0 < \ell < 1/2$, so that the ball of radius $\ell$ is contained in the box $\Lambda= [-1/2 ; 1/2]^3$ (later, we will choose $\ell > 0$ small enough, but always of order one, independent of $N$). We extend then $f_\ell (N.)$ to $\Lambda$, by setting $f_{N,\ell} (x) = f_\ell (Nx)$, if $|x| \leq \ell$ and $f_{N,\ell} (x) = 1$ for $x \in \Lambda$, with $|x| > \ell$. Then  
\begin{equation}\label{eq:scatlN}
 \left( -\Delta + \frac{N^2}{2} V (Nx) \right) f_{N,\ell} = N^2 \lambda_\ell f_{N,\ell} \chi_\ell  
\end{equation}
where $\chi_\ell$ is the characteristic function of the ball of radius $\ell$. The Fourier coefficients of the function $f_{N,\ell}$ are given by 
\begin{equation}\label{eq:fellN} \widehat{f}_{N,\ell} (p) := \int_\Lambda f_\ell (Nx) e^{-i p \cdot x} dx \end{equation}
for all $p \in \L^*$. It is also useful to introduce the function $w_\ell (x) = 1- f_\ell (x)$ for $|x| \leq N \ell$ and to extend it by setting $w_\ell (x) = 0$ for $|x| > N \ell$. Its rescaled version $w_{N,\ell} : \Lambda \to \bR$ is then defined through $w_{N,\ell} (x) = w_{\ell} (Nx)$ if $|x| \leq \ell$ and $w_{N,\ell} (x) = 0$ if $x \in \L$ with $|x| > \ell$.  The Fourier coefficients of $w_{N,\ell}$ are then given, for $p \in \L^*$, by  
\[  \widehat{w}_{N,\ell} (p) = \int_{\Lambda} w_\ell (Nx) e^{-i p \cdot x} dx = \frac{1}{N^3} \widehat{w}_\ell (p/N) \]
where \[ \widehat{w}_\ell (k) = \int_{\bR^3} w_\ell (x) e^{-ik \cdot x} dx \] denotes the  Fourier transform of the (compactly supported) function $w_\ell$. We find $\widehat{f}_{N,\ell} (p) = \delta_{p,0} - N^{-3} \widehat{w}_\ell (p/N)$. {F}rom (\ref{eq:scatlN}), we obtain  
\begin{equation}\label{eq:wellp}
\begin{split}  
- p^2 \widehat{w}_\ell (p/N) +  \frac{N^2}{2} \sum_{q \in \L^*} \widehat{V} ((p-q)/N) \widehat{f}_{N,\ell} (q) = N^5 \lambda_\ell \sum_{q \in \L^*} \widehat{\chi}_\ell (p-q) \widehat{f}_{N,\ell} (q) 
\end{split} 
\end{equation}

In the next lemma we collect some important properties of $w_\ell, f_\ell$. The proof of the lemma is given in Appendix \ref{appx:sceq}. 
%\cite[Lemma A.1]{ESY0} and in \cite[Lemma 4.1]{BS}. 
% Notice that this lemma is the reason why we require that $V \in L^3 
%(\bR^3)$; for the rest of the analysis $V \in L^2 (\bR^3)$ would be %enough. 
\begin{lemma} \label{3.0.sceqlemma}
Let $V \in L^3 (\bR^3)$ be non-negative, compactly supported and spherically symmetric. Fix $\ell > 0$ and let $f_\ell$ denote the solution of \eqref{eq:scatl}. For $N$ large enough the following properties hold true.
\begin{enumerate}
\item [i)] We have 
\begin{equation}\label{eq:lambdaell} 
  \lambda_\ell = \frac{3\frak{a}_0 }{(\ell N)^3} \left(1 +\mathcal{O} \big(\frak{a}_0  / \ell N\big) \right)
\end{equation}
% \[ \lambda_\ell = \frac{3\frak{a}_0 }{(\ell N)^3} \left(1 + \mathcal{O} (\frak{a}_0  / \ell N) \right) \]
\item[ii)] We have $0\leq f_\ell, w_\ell\leq1$. Moreover there exists a constant $C > 0$ such that 
%\begin{equation}\label{eq:Vfa0} \left|  \int  V(x) f_\ell (x) dx - 8\pi \frak{a}_0  \right| \leq \frac{C \frak{a}_0^2}{\ell N} \, . 
%\end{equation}    
\begin{equation} \label{eq:Vfa0} 
\left|  \int  V(x) f_\ell (x) dx - 8\pi \frak{a}_0   \right| \leq \frac{C \frak{a}_0^2}{\ell N} \, 
\end{equation}
for all $\ell \in (0;1/2)$ and $N \in \bN$.
\item[iii)] There exists a constant $C>0 $ such that 
	\begin{equation}\label{3.0.scbounds1} 
	w_\ell(x)\leq \frac{C}{|x|+1} \quad\text{ and }\quad |\nabla w_\ell(x)|\leq \frac{C }{x^2+1}. 
	\end{equation}
for all $x \in \bR^3$, $\ell \in (0;1/2)$ and all $N$ large enough. 
\item[iv)] There exists a constant $C > 0$ such that 
\[ |\widehat{w}_{N,\ell} (p)| \leq \frac{C}{N p^2} \]
for all $p \in \bR^3$, all $\ell \in (0;1/2)$ and all $N$ large enough (such that $N \geq \ell^{-1}$).  
\end{enumerate}        
\end{lemma}

We define $\eta: \L^* \to \bR$ through
\[ %\label{eq:defeta}
\eta_p = -N \widehat{w}_{N,\ell} (p) = - \frac 1 {N^2} \widehat{w}_\ell(p/N)
\]
With Lemma \ref{3.0.sceqlemma}, we can bound
\be \label{eq:modetap}
|\eta_p| \leq \frac{C}{|p|^2}
\ee
for all $p \in \L_+^*=2\pi \bZ^3 \backslash \{0\}$, and  for some constant $C>0$ independent of $N$ and $\ell\in (0;\frac12)$, if $N$ is large enough. {F}rom (\ref{eq:wellp}), we also find the relation  
\begin{equation}\label{eq:eta-scat0}
\begin{split} 
p^2 \eta_p + \frac{1}{2} (\widehat{V} (./N) *\widehat{f}_{N,\ell}) (p) = N^3 \l_\ell (\widehat{\chi}_\ell * \widehat{f}_{N,\ell}) (p)
\end{split} \end{equation}
or equivalently, expressing the r.h.s. through the coefficients $\eta_p$, 
\begin{equation}\label{eq:eta-scat}
\begin{split} 
p^2 \eta_p + \frac{1}{2} \widehat{V} (p/N) & + \frac{1}{2N} \sum_{q \in \Lambda^*} \widehat{V} ((p-q)/N) \eta_q \\ &\hspace{2cm} = N^3 \lambda_\ell \widehat{\chi}_\ell (p) + N^2 \lambda_\ell \sum_{q \in \Lambda^*} \widehat{\chi}_\ell (p-q) \eta_q
\end{split} \end{equation}
Moreover, with (\ref{3.0.scbounds1}), we find 
\begin{equation}\label{eq:L2eta} \| \eta \|^2 = \| \check{\eta} \|^2 = \int_{|x| \leq \ell}  N^2 |w (N x)|^2 dx \leq C \int_{|x| \leq \ell} \frac{1}{|x|^2} dx \leq C \ell \end{equation}
In particular, we can make $\| \eta \|$ arbitrarily small, choosing $\ell$ small enough. 

%%%%%%%%%%%%%%%%%%%%%%%%%%%%%%%%%%%%%%%%%%%%%%%%%%%%%%%%%%%
%%%%%%%%%%%%%%%%%%%%%%%%%%%%%%%%%%%%%%%%%%%%%%%%%%%%%%%%%%%

For $\alpha > 0$, we now define the momentum set 
\begin{equation}\label{eq:PellHPellL1} P_{H}= \{p\in \Lambda_+^*: |p|\geq \ell^{-\alpha}\}, 		\end{equation}
depending on the parameter $\ell > 0$ introduced in (\ref{eq:scatl})\footnote{At the end, we will need the high-momentum cutoff $\ell^{-\a}$ to be sufficiently large. To reach this goal, we will choose $\ell$ sufficiently small. Alternatively, we could decouple the cutoff from the radius $\ell$ introduced in (\ref{eq:scatl}), keeping $\ell \in (0;1/2)$ fixed and choosing instead the exponent $\a$ sufficiently large.}. We set 
\be \label{eq:defetaH}
\eta_H (p)=\eta_p\, \chi(p \in P_H) = \eta_p \chi (|p| \geq \ell^{-\alpha}) \,.
\ee
Eq. \eqref{eq:modetap} implies that 
\begin{equation}\label{eq:etaHL2}
\| \eta_H \| \leq C \ell^{\a/2}
\end{equation}
For $\a > 1$, the last bound improves (\ref{eq:L2eta}). As we will see later, this improvement, obtained through the introduction of a momentum cutoff, will play an important role in our analysis. Notice, on the other hand, that the $H^1$-norms of $\eta$ and $\eta_{H}$ diverge, as $N \to \infty$. From Lemma \ref{3.0.sceqlemma}, part iii), we find 
\be \label{eq:H1eta}
\sum_{p \in P_H} p^2 |\eta_p|^2  \leq \sum_{p \in \L_+^*} p^2 |\eta_p|^2 \leq C N
\ee
for all $\ell \in (0;1/2)$ and $N \in \bN$ large enough. We will mostly use the coefficients $\eta_p$ with $p\neq 0$. Sometimes, however, it will be useful to have an estimate on $\eta_0$ (because Eq. \eqref{eq:eta-scat} involves $\eta_0$). From Lemma \ref{3.0.sceqlemma}, part iii) we find
\begin{equation}\label{eq:wteta0}  |\eta_0| \leq N^{-2} \int_{\bR^3} w_\ell (x) dx \leq C \ell^2 \end{equation}

It will also be useful to have bounds for the function $\check{\eta}_H : \L \to \bR$, having Fourier coefficients $\eta_H (p)$ as defined in (\ref{eq:defetaH}). Writing $\eta_H (p) = \eta_p - \eta_p \chi (|p| \leq \ell^{-\a})$, we obtain 
\[ \check{\eta}_H (x) = \check{\eta} (x) - \sum_{\substack{p \in \L^* :\\  |p| \leq \ell^{-\a}}} \eta_p e^{i p \cdot x} = -N w_\ell (Nx) - \sum_{\substack{p \in \L^* :\\  |p| \leq \ell^{-\a}}} \eta_p e^{i p \cdot x} \]
We obtain 
\begin{equation}\label{eq:etax}
|\check{\eta}_H (x)| \leq C N + \sum_{\substack{p \in \L^* :\\  |p| \leq \ell^{-\a}}} |p|^{-2} \leq C (N + \ell^{-\a}) \leq C N
\end{equation}
for all $x \in \L$, if $N \in \bN$ is large enough. 

With the coefficients (\ref{eq:defetaH}), we construct the generalized Bogoliubov transformation $e^{B(\eta_H)} : \cF_+^{\leq N} \to \cF^{\leq N}_+$, defined as in (\ref{eq:eBeta}). Furthermore, we define a new, renormalized, excitation Hamiltonian $\cG_{N,\ell} : \cF^{\leq N}_+ \to \cF^{\leq N}_+$ by setting 
\begin{equation}\label{eq:GN} 
\cG_{N,\ell} = e^{-B(\eta_H)} \cL_N e^{B(\eta_H)} = e^{-B(\eta_H)} U_N H_N U_N^* e^{B(\eta_H)} 
\end{equation}

In the next proposition, we collect some important properties of the renormalized excitation Hamiltonian $\cG_{N,\ell}$. In the following, we will use the notation 
\begin{equation}\label{eq:KcVN}  
\cK = \sum_{p \in \Lambda^*_+} p^2 a_p^* a_p \qquad \text{and } \quad \cV_N = \frac{1}{2N} \sum_{\substack{p,q \in \Lambda_+^*, r \in \Lambda^* : \\ r \not = -p, -q}} \widehat{V} (r/N)  a_{p+r}^* a_q^* a_{q+r} a_p \end{equation}
for the kinetic and potential energy operators, restricted on $\cF_+^{\leq N}$. We will also write $\cH_N = \cK + \cV_N$. 

\begin{prop} \label{prop:GNell} Let $V\in L^3(\bR^3)$ be compactly supported, pointwise non-negative and spherically symmetric. Then 
\be \label{eq:GNell-prel}
\cG_{N,\ell}  = 4 \pi \frak{a}_0 N + \cH_N + \theta_{\cG_{N,\ell}}
\ee
where for every $\delta > 0$ there exists a constant $C > 0$ such that  
\begin{equation}\label{eq:Gbd0}
\pm \theta_{\cG_{N,\ell}} \leq \delta \cH_N + C \ell^{-\a} (\cN_+ + 1)
\end{equation}
and the improved lower bound 
\begin{equation}\label{eq:theta-err}
\theta_{\cG_{N,\ell}} \geq  - \delta \cH_N - C \cN_+ - C \ell^{-\a}  
\end{equation}
hold true for all $\a >3$, $\ell \in (0;1/2)$ small enough, $N \in \bN$ large enough.  

Furthermore, let 
\begin{equation}\begin{split}
		\label{eq:GNeff} \cG^{\text{eff}}_{N,\ell}  :=&\;4\pi \frak{a}_0 (N-\cN_+) + \big[\widehat V(0)-4\pi \frak{a}_0\big]\cN_+\frac{(N-\cN_+)}{N}\\
		&   + \widehat V(0)\sum_{p\in P_H^c}   a^*_pa_p (1-N/\cN_+) + 4\pi \frak{a}_0\sum_{p\in P_H^c}  \big[ b^*_p b^*_{-p} + b_p b_{-p} \big] \\
		& +  \frac{1}{\sqrt N}\sum_{p,q\in\Lambda_+^*: p+q\neq 0} \widehat V(p/N)\big[ b^*_{p+q}a^*_{-p}a_q+ \emph{h.c.}\big]+\cH_N
		\end{split}\end{equation} 
Then there exists a constant $C > 0$ such that $\cE_{\cG_{N,\ell}} = \cG_{N,\ell} - \cG^\text{eff}_{N,\ell}$ is bounded by 
\begin{equation}\label{eq:GeffE} \pm \cE_{\cG_{N,\ell}} \leq  C \ell^{(\a-3)/2} \cH_N + C \ell^{-\a} \end{equation}
for all $\a > 3$, $\ell \in (0;1/2)$ small enough, and $N$ large enough. 

Finally, there exists a constant $C > 0$ such that 
\begin{equation} \begin{split}\label{eq:errComm}
\pm \left[ f (\cN_+/M), \left[ f (\cN_+ /M) , \theta_{\cG_{N,\ell}} \right] \right] & \leq C \ell^{-\a/2} M^{-2} \|f'\|^2_{\infty}\,  \big( \cH_N + 1 \big) \\
\pm \left[ f (\cN_+/M), \left[ f (\cN_+/M) , \cE_{\cG_{N,\ell}} \right] \right] &\leq C \ell^{(\a-3)/2} M^{-2} \|f' \|^2_{\infty}\, \big( \cH_N +1 \big)
\end{split}\end{equation}
for all $\a > 3$, $\ell \in (0;1/2)$ small enough, $f: \bR \to \bR$ smooth and bounded, $M \in \bN$ and $N\in \bN$ large enough.
\end{prop}

The proof of Prop. \ref{prop:GNell} is technical and quite long; it is deferred to Section \ref{sec:GN} below. 
Eq. \eqref{eq:errComm} allows us to prove a localization estimate for $\cG_{N,\ell}$. 
\begin{prop} \label{prop:GNell-loc}
Let $f,g : \bR \to [0;1]$  be smooth, with $f^2 (x) + g^2 (x) =1$ for all $x \in \bR$. For $M \in \bN$, let $f_M := f(\cN_+/M)$ and $g_M:= g(\cN_+/M)$. There exists $C > 0$ such that  
\[
\cG_{N,\ell} =  f_M\, \cG_{N, \ell}\, f_M + g_M\, \cG_{N, \ell}\, g_M + \cE_{M}
\]
with
\begin{equation*}%\label{eq:errorLoc}
 \pm \cE_M \leq \frac{C\ell^{-\a/2}}{M^2}\big(\|f'\|^2_{\infty} +\|g'\|^2_{\infty}\big) \big( \cH_N +1 \big)
\end{equation*}
for all $\a > 3$, $\ell \in (0;1/2)$ small enough, $M \in \bN$ and $N \in \bN$ large enough. 
 \end{prop}

\begin{proof}
 An explicit computation shows that
 \begin{equation*}%\label{eq:IMS}
  \begin{split}
  \cG_{N,\ell}=f_M \cG_{N,\ell}f_M +g_M \cG_{N,\ell}g_M+\frac{1}{2}\Big([f_M,[f_M,\cG_{N,\ell}]]+[g_M,[g_M,\cG_{N,\ell}]]\Big)
  \end{split}
 \end{equation*}
Writing as in (\ref{eq:GNell-prel}), $\cG_{N,\ell} = 4\pi \frak{a}_0 N + \cH_N + \theta_{\cG_{N,\ell}}$, noticing that $4\pi \frak{a}_0 N$ and $\cH_N$ commute with $f_M, g_M$, and using the first bound in \eqref{eq:errComm}, we conclude that    
\begin{equation*}%\label{eq:locErr2}
  \begin{split}
\pm\Big([f_M,[f_M,\cG_{N,\ell}]]+[g_M,[g_M,\cG_{N,\ell}]]\Big)\leq \frac{C\ell^{-\a/2}}{M^2} \big(\|f_M'\|^2_{\infty} +\|g_M'\|^2_{\infty}\big) \big( \cH_N +1 \big)
  \end{split}
 \end{equation*}
\end{proof}

\section{Cubic Renormalization}
\label{sec:cubic} 

The quadratic renormalization leading to the excitation Hamiltonian $\cG_{N,\ell}$ is not enough to show Theorem \ref{thm:main}. In (\ref{eq:theta-err}), the error term proportional to the number of particles operator cannot be controlled by the gap in the kinetic energy (in \cite{BBCS1} this was possible, because the constant multiplying $\cN_+$ is small, if the interaction potential is weak). To circumvent this problem, we have to conjugate the main part $\cG_{N,\ell}^\text{eff}$ of $\cG_{N,\ell}$, as defined in (\ref{eq:GNeff}), with an additional unitary operator, given by the exponential of an expression cubic in creation and annihilation operators. 

For a parameter $0< \beta < \alpha$ we define the low-momentum set 
\[  P_{L} = \{p\in \Lambda_+^*: |p| \leq \ell^{-\beta}\} \]
depending again on the parameter $\ell > 0$ introduced in (\ref{eq:scatl})\footnote{At the end, we will need the low-momentum cutoff $\ell^{-\beta}$ to be sufficiently large (preserving however certain relations with the high-momentum cutoff). We will reach this goal by choosing $\ell$ small enough. Alternatively, as already remarked in the footnote after (\ref{eq:PellHPellL1}), also here we could decouple the low-momentum cutoff from the radius $\ell$ introduced in (\ref{eq:scatl}), by keeping $\ell \in (0;1/2)$ fixed and varying instead the exponent $\beta$.}. Notice that the high-momentum set $P_H$ defined in (\ref{eq:PellHPellL1}) and $P_{L}$ are separated by a set of intermediate momenta $\ell^{-\beta} < |p| < \ell^{-\alpha}$. We introduce the operator $A : \cF_+^{\leq N} \to \cF_+^{\leq N}$, by 
\begin{equation}\label{eq:Aell1} A := \frac1{\sqrt N} \sum_{r\in P_{H}, v \in P_{L}} 
\eta_r \big[b^*_{r+v}a^*_{-r}a_v - \text{h.c.}\big]\end{equation}

An important observation for our analysis is the fact that conjugation with $e^{A}$ does not substantially change the number of excitations. 
\begin{prop} \label{prop:AellNgrow}
Suppose that $A$ is defined as in (\ref{eq:Aell1}). For any $k\in \bN$ there exists a constant $C >0$ such that the operator inequality  
\[ e^{-A} (\cN_++1)^k e^{A} \leq C (\cN_+ +1)^k   \]
holds true on $\cF_+^{\leq N}$, for all $\alpha > \b > 0$, $\ell \in (0;1/2)$, and $N$ large enough. 
\end{prop}
\begin{proof} Let $\xi\in\cF_+^{\leq N}$ and define $\varphi_{\xi}:\mathbb R\to \mathbb R$ by 
 \[\varphi_{\xi}(s):= \langle \xi, e^{-sA} (\cN_+ + 1)^k e^{sA} \xi \rangle \]
Then we have, using the notation $A _\g = N^{-1/2} \sum_{r \in P_H, v \in P_L} \eta_r b_{r+v}^* a_{-r}^* a_v$,  
\[ \partial_s\varphi_{\xi}(s) = 2 \text{Re } \langle \xi, e^{-sA} \big[(\cN_+ + 1)^k, A_{\g} \big]  e^{sA} \xi \rangle 
         \]
We find 
\[\begin{split}
\langle \xi, e^{-sA} &\big[(\cN_+ + 1)^k, A_{\g} \big]  e^{sA} \xi \rangle \\ &= \frac{1}{\sqrt{N}} \sum_{ r\in P_H, v\in P_L } \eta_r  \langle e^{sA} \xi , b^*_{r+v} a^*_{-r} a_{-v} \big[(\cN_+ + 2)^k - (\cN_+ + 1)^k \big]  e^{sA} \xi  \rangle \end{split}\]
With the mean value theorem, we find a function $\theta:\mathbb N\to (0;1)$ such that 
        \[ (\cN_+ + 2)^k - (\cN_+ + 1)^k= k (\cN_+ +  \theta(\cN_+) +1)^{k-1} \]
Since $b_p \cN_+ =(\cN_+ + 1) b_p$ and $b_p^* \cN_+ = (\cN_+ - 1) b_p^*$, we obtain, using Cauchy-Schwarz and the boundedness of $\theta$,
        \begin{equation*}%\label{eq:N+k-As} 
        \begin{split} 
        &\Big| \langle \xi, e^{-sA} \big[(\cN_+ + 1)^k, A_{\g} \big]  e^{sA} \xi \rangle  \Big|\\
        &\hspace{0.2cm} \leq   \frac{C}{\sqrt{N}}  \sum_{ r\in P_H, v\in P_L } |\eta_r|  \big\| b_{r+v} a_{-r}  (\cN_+ +1)^{-1/4 +(k-1)/2} e^{sA} \xi \big\| \\ &\hspace{6cm} \times \big\| a_{-v} (\cN_+ +1)^{1/4 +(k-1)/2}  e^{sA} \xi \big\| \\
                &\hspace{.2cm} \leq \frac{C}{\sqrt{N}} \Big[ \sum_{  r\in P_H , v \in P_L } \big\| b_{r+v} a_{-r}  (\cN_+ +1)^{-1/4 +(k-1)/2} e^{sA} \xi \big\|^2 \Big]^{1/2} \\ &\hspace{5cm} \times \Big[ \sum_{  r\in P_H, v 
                \in P_L } |\eta_r|^2  
         \big\| a_{-v} (\cN_+ +1)^{1/4 +(k-1)/2}  e^{sA} \xi \big\|^2 \Big]^{1/2} \\
        &\hspace{0.2cm} \leq \frac{C}{\sqrt{N}} \|\eta_H \| \big\|  (\cN_+ +1)^{3/4 +(k-1)/2} e^{sA} \xi \big\|^2 \\   &\hspace{0.2cm}\leq   \frac{C}{\sqrt{N}} \langle e^{sA} \xi, (\cN_+ +1)^{ k +1/2 } e^{sA} \xi \rangle\\
        &\hspace{0.2cm} \leq C \langle e^{sA} \xi, (\cN_+ +1)^k e^{sA} \rangle 
        \end{split}\end{equation*}
        for a constant $C>0$ depending on $k$, but not on $N$ or $\ell$. This proves that 
        \[ \partial_s\varphi_{\xi}(s) \leq C \varphi_{\xi}(s) \] so that, by Gronwall's lemma, we find a constant $C$ with
        \[\langle \xi, e^{-A} (\cN_+ + 1)^k e^{A}   \xi \rangle = C \langle \xi, (\cN_+ + 1)^k  \xi \rangle \, . \]
\end{proof}

We use now the cubic phase $e^{A}$ to introduce a new excitation Hamiltonian, defining 
		\begin{equation*} %\label{eq:RNell}
		\cR_{N,\ell}:= e^{-A} \,\cG^{\text{eff}}_{N,\ell}\,e^{A} \end{equation*}
on a dense subset of $\cF_+^{\leq N}$. The operator $\cG_{N,\ell}^\text{eff}$ is defined as in (\ref{eq:GNeff}). As explained in the introduction, conjugation with $e^{A}$ renormalizes the cubic term on the r.h.s. of (\ref{eq:GNeff}), effectively replacing the singular potential $\widehat{V} (p/N)$ by a potential decaying already on momenta of order one. This allows us to show the following proposition.
\begin{prop}\label{prop:RNell} Let $V\in L^3(\bR^3)$ be compactly supported, pointwise non-negative and spherically symmetric. Then, for all $\a > 3$ and $\a/2 < \b < 2\a/3$, there exists $\kappa > 0$ and a constant $C>0$ such that  
\begin{equation*} %\label{eq:propRnell}
\begin{split} 
\cR_{N,\ell} \geq &\; 4\pi\mathfrak{a}_0 N +  \big(1 - C\ell^{\kappa}\big) \cH_N   - C\ell^{-3\alpha}\cN_+^2/N   - C\ell^{-3\alpha}
\end{split}
\end{equation*}	
for all $\ell \in (0;1/2)$ small enough and $N$ large enough.  
\end{prop}

The proof of Proposition \ref{prop:RNell} will be given in Section \ref{sec:RN}. In the next section, we show how Prop. \ref{prop:RNell}, together with Prop. \ref{prop:GNell} and Prop. \ref{prop:GNell-loc}, implies Theorem \ref{thm:main}.

\section{Proof of Theorem \ref{thm:main}}
\label{sec:main}

The next proposition combines the results of Prop. \ref{prop:GNell}, Prop. \ref{prop:GNell-loc} and of Prop. \ref{prop:RNell}.
\begin{prop}\label{prop:fin}
Let $V\in L^3(\bR^3)$ be compactly supported, pointwise non-negative and spherically symmetric. Let $\cG_{N,\ell}$ be the renormalized excitation Hamiltonian defined as in (\ref{eq:GN}). Then, for every $\a > 3$, $\ell \in (0;1/2)$ small enough, there exist constants $C,c > 0$ such that 
\begin{equation}\label{eq:cGN-fin} \cG_{N,\ell} -4\pi \frak{a}_0 N \geq c \cN_+ - C \end{equation}
for all $N \in \bN$ sufficiently large.   
\end{prop}
\begin{proof}
As in Proposition \ref{prop:GNell-loc}, let $f,g: \bR \to [0;1]$ be smooth, with $f^2 (x) + g^2 (x)= 1$ for all $x \in \bR$. Moreover, assume that $f (x) = 0$ for $x > 1$ and $f (x) = 1$ for $x < 1/2$. We fix $M  = \ell^{3\alpha + \kappa} N$ (with $\kappa > 0$ as in Prop. \ref{prop:RNell}) and we set $f_M = f (\cN_+ / M), g_M = g (\cN_+ / M)$. It follows from Proposition  \ref{prop:GNell-loc} that 
\begin{equation}\label{eq:cGN-1} \cG_{N,\ell} - 4 \pi \frak{a}_0 N  \geq f_M (\cG_{N,\ell} - 4 \pi \frak{a}_0 N) f_M + g_M (\cG_{N,\ell} - 4 \pi \frak{a}_0 N) g_M - C \ell^{-13\alpha/2 -2\kappa} N^{-2} (\cH_N + 1) \end{equation}

Let us consider the first term on the r.h.s. of (\ref{eq:cGN-1}). From Prop. \ref{prop:GNell}, there exists $C> 0$ such that 
\[ \cG_{N,\ell} - 4\pi \frak{a}_0 N \geq \cG_{N,\ell}^\text{eff} - 4 \pi \frak{a}_0 N - C \ell^{(\alpha-3)/2} \cH_N - C \ell^{-\alpha} \]
and also, from (\ref{eq:GNell-prel}),
\begin{equation}\label{eq:second-bd} \cG_{N,\ell} - 4\pi \frak{a}_0 N \geq \frac{1}{2} \cH_N - C \cN_+ - C \ell^{-\a} \end{equation}
for all $\a > 3$, $\ell \in (0;1/2)$ small enough and $N$ large enough. 
Together, the last two bounds imply that 
\[ \begin{split} \cG_{N,\ell} - 4\pi \frak{a}_0 N & \geq (1- C \ell^{(\a-3)/2}) (\cG_{N,\ell}^\text{eff} - 4 \pi \frak{a}_0 N ) - C \ell^{(\alpha-3)/2} \cN_+ - C \ell^{-\alpha} \end{split} \]
Hence, for $\ell > 0$ small enough, 
\[  \cG_{N,\ell} - 4\pi \frak{a}_0 N   \geq  \frac{1}{2} (\cG_{N,\ell}^\text{eff} - 4 \pi \frak{a}_0 N ) - C \ell^{(\alpha-3)/2} \cN_+ - C \ell^{-\alpha} \]
With Prop. \ref{prop:RNell}, choosing $\a > 3$ and $\a/2 < \b < 2\a/3$, we find $\kappa > 0$ such that 
\[ \begin{split} f_M &(\cG_{N,\ell} - 4 \pi \frak{a}_0 N) f_M \\ &\geq \frac{1}{2} f_M (\cG_{N,\ell}^\text{eff} - 4 \pi \frak{a}_0 N ) f_M - C \ell^{(\alpha-3)/2} f^2_M \cN_+ - C \ell^{-\alpha} f_M^2 \\ &\geq \frac{1}{2} f_M e^{A}  \left[ (1-C\ell^\kappa) \cH_N  - C \ell^{-3\alpha} \frac{\cN_+^2}{N} -  C \ell^{-3\alpha} \right] e^{-A} f_M - C \ell^{(\alpha-3)/2} f^2_M \cN_+ - C \ell^{-\alpha} f_M^2 \\  &\geq  \frac{1}{2} f_M e^{A} \left[ (1-C\ell^\kappa) \cH_N  - C \ell^\kappa \cN_+ \right] e^{-A} f_M - C \ell^{(\alpha-3)/2} f^2_M \cN_+ - C \ell^{-3\alpha} f_M^2 
\end{split} \]
In the last inequality, we used Prop. \ref{prop:AellNgrow} to estimate
\[ \begin{split} f_M e^{-A} \cN_+^2 e^A f_M &\leq C f_M (\cN_+ +1)^2 f_M \\ &\leq C N \ell^{3\a+\kappa}  f_M (\cN_+ + 1) f_M \leq  C N  \ell^{3\a+\kappa} f_M e^{-A} (\cN_+ + 1) e^A f_M \end{split} \]
because we chose $M = \ell^{3\alpha + \kappa} N$. Since now $\cN_+ \leq C \cK \leq C \cH_N$, we obtain that, for $\ell \in (0;1/2)$ small enough,  
\[ f_M (\cG_{N,\ell} - 4 \pi \frak{a}_0 N) f_M \geq C f_M e^{A} \cN_+ e^{-A} f_M - C \ell^{(\alpha-3)/2} f^2_M \cN_+ - C \ell^{-3\alpha} f_M^2 \]
With Prop. \ref{prop:AellNgrow}, we conclude that, again for $\ell >0 $ small enough, 
\begin{equation}\label{eq:fMGN} f_M (\cG_{N,\ell} - 4 \pi \frak{a}_0 N) f_M \geq C f_M^2 \cN_+ - C \ell^{-3\alpha} f_M^2 \end{equation}

Let us next consider the second term on the r.h.s. of (\ref{eq:cGN-1}). From now on, we keep $\ell > 0$ fixed (so that (\ref{eq:fMGN}) holds true), and we will only worry about the dependence of $N$. We claim that there exists a constant $C > 0$ such that 
\begin{equation}\label{eq:gMGb} g_M (\cG_{N,\ell} - 4\pi \frak{a}_0 N ) g_M \geq C N g_M^2 \end{equation}
for all $N$ sufficiently large. To prove (\ref{eq:gMGb}) we observe that, since $g(x) = 0$ for all $x \leq 1/2$,   
\[ g_M ( \cG_{N,\ell} - 4 \pi \frak{a}_0 N ) g_M \geq \left[ \inf_{\xi \in \cF_{\geq M/2}^{\leq N} : \| \xi \| = 1} \frac{1}{N} \langle \xi, \cG_{N,\ell} \xi \rangle - 4 \pi \frak{a}_0 \right] N g_M^2 \]
where $\cF_{\geq M/2}^{\leq N} = \{ \xi \in \cF_+^{\leq N} : \xi = \chi (\cN_+ \geq M/2) \xi \}$ 
is the subspace of $\cF_+^{\leq N}$ where states with at least $M/2$ excitations are described (recall that $M = \ell^{3\a + \kappa} N$). To prove (\ref{eq:gMGb}) it is enough to show that there exists $C > 0$ with 
\begin{equation}\label{eq:gMGN2}  \inf_{\xi \in \cF_{\geq M/2}^{\leq N} : \| \xi \| = 1} \frac{1}{N} \langle \xi, \cG_{N,\ell} \xi \rangle - 4 \pi \frak{a}_0 \geq C \end{equation}
for all $N$ large enough. From the result (\ref{eq:BEC1}) of \cite{LS,LS2,NRS}, we already know that 
\[  \inf_{\xi \in \cF_{\geq M/2}^{\leq N} : \| \xi \| = 1} \frac{1}{N} \langle \xi, \cG_{N,\ell} \xi \rangle - 4 \pi \frak{a}_0  \geq  \inf_{\xi \in \cF_{+}^{\leq N} : \| \xi \| = 1} \frac{1}{N} \langle \xi, \cG_{N,\ell} \xi \rangle - 4 \pi \frak{a}_0 = \frac{E_N}{N} - 4 \pi \frak{a}_0 \to 0 \]
as $N \to \infty$. Hence, if we assume by contradiction that (\ref{eq:gMGN2}) does not hold true, then we can find a subsequence $N_j \to \infty$ with
\[  \inf_{\xi \in \cF_{\geq M_j/2}^{\leq N_j} : \| \xi \| = 1} \frac{1}{N_j} \langle \xi, \cG_{N_j ,\ell} \xi \rangle - 4 \pi \frak{a}_0 \to 0 \]
as $j \to \infty$ (here we used the notation $M_j = \ell^{3\alpha + \kappa} N_j$). This implies that there exists a sequence $\xi_{N_j} \in \cF^{\leq N_j}_{ \geq M_j /2}$ with $\| \xi_{N_j} \| = 1$ for all $j \in \bN$ such that 
\begin{equation*}%\label{eq:subse} 
\lim_{j \to \infty} \frac{1}{N_j} \langle \xi_{N_j}, \cG_{N_j, \ell} \xi_{N_j} \rangle = 4\pi \frak{a}_0  \, . \end{equation*}
Let now $S:= \{N_j: j\in \NN\} \subset\NN$ and denote by $\xi_N$ a normalized minimizer of $\cG_{N,\ell}$ for all $N\in \NN\setminus S$. Setting $\psi_N = U_N^* e^{B(\eta_H)}  \xi_N$, for all $N \in \bN$, we obtain that $\| \psi_N \| = 1$ and that 
\[ \lim_{N \to \infty} \frac{1}{N} \langle \psi_N, H_N \psi_N \rangle = \lim_{N \to \infty} \frac{1}{N} \langle \xi_N, \cG_{N,\ell} \xi_N \rangle = 4\pi \frak{a}_0 \]
In other words, the sequence $\psi_N$ is an approximate ground state of $H_N$. From \ref{eq:BEC1}, we conclude that $\psi_N$ exhibits complete Bose-Einstein condensation in the zero-momentum mode $\ph_0$, meaning that
\[ \lim_{N \to \infty} 1 - \langle \ph_0, \gamma_N \ph_0 \rangle = 0 \]
Using Lemma \ref{lm:Ngrow} and the rules (\ref{eq:U-rules}), we observe that 
\begin{equation}\label{eq:contra1} \begin{split} \frac{1}{N} \langle \xi_N, \cN_+ \xi_N \rangle &= \frac{1}{N} \langle e^{-B(\eta_H)} U_N \psi_N , \cN_+ e^{-B(\eta_H)} U_N \psi_N \rangle \\ &\leq \frac{C}{N} \langle \psi_N , U_N^* (\cN_+ +1) U_N \psi_N \rangle = \frac{C}{N} + C \left[ 1 - \frac{1}{N} \langle \psi_N, a^* (\ph_0) a(\ph_0) \psi_N \rangle \right] \\ &= \frac{C}{N} + C \left[ 1 - \langle \ph_0 , \gamma_N \ph_0 \rangle \right]  \to 0 \end{split} \end{equation}
as $N \to \infty$. On the other hand, for $N \in S = \{ N_j : j \in \bN \}$, we have $\xi_N = \chi (\cN_+ \geq M/2) \xi_N$ and therefore
\[ \frac{1}{N} \langle \xi_N, \cN_+ \xi_N \rangle \geq \frac{M}{2N} = \frac{\ell^{3\alpha+\kappa}}{2} \]
in contradiction with (\ref{eq:contra1}). This proves (\ref{eq:gMGN2}), (\ref{eq:gMGb}) and therefore also
\begin{equation}\label{eq:gM-bd}  g_M ( \cG_{N,\ell} - 4 \pi \frak{a}_0 N ) g_M \geq C \cN_+ g_M^2 \end{equation}

Inserting (\ref{eq:fMGN}) and (\ref{eq:gM-bd}) on the r.h.s. of (\ref{eq:cGN-1}), we obtain that 
\begin{equation}\label{eq:lbd} \cG_{N,\ell} - 4\pi \frak{a}_0 N \geq C \cN_+  - C N^{-2} \cH_N - C \end{equation}
for $N$ large enough (the constants $C$ are now allowed to depend on $\ell$, since $\ell$ has been fixed once and for always after (\ref{eq:fMGN})). Interpolating (\ref{eq:lbd}) with (\ref{eq:second-bd}), we obtain (\ref{eq:cGN-fin}).
\end{proof}

We are now ready to show our main theorem.
\begin{proof}[Proof of Theorem \ref{thm:main}]
First of all, (\ref{eq:GNell-prel}) and (\ref{eq:Gbd0}) in Prop. \ref{prop:GNell} imply that
\begin{equation*} %\label{eq:upper} 
\cG_{N,\ell} - 4 \pi \frak{a}_0 N \leq 2 \cH_N + C \cN_+ + C \end{equation*}
With the vacuum $\Omega$ as trial state, we obtain the upper bound $E_N \leq 4\pi \frak{a}_0 N + C$ for the ground state energy $E_N$ of $\cG_{N,\ell}$ (which coincides with the ground state energy of $H_N$). With Eq. (\ref{eq:cGN-fin}), we also find the lower bound $E_N \geq 4 \pi \frak{a}_0 N - C$. This proves (\ref{eq:Enbd}). 

Let now $\psi_N \in L^2_s (\Lambda^N)$ with $\| \psi_N \| =1$ and
\[ \langle \psi_N , H_N \psi_N \rangle \leq 4 \pi \frak{a}_0 N + K \]
We define the excitation vector $\xi_N = e^{-B(\eta_H)} U_N \psi_N$. Then $\| \xi_N \| = 1$ and, recalling that $\cG_{N,\ell} = e^{-B(\eta_H)}  U_N H_N U_N^* e^{B(\eta_H)}$, we have 
\[ \langle \xi_N, \cN_+ \xi_N \rangle \leq C \langle \xi_N, (\cG_{N,\ell} - 4\pi \frak{a}_0 N) \xi_N \rangle + C \leq C  (K + 1) \]
If $\gamma_N$ denotes the one-particle reduced density matrix associated with $\psi_N$, we obtain 
\[ \begin{split} 
1 - \langle \ph_0, \gamma_N \ph_0 \rangle &= 1 - \frac{1}{N} \langle \psi_N, a^* (\ph_0) a (\ph_0) \psi_N \rangle \\ &= 1 - \frac{1}{N} \langle U_N^* e^{B(\eta_H)} \xi_N, a^* (\ph_0) a(\ph_0) U_N^* e^{B(\eta_H)} \xi_N \rangle \\ &= \frac{1}{N} \langle e^{B(\eta_H)} \xi_N, \cN_+ e^{B(\eta_H)} \xi_N \rangle \leq \frac{C}{N} \langle \xi_N , \cN_+ \xi_N \rangle \leq \frac{C(K+1)}{N} \end{split} \]
which concludes the proof of (\ref{eq:BEC}). 
\end{proof}

\section{Analysis of $ \cG_{N,\ell}$} \label{sec:GN}

From (\ref{eq:cLN}) and (\ref{eq:GN}), we can decompose
\[ \cG_{N,\ell} = e^{-B(\eta_H)} \cL_N e^{B(\eta_H)} = \cG^{(0)}_{N,\ell} + \cG_{N,\ell}^{(2)} + \cG_{N,\ell}^{(3)} + \cG_{N,\ell}^{(4)} \]
with 
\[ \cG_{N,\ell}^{(j)} = e^{-B(\eta_H)} \cL_N^{(j)} e^{B(\eta_H)} \]
In the next subsections, we prove separate bounds for the operators $\cG_{N,\ell}^{(j)}$, $j=0,2,3,4$. In Subsection \ref{sub:proofGN}, we combine these bounds to prove Prop. \ref{prop:GNell} 
and Prop. \ref{prop:GNell-loc}. Throughout this section, we will assume the potential $V \in L^3 (\bR^3)$ to be compactly supported, pointwise non-negative and spherically symmetric. 

\subsection{Analysis of $ \cG_{N,\ell}^{(0)}=e^{-B(\eta_H)}\cL^{(0)}_N e^{B(\eta_H)}$} 
\label{sub:G0}

From (\ref{eq:cLNj}), recall that 
\begin{equation}\label{eq:cLN02} \cL_{N}^{(0)} =\; \frac{(N-1)}{2N} \widehat{V} (0) (N-\cN_+ ) + \frac{\widehat{V} (0)}{2N} \cN_+  (N-\cN_+ ) \end{equation}
We define the error operator $\cE_{N,\ell}^{(0)}$  through the identity
\begin{equation}\label{eq:G0C}
 \begin{split}
\cG^{(0)}_{N,\ell} &= e^{-B(\eta_H)} \cL^{(0)}_N e^{B(\eta_H)}= \frac{(N-1)}{2N} \widehat{V} (0) (N-\cN_+)+\frac{\widehat{V} (0)}{2N} \cN_+ (N-\cN_+) + \cE_{N,\ell}^{(0)}
 \end{split}
\end{equation}
% and 
% \begin{equation}\label{eq:G0S}
%  \begin{split}
% \cG^{(0)}_{N,\ell} &= e^{-B_\ell(\eta)} \cL^{(0)}_N e^{B_\ell(\eta)}= \frac{(N-1)}{2} \widehat{V} (0)+ \D^{(0)}_{N,\ell}
%  \end{split}
% \end{equation}
In the next proposition, we estimate $\cE_{N,\ell}^{(0)}$ and its double commutator with a smooth and bounded function of $\cN_+$. 
\begin{prop}\label{prop:G0}
There exists a constant $C > 0$ such that  
\begin{equation}\label{eq:E0C}
\pm \cE_{N,\ell}^{(0)} \leq C \ell^{\a/2} (\cN_+ +1)
\end{equation}
and 
\begin{equation}\label{eq:E0Cff}
\pm [f (\cN_+/M) , [f (\cN_+ /M) ,\cE_{N,\ell}^{(0)}]] \leq C \ell^{\a/2} M^{-2} \|f'\|^2_{\infty} (\cN_+ +1)
\end{equation}
for all $\a > 0$, $\ell \in (0;1/2)$, $f$ smooth and bounded, $M \in \bN$ and $N \in \bN$ large enough. 
\end{prop}

\begin{proof}
From (\ref{eq:cLN02}) we have 
\begin{equation}\label{eq:cLN0-rew} \cL_N^{(0)} = \frac{(N-1)}{2} \widehat{V} (0) + \frac{1}{2N} \widehat{V} (0) \cN_+ - \frac{1}{2N} \widehat{V} (0) \cN_+^2 \end{equation}
In the last term, we rewrite 
\[ -\frac{\cN_+^2}{N} = \cN_+ \frac{N-\cN_+}{N} - \cN_+ =  \sum_{q \in \L^*_+} b_q^* b_q - \frac{\cN_+}{N} - \cN_+ \]
Inserting in (\ref{eq:cLN0-rew}), we obtain 
\[ \cL_N^{(0)} = \frac{(N-1)}{2} \widehat{V} (0) + \frac{\widehat{V} (0)}{2} \left[ \sum_{q \in \L^*_+} b_q^* b_q - \cN_+ \right] \]
From \eqref{eq:G0C}, it follows that
\begin{equation}\label{eq:e01}
 \begin{split}
\cE_{N,\ell}^{(0)} &= \frac{\widehat{V} (0)}{2} \sum_{q \in \L_+^*} \left[ e^{-B(\eta_H)} b_q^* b_q e^{B(\eta_H)} - b_q^* b_q \right] - \frac{\widehat{V} (0)}{2}  \left[ e^{-B(\eta_H)} \cN_+ e^{B(\eta_H)} - \cN_+ \right] \end{split} \end{equation}
With (\ref{eq:ebe}), we can express 
\[ \sum_{q \in \L_+^*}  e^{-B(\eta_H)} b_q^* b_q e^{B(\eta_H)} = \sum_{q \in \L^*_+} \left[ \g_q b_q^* + \s_q b_{-q} + d^*_q \right] \left[ \g_q b_q + \s_q b_{-q}^* + d_q \right]  \]
where we set $\g_q = \cosh \eta_H (q)$, $\s_q = \sinh \eta_H (q)$ and where $d_q, d^*_q$ are defined as in (\ref{eq:defD}), with $\eta$ replaced by $\eta_H (q) = \eta_q \chi (q \in P_H)$. Using $|\g_q^2 - 1| \leq C \eta_H (q)^2$, $|\s_q| \leq C |\eta_H (q)|$, the first bound in (\ref{eq:d-bds}), Cauchy-Schwarz and the estimate $\| \eta_H \| \leq C \ell^{\a/2}$ from (\ref{eq:etaHL2}), we conclude that first term on the r.h.s. of (\ref{eq:e01}) can be bounded by 
\[ \Big| \sum_{q \in \L_+^*}  \langle \xi , \big[ e^{-B(\eta_H)} b_q^* b_q e^{B(\eta_H)} - b_q^* b_q \big] \xi \rangle \Big| \leq C \ell^{\a/2} \| (\cN_+ + 1)^{1/2} \xi \|^2 \]
As for the second term on the r.h.s. of (\ref{eq:e01}), we expand using again (\ref{eq:ebe}), 
\begin{equation*}
\begin{split}
e^{-B(\eta_H)} &\cN_+  e^{B(\eta_H)} - \cN_+ \\ = \; &\int_0^1 e^{-sB(\eta_H)} [\cN_+ ,B(\eta_H)] e^{s B(\eta_H)}ds \\
=\; &\int_0^1 \sum_{p \in P_{H}}  \eta_p \, e^{-s B(\eta_H)} ( b_p b_{-p}  + b^*_p b^*_{-p} ) e^{s B(\eta_H)} \, ds \\
= \; &\int_0^1 ds \sum_{p \in P_{H}}  \eta_p \, \left[ (\g_p^{(s)} b_p + \s_p^{(s)} b_{-p}^* + d_p^{(s)}) ( \g_p^{(s)} b_{-p} + \s_p^{(s)} b_{-p}^* + d_{-p}^{(s)}) + \hc \right] 
\end{split}
\end{equation*}
with $\g_p^{(s)} = \cosh (s \eta_H (p))$, $\s_p^{(s)} = \sinh (s \eta_H (p))$ and where the operators $d_p^{(s)}$ are defined as in (\ref{eq:defD}), with $\eta$ replaced by $s \eta_H$. Using 
$|\g^{(s)}_p| \leq C$ and $|\s_p^{(s)}| \leq C |\eta_p|$, (\ref{eq:d-bds}) in Lemma \ref{lm:dp} and again (\ref{eq:etaHL2}), we arrive at  
\[ \begin{split} \Big| \langle \xi, \big[ e^{-B(\eta_H)} \cN_+  &e^{B(\eta_H)} - \cN_+ \big] \xi \rangle \Big| \\ &\leq C \| (\cN_+ + 1)^{1/2} \xi \|  \sum_{p \in P_H} |\eta_p| \left[ |\eta_p| \| (\cN_+ + 1)^{1/2} \xi \| + \| b_{p} \xi \| \right] \\ &\leq C \ell^{\a/2} \| (\cN_+ + 1)^{1/2} \xi \|^2 \end{split} \]
This concludes the proof of (\ref{eq:E0C}). 

The bound \eqref{eq:E0Cff} follows analogously, because, as observed in Cor. \ref{cor:ffd}, the estimates (\ref{eq:d-bds}) in Lemma \ref{lm:dp} remain true if we replace $d_p$ and $d_p^*$ by $[f (\cN_+/M), [f (\cN_+/M) , d_p]]$ and, respectively, $[f (\cN_+/M) , [ f (\cN_+/M), d_p^*]]$, provided we multiply the r.h.s. by an additional factor $M^{-2} \| f' \|^2_\infty$. The same observation holds true for bounds involving the operators $b_p, b_p^*$, since, for example, 
\begin{equation}\label{eq:ffbp} [ f (\cN_+ /M), [ f (\cN_+/M), b_p]] = (f(\cN_+/M) - f((\cN_+ +1)/M))^2 b_p \end{equation}
and $\| f(\cN_+ / M) - f ((\cN_+ + 1)/M) \| \leq C M^{-1} \| f' \|_\infty$.  
\end{proof}

\subsection{Analysis of $\cG_{N,\ell}^{(2)}=e^{-B(\eta_H)}\cL^{(2)}_N e^{B(\eta_H)}$} 
\label{sub:G2}

With (\ref{eq:cLNj}), we decompose $\cL_N^{(2)} = \cK + \cL_N^{(2,V)}$, where $\cK = \sum_{p \in \Lambda_+^*} p^2 a_p^* a_p$ is the kinetic energy operator and
\begin{equation}\label{eq:L2VN} 
\cL^{(2,V)}_N =  \sum_{p \in \L^*_+} \widehat{V} (p/N) a^*_pa_p \frac{N-\cN_+}{N} + \frac{1}{2} \sum_{p \in \L^*_+} \widehat{V} (p/N) \left[ b_p^* b_{-p}^* + b_p b_{-p} \right] 
\end{equation}
Accordingly, we have 
\begin{equation}\label{eq:dec-G2} \cG_{N,\ell}^{(2)} = e^{-B(\eta_H)} \cK   e^{B(\eta_H)} + e^{-B(\eta_H)} \cL_N^{(2,V)} e^{B(\eta_H)} \end{equation}
In the next two propositions, we analyse the two terms on the r.h.s. of the last equation.
\begin{prop}\label{prop:K}
There exists $C > 0$ such that   
\begin{equation} \label{eq:K-dec} \begin{split} e^{-B(\eta_H)}\cK e^{B(\eta_H)} = \; &\cK + \sum_{p \in P_{H}} p^2 \eta_p ( b_p b_{-p} + b^*_p b^*_{-p} ) \\&+ \sum_{p \in  P_{H}} p^2 \eta_p^2 \Big(\frac{N-\cN_+}{N}\Big) \Big(\frac{N-\cN_+ -1}{N}\Big) +\cE^{(K)}_{N,\ell}
 \end{split}
\end{equation}
where 
\begin{equation}\label{eq:errorKc}
 \begin{split}
\pm \cE^{(K)}_{N,\ell} \leq  C \ell^{(\a-3)/2} (\cH_N +1) 
 \end{split}
\end{equation}
and 
\begin{equation} \begin{split}\label{eq:errCommK}
\pm \left[ f (\cN_+/M), \left[ f (\cN_+ /M) ,\cE^{(K)}_{N,\ell} \right] \right] & \leq C M^{-2} \|f'\|^2_{\infty}\, \ell^{(\a-3)/2} \big( \cH_N + 1 \big)
\end{split}\end{equation} 
for all $\a > 3$, $\ell \in (0;1/2)$ small enough, $f$ smooth and bounded, $M \in \bN$ and $N \in \bN$ large enough.
\end{prop}
\begin{proof}
To show \eqref{eq:errorKc}, we write
\[ \begin{split} 
e^{-B(\eta_H)} \cK e^{B(\eta_H)} -\cK &= \int_0^1 e^{-s B(\eta_H)} [\cK , B(\eta_H)] e^{sB(\eta_H)} ds \\
&=\int_0^1 \sum_{p \in P_{H}} p^2 \eta_p \left[ e^{-s B(\eta_H)} b_p b_{-p} e^{s B(\eta_H)} + e^{-s B(\eta_H)} b_p^* b_{-p}^* e^{sB(\eta_H)} \right] ds . \end{split}  \]
With relations \eqref{eq:ebe}, we can write
\begin{equation} \label{eq:cKterms}
\begin{split} 
&e^{-B(\eta_H)} \cK e^{B(\eta_H)}-\cK \\ &= \int_0^1ds \sum_{p \in P_{H}} p^2 \eta_p \Big [\big(\gamma_p^{(s)} b_p+ \s^{(s)}_p b^*_{-p}\big) \big(\gamma_p^{(s)} b_{-p} + \s^{(s)}_p b^*_{p}\big)\, +\hc\Big ]\\
&\hspace{.3cm} + \int_0^1ds \sum_{p \in P_{H}} p^2 \eta_p \big[\big(\gamma_p^{(s)} b_p+ \s^{(s)}_p b^*_{-p}\big) 
d_{-p}^{(s)}+ d_p^{(s)} \big( \gamma_p^{(s)} b_{-p}+ \s^{(s)}_p b^*_{p}\big)+\hc\big]\\
&\hspace{.3cm}+ \int_0^1ds \sum_{p \in P_{H}} p^2 \eta_p\big[ d_p^{(s)} d_{-p}^{(s)} + \hc \big]\\
&=: \text{G}_1+\text{G}_2+\text{G}_3
\end{split}  
\end{equation}
with the notation $\gamma_p^{(s)} = \cosh(s \eta_H (p))$, $\s^{(s)}_p =\sinh(s \eta_H (p))$ and where $d^{(s)}_p$ is defined as in (\ref{eq:defD}), with $\eta_p$ replaced by $s \eta_H (p)$ (recall that $\eta_H (p) = \eta_p \chi (p \in P_H)$). We start by analysing $\text{G}_1$. Expanding the product, we obtain
\begin{equation} \label{eq:first2cK}
\begin{split} 
\text{G}_1&=\,\int_0^1ds \sum_{p \in P_{H}} p^2 \eta_p \Big [\big (\gamma_p^{(s)})^2+(\s^{(s)}_p)^2\big)\big(b_pb_{-p}+b^*_{-p}b^*_{p}\big)+ \gamma_p^{(s)} \s^{(s)}_p (4b_p^*b_{p}-2N^{-1}a^*_pa_p)\big)\Big ]\\
&\quad+2\int_0^1ds \sum_{p \in P_{H}} p^2 \eta_p \gamma_p^{(s)} \s^{(s)}_p \left(1-\frac{\cN_+}{N}\right)\\
&= \sum_{p \in P_{H}} p^2 \eta_p \big(b_p b_{-p} + b^*_{-p} b^*_{p} \big)+ \sum_{p \in P_{H}} p^2 \eta_p^2 \left(1-\frac{\cN_+}{N}\right)+\cE^K_{1}
\end{split}  
\end{equation}
with
\begin{equation} \nonumber
\begin{split} 
\cE^K_{1}=&\,\int_0^1ds \sum_{p \in P_{H}} p^2 \eta_p \big[\big(( \gamma_p^{(s)})^2-1\big)+(\s^{(s)}_p)^2\big]\big(b_p b_{-p} + b^*_{-p} b^*_{p} \big)\\
&+\int_0^1ds \sum_{p \in P_{H}} p^2 \eta_p \gamma_p^{(s)} \s^{(s)}_p (4b_p^*b_{p}-2N^{-1}a^*_pa_p)\big)\\
&+2\int_0^1ds \sum_{p \in P_{H}} p^2 \eta_p  \left[ (\gamma_p^{(s)}-1) \s^{(s)}_p + (\s^{(s)}_p-s \eta_p) \right] \Big(1-\frac{\cN_+}{N}\Big) 
\end{split}  
\end{equation}
For an arbitrary $\xi\in\cF_+^{\leq N}$, we bound 
\begin{equation} \label{eq:cE0K}
\begin{split} 
|\langle\xi, &\cE^K_{1} \xi\rangle| \\ \leq \; & C \sum_{p \in P_{H}} p^2 |\eta_p|^3 \|b_p\xi\|\|(\cN_++1)^{1/2} \xi\|+C \sum_{p \in P_{H}} p^2 \eta_p^2 \|a_p\xi\|^2+C\sum_{p \in P_{H}} p^2 \eta_p^4\\
\leq \; &C\ell^{2\a}\|(\cN_++1)^{1/2}\xi\|^2,
\end{split}  
\end{equation}
since $|\big((\gamma_p^{(s)})^2-1\big)|\leq C \eta_p^2$, $(\s^{(s)}_p)^2\leq C \eta_p^2$ and $p^2 \eta_p^2 \leq C \ell^{2\alpha}$, for all $p \in P_H$.

We consider now $\text{G}_2$ in \eqref{eq:cKterms}. We split it as $\text{G}_2 = \text{G}_{21} + \text{G}_{22} + \text{G}_{23} + \text{G}_{24}$, with 
\begin{equation} \label{eq:cKtermsG2}
\begin{split} 
\text{G}_{21} &= \int_0^1ds \sum_{p \in P_{H}} p^2 \eta_p \left( \gamma_p^{(s)} b_p d_{-p}^{(s)} + \hc \right), \hspace{.23cm} \text{G}_{22} =  \int_0^1ds \sum_{p \in P_{H}} p^2 \eta_p \left( \s^{(s)}_p b^*_{-p} d_{-p}^{(s)} + \hc \right) \\
\text{G}_{23} &= \int_0^1ds \sum_{p \in P_{H}} p^2 \eta_p \left( \gamma_p^{(s)} d_p^{(s)} b_{-p}+ \hc \right), \hspace{.05cm} \text{G}_{24} =  \int_0^1ds \sum_{p \in P_{H}} p^2 \eta_p \left( \s^{(s)}_p d_p^{(s)} b^*_{p}+ \hc \right) 
\end{split} \end{equation}

We consider $\text{G}_{21}$ first. We write
\begin{equation}\nonumber
\begin{split} 
\text{G}_{21} = \; &\int_0^1ds \sum_{p \in P_{H}} p^2 \eta_p (\gamma_p^{(s)} -1) b_p d_{-p}^{(s)} + 
\int_0^1 ds \sum_{p \in \L^*_+} p^2 \eta_p b_p d_{-p}^{(s)} \\
&-\int_0^1 ds \sum_{p \in P_H^c} p^2 \eta_p b_p \left[ d_{-p}^{(s)} + \frac{1}{N} \sum_{q \in P_H} s \eta_q b_q^* a^*_{-q} a_{-p} \right]  \\
&+ \int_0^1 ds \, \frac{s}{N} \sum_{p\in P_H^c , q \in P_H} p^2 \eta_p \eta_q b_p b_q^* a_{-q}^* a_{-p} +\hc 
\end{split} 
\end{equation}
Massaging a bit the second term (similarly as we do below, in (\ref{eq:mainF321}), (\ref{eq:bNb-comm}) 
in the proof of Prop. \ref{prop:G2V}), we arrive at 
\begin{equation} 
\label{eq:G21f} \text{G}_{21} = - \sum_{p \in P_H} p^2 \eta_p \frac{\cN_+ + 1}{N} \frac{N-\cN_+}{N} + \left[  \cE_{2}^K + \hc \right] \end{equation}
where $\cE_{2}^K = \sum_{j=1}^5 \cE_{2j}^K$, with  
\begin{equation}\label{eq:cE2} \begin{split} 
\cE_{21}^K = \; &\frac{1}{2N} \sum_{p \in P_H} p^2 \eta_p^2 (\cN_++1) \big( b^*_p b_p - \frac 1 N a_p^* a_p\big) , \hspace{.3cm} 
\cE_{22}^K = \; \int_0^1 ds \sum_{p \in P_H} p^2 \eta_p (\gamma_p^{(s)} - 1)  b_p d_{-p}^{(s)} \\
\cE_{23}^K = \; & \int_0^1 ds \sum_{p \in \Lambda_+^*} p^2 \eta_p b_p \bar{d}_{-p}^{(s)} , \hspace{3.2cm} \cE_{24}^K = \; - \int_0^1 ds \sum_{p \in P_H^c} p^2 \eta_p b_p \bar{\bar{d}}_{-p}^{(s)} \\
\cE_{25}^K =\; & \frac{1}{2N} \sum_{p \in P_H^c , q \in P_H} p^2 \eta_p \eta_q b_p b_q^* a_{-q}^* a_{-p}  \end{split} \end{equation}
Here we introduced the notation \begin{equation}\label{eq:barbarnot} \bar{d}^{(s)}_{-p} = d_{-p}^{(s)} +s \eta_{H} (p) \frac{\cN_+}{N} b_p^*, \quad \text{ and } \quad \bar{\bar{d}}^{(s)}_{-p} = d_{-p}^{(s)} + \frac{1}{N} \sum_{q \in P_H} s \eta_q b_q^* a_{-q}^* a_{-p}.\end{equation} We can easily bound
\begin{equation}\label{eq:cEK21} |\langle \xi , \cE_{21}^K \xi \rangle | \leq C \sum_{p \in P_H} p^2 \eta_p^2 \| a_p \xi \|^2 \leq C \ell^{2\a}  \|  \cN_+^{1/2} \xi \|^2 \end{equation}
and, using $|\gamma_p^{(s)} - 1| \leq C \eta_p^2$ and (\ref{eq:d-bds}) in Lemma \ref{lm:dp}, 
\begin{equation}\label{eq:cEK22} \begin{split} |\langle \xi , \cE_{22}^K \xi \rangle | &\leq \sum_{p \in P_H} p^2 |\eta_p|^3 \| \cN_+^{1/2} \xi \| \| d^{(s)}_{-p} \xi \| \\ &\leq \sum_{p \in P_H} p^2 |\eta_p|^3 \| \cN_+^{1/2} \xi \| \left[ |\eta_p| \| \cN^{1/2}_+ \xi \|  + \| \eta \|  \| a_p \xi \| \right] \leq  C \ell^{3\a/2} \| \cN_+ + 1)^{1/2} \xi \|^2 \end{split} \end{equation}
With (\ref{eq:off}) in Lemma \ref{lm:dp}, we can also estimate
\begin{equation}\label{eq:cEK24} \begin{split} |\langle \xi , \cE_{24}^K \xi \rangle | &\leq \int_0^1 ds \sum_{p \in P_H^c} p^2 |\eta_p| \| (\cN_+ + 1)^{1/2} \xi \| \| (\cN_+ + 1)^{-1/2} \bar{\bar{d}}_{-p}^{(s)} \xi \|  \\ &\leq C \| \eta_H \|^2 \| (\cN_+ + 1)^{1/2} \xi \| \sum_{p \in P_H^c} p^2 |\eta_p | \| a_p \xi \| \\ &\leq C \ell^\a  \| (\cN_+ + 1)^{1/2} \xi \| \| \cK^{1/2} \xi \| \Big[ \sum_{|p| \leq \ell^{-\a} } p^2 \eta_p^2 \Big]^{1/2} \\ &\leq C \ell^{\a/2} \| (\cN_+ + 1)^{1/2} \xi \| \| \cK^{1/2} \xi \| \end{split} \end{equation}
To bound the last term in (\ref{eq:cE2}), we commute $b_p$ to the right (note that $p \not = q$). 
We find 
\begin{equation}\label{eq:cEK25} \begin{split} |\langle  \xi , \cE_{25}^K \xi \rangle | &\leq C N^{-1}  \sum_{p \in P_H^c, q \in P_H} p^2 |\eta_p|  |\eta_q| \| a_q a_{-q} \xi \| \| a_p a_{-p} \xi \| \\ &\leq C \sum_{p \in P_H^c, q \in P_H} p^2 |\eta_p| |\eta_q| \| a_q \xi \| \| a_p \xi \| \\ &\leq C \Big[ \sum_{p \in P_H^c, q \in P_H} p^2 \eta^2_p q^2 \| a_q \xi \|^2 \Big]^{1/2}  \Big[  \sum_{p \in P_H^c, q \in P_H} q^{-2} \eta_q^2 p^2  \| a_p \xi \|^2 \Big]^{1/2} \\ &\leq C \ell^\a \| \cK^{1/2} \xi \|^2 \end{split} \end{equation}
To control the third term in (\ref{eq:cE2}), we first use (\ref{eq:eta-scat0}) to write 
\[ \cE_{23}^K = \int_0^1 ds  \sum_{p \in \L_+^*} \bigg( \widehat{V} (./N) * \widehat{f}_{N,\ell} \bigg) (p) b_p \bar{d}^{(s)}_{-p} + \int_0^1 ds \, N^3 \lambda_\ell \sum_{p \in \L_+^*} \bigg( \widehat{\chi}_\ell * \widehat{f}_{N,\ell} \bigg) (p) b_p \bar{d}^{(s)}_{-p} \]
Switching to position space, we obtain
\[ \begin{split} \cE_{23}^K = \; &\int_0^1 ds \int_{\Lambda^2}  dx dy N^3 V(N(x-y)) f_{N,\ell} (x-y) \check{b}_x \check{\bar{d}}^{(s)}_y \\ &+  \int_0^1 ds N^3 \lambda_\ell \int_{\Lambda^2} dx dy \chi_\ell (x-y) f_{N,\ell} (x-y) \check{b}_x \check{\bar{d}}^{(s)}_y \end{split} \]
With Lemma \ref{3.0.sceqlemma}, we find
\[ \begin{split} 
|\langle \xi , \cE_{23}^K \xi \rangle | \leq \; &\int_0^1 ds \int_{\Lambda^2} dx dy \left[ N^3 V (N (x-y)) + \ell^{-3} \chi_\ell (x-y) \right]  \\ &\hspace{4cm}  \times \| (\cN_+ + 1)^{1/2} \xi \| \| (\cN_+ + 1)^{-1/2}  \check{a}_x \check{\bar{d}}^{(s)}_y \xi \|  \end{split} \]
Hence, with Eq. (\ref{eq:splitdbd}) in Lemma \ref{lm:dp}, 
\[ \begin{split} 
|\langle \xi , \cE_{23}^K \xi \rangle | \leq \; &C N^{-1} \| \eta_H \| 
\int_0^1 ds \int_{\Lambda^2} dx dy \left[ N^3 V (N (x-y)) + \ell^{-3} \chi_\ell (x-y) \right] \\ &\times \| (\cN_+ + 1)^{1/2} \xi \|  \Big[ N \| (\cN_++1)^{1/2} \xi \| +  \| \check{a}_x \cN_+ \xi \| +  \| \check{a}_y \cN_+ \xi \| + \| \check{a}_x \check{a}_y \cN_+^{1/2} \xi \| \Big] \\
\leq \; &C \ell^{(\a-3)/2} \| (\cN_+ + 1)^{1/2} \xi \|^2 + C \ell^{\a/2} \|(\cN_+ + 1)^{1/2} \xi \| \| \cV_N^{1/2} \xi \| \end{split} \]
Combining the last bound with (\ref{eq:cEK21}), (\ref{eq:cEK22}), (\ref{eq:cEK24}), (\ref{eq:cEK25}), we conclude that 
\begin{equation}\label{eq:cE2f} \pm \left[ \cE_2^K  + \hc \right] \leq C \ell^{(\a-3)/2} (\cH_N + 1) \end{equation}

Next, we consider the term $\text{G}_{22}$ in \eqref{eq:cKtermsG2}. With (\ref{eq:d-bds}) in Lemma \ref{lm:dp}, we find 
\begin{equation}\label{eq:G22}
\begin{split} 
|\langle\xi,\text{G}_{22}\xi\rangle|&\leq C \sum_{p \in P_{H}} p^2 \eta_p^2\|b_{-p}\xi\|\| d_{-p} \xi\|\\
&\leq C \sum_{p \in P_{H}} p^2 \eta_p^2\|b_{-p}\xi\|\left[ |\eta_p| \| (\cN_+ + 1)^{1/2} \xi \| + \| \eta \| \| b_p \xi \| \right]\\
&\leq C\ell^{5\a/2}\|(\cN_++1)^{1/2}\xi\|^2
\end{split}  
\end{equation}
As for the term $\text{G}_{23}$, defined in \eqref{eq:cKtermsG2}, we split it as $\text{G}_{23} = \sum_{j=1}^4  \cE_{3j}^K + \hc$, with 
\begin{equation*} %\label{eq:G23}
\begin{split} 
\cE_{31}^K = \; &\int_0^1ds \sum_{p \in P_{H}} p^2 \eta_p \big(\gamma_p^{(s)}-1\big) d_p^{(s)}b_{-p} \, , \hspace{1.3cm} \cE_{32}^K =  \int_0^1ds \sum_{p \in \L_+^*} p^2 \eta_p d_p^{(s)} b_{-p} \\  \cE_{33}^K =  \; & \frac{1}{2N} \sum_{p \in P_H^c, q \in P_H} p^2 \eta_p \eta_q b_q^* a_{-q}^* a_p b_{-p} \, , \hspace{1.4cm}
\cE_{34}^K = -\int_0^1 ds \sum_{p\in P_H^c} p^2 \eta_p \bar{\bar{d}}^{(s)}_p b_{-p}
\end{split}  
\end{equation*}
with the notation for $\bar{\bar{d}}^{(s)}_p$ introduced in (\ref{eq:barbarnot}). With (\ref{eq:d-bds}) in Lemma \ref{lm:dp}, we find 
\[ \begin{split} 
|\langle \xi, \cE_{31}^K \xi \rangle | &\leq C \int_0^1 \sum_{p \in P_H} p^2 |\eta_p |^3 \| (\cN_+ + 1)^{1/2} \xi \| \| (\cN_++1)^{-1/2}  d_p^{(s)} b_{-p} \xi \| \\ &\leq C \| \eta_H \|  \| (\cN_+ + 1)^{1/2} \xi \| \sum_{p \in P_H} p^2 | \eta_p|^3  \| b_p \xi \| \leq C \ell^{3\a} \| (\cN_+ + 1)^{1/2} \xi \|^2 \end{split} \]
and also, proceeding as in (\ref{eq:cEK24}),  
\begin{equation}\label{eq:cEK34f} \begin{split} 
|\langle \xi, \cE_{34}^K \xi \rangle | &\leq C \int_0^1 ds \sum_{p \in P^c_H} p^2 |\eta_p | \| (\cN_+ + 1)^{1/2} \xi \| \| (\cN_+ + 1)^{-1/2} \bar{\bar{d}}_p^{(s)} b_{-p} \xi \| \\ &\leq C \| \eta_H \|^2 \| (\cN_+ + 1)^{1/2} \xi \| \sum_{p \in P_H^c} p^2 |\eta_p | \| b_{-p} \xi \| \\ &\leq C \ell^{\a/2}   \|  (\cN_+ + 1)^{1/2} \xi \| \| \cK^{1/2} \xi \| \end{split} \end{equation}
The term $\cE^K_{33}$ coincides with the contribution $\cE_{25}^K$ in (\ref{eq:cE2}); from (\ref{eq:cEK25}) we obtain $\pm \cE_{33}^K \leq C \ell^\a \cK$. As for $\cE_{32}^K$, we use (\ref{eq:eta-scat0}) and we switch to position space. Proceeding as we did above to control the term $\cE_{23}^K$, we arrive at 
\[\begin{split} |\langle \xi , \cE_{32}^K \xi \rangle | \leq \int_0^1 ds \int_{\Lambda^2} dx dy &\left[ N^3 V(N(x-y)) + \ell^{-3} \chi_\ell (x-y) \right] \\ &\hspace{2cm} \times \| (\cN_+ + 1)^{1/2} \xi \| \| (\cN_+ + 1)^{-1/2} \check{d}_x^{(s)} \check{a}_y \xi \| \end{split} \]
With (\ref{eq:dxy-bds}) in Lemma \ref{lm:dp}, we find 
\[\begin{split} |\langle \xi , \cE_{32}^K \xi \rangle | &\leq C N^{-1}  \| \eta_H \| \int_{\Lambda^2} dx dy \left[ N^3 V(N(x-y)) + \ell^{-3} \chi_\ell (x-y) \right] \\ &\hspace{2cm} \times \| (\cN_+ + 1)^{1/2} \xi \| \left[ \| \check{a}_y (\cN_+ + 1) \xi \| + \| \check{a}_x \check{a}_y (\cN_+ + 1)^{1/2} \xi \| \right]  \\ 
&\leq C \ell^{(\a-3)/2} \| (\cN_+ + 1)^{1/2} \xi \|^2 + C \ell^{\a/2} \| (\cN_+ + 1)^{1/2} \xi \| \| \cV_N^{1/2} \xi \| 
 \end{split} \]
Combining the last bounds, we conclude that 
\begin{equation}\label{eq:G23f} \pm \text{G}_{23} \leq C \ell^{(\a-3)/2} (\cH_N + 1) \end{equation} 

To estimate the term $\text{G}_{24}$ in \eqref{eq:cKtermsG2}, we use (\ref{eq:d-bds}) in Lemma \ref{lm:dp}; with (\ref{eq:H1eta}), we find 
\begin{equation*} 
\begin{split} 
|\langle &\xi,\text{G}_{24}\xi\rangle| \\ &\leq C \int_0^1 ds  \sum_{p \in P_{H}} p^2 \eta_p^2 \|(\cN_++1)^{1/2}\xi\|\|(\cN_++1)^{-1/2} d_p^{(s)}b^*_{p}\xi\| \\
&\leq C N^{-1} \|(\cN_++1)^{1/2}\xi\|  \sum_{p \in P_{H}} p^2 \eta_p^2  \left[ |\eta_p| \| (\cN_+ + 1)^{3/2} \xi \| + \| \eta_H \| \| b_p b_p^* (\cN_+ +1)^{1/2} \xi \| \right] \\
&\leq C  N^{-1} \|(\cN_++1)^{1/2}\xi\| \\ &\hspace{.5cm} \times \sum_{p \in P_{H}} p^2 \eta_p^2  \left[ |\eta_p| \| (\cN_+ + 1)^{3/2} \xi \| + \| \eta_H \| \| (\cN_+ +1)^{1/2} \xi \| + \| \eta_H \| \| a_p (\cN_+ + 1) \xi \| \right] \\
&\leq C \ell^{\a/2} \| (\cN_+ + 1)^{1/2} \xi \|^2 \end{split} \end{equation*}
Together with (\ref{eq:G21f}), (\ref{eq:cE2f}), (\ref{eq:G22}), (\ref{eq:G23f}), this implies that 
\[ \text{G}_2 =  - \sum_{p \in P_H} p^2 \eta_p \frac{\cN_+ + 1}{N} \frac{N-\cN_+}{N} + \cE_4^K \]
where 
\begin{equation}\label{eq:G2fi} \pm \cE_4^K \leq C \ell^{(\a-3)/2} (\cH_N + 1) \end{equation}

Finally, we consider $\text{G}_3$, defined in \eqref{eq:cKterms}. We split it as $\text{G}_3 = \cE_{51}^K + \cE_{52}^K + \hc$, with 
\[ \cE_{51}^K = \int_0^1 ds \sum_{p \in \L_+^*} p^2 \eta_p d^{(s)}_p d_{-p}^{(s)}, \qquad \cE_{52}^K= - \int_0^1 ds \sum_{p \in P_H^c} p^2 \eta_p d^{(s)}_p d^{(s)}_{-p}  \]
With (\ref{eq:d-bds}) in Lemma \ref{lm:dp} (using $\eta_{H} (p) = 0$ for $p \in P_H^c$) and proceeding as in (\ref{eq:cEK34f}), we obtain 
\[ \begin{split} 
|\langle \xi, \cE_{52}^K \xi \rangle | &\leq C \| \eta_H \| \sum_{p \in P_H^c} p^2 |\eta_p| \| (\cN_+ + 1)^{1/2} \xi \| \| d_{-p} \xi \|  \\ &\leq C \| \eta_H \|^2 \| (\cN_+ + 1)^{1/2} \xi \| \sum_{p \in P_H^c} p^2 |\eta_p | \| b_{-p} \xi \| \leq C \ell^{\a/2} \| (\cN_+ + 1)^{1/2} \xi \| \| \cK^{1/2} \xi \| \end{split} \]
To estimate $\cE_{51}^K$, we use (\ref{eq:eta-scat0}) and we switch to position space. Similarly as in the analysis of the terms $\cE_{23}^K$ and $\cE_{32}^K$ above, we obtain
\[ \begin{split}  | \langle \xi , \cE_{51}^K \xi \rangle | \leq \; &C \| (\cN_+ + 1)^{1/2} \xi \| \int_0^1 ds \int_{\L^2}  dx dy \left[ N^3 V(N(x-y)) + \ell^{-3} \chi_\ell (x-y) \right] \\ &\hspace{8cm} \times \| (\cN_+ + 1)^{-1/2} \check{d}^{(s)}_x \check{d}^{(s)}_y \xi \| \end{split} \]
With (\ref{eq:ddxy}) in Lemma \ref{lm:dp}, we arrive at
\[ \begin{split} | \langle &\xi , \cE_{51}^K \xi \rangle | \\ &\leq C N^{-2} \| (\cN_+ + 1)^{1/2} \xi \| \int_0^1 ds \int_{\L^2}  dx dy \left[ N^3 V(N(x-y)) + \ell^{-3} \chi_\ell (x-y) \right] \\ &\hspace{.3cm} \times \Big[ N \| \eta_H \| \| (\cN_+ + 1)^{3/2} \xi \|  + \| \eta_H \|^2 \| \check{a}_x \cN_+^2 \xi \| + \| \eta_H \|^2 \| \check{a}_y \cN_+^2 \xi \| + \| \eta_H \|^2 \| \check{a}_x \check{a}_y \cN_+^{3/2} \xi \| \Big] 
\\ &\leq C ( \ell^{\a /2} + \ell^{\alpha-3/2} ) \| (\cN_+ + 1)^{1/2} \xi \|^2 + C \ell^{3\a/2} \| (\cN_+ + 1)^{1/2} \xi \| \| \cV_N^{1/2} \xi \| 
\end{split} \]
Hence, $\pm \text{G}_3 \leq C (\ell^{\a/2} + \ell^{3(\alpha-1)/2}) (\cH_N + 1)$. With (\ref{eq:first2cK}),
(\ref{eq:cE0K}), (\ref{eq:G2fi}), we obtain (\ref{eq:K-dec}) and (\ref{eq:errorKc}), as desired. 

As explained in Corollary \ref{cor:ffd}, the bounds in Lemma \ref{lm:dp} continue to hold, with an additional factor $M^{-2} \| f' \|_\infty^2$ on the r.h.s., if we replace the operators $d_p$, $d^*_p$, $\bar{\bar{d}}_p$,  $\check{a}_y \check{\bar{d}}_x$, $\check{d}_x \check{d}_y$ by their double commutators with $f (\cN_+ /M)$.
From (\ref{eq:ffbp}) we conclude that also bounds involving $b_p$ and $b_p^*$ or, analogously $\check{b}_x$ and $\check{b}^*_x$ remain true if we replace them by their double commutator with $f(\cN_+/M)$. As a consequence, \eqref{eq:errCommK} follows through the same arguments that led us to (\ref{eq:errorKc}).
\end{proof}

In the next proposition, we study the second term on the r.h.s. of (\ref{eq:dec-G2}).
\begin{prop}\label{prop:G2V}
There is a constant $C > 0$ such that 
\begin{equation}\label{eq:cEV-defc} 
\begin{split}
 e^{-B (\eta_H )} &\cL^{(2,V)}_N  e^{B(\eta_H)} \\ = &\sum_{p \in P_{H}} \widehat{V} (p/N) \eta_p\Big(\frac{N-\cN_+}{N}\Big)\Big(\frac{N-\cN_+-1}{N}\Big) +\sum_{p \in \Lambda^*_+}  \widehat{V} (p/N)  a^*_pa_p \frac{N-\cN_+}{N} \\ &+ \frac{1}{2}\sum_{p \in \Lambda^*_+} \widehat{V} (p/N) \big( b_p b_{-p}+ b_{-p}^* b_p^*\big) +\cE_{N,\ell}^{(V)}
\end{split}
\end{equation}
where 
\begin{equation}\label{eq:errorVc}
  \pm \cE_{N,\ell}^{(V)}\leq  C \ell^{\a/2} (\cH_N +1)
\end{equation}
and 
\begin{equation} \begin{split}\label{eq:errCommV}
\pm \left[ f (\cN_+/M), \left[ f (\cN_+/M) ,\cE^{(V)}_{N,\ell} \right] \right] & \leq C \ell^{\a/2} M^{-2} \|f'\|^2_{\infty}\,  \big( \cH_N + 1 \big)
\end{split}\end{equation} 
for all $\a > 0$, $\ell \in (0;1/2)$ small enough, $f$ smooth and bounded, $M \in \bN$ and $N \in \bN$ large enough.
\end{prop}

\begin{proof}
To show \eqref{eq:errorVc}, we start from \eqref{eq:L2VN} and we decompose
\begin{equation}\label{eq:G2-deco} \begin{split} 
e^{-B(\eta_H)} \cL^{(2,V)}_{N}  e^{B(\eta_H)}  =\; & \sum_{p \in \Lambda^*_+}  \widehat{V} (p/N) e^{-B(\eta_H)} b_p^* b_p e^{B(\eta_H)} \\ &- \frac{1}{N} \sum_{p \in \Lambda^*_+} \widehat{V} (p/N) e^{B(\eta_H)} a_p^* a_p e^{-B(\eta_H)} \\ 
&+ \frac{1}{2} \sum_{p \in \Lambda^*_+} \widehat{V} (p/N) e^{-B(\eta_H)} \big[ b_p b_{-p} + b_p^* b_{-p}^* \big] e^{B(\eta_H)} \\ =: \; &\text{F}_1 + \text{F}_2 +\text{F}_3 
\end{split} \end{equation} 
With equations (\ref{eq:ebe}), we split $\text{F}_1$ as
\begin{equation*} %\label{eq:G21-deco} 
\begin{split} 
\text{F}_1 =\; & \sum_{p \in  \L_+^*}  \widehat{V} (p/N) \big[ \gamma_p b_p^* + \sigma_p b_{-p} \big] \big[ \gamma_p b_p + \sigma_p b_{-p}^*] \\
&+\sum_{ p\in \L_+^*} \widehat{V} (p/N) \big[ (\gamma_p b_p^* + \sigma_p b_{-p})  d_p+  d_p^*(\gamma_p b_p + \sigma_p b_{-p}^*)+ d_p^* d_p\big]\\
=: &\, \text{F}_{11} + \text{F}_{12} 
\end{split} \end{equation*}
with the notation $\gamma_p = \cosh \eta_H (p)$, $\s_p = \sinh \eta_H (p)$ and the operators $d_p$, as defined in (\ref{eq:defD}), with $\eta$ replaced by $\eta_H$. We decompose 
\begin{equation*}%\label{eq:G21-F11} 
\begin{split} 
\text{F}_{11}=\; & \sum_{ p\in  \L_+^*} \widehat{V} (p/N)  a_p^*  a_p \frac{N-\cN_+}{N} + \cE_{1}^V
\end{split} \end{equation*}
with
\begin{equation*} %\label{eq:G21-F11err} 
\begin{split} 
\cE_{1}^V = &\frac{1}{N} \sum_{p \in \L_+^*} \widehat{V} (p/N) a_p^* a_p  + \sum_{ p\in P_H} \widehat{V} (p/N) \Big[ (\g_p^2-1) b_p^*b_p + \gamma_p \sigma_p (b_{-p}b_p+b_p^*b_{-p}^*) \\
&\hspace{5cm}+  \sigma_p^2 ( b_{p}^*b_p-N^{-1}a_{p}^*a_p)+ \sigma_p^2 \Big(\frac{N-\cN_+}{N}\Big)\Big] 
\end{split} \end{equation*}
where we used $\gamma_p = 1$ and $\sigma_p = 0$ for $p \in P_H^c$ to restrict the second sum. With $|\g_p^2-1|\leq C \eta_p^2$, $|\s_p|\leq C|\eta_p|$ for all $p \in P_H$ and since $\| \eta_H \| \leq \ell^{\a/2}$, we find 
\begin{equation*}%\label{eq:G21-F11err1} 
\begin{split} 
\pm\cE^V_{1}\leq C(\ell^{\a/2} + N^{-1}) (\cN_++1) \leq C \ell^{\a/2} (\cN_++ 1) 
\end{split} \end{equation*}
if $N$ is large enough. With Lemma \ref{lm:dp} (with $\eta$ replaced by $\eta_H$),  we can also bound $\pm \text{F}_{12} \leq C\ell^{\a/2}(\cN_++1)$. We conclude that 
\begin{equation}\label{eq:F1} \text{F}_1 = \sum_{p \in \L_+^*} \widehat{V} (p/N) a_p^* a_p \frac{N-\cN_+}{N} + \cE_2^V \end{equation}
with $\pm \cE_2^V \leq C \ell^{\a/2} (\cN_+ + 1)$. Let us now consider the second contribution on the r.h.s. of (\ref{eq:G2-deco}). We have $-\text{F}_2 \geq 0$ and, by Lemma \ref{lm:Ngrow}, 
\begin{equation}\label{eq:fin-G22} \begin{split} 
-\text{F}_2 = \; & \frac{1}{N} \sum_{p \in \Lambda^*_+} \widehat{V} (p/N) e^{-B(\eta_H)} a_p^* a_p e^{B (\eta_H)}  \leq  \frac{ \| \widehat{V} \|_\infty}{N} e^{-B (\eta_H)} \cN_+ e^{B (\eta_H)} \leq C \ell^{\a/2} (\cN_++1) 
\end{split} \end{equation}
if $N \in \bN$ is large enough, 
%%%%%%%%%%%%%%%%%%%%%%%%%%%%%%%%%%%%%%%%%%%%
Finally, we turn our attention to the last term on the r.h.s. of \eqref{eq:G2-deco}.  With (\ref{eq:ebe}), we decompose  $\text{F}_3 $ as
\begin{equation}\label{eq:G23split}
\begin{split}
\text{F}_3 =\; &\frac{1}{2} \sum_{p \in \Lambda^*_+} \widehat{V} (p/N) \left[ \gamma_p b_p + \sigma_p b_{-p}^* \right] \left[ \gamma_p b_{-p} + \sigma_p b_p^* \right]  +\hc \\ 
&+ \frac{1}{2}  \sum_{p \in \Lambda^*_+} \widehat{V} (p/N) \, \left[ (\g_p b_p+ \s_p b^*_{-p}) \, d_{-p} + 
d_p\, (\g_p b_{-p} + \s_p b^*_{p}) \right] +\hc \\ 
&+\frac{1}{2}  \sum_{p \in \Lambda^*_+} \widehat{V} (p/N) d_p d_{-p} + \hc \\
=:  &\,\text{F}_{31} + \text{F}_{32}+\text{F}_{33}+\hc
\end{split} \end{equation}
We decompose the first term as 
\begin{equation}\label{eq:fin-G23}
\begin{split}
\text{F}_{31}
% = \; &\frac{1}{2} \sum_{p \in \Lambda^*_+} \widehat{V} (p/N) \left[ \tilde\gamma_p^2 b_p b_{-p} + \tilde\sigma_p\tilde\gamma_p b_{-p}^* b_{-p} + \tilde\sigma_p^2 b_{-p}^* b_p^* + \tilde\gamma_p\tilde\sigma_p\big( b_p^* b_p-N^{-1}a^*_pa_p\big)+\tilde\gamma_p\tilde\sigma_p\left(\frac{N-\cN_+}{N}\right)\right]   +\hc \\ 
% = \; & \sum_{p \in \Lambda^*_+} \widehat{V} (p/N) \left[ \frac{1}{2}\big(\tilde\gamma_p^2+ \tilde\sigma_p^2\big)\big( b_p b_{-p}+ b_{-p}^* b_p^*\big) +2\tilde\sigma_p\tilde\gamma_p b_{p}^* b_{p}  -N^{-1}\tilde\gamma_p\tilde\sigma_pa^*_pa_p+\tilde\gamma_p\tilde\sigma_p\left(\frac{N-\cN_+}{N}\right)\right] \\ 
= \; & \frac{1}{2}\sum_{p \in \Lambda^*_+} \widehat{V} (p/N) \big( b_p b_{-p}+ b_{-p}^* b_p^*\big) +\sum_{p \in P_H} \widehat{V} (p/N) \eta_p \frac{N-\cN_+}{N} + \cE^V_3\\ 
\end{split} \end{equation}
with (recall that $\gamma_p =1$ and $\sigma_p = 0$ for $p \in P_H^c$) 
\begin{equation*}%\label{eq:G23-deco1} 
\begin{split}
\cE_3^V= \; &\sum_{p \in P_H} \widehat{V} (p/N) \left[ \frac{1}{2}(\gamma_p^2-1+ \sigma_p^2)\big( b_p b_{-p}+ b_{-p}^* b_p^*\big)+2 \sigma_p \gamma_p b_{p}^* b_{p} \right. \\
 &\hspace{5cm} \left.-N^{-1} \gamma_p \sigma_p a^*_p a_p+ (\gamma_p \sigma_p - \eta_p) \frac{N-\cN_+}{N} \right]  
\end{split} \end{equation*}
Using again the estimates $|\g_p^2-1|\leq C\eta_p^2$ and $|\s_p|\leq C |\eta_p|$ for all $p \in P_H$, we find 
\begin{equation}\label{eq:G23-deco1} 
\begin{split}
\pm\cE_3^V \leq C\ell^{\a/2}(\cN_++1)
\end{split} \end{equation}
Let us now consider $\text{F}_{32}$ in \eqref{eq:G23split}. We divide it into four parts
\begin{equation}\label{eq:splitF32}
\begin{split}
\text{F}_{32} = \; & \frac{1}{2}  \sum_{p \in \Lambda^*_+} \widehat{V} (p/N) \, \left[ ( \g_p b_p+ \s_p b^*_{-p}) \, d_{-p} + d_p\, (\g_p b_{-p} + \s_p b^*_{p}) \right] +\hc \\ 
=: \; & \text{F}_{321}+\text{F}_{322}+\text{F}_{323}+\text{F}_{324} 
\end{split} \end{equation}
We start with $\text{F}_{321}$, which we decompose as
\begin{equation}\label{eq:mainF321}
\begin{split}
\text{F}_{321} = \; & \frac{1}{2}  \sum_{p \in \Lambda^*_+} \widehat{V} (p/N) (\g_p - 1) b_p d_{-p} + \frac{1}{2} \sum_{p \in \L_+^*} \widehat{V} (p/N) b_p \left[ d_{-p} + \eta_H (p) \frac{ \cN_+}{N} b_p^* \right] \\ &- \frac{1}{2N} \sum_{p \in \L_+^*} \widehat{V} (p/N) \eta_H (p) b_p \, \cN_+ b_p^* + \hc \end{split} \end{equation}
Using (\ref{eq:comm-bp}), we commute
\begin{equation}\label{eq:bNb-comm} b_p\,  \cN_+ b_p^* = (\cN_+ + 1) b_p b_p^* = (\cN_+ + 1) (1- \cN_+/N) + (\cN_+ + 1) (b_p^* b_p -N^{-1} a_p^* a_p) \end{equation}
We arrive at
\begin{equation*} %\label{eq:F321} 
\text{F}_{321} = - \sum_{p \in P_H} \widehat{V} (p/N) \eta_p \, \left(\frac{N-\cN_+}{N}\right)\left(\frac{\cN_+ +1}{N}\right) +\cE_4^V \end{equation*} 
where $\cE_4^V = \cE_{41}^V + \cE_{42}^V + \cE_{43}^V + \hc$, with 
\begin{equation*} %\label{eq:F321error}
\begin{split}
\cE^V_{41} = \; &\frac{1}{2}  \sum_{p \in \Lambda^*_+} \widehat{V} (p/N) \, (\g_p - 1) b_p d_{-p} \, , \qquad 
\cE_{42}^V = \frac{1}{2} \sum_{p \in \L_+^*} \widehat{V} (p/N) b_p \bar{d}_{-p} \\
\cE_{43}^V = \; &- \frac{1}{2} \sum_{p \in P_H} \widehat{V} (p/N) \eta_p \frac{\cN_+ + 1}{N} (b_p^* b_p - N^{-1} a_p^* a_p ) 
\end{split} \end{equation*}
and with the notation $\bar{d}_{-p} = d_{-p} + N^{-1}  \eta_H (p) \, \cN_+ b_p^*$. Since $|\g_p - 1| \leq C \eta_p^2 \chi (p \in P_H)$, we find easily with (\ref{eq:d-bds}) in Lemma \ref{lm:dp} that 
\[ \begin{split} |\langle \xi, \cE_{41}^V \xi \rangle |  &\leq C \sum_{p \in P_H} \eta_p^2 \| (\cN_+ + 1)^{1/2} \xi \| \left[ |\eta_p | \| (\cN_+ + 1)^{1/2} \xi \| + \| \eta_H \| \| a_p \xi \| \right] \\ &\leq C \ell^{3\a} \| (\cN_+ + 1)^{1/2} \xi \|^2 \end{split} \]
Furthermore
\[ |\langle \xi , \cE_{43}^V \xi \rangle | \leq C \sum_{p \in P_H} \eta_p \| a_p \xi \|^2 \leq C \ell^{2\a} \| \cN_+^{1/2} \xi \|^2 \]
To control $\cE_{42}^V$ we switch to position space. With (\ref{eq:splitdbd}) in Lemma \ref{lm:dp}, 
we find 
\[ \begin{split} 
|\langle \xi , \cE_{42}^V \xi \rangle | &\leq C \int_{\L^2}  dx dy \, N^3 V(N(x-y)) \| (\cN_+ + 1)^{1/2} \xi \| \| (\cN_+ + 1)^{-1/2} \check{a}_x \check{\bar{d}}_y \xi \| \\ &\leq C \| \eta_H \|  \int_{\L^2}  dx dy \, N^2 V(N(x-y)) \| (\cN_+ + 1)^{1/2} \xi \| \\ &\hspace{2cm} \times \Big[  N \| (\cN_+ + 1)^{1/2} \xi \| + \| \check{a}_x \cN_+ \xi \| +  \| \check{a}_y \cN_+ \xi \| +  \| \check{a}_x \check{a}_y \cN_+^{1/2} \xi \| \Big] \\
&\leq C \ell^{\a/2} \|(\cN_+ + 1)^{1/2} \xi \|^2 + C \ell^{\a/2} \| (\cN_+ + 1)^{1/2} \xi \| \| \cV_N^{1/2} \xi \|
\end{split} \]
We conclude that \begin{equation*}%\label{eq:cEV4} 
\pm \cE^V_4 \leq C \ell^{\a/2} (\cH_N + 1). 
\end{equation*}
To estimate the term $\text{F}_{322}$ in (\ref{eq:splitF32}), we use (\ref{eq:d-bds}) in Lemma \ref{lm:dp} and $|\s_p|\leq C|\eta_H (p)|$; we obtain 
\begin{equation*} %\label{eq:F322}
\begin{split}
|\langle\xi,\text{F}_{322}\xi\rangle| \leq \; &C \sum_{p \in P_H} |\eta_p| \| b_{-p} \xi\|\| d_{-p} \xi\| \\ &\leq C   \sum_{p \in P_H} |\eta_p| \| b_{-p} \xi\| \left[ |\eta_p| \| (\cN_+ + 1)^{1/2} \xi \| + \| \eta_H \| \| b_{-p} \xi \| \right] \\ &\leq C \ell^{5\a/2} \| (\cN_+ + 1)^{1/2} \xi \|^2  \end{split} \end{equation*}
Let us now consider the term $\text{F}_{323}$ on the r.h.s. of (\ref{eq:splitF32}). Here, we proceed as we did above to estimate $\text{F}_{321}$. We write  $\text{F}_{323} = \cE_{51}^V + \cE_{52}^V + \hc$, with
\begin{equation*} %\label{eq:F323}
\begin{split}
\cE_{51}^V = \frac{1}{2}  \sum_{p \in \Lambda^*_+} \widehat{V} (p/N) \, (\g_p-1) \, d_p  b_{-p} \, , \qquad \cE_{52}^V = \frac{1}{2}  \sum_{p \in \Lambda^*_+} \widehat{V} (p/N) \,  d_p  b_{-p} 
\end{split} \end{equation*}
With $|\g_p - 1| \leq C \eta_p^2 \chi (p \in P_H)$, we obtain 
\[ |\langle \xi , \cE_{51}^V \xi \rangle | \leq C \sum_{p \in P_H} \eta_p^2 \| (\cN_+ + 1)^{1/2} \xi \| \| a_p \xi \| \leq C \ell^{5\a/2} \| (\cN_+ + 1)^{1/2} \xi \|^2 \]
Switching to position space, we find, by (\ref{eq:dxy-bds}), 
\[ \begin{split} |\langle \xi, \cE_{52}^V \xi \rangle | &\leq C \int_{\L^2}  dx dy \, N^3 V(N(x-y)) \| (\cN_+ + 1)^{1/2} \xi \| \| (\cN_+ +1)^{-1/2} \check{d}_x \check{a}_y \xi \| \\ &\leq C \| \eta_H \| \| (\cN_+ + 1)^{1/2} \xi \| \int_{\L^2}  dx dy \, N^2 V(N(x-y)) \left[ \| \check{a}_y \cN_+  \xi \| + \| \check{a}_x \check{a}_y \cN_+^{1/2} \xi \| \right] \\&\leq C \ell^{\a/2} \| (\cN_+ + 1)^{1/2} \xi \|^2 + C \ell^{\a/2} \| (\cN_+ + 1)^{1/2} \xi \| \| \cV_N^{1/2} \xi \| 
\end{split} \]
Hence, $\pm \text{F}_{323} \leq C \ell^{\a/2} (\cH_N + 1)$. 

To estimate the term $\text{F}_{324}$ in \eqref{eq:splitF32}, we use (\ref{eq:d-bds}) in Lemma \ref{lm:dp} and the estimate $\sum_{p \in \Lambda^*_+} \big|\widehat{V} (p/N) \big||\eta_p|\leq CN$; we find 
\begin{equation}\nonumber
\begin{split}
|\langle\xi, \text{F}_{324} \xi \rangle|  &\leq \; C\sum_{p \in P_H} \big|\widehat{V} (p/N) \big||\eta_p|\|(\cN_++1)^{1/2}\xi\|\|(\cN_++1)^{-1/2} d_p\, b^*_{p}\xi\|
\\ 
&\leq \; \frac{C}{N}  \sum_{p \in P_H} \big|\widehat{V} (p/N) \big||\eta_p| \|(\cN_++1)^{1/2}\xi \|  \\ &\hspace{1cm} \times 
\left[ |\eta_p| \| (\cN_+ + 1)^{3/2} \xi \| + \| \eta_H \| \| b_p b^*_p (\cN_+ + 1)^{1/2} \xi \| \right] \\ 
&\leq \; \frac{C}{N}   \sum_{p \in P_H} \big|\widehat{V} (p/N) \big||\eta_p|\|(\cN_++1)^{1/2}\xi \|  \\ &\hspace{1cm} \times  
\left[ |\eta_p| \| (\cN_+ + 1)^{3/2} \xi \| + \| \eta_H \| \| (\cN_+ + 1)^{1/2} \xi \| + \| \eta_H \| \| a_p (\cN_+ + 1) \xi \| \right] \\ &\leq \; C \ell^{\a/2}\|(\cN_++1)^{1/2}\xi\|^{2}\\ 
\end{split} \end{equation}
Combining the last bounds, we conclude that  
\[ \text{F}_{32} = \sum_{p \in P_H} \widehat{V} (p/N) \eta_p \left( \frac{N-\cN_+}{N} \right) \left( \frac{-\cN_+ -1}{N} \right) + \cE_6^V \]
with 
\begin{equation}\label{eq:F32} \pm \cE_6^V \leq C \ell^{\alpha/2} (\cH_N + 1) \end{equation}

To bound the last term $\text{F}_{33}$ in \eqref{eq:G23split}, we switch to position space. With Lemma \ref{lm:dp}, specifically (\ref{eq:ddxy}), and (\ref{eq:etax}), we obtain 
\[ \begin{split}  |\langle \xi , \text{F}_{33} \xi \rangle | &\leq C \| (\cN_+ + 1)^{1/2} \xi \| \int_{\L^2}  dx dy \, N^3 V(N(x-y)) \| (\cN_+ + 1)^{-1/2} \check{d}_x \check{d}_y \xi \| \\ &\leq C \| \eta_H \| \| (\cN_+ + 1)^{1/2} \xi \| \int_{\L^2}  dx dy \, N V(N(x-y)) \\ &\hspace{2cm} \times \left[ N \| (\cN_+ +1)^{3/2} \xi \| + \| \check{a}_x \cN_+^2 \xi \| + \| \check{a}_x \check{a}_y \cN^{3/2}_+ \xi \| \right] \\
&\leq C \ell^{\a/2} \| (\cN_+ + 1)^{1/2} \xi \|^2 + C \ell^{\a/2} \| \cN_+^{1/2} \xi \| \| \cV_N^{1/2} \xi \| \end{split} \]    
The last equation, combined with (\ref{eq:G23split}), (\ref{eq:fin-G23}), (\ref{eq:G23-deco1}) and 
(\ref{eq:F32}), implies that 
\[ \begin{split} 
\text{F}_3 = \; &\frac{1}{2} \sum_{p \in \L_+^*} \widehat{V} (p/N) (b_p b_{-p} + b^*_{-p} b^*_p) \\ &+ \sum_{p \in P_H} \widehat{V} (p/N) \eta_p \left( \frac{N-\cN_+}{N} \right) \left( \frac{N-\cN_+ - 1}{N} \right) + \cE_7^V \end{split} \]
with 
\[ \pm \cE_7^V \leq C \ell^{\a/2} (\cH_N + 1) \]
Together with (\ref{eq:F1}) and with (\ref{eq:fin-G22}), we obtain (\ref{eq:cEV-defc}) with (\ref{eq:errorVc}). Eq. (\ref{eq:errCommV}) follows similarly, arguing as we did at the end of the proof of Prop. \ref{prop:K} to show (\ref{eq:errCommK}). 
\end{proof}
 We conclude this section, summarizing the results of Prop. \ref{prop:K} and Prop. \ref{prop:G2V}.
 \begin{prop}\label{prop:G2}
 There exists a constant $C > 0$ such that 
 \begin{equation*} %\label{eq:def-EK}  
\begin{split}
\cG_{N,\ell}^{(2)} =&\; \cK + \sum_{p \in P_{H}} \Big[p^2 \eta_p^2 + \widehat{V} (p/N) \eta_p\Big]\Big(\frac{N-\cN_+}{N}\Big) \Big(\frac{N-\cN_+ -1}{N}\Big)\\
&+ \sum_{p \in P_{H}} p^2 \eta_p \big( b^*_p b^*_{-p} + b_p b_{-p} \big) +\sum_{p \in \Lambda^*_+}  \widehat{V} (p/N)  a^*_pa_p \frac{N-\cN_+}{N} \\ &+ \frac{1}{2}\sum_{p \in \Lambda^*_+} \widehat{V} (p/N) \big( b_p b_{-p}+ b_{-p}^* b_p^*\big) + \cE_{N,\ell}^{(2)} 
\end{split}
\end{equation*}
where  
\begin{equation*}%\label{eq:propG2c}
\pm \cE^{(2)}_{N,\ell} \leq C \ell^{(\a-3)/2} (\cH_N + 1)  \,  \end{equation*}
and 
\begin{equation*} \begin{split} %\label{eq:errComm2}
\pm \left[ f (\cN_+/M), \left[ f (\cN_+ /M) ,\cE^{(2)}_{N,\ell} \right] \right] & \leq C \ell^{(\a-3)/2} M^{-2} \|f'\|^2_{\infty}\,  \big( \cH_N + 1 \big)
\end{split}\end{equation*} 
for all $\alpha > 3$, $\ell \in (0;1/2)$ small enough, $f$ smooth and bounded, $M \in \bN$, $N \in \bN$ large enough. 
\end{prop}

\subsection{Analysis of $ \cG_{N,\ell}^{(3)}=e^{-B(\eta_H)}\cL^{(3)}_N e^{B(\eta_H)}$}
\label{sub:G3}

From (\ref{eq:cLNj}), we have
\be \label{eq:G3N2}
 \cG_{N,\ell}^{(3)} = \frac{1}{\sqrt{N}} \sum_{p,q \in \L^*_+ : p + q \not = 0} \widehat{V} (p/N) e^{-B(\eta_H)} b^*_{p+q} a^*_{-p} a_q e^{B(\eta_H)} + \hc  \ee
 
\begin{prop}\label{prop:GN-3} 
There exists a constant $C > 0$ such that 
\begin{equation}\label{eq:def-E3}
\cG_{N,\ell}^{(3)} = \frac{1}{\sqrt{N}} \sum_{p,q \in \L^*_+ : p + q \not = 0} \widehat{V} (p/N) \left[ b_{p+q}^* a_{-p}^* a_q  + \hc \right] + \cE^{(3)}_{N,\ell} \end{equation}
where  
\begin{equation}\label{eq:lm-GN31} \pm \cE_{N,\ell}^{(3)} \leq C \ell^{\a/2} \big(\cH_N + 1\big)  \end{equation}
and 
\begin{equation}\label{eq:E3Cff}
\pm [f (\cN_+/M) , [ f (\cN_+ / M) , \cE_{N,\ell}^{(3)}]] \leq C M^{-2} \|f'\|^2_{\infty}\ell^{\a/2} \big(\cH_N + 1\big) 
\end{equation}
for all $\a > 0$, $\ell \in (0;1/2)$ small enough, $f$ smooth and bounded, $M \in \bN$, $N \in \bN$ large enough.
\end{prop} 

\begin{proof}[Proof of Proposition \ref{prop:GN-3}] We start by writing 
\[ \begin{split} e^{-B(\eta_H)} a_{-p}^* a_q e^{B(\eta_H)} &= a_{-p}^* a_q + \int_0^1 ds \, e^{-sB(\eta_H)} [a_{-p}^* a_q , B(\eta_H)] e^{sB(\eta_H)} \\ &= a_{-p}^* a_q + \int_0^1 ds e^{-sB(\eta_H)} ( \eta_H (p) b_q b_p + \eta_H (q) b_{-p}^* b^*_{-q} ) e^{s B(\eta_H)} \end{split} \]
{F}rom (\ref{eq:G3N2}), we find
\begin{equation*}%\label{eq:GN3-in} 
\begin{split} \cG_{N,\ell}^{(3)} &=  \frac{1}{\sqrt{N}} \sum_{p,q \in \Lambda^*_+ : p+q \not = 0} \widehat{V} (p/N) e^{-B(\eta_H)}  b^*_{p+q} e^{B(\eta_H )}\,a_{-p}^* a_q \\
&\hspace{.2cm}+ \frac{1}{\sqrt{N}} \sum_{p,q \in \Lambda_+^* , p+q \not = 0} \widehat{V} (p/N) \eta_H (p) \, e^{-B(\eta_H)}b^*_{p+q}e^{B(\eta_H)} \int_0^1 ds\, e^{-sB(\eta_H)} b_{p} b_{q} e^{sB(\eta_H)} \\
&\hspace{.2cm}+ \frac{1}{\sqrt{N}} \sum_{p,q \in \Lambda_+^* , p+q \not = 0} \widehat{V} (p/N) \eta_H (q) \, e^{-B(\eta_H)} b^*_{p+q}e^{B(\eta_H)} \int_0^1 ds\,  e^{-sB(\eta_H)}b_{-p}^*b^*_{-q} e^{sB(\eta_H)}
\\ &\hspace{.2cm}+\text{h.c.} 
\end{split}
\end{equation*}
Using \eqref{eq:ebe} we arrive at (\ref{eq:def-E3}), with 
\begin{equation} \label{eq:deco-cE3} \begin{split} 
\cE^{(3)}_{N,\ell} &=  \frac{1}{\sqrt{N}} \sum_{p,q \in \Lambda^*_+ : p+q \not = 0} \widehat{V} (p/N) \big((\g_{p+q}-1) b^*_{p+q} + \s_{p+q} b_{-p-q} + d_{p+q}^* \big) \, a_{-p}^* a_q  \\
&\hspace{.3cm}+ \frac{1}{\sqrt{N}} \sum_{p,q \in \Lambda_+^* , p+q \not = 0} \widehat{V} (p/N) \eta_H (p) \, e^{-B(\eta_H)}b^*_{p+q}e^{B(\eta_H)} \int_0^1 ds\, e^{-sB(\eta_H)} b_{p} b_{q} e^{sB(\eta_H)}\\
&\hspace{.3cm}+ \frac{1}{\sqrt{N}} \sum_{p,q \in \Lambda_+^* , p+q \not = 0} \widehat{V} (p/N) \eta_H (q) \, e^{-B(\eta_H)} b^*_{p+q}e^{B(\eta_H)} \int_0^1 ds\,  e^{-sB(\eta_H)}b_{-p}^*b^*_{-q} e^{sB(\eta_H)}  \\
&\hspace{.3cm}+ \hc \\ &=:  \; \cE^{(3)}_1 + \cE_2^{(3)} + \cE_3^{(3)} + \hc 
\end{split} \end{equation}
where we defined $\gamma_p = \cosh \eta_H (p)$, $\s_p = \sinh \eta_H (p)$ and where the operator $d_p$ is defined as in (\ref{eq:defD}), with $\eta$ replaced by $\eta_H$. To complete the proof of the proposition, we have to show that the three error terms $\cE_1^{(3)}, \cE_2^{(3)}, \cE_3^{(3)}$ all satisfy the bounds (\ref{eq:lm-GN31}), (\ref{eq:E3Cff}). We start by considering $\cE_1^{(3)}$. We decompose it as
\[ \begin{split}
\cE^{(3)}_1 =\; &  \frac{1}{\sqrt{N}} \sum_{p,q \in \Lambda^*_+ : p+q \not = 0} \widehat{V} (p/N) \big((\g_{p+q}-1) b^*_{p+q} + \s_{p+q} b_{-p-q} + d_{p+q}^* \big) \, a_{-p}^* a_q \\
=: & \; \cE^{(3)}_{11} + \cE^{(3)}_{12} +\cE^{(3)}_{13} 
\end{split}\]
Since $|\g_{p+q}-1|\leq |\eta_H (p+q)|^2$ and $\| \eta_H \| \leq C \ell^{\a/2}$, we have
\be \begin{split} \label{eq:GN3-P11}
|\langle \xi, \cE^{(3)}_{11} \xi \rangle| 
& \leq  \frac{C}{\sqrt{N}} \sum_{p,q \in \Lambda^*_+ : p+q \not = 0} |\widehat{V} (p/N)| | \eta_H (p+q)|^2 \, \| b_{p+q} a_{-p} \xi \| \| a_q \xi \| \\
& \leq   \frac{C}{\sqrt{N}} \Big[\sum_{p,q \in \Lambda^*_+ : p+q \not = 0}   |\eta_H (p+q)|^2 \, \| a_{-p}(\cN_++1)^{1/2} \xi \|^2 \Big]^{1/2}\\ &\hspace{4cm} \times  \Big[\sum_{p,q \in \Lambda^*_+ : p+q \not = 0} |\eta_H (p+q)|^2   \| a_q \xi \|^2 \Big]^{1/2}\\
& \leq C \| \eta_H \|^2 \| (\cN_++1)^{1/2} \xi \|^2 \leq C \ell^\a \| (\cN_+ + 1)^{1/2} \xi \|^2 
\end{split}\ee
To bound  $\cE^{(3)}_{12}$ we move $a^*_{-p}$ to the left  of $b_{-p-q}$ (using $[a_{-p-q}, a_{-p}^*] = 0$, since $q \not = 0$). With $|\s_{p+q}| \leq C |\eta_H (p+q)|$, we obtain 
\be \begin{split} \label{eq:GN3-P12}
|\langle \xi, \cE^{(3)}_{12} \xi \rangle| & \leq  \frac{C}{\sqrt{N}} \sum_{p,q \in \Lambda^*_+ : p+q \not = 0} |
\widehat{V} (p/N)| | \eta_H (p+q)| \, \| a_{-p} \xi \| \| a_q b_{-p-q} \xi \| \\
& \leq   \frac{C}{\sqrt{N}} \Big[\sum_{p,q \in \Lambda^*_+ : p+q \not = 0}   |\eta_H (p+q)|^2 \, \| a_{-p} \xi \|^2 \Big]^{1/2} \Big[\sum_{p,q \in \Lambda^*_+ : p+q \not = 0}  \| a_q b_{-p-q} \xi \|^2 \Big]^{1/2}\\
& \leq C \ell^{\a/2} \| (\cN_++1)^{1/2} \xi \|^2
\end{split}\ee
In $ \cE^{(3)}_{13} $, on the other hand, we write $d^*_{p+q}= \bar{d}^*_{p+q} - \frac{(\cN_++1)}{N} \eta_H (p+q) b_{-p-q}$. We obtain $\cE^{(3)}_{13} = \cE^{(3)}_{131} + \cE^{(3)}_{132}$, with 
\[ \begin{split} %\label{eq:G3-P13}
\cE^{(3)}_{131} =\; & \frac{1}{\sqrt{N}} \sum_{p,q \in \Lambda^*_+ : p+q \not = 0} \widehat{V} (p/N) \, \bar d^*_{p+q}   a^*_{-p} a_q \\
\cE^{(3)}_{132} = \; & -\frac{(\cN_++1)}{N}  \frac{1}{\sqrt{N}} \sum_{p,q \in \Lambda^*_+ : p+q \not = 0} \widehat{V} (p/N) \eta_H (p+q) \, b_{-p-q} a^*_{-p} a_q 
\end{split}\]
The term $\cE^{(3)}_{132}$ can be bounded like $\cE_{12}^{(3)}$, commuting $a_{-p}^*$ to the left of $b_{-p-q}$; we find $\pm \cE^{(3)}_{132} \leq C \ell^{\a/2} (\cN_+ + 1)$. As for the term $\cE^{(3)}_{131}$, we switch to position space:
\[ \begin{split}
\cE^{(3)}_{131}& =\frac{1}{\sqrt{N}} \sum_{p,q \in \Lambda^*_+ : p+q \not = 0} \widehat{V} (p/N)\, \bar d^*_{p+q}   a^*_{-p} a_q =  \int_{\L^2}  dx dy N^{5/2} V(N(x-y)) \check{\bar d}^*_x \check{a}^*_y \check{a}_x
\end{split}\]
With \eqref{eq:splitdbd}, we bound
\[\begin{split} %\label{eq:GN3-P13tl}
|\langle \xi, \cE^{(3)}_{131} \xi \rangle|  & \leq 
 \int_{\L^2}  dx dy N^{5/2} V(N(x-y))  \| \check{a}_x \xi \| \|\check{a}_y \check{\bar d}_x  \xi \|  \\
 & \leq  C \| \eta_H \| \int_{\L^2}  dx dy N^{5/2} V(N(x-y))\| \check{a}_x \xi \| \\ & \hspace{1cm} \times  \big[\|(\cN_++1)\xi \| + \| \check{a}_x (\cN_+ + 1)^{1/2} \xi \| + \| \check{a}_y (\cN_+ + 1)^{1/2} \xi \| + \| \check{a}_x \check{a}_y \xi \| \big]  \\
 & \leq  C \| \eta_H \| \| \cN_+^{1/2} \xi \| \Big[  \int_{\L^2}  dx dy N^{2} V(N(x-y)) \\ & \hspace{3.5cm} \times  \big[\|(\cN_++1)\xi \|^2 + \| \check{a}_x (\cN_+ + 1)^{1/2} \xi \|^2 + \| \check{a}_x \check{a}_y \xi \|^2 \big] \Big]^{1/2}   \\
 & \leq C \| \eta_H \| \| \cN_+^{1/2} \xi \| \big[ \| (\cN_+ + 1)^{1/2} \xi \| + \| \cV_N^{1/2} \xi \| \big] \\
 & \leq C \ell^{\a/2} \big[ \| (\cN_+ + 1)^{1/2} \xi \|^2 + \| \cV_N^{1/2} \xi \|^2 \big] 
 \end{split}\]
With \eqref{eq:GN3-P11} and \eqref{eq:GN3-P12} we conclude that 
\be \label{eq:G3N-P1end}
\pm \cE^{(3)}_1 \leq C \ell^{\a /2} (\cV_N + \cN_+ + 1) \leq C \ell^{\a/2} (\cH_N + 1) 
\ee

Next, we consider the term $\cE^{(3)}_2$, defined in (\ref{eq:deco-cE3}). Using Eq. \eqref{eq:ebe} we rewrite
\begin{equation}\label{eq:cE32-deco} \begin{split}
\cE^{(3)}_2 =\; & \frac{1}{\sqrt{N}} \sum_{p,q \in \Lambda_+^* , p+q \not = 0} \widehat{V} (p/N) \eta_H (p) \, e^{-B( \eta_H )}b^*_{p+q}e^{B( \eta_H)}\\
& \hskip 1cm \times \int_0^1 ds\, \big( \g^{(s)}_p  \g^{(s)}_{q} b_{p} b_q + \s^{(s)}_p  \s^{(s)}_{q}b^*_{-p} b^*_{-q} +\g^{(s)}_p   \s^{(s)}_{q} b^*_{-q} b_{p} +  \s^{(s)}_p  \g^{(s)}_{q}  b^*_{-p} b_{q}  \big)  \\
& +\frac{1}{\sqrt{N}} \sum_{p,q \in \Lambda_+^* , p+q \not = 0} \widehat{V} (p/N) \eta_H (p) \, e^{-B( \eta_H )}b^*_{p+q}e^{B( \eta_H)} \int_0^1 ds\,  \g^{(s)}_p  \s^{(s)}_{q} [b_{p},  b^*_{-q}]  \\
& + \frac{1}{\sqrt{N}} \sum_{p,q \in \Lambda_+^* , p+q \not = 0} \widehat{V} (p/N)  \eta_H (p) \, e^{-B( \eta_H )}b^*_{p+q}e^{B( \eta_H )}\\
& \hskip 1cm\times \int_0^1 ds\, \Big[ d^{(s)}_p \big( \g^{(s)}_{q} b_{q} + \s^{(s)}_{q} b^*_{-q} \big) + \big( \g^{(s)}_p b_{p} + \s^{(s)}_p b^*_{-p} \big)   d^{(s)}_{q} +   d^{(s)}_{p}  d^{(s)}_{q}   \Big]\\
=: &\; \cE^{(3)}_{21} + \cE^{(3)}_{22} + \cE^{(3)}_{23}
\end{split}\end{equation}
where, for any $s \in [0;1]$ and $p \in \L_+^*$, $\gamma^{(s)}_p = \cosh (s \eta_H (p))$, $\s^{(s)}_p = \sinh (s \eta_H (p))$ and $d^{(s)}_p$ is the operator defined as in (\ref{eq:defD}), with $\eta$ replaced by $s \eta_H$. We have 
\be \begin{split} \label{eq:G3N-P31} 
|\langle \xi , \cE^{(3)}_{21} \xi \rangle | \leq \;  &\frac{C}{\sqrt{N}} \sum_{\substack{p,q \in \Lambda^*_+ : p \not =-q}} | \eta_H (p)|  \| b_{p+q}e^{B( \eta_H )}\xi \| \Big[ \| b_p b_q \xi \|  + | \eta_H (p)| \| b_q (\cN_++1)^{1/2}\xi \|  \\ & \hspace{2.5cm} + | \eta_H (q)| \| b_p (\cN_++1)^{1/2}\xi \|  + | \eta_H (p)| | \eta_H (q)| \| (\cN_++1)\xi\|  \Big] \\
 \leq \; &C \| \eta_H \| \| (\cN_++1)^{1/2}\xi \|^2  \leq C \ell^{\a/2} \| (\cN_++1)^{1/2}\xi \|^2 
\end{split}\ee
Since $[b_{p},b^*_{-q}] = - a^*_{-q} a_{p} /N $ for all $p \neq -q$, we find 
\be \begin{split}\label{eq:G3N-P32} 
|\langle \xi , \cE^{(3)}_{22} \xi \rangle | &\leq \frac{C}{N^{3/2}} \sum_{p,q \in \Lambda_+^* , p+q \not = 0} | \eta_H (p) |  | \eta_H (q) |  \| b_{p+q} e^{B( \eta_H )}\xi \| \|  a_{p} (\cN_++1)^{1/2}  \xi \| \\
& \leq \frac{C}{N}  \| \eta_H \|^2  \|(\cN_++1)^{1/2} \xi \|^2 \leq \frac{C \ell^\alpha}{N} \| (\cN_+ + 1)^{1/2} \xi \|^2 
\end{split}\ee
To bound the third term on the r.h.s. of (\ref{eq:cE32-deco}), we switch to position space. We obtain 
\[ \begin{split}
\cE^{(3)}_{23} =\;&  \int_{\L^3} dx dy dz  N^{5/2} V(N(x-z))  \check{\eta}_H (z-y) \,e^{-B( \eta_H )} \check{b}^*_x e^{B( \eta_H )}\\
& \hskip 1cm\times \int_0^1 ds\, \Big[  \check{d}^{(s)}_y \big( b(\check{\g}^{(s)}_{x}) +b^*(\check{\s}^{(s)}_{x}) \big) +  \big( b(\check{ \g}^{(s)}_{y}) +b^*(\check{\s}^{(s)}_{y}) \big) \check{ d}^{(s)}_x    +  \check{ d}^{(s)}_y \check{ d}^{(s)}_x  \Big]
\end{split}\]
Using the bounds \eqref{eq:dxy-bds}, \eqref{eq:splitdbd}, \eqref{eq:ddxy} and Lemma \ref{lm:Ngrow} we arrive at 
\begin{equation*}\begin{split} 
|\langle \xi , \cE^{(3)}_{23} \xi \rangle | \leq \; & C \| \eta_H \|  \int_{\Lambda^3} dx dy dz \, N^{5/2} V(N(x-z))   | \check{\eta}_H (y-z)|\, \| \check{b}_x e^{B( \eta_H)}  \xi \| \\ & \times \, \Big[ \| \check{b}_{x} \check{b}_{y} \xi \| + \| (\cN_++1) \xi \| + \| \check{b}_x (\cN_++1)^{1/2} \xi \| + \| \check{b}_y (\cN_++1)^{1/2} \xi \| \Big] \\
\leq \; & \frac{C \| \eta_H \|^2}{\sqrt{N}} \|  \cN_+^{1/2} e^{B (\eta_H)} \xi \| \| (\cN_++ 1) \xi \| \\
\leq \; &C \ell^\a \| (\cN_+ + 1)^{1/2} \xi \|^2 
\end{split} \end{equation*}
Combined with (\ref{eq:G3N-P31}) and (\ref{eq:G3N-P32}), the last bound implies that 
\begin{equation}\label{eq:cE32f} \pm \cE_2^{(3)} \leq C \ell^{\a/2} (\cN_+ + 1) \end{equation}
Finally, we consider the last term on the r.h.s. of (\ref{eq:deco-cE3}). In fact, it is convenient to bound (in absolute value) the expectation of its adjoint, which we decompose as  
\begin{equation*} \begin{split} %\label{eq:G3N-P2}  
\cE^{(3)*}_3 =\; & \frac{1}{\sqrt{N}} \sum_{p,q \in \Lambda_+^* , p+q \not = 0} \widehat{V} (p/N)  \eta_H (q) \int_0^1 ds\, e^{-sB(\eta_H)} b_{-q}e^{sB( \eta_H )}\\
& \hskip 1cm \times  \big( \g^{(s)}_p b_{-p} + \s^{(s)}_p b^*_{p} +  d^{(s)}_{-p}\big) \big( \g_{p+q} b_{p+q} + \s_{p+q} b^*_{-p-q} +  d_{p+q}\big) \\
 =\; &\frac{1}{\sqrt{N}} \sum_{p,q \in \Lambda_+^* , p+q \not = 0} \widehat{V} (p/N)  \eta_H (q) \, \int_0^1 ds\,e^{-sB(\eta_H)} b_{-q}e^{sB( \eta_H)}\\
&  \times \Big[ \, \g^{(s)}_p  \g_{p+q} b_{-p} b_{p+q} + \s^{(s)}_p  \s_{p+q} b^*_{p} b^*_{-p-q}  + \g^{(s)}_p  \s_{p+q} b^*_{-p-q} b_{-p}+  \g_{p+q}  \s^{(s)}_p b^*_{p} b_{p+q} \\
& \hskip 1.5cm +  d^{(s)}_{-p} \big( \g_{p+q} b_{p+q} + \s_{p+q} b^*_{-p-q}\big) +  \big( \g^{(s)}_p b_{-p} + \s^{(s)}_p b^*_{p}\big)  d_{p+q} +   d^{(s)}_{-p}  d_{p+q}\Big]  \\
& +\frac{1}{\sqrt{N}} \sum_{p,q \in \Lambda_+^* , p+q \not = 0} \widehat{V} (p/N)  \eta_H (q) \, \int_0^1 ds\,e^{-sB( \eta_H )} b_{-q}e^{sB( \eta_H)}   \g^{(s)}_p  \s_{p+q} [b_{-p},b^*_{-p-q}] 
\\ =: & \, \cE_{31}^{(3)} + \cE_{32}^{(3)} 
\end{split} \end{equation*}
Using that $q \neq 0$ and thus that $[b_{-p},b^*_{-p-q}] = - a^*_{-p-q} a_{-p} /N $, we can estimate the second term by 
\begin{equation}\label{eq:cE332f}\begin{split}
| \langle \xi,  &\cE_{32}^{(3)}  \xi\rangle| \\ & \leq  \frac{C}{N^{3/2}}  \int_0^1 ds \sum_{p,q \in \Lambda_+^* , p+q \not = 0} | \eta_H (q) | | \eta_H (p+q)| \, \| a_{-p-q}\,e^{-sB( \eta_H )} b^*_{-q}e^{sB( \eta_H )} \xi \| \|  a_{-p} \xi \| \\
& \leq  \frac{C}{N^{3/2}} \int_0^1 ds  \Big[ \sum_{\substack{p,q \in \Lambda_+^* \\ p+q \not = 0}} | \eta_H (q) |^2  \, \| a_{-p-q}\,e^{-sB( \eta_H )} b^*_{-q}e^{sB( \eta_H )} \xi \|^2 \Big]^{1/2} \\ &\hspace{6cm} \times \Big[ \sum_{\substack{p,q \in \Lambda_+^*\\ p+q \not = 0}}  | \eta_H (p+q)|^2  \|  a_{-p} \xi \|^2 \Big]^{1/2}  \\
& \leq \frac{C}{N}  \| \eta_H\|^2  \|(\cN_++1)^{1/2} \xi \|^2 \leq \frac{C\ell^\a}{N} \| (\cN_+ + 1)^{1/2} \xi \|^2 
\end{split}\end{equation}
To bound the expectation of $\cE^{(3)}_{31}$, it is convenient  to switch to position space. We find 
\[\begin{split}
 \cE_{31}^{(3)} &= \int_0^1 ds \int_{\Lambda^2} dx dy \, N^{5/2} V(N(x-y)) \, e^{-sB( \eta_H)} b(\check{\eta}_{H,x}) e^{sB( \eta_H)} \\
& \hskip 1cm \times \, \Big[  b(\check{ \g}^{(s)}_x) b(\check{ \g}_y) + b^*(\check{ \s}^{(s)}_x) b^*(\check{ \s}_y) +b^*(\check{ \s}_y) b(\check{ \g}^{(s)}_x)+ b^*(\check{ \s}^{(s)}_x) b(\check{ \g}_y)  \\
& \hskip 3 cm + \check{ d}_x^{(s)} \big(  b(\check{ \g}_y) + b^*(\check{ \s}_y) \big) + \big(  b(\check{ \g}^{(s)}_x) + b^*(\check{ \s}^{(s)}_x) \big) \check{ d}_y +\check{ d}_x^{(s)}\check{ d}_y\Big]
\end{split} \]
where we used the notation $\check{\eta}_H$, $\check{\g}^{(s)}$ and $\check{\s}^{(s)}$ to indicate the functions on $\Lambda$ with Fourier coefficients $\eta_H (p)$, $\cosh (s \eta_H (p))$ and, respectively, $\sinh (s \eta_H (p))$, and where $\check{\eta}_{H,x}$, $\check{\g}_{x}$ and $\check{\s}_{x}$ denote the functions defined by $\check{\eta}_{H,x} (z) = \check{\eta}_H (z-x)$, $\check{\g}_{x} (z) = \check{\g}(z-x)$ and $\check{\s}_{x} (z) = \check{\s}^s(z-x)$. Using \eqref{eq:dxy-bds}, \eqref{eq:splitdbd}, \eqref{eq:ddxy} and the bound (\ref{eq:etax}), we find, for $N$ large enough,  
\begin{equation*} \begin{split}  \label{eq:GN3-P2} 
|\langle \xi , \cE^{(3)}_{31} \xi \rangle | \leq \; & \int_0^1  ds\,  \int_{\Lambda^2} dx dy \, N^{5/2} V(N(x-y))\|  b^*(\check{\eta}_{H,x}) e^{sB( \eta_H)}  \xi \| \\ & \times \, \Big[  \| \check{b}_{x}\check{b}_{y} \xi \| +  \| \check{b}_x (\cN_++1)^{1/2} \xi \| + \| \check{b}_y (\cN_++1)^{1/2} \xi \| + \| (\cN_++1) \xi \| \Big]
\end{split} \end{equation*}
With Lemma \ref{lm:Ngrow}, we estimate
\begin{equation*} %\label{eq:GN3-P2-2}
\|  b^*(\check{\eta}_{H,x}) e^{sB( \eta_H)}  \xi \| \leq C \|  \eta_H \| \| (\cN_+ +1)^{1/2} \xi \| \end{equation*}
We conclude that 
\begin{equation*} \begin{split}  %\label{eq:GN3-P2-3} 
& |\langle \xi , \cE_{31}^{(3)} \xi \rangle | \leq  C \ell^{\a/2} \Big[ \| (\cN_+ +1)^{1/2} \xi \|^2 + \| \cV_N^{1/2}\xi \|^2 \Big]
\end{split} \end{equation*}
From (\ref{eq:cE332f}), we find 
\[ \pm \cE^{(3)}_3 \leq C \ell^{\a/2} (\cH_N + 1) \]
and thus, combining this bound with (\ref{eq:deco-cE3}), (\ref{eq:G3N-P1end}) and (\ref{eq:cE32f}), we arrive at 
\[ \pm \cE^{(3)}_{N,\ell} \leq C  \ell^{\a/2} (\cH_N + 1) \]
This proves (\ref{eq:lm-GN31}). The bound (\ref{eq:E3Cff}) follows similarly, arguing as we did at the end of the proof of Prop. \ref{prop:K} to show (\ref{eq:errCommK}). 
\end{proof}

\subsection{Analysis of $ \cG_{N,\ell}^{(4)}=e^{-B(\eta_H)}\cL^{(4)}_N e^{B(\eta_H)}$}
\label{sub:G4}

With $\cL^{(4)}_N$ as defined in (\ref{eq:cLNj}), we write
\begin{equation*} \begin{split} %\label{eq:def-E4}
\cG_{N,\ell}^{(4)} = \; &e^{-B(\eta_H)} \cL^{(4)}_{N} e^{B(\eta_H)} \\ = \; &\cV_N + \frac{1}{2N} \sum_{\substack{q \in \Lambda^*_+, r\in \Lambda^* \\ q,\, q+ r \in P_H}} \widehat{V} (r/N) \eta_{q+r} \eta_q \left( 1-\frac{\cN_+ }{N} \right) \left( 1 - \frac{\cN_+ +1}{N} \right)  \\ &+ \frac{1}{2N} \sum_{\substack{q \in \Lambda^*_+, r \in \L^*: \\ q +r \in P_H}} \widehat{V} (r/N) \, \eta_{q+r} \left(  b_q b_{-q} + b^*_q b^*_{-q} \right)  + \cE^{(4)}_{N,\ell} \end{split} \end{equation*}
\begin{prop}\label{prop:GN-4} 
There exists a constant $C > 0$ such that 
\begin{equation}\label{eq:E4bound2}
  \pm \cE_{N,\ell}^{(4)} \leq C \ell^{\a/2} \big( \cH_N +1 \big) 
\end{equation}
and 
\begin{equation}\label{eq:E4Cff}
\pm [f (\cN_+/M), [f (\cN_+/M),\cE_{N,\ell}^{(4)}]] \leq C M^{-2} \|f'\|^2_{\infty}\ell^{\a/2} \big(\cH_N + 1\big) 
\end{equation}
for all $\a > 0$, $\ell \in (0;1/2)$ small enough, $f$ smooth and bounded, $M \in \bN$, $N \in \bN$ large enough.  
\end{prop}

The following lemma will be useful in the proof of Prop. \ref{prop:GN-4}.
\begin{lemma}\label{lm:prel4}
Let $\eta_H \in \ell^2 (\L^*)$, as defined in (\ref{eq:defetaH}). Then there 
exists a constant $C > 0$ such that 
\be \begin{split} \label{eq:prel4-2}
&\| (\cN_++1)^{n/2} e^{-B(\eta_H)} \check{b}_x \check{b}_y  e^{B( \eta_H)}\xi \| \\
&\leq  C  \Big [\; \| (\cN_+ +1)^{(n+2)/2} \xi \| + N \| (\cN_+ +1)^{n/2} \xi \| \\
&\hskip 1cm + \| \check{a}_y (\cN_+ +1)^{(n+1)/2} \xi \| + \| \check{a}_x (\cN_+ +1)^{(n+1)/2} \xi \|   +  \| \check{a}_x \check{a}_y (\cN_++1)^{n/2} \xi \| \Big ]
\end{split}\ee
for all $\xi \in \cF_+^{\leq N}$, $n \in \bZ$. \end{lemma}

\begin{proof} 
We consider $n=0$, the general case follows similarly. With the notation $\g_p = \cosh \eta_H (p)$, $r_p = 1 - \g_p$, $\s_p = \sinh \eta_H (p)$ and denoting by $\check{\s}$, $\check{r}$ the functions in $L^2 (\L)$ with Fourier coefficients $\s_p$ and $r_p$, we use (\ref{eq:ebe}) to write 
\[\begin{split} 
 \| e^{-B(\eta)} \check{b}_x \check{b}_y  e^{B( \eta)}\xi \|   = \; &  \| \big( \check{b}_x + b (\check{r}_x) + b^*(\check{\s}_x) + \check{d}_x \big) \big( \check{b}_y + b(\check{r}_y) + b^*(\check{\s}_y) + \check{d}_y\big) \xi \| \\
\leq \; &\| \check{b}_x \check{b}_y \xi \| + C (\| \check{b}_x  \cN_+^{1/2}  \xi \|  + \| \check{b}_y \cN_+^{1/2} \xi \|) + C |\check{\s} (x-y)| \| \xi \| \\ &+ \| \check{b}_x \check{d}_y \xi \| +  \| \check{d}_x  \big( \check{b}_y + b(\check{r}_y) + b^*(\check{\s}_y) + \check{d}_y\big) \xi \| 
\end{split} \]
because $\| r \|, \| \s \| \leq C \| \eta_H \| \leq C$. Using Eq. (\ref{eq:ddxy}) and (after writing $\check{b}_x \check{d}_y = \check{b}_x \check{\bar{d}}_y - \check{b}_x (\cN_+/ N) b^* (\check{\eta}_y)$) Eq. (\ref{eq:splitdbd}), and with the bound (\ref{eq:etax}) (which also implies $|\check{\sigma} (x)| \leq C N$), we obtain (\ref{eq:prel4-2}). 
\end{proof}

\begin{proof}[Proof of Prop. \ref{prop:GN-4}] We start by writing
\begin{equation}\label{eq:quartic1} \begin{split} 
&e^{-B(\eta_H)} \cL^{(4)}_{N} e^{B(\eta_H)} \\ &= \frac{1}{2N} \sum_{p,q \in \Lambda_+^*, r \in \Lambda^* : r \not = -p,q} \widehat{V} (r/N) e^{-B(\eta_H)} a_p^* a_q^* a_{q-r} a_{p+r}  e^{B( \eta_H)} \\ &= \cV_N + \frac{1}{2N} \sum_{p,q \in \Lambda_+^*, r \in \Lambda^* : r \not = -p,q} \widehat{V} (r/N) \int_0^1 ds \,  e^{-sB( \eta_H)} \left[ a_p^* a_q^* a_{q-r} a_{p+r} , B( \eta_H) \right] e^{sB(\eta_H)}
 \\ & = \cV_N + \frac{1}{2N} \sum_{q \in \Lambda_+^*, r \in \Lambda^* : r \not = -q} \widehat{V} (r/N) \eta_H (q+r) \int_0^1 ds \, \left( e^{-sB(\eta_H)} b_q^* b_{-q}^* e^{sB(\eta_H)} + \text{h.c.} \right)  \\
&\hspace{.3cm}  +\frac{1}{N} \sum_{p,q \in \Lambda_+^* , r \in \Lambda^* : r \not = p,-q} \widehat{V} (r/N) \eta_H (q+r) \int_0^1 ds \left( e^{-s B(\eta_H)} b_{p+r}^* b_q^* a^*_{-q-r} a_p e^{sB(\eta_H)} + \text{h.c.} \right) \end{split} \end{equation}
Now we observe that
\[ \begin{split}  e^{-sB( \eta_H )} &a^*_{-q-r} a_p e^{s B(\eta_H)}  \\ &=
a^*_{-q-r} a_p + \int_0^s d\tau \, e^{-\tau B(\eta_H)} \left[ a^*_{-q-r} a_p , B( \eta_H) \right] e^{-\tau B(\eta_H)} \\ &= a^*_{-q-r} a_p + \int_0^s d\tau \, e^{-\tau B(\eta_H)} \left( \eta_H (p) b^*_{-p} b^*_{-q-r} + \eta_H ({q+r}) b_p b_{q+r} \right) e^{-\tau B(\eta_H)} \end{split}  \]
Inserting in (\ref{eq:quartic1}), we obtain 
\[  \cG^{(4)}_{N,\ell}  -\cV_N = \text{W}_1 + \text{W}_2 + \text{W}_3 + \text{W}_4 \]
where we defined 
\begin{equation}\label{eq:defW} 
\begin{split} 
\text{W}_1 = \; & \frac{1}{2N} \sum_{q \in \Lambda_+^* , r \in \Lambda^* : r \not = -q} \widehat{V} (r/N) \eta_H ({q+r}) \int_0^1 ds \big( e^{-s B(\eta_H)} b_{q} b_{-q}\,e^{sB(\eta_H)}   + \text{h.c.} \big) \\ 
  \text{W}_2 = \; & \frac{1}{N} \sum_{p,q \in \Lambda_+^* , r \in \Lambda^* : r \not = p,-q} \widehat{V} (r/N)\,  \eta_H ({q+r}) \int_0^1 ds  \big(  e^{-s B(\eta_H)} b^*_q b^*_{-q} e^{s B( \eta_H )}  a^*_{-q-r} a_p + \text{h.c.} \big)\\
\text{W}_3 = \; &   \frac{1}{N} \sum_{p,q\in \Lambda^*_+, r \in \Lambda^* : r \not = -p -q} \widehat{V} (r/N) \eta_H ({q+r}) \eta_H (p) \, \\ &\hspace{1.5cm} \times  \int_0^1 ds\,\int_0^s d\t \, \big(e^{-sB( \eta_H)} b^*_{p+r} b^*_q e^{sB( \eta_H)} e^{-\t B( \eta_H)} b^*_{-p} b^*_{-q-r} e^{\t B( \eta_H )}+ \hc \big) \\ 
\text{W}_4 = \; & \frac{1}{N} \sum_{p,q\in \Lambda^*_+, r \in \Lambda^* : r \not = -p -q} \widehat{V} (r/N) \, \eta^2_H ({q+r}) \\ &\hspace{1.5cm} \times \int_0^1 ds\,\int_0^s d\t \,  \big( e^{-sB( \eta_H)} b^*_{p+r} b^*_q e^{sB( \eta_H)} e^{-\t B( \eta_H)} b_{p} b_{q+r} e^{\t B( \eta_H)} + \hc \big) 
  \end{split} \end{equation}
  
First, we consider the term $\text{W}_1$. With (\ref{eq:ebe}), we find
\[ \begin{split}  \text{W}_1 = &\;\frac{1}{2N} \sum_{q \in \Lambda_+^* , r \in \Lambda^* : r \not = -q} \widehat{V} (r/N) \eta_H ({q+r}) \\ &\hspace{2cm} \times \int_0^1 ds \big( \g_q^{(s)} b_q + \s^{(s)}_q b_{-q}^* + d_q^{(s)} \big) \big(  \g_q^{(s)} b_{-q} + \s^{(s)}_q b_{q}^* + d_{-q}^{(s)} \big) + \hc \end{split}  \]
where we defined $\gamma^{(s)}_q = \cosh (s \eta_H (q))$, $\s_q^{(s)} = \sinh (s \eta_H (q))$ and where $d_q^{(s)}$ is defined as in (\ref{eq:defD}), with $\eta$ replaced by $s \eta_H$. We write 
\begin{equation}\label{eq:W1-dec} 
\begin{split} 
\text{W}_1 = \; & \frac{1}{2N} \sum_{q \in \Lambda_+^* , r \in \Lambda^* : r \not = -q} \widehat{V} (r/N) \eta_H (q+r) \int_0^1 ds (\g^{(s)}_q)^2  (b_q b_{-q} + \hc ) \\ &+ \frac{1}{2N} \sum_{q \in \Lambda_+^* , r \in \Lambda^* : r \not = -q} \widehat{V} (r/N) \eta_H (q+r) \int_0^1 ds \, \g_q^{(s)} \s^{(s)}_q \big( [b_q , b_q^*] + \hc \big) \\
&+ \frac{1}{2N} \sum_{q \in \Lambda_+^* , r \in \Lambda^* : r \not = -q} \widehat{V} (r/N) \eta_H (q+r) \int_0^1 ds  \, \g_q^{(s)} \big( b_q d_{-q}^{(s)}+\hc \big) + \cE_{10}^{(4)} \\
=: & \; \text{W}_{11} + \text{W}_{12} + \text{W}_{13} + \cE^{(4)}_{10}  
\end{split} 
\end{equation}
where  
\begin{equation}\label{eq:cE410}
 \cE_{10}^{(4)} =  \; \cE_{101}^{(4)} +  \cE_{102}^{(4)} +  \cE_{103}^{(4)} +  \cE_{104}^{(4)} +  \cE_{105}^{(4)} \end{equation} 
with the errors 
\begin{equation}\label{eq:cE410j} \begin{split} \cE^{(4)}_{101} &= \frac{1}{2N} \sum_{q \in \Lambda_+^* , r \in \Lambda^* : r \not = -q} \widehat{V} (r/N) \eta_H (q+r) \int_0^1 ds \Big[ 2 \gamma_q^{(s)} \s_q^{(s)} b_q^* b_q  + (\s_q^{(s)})^2 b_{-q}^* b_q^* +\hc \Big] \\
 \cE^{(4)}_{102} &= \frac{1}{2N} \sum_{q \in \Lambda_+^* , r \in \Lambda^* : r \not = -q} \widehat{V} (r/N) \eta_H (q+r) \int_0^1 ds \, \sigma_q^{(s)} \big( b_{-q}^* d_{-q}^{(s)} + \hc \big)  \\
  \cE^{(4)}_{103} &= \frac{1}{2N} \sum_{q \in \Lambda_+^* , r \in \Lambda^* : r \not = -q} \widehat{V} (r/N) \eta_H (q+r) \int_0^1 ds \, \sigma_q^{(s)} \big( d^{(s)}_q b_q^* + \hc \big) \\
 \cE^{(4)}_{104} &= \frac{1}{2N} \sum_{q \in \Lambda_+^* , r \in \Lambda^* : r \not = -q} \widehat{V} (r/N) \eta_H (q+r) \int_0^1 ds \, \g_q^{(s)} \big( d^{(s)}_q b_{-q} + \hc \big) \\
   \cE^{(4)}_{105} &= \frac{1}{2N} \sum_{q \in \Lambda_+^* , r \in \Lambda^* : r \not = -q} \widehat{V} (r/N) \eta_H (q+r) \int_0^1 ds \big( d^{(s)}_q d^{(s)}_{-q} + \hc \big) 
\end{split} \end{equation}
Since 
\begin{equation}\label{eq:VetaN} 
\sup_{q \in \Lambda_+^*} \frac{1}{N} \sum_{r \in \L_+^*} |\widehat{V} (r/N)| |\eta_{q+r}|  \leq C \, < \infty \end{equation}
uniformly in $N \in \bN$ and $\ell \in (0;1/2)$, we can bound the first term in (\ref{eq:cE410j}) by 
\[ |\langle \xi, \cE^{(4)}_{101} \xi \rangle | \leq C \sum_{q \in \L_+^*} \left[ |\eta_q| \| b_q \xi \|^2 + \eta_q^2 \| b_q \xi \| \| (\cN_+ + 1)^{1/2} \xi \| \right] \leq C \ell^{2\a} \| (\cN_+ + 1)^{1/2} \xi \|^2 \]
To estimate the second term in (\ref{eq:cE410j}), we use (\ref{eq:VetaN}) and Lemma \ref{lm:dp}; we find 
\[ \begin{split} |\langle \xi, \cE_{102}^{(4)} \xi \rangle | &\leq C \sum_{q \in \L_+^*} |\eta_H (q)| \| b_{-q} \xi \| \left[ |\eta_H (q)| \| (\cN_+ + 1)^{1/2} \xi \| + \| \eta_H \| \| b_{-q} \xi \| \right] \\ &\leq C \ell^{2\a} \| (\cN_+ + 1)^{1/2} \xi \|^2 \end{split} \]
For the third term in (\ref{eq:cE410j}), we use (\ref{eq:VetaN}), Lemma \ref{lm:dp}, and also  
\begin{equation*}%\label{eq:VetaN2} 
\frac{1}{N^2} \sum_{q \in \L_+^*, r \in \L^*, r \not = -q} |\widehat{V} (r/N)| |\eta_H (q+r)| |\eta_H (q)| \leq C < \infty \end{equation*}
uniformly in $N$ and $\ell \in (0;1/2)$. We obtain 
\[ \begin{split} |\langle \xi, \cE_{103}^{(4)} \xi \rangle | & \leq \frac{C \| (\cN_+ + 1)^{1/2} \xi \|}{N} \sum_{q \in \Lambda_+^* , r \in \Lambda^* : r \not = -q} |\widehat{V} (r/N)| |\eta_H (q+r)| |\eta_H (q)|\\ &\hspace{5.5cm} \times \left[ |\eta_q | \| b_q^* \xi \| + N^{-1} \| \eta_H \| \| b_q b_q^* \cN^{1/2}_+ \xi \| \right] \\ & \leq \frac{C \| (\cN_+ + 1)^{1/2} \xi \|}{N} \sum_{q \in \Lambda_+^* , r \in \Lambda^* : r \not = -q} |\widehat{V} (r/N)| |\eta_H (q+r)| |\eta_H (q)| \\ &\hspace{5.5cm} \times \left[ (|\eta_q | + N^{-1} \| \eta_H \|)\| \cN_+^{1/2} \xi \| + \| \eta_H \| \| b_q  \xi \| \right] \\ &\leq C \ell^\a \| (\cN_++1)^{1/2} \xi \|^2 \end{split} \]
Consider now the fourth term in (\ref{eq:cE410j}). We write $\cE_{104}^{(4)} = \cE_{1041}^{(4)} + \cE_{1042}^{(4)}$, with 
\[ \begin{split} \cE_{1041}^{(4)} &= \frac{1}{2N} \sum_{q \in \Lambda_+^* , r \in \Lambda^* : r \not = -q} \widehat{V} (r/N) \eta_H (q+r) \int_0^1 ds \, (\g_q^{(s)} -1) d_q^{(s)} b_{-q} \\
\cE_{1042}^{(4)} &= \frac{1}{2N} \sum_{q \in \Lambda_+^* , r \in \Lambda^* : r \not = -q} \widehat{V} (r/N) \eta_H (q+r) \int_0^1 ds  \, d^{(s)}_q b_{-q} 
\end{split} \]
With $|\g_q^{(s)} - 1| \leq C |\eta_H (q)|^2$, (\ref{eq:VetaN}) and Lemma \ref{lm:dp}, we easily find 
\[|  \langle \xi , \cE_{1041}^{(4)} \xi \rangle|  \leq C \ell^{3\a} \| (\cN_+ + 1)^{1/2} \xi \|^2 \]  
As for the term $\cE_{1042}^{(4)}$, we switch to position space. Using (\ref{eq:etax}) and (\ref{eq:dxy-bds}) in Lemma \ref{lm:dp}, we obtain 
\begin{equation*} %\label{eq:cE41042} 
\begin{split}
| \langle \xi , \cE_{1042}^{(4)} \xi \rangle | &= \Big|\frac{1}{2} \int_0^1 ds \int_{\Lambda^2} dx dy N^2 V(N(x-y)) \check{\eta}_H (x-y) \langle \xi, \check{d}^{(s)}_x \check{b}_y \xi \rangle \Big| \\ &\leq C \int_0^1 \int_{\Lambda^2} dx dy N^3 V(N(x-y)) \| (\cN_+ + 1)^{1/2} \xi \| \| (\cN_+ + 1)^{-1/2} \check{d}^{(s)}_x \check{b}_y \xi \| \\ &\leq C \| \eta_H \|  \int_0^1 \int_{\Lambda^2} dx dy N^2 V(N(x-y)) \| (\cN_+ + 1)^{1/2} \xi \|  \\ &\hspace{5cm} \times \left[ \| \check{a}_y \cN_+ \xi \| + \| \check{a}_x \check{a}_y \cN_+^{1/2} \xi \| \right] \\ 
&\leq C \ell^{\a/2} \| (\cN_+ + 1)^{1/2} \xi \|^2 + C \ell^{\a/2} \| (\cN_+ + 1)^{1/2} \xi \| \| \cV_N^{1/2} \xi \| 
\end{split} \end{equation*} 
Let us now consider the last term in (\ref{eq:cE410j}). Switching to position space and using (\ref{eq:ddxy}) in Lemma \ref{lm:dp} and again (\ref{eq:etax}), we arrive at 
\[ \begin{split} |\langle \xi , \cE_{105}^{(4)} \xi \rangle | &\leq C \int_{\L^2} dx dy \, N^3 V(N(x-y))  \| (\cN_+ + 1)^{1/2} \xi \| \| (\cN_+ + 1)^{-1/2} \check{d}_x \check{d}_y \xi \| \\ &\leq C \| \eta_H \|  \| (\cN_+ + 1)^{1/2} \xi \| \int_{\L^2} dx dy \, N V(N(x-y)) \\ &\hspace{2cm} \times \left[ N \| (\cN_+ + 1)^{3/2} \xi \| + \| \check{a}_x \cN_+^2 \xi \| +  \| \check{a}_y \cN_+^2 \xi \| +  \| \check{a}_x \check{a}_y \cN_+^{3/2} \xi \| \right] \\ &\leq C \ell^{\a/2} \| (\cN_++ 1)^{1/2} \xi \|^2  + C \ell^{\a/2} \| (\cN_+ + 1)^{1/2} \xi \| \| \cV_N^{1/2} \xi \| \end{split} \]
We conclude that the error term (\ref{eq:cE410}) can be estimated by 
\[ \pm \cE_{10}^{(4)} \leq C \ell^{\a/2} (\cH_N + 1) \]

Next, we come back to the terms $\text{W}_{11}, \text{W}_{12}, \text{W}_{13}$ defined  in (\ref{eq:W1-dec}). 
Using (\ref{eq:VetaN}) and $|\gamma_q^{(s)} -1| \leq C \eta_H (q)^2$, we can write  
\begin{equation}\label{eq:W11} \text{W}_{11} =  \frac{1}{2N} \sum_{q \in \Lambda_+^* , r \in \Lambda^* : r \not = -q} \widehat{V} (r/N) \eta_H (q+r) (b_q b_{-q} + \hc ) + \cE_{11}^{(4)} \end{equation}
where $\cE_{11}^{(4)}$ satisfies the estimate
\[ \begin{split} | \langle \xi , \cE_{11}^{(4)} \xi \rangle | &\leq \frac{C}{N} \sum_{q \in \Lambda_+^* , r \in \Lambda^* : r \not = -q} |\widehat{V} (r/N)| |\eta_H (q+r)| |\eta_H (q)|^2 \| b_q \xi \| \| ( \cN_+ +1)^{1/2} \xi \| \\ &\leq C \ell^{5\a/2} \| (\cN_+ + 1) \xi \|^2 \end{split} \]
The second term in (\ref{eq:W1-dec}) can be decomposed as 
\begin{equation}\label{eq:W12}  \text{W}_{12} = \frac{1}{2N} \sum_{q \in \Lambda_+^* , r \in \Lambda^* : r \not = -q} \widehat{V} (r/N) \eta_H (q+r) \eta_H (q)  \left(1- \frac{\cN_+}{N} \right) + \cE_{12}^{(4)} \end{equation}
where the error 
\[ \begin{split} \cE_{12}^{(4)}  = \; & -\frac{1}{2N^2} \sum_{q \in \Lambda_+^* , r \in \Lambda^* : r \not = -q} \widehat{V} (r/N) \eta_H (q+r) \int_0^1 ds \g_q^{(s)} \s^{(s)}_q a_q^* a_q \\ &+ \frac{1}{2N} \sum_{q \in \Lambda_+^* , r \in \Lambda^* : r \not = -q} \widehat{V} (r/N) \eta_H (q+r) \int_0^1 ds (\g_q^{(s)} \s^{(s)}_q -s \eta_H (q)) \left( 1- \frac{\cN_+}{N} \right)  \end{split} \]
can be bounded, using (\ref{eq:VetaN}) and $| \g_q^{(s)} \s^{(s)}_q -s \eta_H (q))| | \leq C |\eta_H (q)|^3$, by
\[ \pm \cE^{(4)}_{12}  \leq C \ell^{2\a} (\cN_+ + 1) \]
As for the third term on the r.h.s. of (\ref{eq:W1-dec}), we write 
\begin{equation}\label{eq:W13} \text{W}_{13} = - \frac{1}{2N} \sum_{q \in \Lambda_+^* , r \in \Lambda^* : r \not = -q} \widehat{V} (r/N) \eta_H (q+r) \eta_H (q) \left(1- \frac{\cN_+}{N} \right) \frac{\cN_+ +1}{N}  + \cE^{(4)}_{13} \end{equation}
where $\cE^{(4)}_{13} = \cE_{131}^{(4)} +  \cE_{132}^{(4)} + \cE_{133}^{(4)} + \cE_{134}^{(4)}$, with 
\begin{equation*}%\label{eq:cE413} 
\begin{split} 
\cE^{(4)}_{131} = \; &\frac{1}{2N} \sum_{q \in \Lambda_+^* , r \in \Lambda^* : r \not = -q} \widehat{V} (r/N) \eta_H (q+r) \int_0^1 ds (\g_q^{(s)} -1) b_q d_{-q}^{(s)} +\hc  \\
\cE^{(4)}_{132} = \; &\frac{1}{2N} \sum_{q \in \Lambda_+^* , r \in \Lambda^* : r \not = -q} \widehat{V} (r/N) \eta_H (q+r) \int_0^1 ds \, b_q \left[ d_{-q}^{(s)} + s \eta_H (q) \frac{\cN_+}{N} b_{q}^* \right] + \hc \\
\cE^{(4)}_{133} = \;&- \frac{1}{2N} \sum_{q \in \Lambda_+^* , r \in \Lambda^* : r \not = -q} \widehat{V} (r/N) \eta_H (q+r) \eta_H (q) b^*_q  b_{q} \frac{\cN_+ +1}{N} \\
\cE^{(4)}_{134} = \; & \frac{1}{2N^2} \sum_{q \in \Lambda_+^* , r \in \Lambda^* : r \not = -q} \widehat{V} (r/N) \eta_H (q+r) \eta_H (q) a_q^* a_q \frac{\cN_+ +1}{N} 
\end{split} \end{equation*}
It is easy to estimate the last two terms: with (\ref{eq:VetaN}), we have 
\[  \pm \cE_{133}^{(4)} \leq C \ell^{2\a} (\cN_+ + 1) , \qquad \pm \cE_{134}^{(4)} \leq C \ell^{2\a} 
(\cN_+ + 1)  \]
With $|\g_q^{(s)} -1| \leq C \eta_H (q)^2$, Lemma \ref{lm:dp} and, again, (\ref{eq:VetaN}), we also find 
\[  \begin{split} |\langle \xi , \cE_{131}^{(4)} \xi \rangle | &\leq \frac{C}{N} \sum_{q \in \Lambda_+^* , r \in \Lambda^* : r \not = -q} |\widehat{V} (r/N)| |\eta_H (q+r)| |\eta_H (q)|^2  \| (\cN_+ +1)^{1/2} \xi \| \\ &\hspace{5cm} \times \left[ |\eta_H (q)| \| (\cN_+ +1)^{1/2} \xi \| + \| \eta_H \| \| b_q \xi \| \right] \\ &\leq C \ell^{3\a} \| (\cN_+ + 1)^{1/2} \xi \|^2 \end{split} \]
Let us now focus on $\cE_{132}^{(4)}$. Switching to position space, making use of the notation $\check{\bar{d}}^{(s)}_y = d^{(s)}_y + s (\cN_+ / N) b^* (\check{\eta}_{H,y})$ and using Lemma \ref{lm:dp}, specifically (\ref{eq:splitdbd}), we obtain 
\[ \begin{split} 
 |\langle \xi , \cE_{132}^{(4)} \xi \rangle |  &= \Big| \int_0^1 ds \int_{\L^2} dx dy N^2 V(N(x-y)) \check{\eta}_H (x-y) \langle \xi, \check{b}_x \check{\bar{d}}_y \xi \rangle \Big| \\ &\leq C \| \eta_H \| 
\int_{\L^2}  dx dy N^2 V(N(x-y)) \| (\cN_+ + 1)^{1/2} \xi \| \\ &\hspace{2cm} \times \left[ N \| (\cN_+ + 1)^{1/2} \xi \|  + \| \check{a}_x \cN_+ \xi \| + \| \check{a}_y \cN_+ \xi \| + \| \check{a}_x \check{a}_y \cN_+^{1/2} \xi \| \right] \\
&\leq C \ell^{\a/2} \| (\cN_+ + 1)^{1/2} \xi \|^2 + C \ell^{\a/2} \| (\cN_+ + 1)^{1/2} \xi \| \| \cV_N^{1/2} \xi \| \end{split} \]
We conclude that $\pm \cE_{13}^{(4)} \leq C \ell^{\a/2} (\cH_N + 1)$. Combining this with (\ref{eq:W11}), (\ref{eq:W12}), (\ref{eq:W13}), we obtain 
\begin{equation}\label{eq:W1f} \begin{split}  \text{W}_1 = \; &\frac{1}{2N} \sum_{q \in \L_+^*, r \in \L^* : r \not = -q} \widehat{V} (r/N) \eta_H (q+r)  \big(b_q b_{-q} + \hc \big) \\ &+\frac{1}{2N} \sum_{q \in \Lambda_+^* , r \in \Lambda^* : r \not = -q} \widehat{V} (r/N) \eta_H (q+r) \eta_H (q)  \left(1- \frac{\cN_+}{N} \right) \left(1- \frac{\cN_+ + 1}{ N} \right) + \cE^{(4)}_1 \end{split} \end{equation}
with 
\[ \pm \cE_{1}^{(4)} \leq C \ell^{\a/2} (\cH_N + 1) \]

Next, we consider the term $\text{W}_2$, in (\ref{eq:defW}). To this end, it is convenient to switch to position space. We find
\begin{equation*} \begin{split} 
\text{W}_2 = \int_{\Lambda^2} dx dy N^2 V(N(x-y))  \int_0^1  ds \big( e^{-sB( \eta_H)} \check{b}^*_x \check{b}^*_y  e^{s B( \eta_H)} 
a^* (\check{\eta}_{H,x}) \check{a}_y  + \text{h.c.}   \big)
\end{split} 
\end{equation*}
with the notation $\check{\eta}_{H,x} (z) = \check{\eta}_H (x-z)$. By Cauchy-Schwarz, we have 
\[ \begin{split} 
 |\langle  \xi, \text{W}_2 \xi \rangle | &\leq \int_{\Lambda^2} dx dy \, N^2 V(N(x-y))  \int_0^1  ds  \\ &\hspace{1cm} \times  \| (\cN_+ +1)^{1/2} e^{-sB(\eta_H)} \check{b}_x \check{b}_y  e^{s B(\eta_H)}  \xi \| \| (\cN_+ +1)^{-1/2} a^* (\check{\eta}_{H,x}) \check{a}_y \xi \| 
\end{split} \]
With 
\[ \| (\cN_+ +1)^{-1/2} a^* (\check{\eta}_{H,x}) \check{a}_y \xi \| \leq C \|  \eta_H \| \| \check{a}_y \xi \| \leq C \ell^{\a/2} \| \check{a}_y \xi \|  \]
and using Lemma \ref{lm:prel4}, we obtain 
\be \begin{split}  \label{eq:W2end}
|\langle \xi, \text{W}_2 \xi \rangle |&\leq  C \ell^{\a/2} \int_{\Lambda^2}  dx dy \, N^2 V(N(x-y))  \| \check{a}_y \xi \| \\ &\hspace{1cm} \times \Big\{ N \| (\cN_+ +1)^{1/2} \xi \| + N  \| \check{a}_x \xi \| + N \| \check{a}_y \xi \| + N^{1/2}  \| \check{a}_x \check{a}_y \xi \| \Big\}  \\ 
& \leq C \ell^{\a/2}  \,  \| (\cN_+ +1)^{1/2} \xi \| \| (\cV_N + \cN_+  + 1 )^{1/2} \xi \| \end{split} \ee
Also for the term $\text{W}_3$ in (\ref{eq:defW}), we switch to position space. We find
\begin{equation*}%\label{eq:W3}
\begin{split}  \text{W}_3 = &\;  \int_{\L^2}  dx dy \, N^2 V(N(x-y))  \\
&  \times \int_0^1 ds\, \int_0^s d \t \, \big( e^{-sB( \eta_H)} \check{b}^*_x \check{b}^*_y e^{sB( \eta_H)} \, e^{- \t B( \eta_H)} b^*(\check{ \eta}_{H,x}) b^* (\check{\eta}_{H,y}) e^{\t B( \eta_H)} + \text{h.c.} \big) \end{split} \end{equation*}
and thus 
\[ \begin{split} & \left|\langle \xi, \text{W}_3 \xi \rangle \right|  \leq  \int_{\L^2} dx dy \, N^2 V(N(x-y))  \int_0^1 ds\, \int_0^s d \t \, \| (\cN_+ +1)^{1/2} e^{-sB( \eta_H)} \check{b}_x \check{b}_y e^{sB( \eta_H)} \xi \| \\ &\hspace{5cm} \times  \|  (\cN_+ +1)^{-1/2}  e^{- \t B( \eta_H)}  b^*(\check{ \eta}_{H,x})) b^* (\check{ \eta}_{H,y}) e^{ \t B( \eta_H)} \xi \|\,, 
\end{split} \]
With Lemma \ref{lm:Ngrow}, we find 
\begin{equation*}%\label{eq:W3-1} 
\begin{split} 
\| (\cN_+ +1)^{-1/2}  e^{- \t B( \eta_H)}  b^*(\check{ \eta}_{H,x})) b^*(\check{ \eta}_{H,y}) e^{ \t B( \eta_H)} \xi \| &  \leq C \| \eta_H\|^2 \| (\cN_+ +1)^{1/2} \xi \|  
\end{split} \end{equation*}
Using Lemma \ref{lm:prel4}, we conclude that 
\be \begin{split}  \label{eq:W3end}
|\langle \xi, \text{W}_3 \xi \rangle |&\leq  C \ell^\a \, \int_{\Lambda^2}  dx dy \, N^2 V(N(x-y))  \| (\cN_+ +1)^{1/2}\xi \| \\ &\hspace{1cm} \times \Big\{ N \| (\cN_+ +1)^{1/2} \xi \| + N  \| \check{a}_x \xi \| + N \| \check{a}_y \xi \| + N^{1/2}  \| \check{a}_x \check{a}_y \xi \| \Big\}  \\ 
& \leq C \ell^\a \,  \| (\cN_+ +1)^{1/2} \xi \| \| (\cV_N + \cN_+  + 1 )^{1/2} \xi \| \end{split} \ee
The term $\text{W}_4$ in (\ref{eq:defW}) can be bounded similarly. Switching to position space, we find 
\begin{equation*}%\label{eq:W4} 
\begin{split} \text{W}_4 = \; & \int dxdy \, N^2 V(N(x-y)) \\
& \times\int_0^1 ds \int_0^s d \t \,  \big( e^{-sB(\eta_H)} \check{b}^*_x \check{b}^*_y  \, e^{sB( \eta_H )} \,  e^{-\t B(\eta_H)} b(\check{\eta}^2_{H,x}) \check{b}_y e^{\t B(\eta_H)} + \text{h.c.} \big) 
 \end{split} \end{equation*}
where $\check{\eta}^2_H$ denotes the function with Fourier coefficients $\eta_H^2 (q)$, for $q \in \L^*$, and where $\check{\eta}^2_{H,x} (y) := \check{\eta}^2_H (x-y)$. We conclude that $\| \check{\eta}^2_{H,x} \|  = \| \eta^2_H \| \leq C \ell^{5\a/2}$. With Cauchy-Schwarz, we arrive at
\[ \begin{split} |\langle \xi , \text{W}_4 \xi \rangle | \leq \; & C \ell^{5\a/2} \int_0^1 ds \int_0^s d\tau \int dx dy N^2 V(N(x-y)) \\ &\hspace{2cm} \times  \| (\cN_+ + 1)^{1/2} e^{-sB(\eta_H)} \check{b}_y \check{b}_x e^{sB(\eta_H)} \xi \|  \| \check{b}_y e^{\tau B(\eta_H)} \xi \| \end{split} \] 
Applying Lemma \ref{lm:prel4} and then Lemma \ref{lm:Ngrow}, we obtain 
\[ \begin{split} 
|\langle \xi , \text{W}_4 \xi \rangle | \leq &\; C \ell^{5\a/2} \int_0^1 ds \int_0^s d\tau \int dx dy N^2 V(N(x-y))
\| \check{b}_y e^{\tau B(\eta_H)} \xi \| \\ &\hspace{2cm} \times \left\{  N \| (\cN_+ +1)^{1/2} \xi \| + N \| \check{a}_x \xi \| + N \| \check{a}_y \xi \| + N^{1/2} \| \check{a}_x \check{a}_y \xi \| \right\} 
\\ \leq \; &C \ell^{5\a /2}  \int_0^1 ds \int_0^s d\tau \, \| (\cN_++1)^{1/2} e^{-\t B(\eta_H)} \xi \| \| (\cV_N + \cN_+ + 1)^{1/2} \xi \| 
\\ \leq \; &C  \ell^{5\a /2} \| (\cN_+ + 1)^{1/2} \xi \|   \| (\cV_N + \cN_+ + 1)^{1/2} \xi \| 
\end{split} \]
Combining (\ref{eq:W1f}), (\ref{eq:W2end}), (\ref{eq:W3end}) with the last bound, we find 
\[ \begin{split} \cG^{(4)}_{N,\ell} = \; &\cV_N + \frac{1}{2N} \sum_{q \in \L_+^*, r \in \L^* : r \not = -q} \widehat{V} (r/N) \eta_H (q+r)  \big(b_q b_{-q} + \hc \big) \\ &+\frac{1}{2N} \sum_{q \in \Lambda_+^* , r \in \Lambda^* : r \not = -q} \widehat{V} (r/N) \eta_H (q+r) \eta_H (q)  \left(1- \frac{\cN_+}{N} \right) \left(1- \frac{\cN_+ + 1}{ N} \right) + \cE^{(4)}_{N,\ell} \end{split} \] 
where $\cE_{N,\ell}^{(4)}$ satisfies \eqref{eq:E4bound2}.  As for the bound \eqref{eq:E4Cff}, it follows similarly, arguing as we did at the end of the proof of Prop. \ref{prop:K} to show (\ref{eq:errCommK}). 
\end{proof}

\subsection{Proof of Propositions \ref{prop:GNell}}
\label{sub:proofGN}

We now combine the results of Subsections \ref{sub:G0} - \ref{sub:G4} to prove Proposition \ref{prop:GNell}. From Propositions \ref{prop:G0}, \ref{prop:G2}, \ref{prop:GN-3}, \ref{prop:GN-4}, we conclude that the excitation Hamiltonian $\cG_{N,\ell}$ defined in \eqref{eq:GN} is such that 
\begin{equation} \begin{split} \label{eq:proofGNell-1}
\cG_{N,\ell} = \; &    \frac{\widehat{V} (0)}{2}\, (N +\cN_+ -1) \, \frac{N-\cN_+}{N}\\
%%%
& + \sum_{p \in P_{H}} \eta_p \Big[p^2 \eta_p + \widehat{V} (p/N) + \frac 1 {2N} \sum_{\substack{r \in \L^*\\ p+r \in P_H}} \widehat{V} (r/N)  \eta_{p+r}\Big]\Big(\frac{N-\cN_+}{N}\Big) \Big(\frac{N-\cN_+ -1}{N}\Big)\\
& +\cK +\sum_{p \in \Lambda^*_+}  \widehat{V} (p/N)  a^*_pa_p \frac{N-\cN_+}{N} \\
&+ \sum_{p \in P_{H}} \Big[\; p^2 \eta_p + \frac 12 \widehat{V} (p/N) + \frac 1 {2N} \sum_{r \in \L^*:\; p+r \in P_H}\hskip -0.5cm \widehat{V} (r/N)  \eta_{p+r} \; \Big]  \big( b^*_p b^*_{-p} + b_p b_{-p} \big) \\
 & + \frac{1}{2}\sum_{p \in P_H^c} \Big[ \widehat{V} (p/N) + \frac 1 {2N} \sum_{r \in \L^*:\; p+r \in P_H}\hskip -0.5cm \widehat{V} (r/N)  \eta_{p+r} \Big]\big( b_p b_{-p}+ b_{-p}^* b_p^*\big) \\
& + \frac{1}{\sqrt{N}} \sum_{p,q \in \L^*_+ : p + q \not = 0} \widehat{V} (p/N) \left[ b_{p+q}^* a_{-p}^* a_q  + \hc \right]  +\cV_N   + \cE_{1}
\end{split} \end{equation}
where 
\begin{equation*}%\label{eq:cE10} 
\pm \cE_1 \leq C \ell^{(\a-3)/2} \big(\cH_N + 1 \big) \end{equation*}
and, with the notation $f_M = f(\cN_+ / M)$,
\[  %\label{eq:errCommtlE}
\pm [f_M,[f_M, \cE_1]] \leq C  \ell^{(\a-3)/2}  M^{-2} \|f'\|^2_{\infty} \big(\cH_N+ 1\big) 
\]
for every $f$ bounded and smooth and $M \in \bN$.

Our first goal is to show (\ref{eq:GeffE}). With \eqref{eq:eta-scat}, we have 
\begin{equation*} %\label{eq:proofGNellconst}
\begin{split} 
 \sum_{p \in P_{H}} \eta_p \Big[p^2 \eta_p &+ \widehat{V} (p/N) + \frac 1 {2N} \sum_{r \in \L^*:\; p+r \in P_H}\hskip -0.5cm \widehat{V} (r/N)  \eta_{p+r}\Big] \\
 & = \sum_{p \in P_{H}} \eta_p \Big[ \;\frac 1 2 \widehat{V} (p/N)  + \l_\ell N^3 \widehat \chi_\ell(p) + \l_\ell N^2 \sum_{q \in \L^*} \widehat \chi_\ell(p-q) \eta_q \;\Big]\\
 & \quad - \frac 1 {2N} \sum_{\substack{p,q \in \L^*: \\ p\in P_H,\, q \in P_H^c}}\hskip -0.5cm \widehat{V} ((p-q)/N) \eta_p  \eta_{q} 
\end{split} \end{equation*}
With Lemma  \ref{3.0.sceqlemma} and estimating \begin{equation}\label{eq:estichi} \| \widehat{\chi}_\ell \| = \| \chi_\ell \| \leq C \ell^{3/2}, \qquad \| \eta_H \| \leq \ell^{\a/2}, \qquad \| \widehat{\chi}_\ell * \eta_H \| = \| \chi_\ell \check{\eta}_H \| \leq \| \check{\eta}_H \| \leq \ell^{\a/2},\end{equation} we conclude that
\begin{equation*} \begin{split}  \sum_{p \in P_{H}} \eta_p \Big[p^2 \eta_p + \widehat{V} (p/N) + \frac 1 {2N} &\sum_{\substack{r \in \L^*\\ p+r \in P_H}} \widehat{V} (r/N)  \eta_{p+r}\Big]\Big(\frac{N-\cN_+}{N}\Big) \Big(\frac{N-\cN_+ -1}{N}\Big) \\ & = 
\frac{1}{2} \sum_{p \in P_H} \widehat{V} (p/N) \eta_p \left( \frac{N-\cN_+}{N} \right) \left( \frac{N-\cN_+ - 1}{N} \right) + \cE_2 \end{split} \end{equation*}
with $\pm \cE_2 \leq C \ell^{-\a}$ (and with $[f_M, [f_M, \cE_2]] = 0$). Since $\sum_{p \in P^c_H} |V(p/N)|  |\eta_p | \leq C \ell^{-\a}$, and from (\ref{eq:Vfa0}), we further obtain 
\begin {equation}\label{eq:1lineG} \begin{split}   \sum_{p \in P_{H}} \eta_p \Big[p^2 \eta_p + \widehat{V} (p/N) &+ \frac 1 {2N} \sum_{\substack{r \in \L^*\\ p+r \in P_H}} \widehat{V} (r/N)  \eta_{p+r}\Big]\Big(\frac{N-\cN_+}{N}\Big) \Big(\frac{N-\cN_+ -1}{N}\Big) \\ & = \left[ 4\pi \frak{a}_0 - \frac{\widehat{V} (0)}{2} \right] (N-\cN_+ -1) \left( \frac{N-\cN_+}{N} \right)  + \cE_3 \end{split} \end{equation}
where $\pm \cE_3 \leq C \ell^{-\a}$ (and $[f_M, [f_M, \cE_3]] = 0$). Using \eqref{eq:eta-scat}, we can also handle the fourth line of \eqref{eq:proofGNell-1}; we find 
\begin{equation}\begin{split} \label{eq:proofGNellQ}
\sum_{p \in P_{H}} \Big[\, p^2 \eta_p + \frac 12 &\widehat{V} (p/N) + \frac 1 {2N} \sum_{r \in \L^*:\; p+r \in P_H}\hskip -0.5cm \widehat{V} (r/N)  \eta_{p+r} \, \Big]  \big( b^*_p b^*_{-p} + b_p b_{-p} \big) \\
& = \sum_{p \in P_H} \Big[ N^3 \lambda_\ell \widehat{\chi}_\ell (p) + N^2 \lambda_\ell \sum_{q \in \Lambda^*} \widehat{\chi}_\ell (p-q) \eta_q \Big] \big( b^*_p b^*_{-p} + b_p b_{-p} \big) \\
& \hskip .3cm -  \frac 1 {2N} \sum_{\substack{p,q \in \L^*: \\ p\in P_H,\, q \in P_H^c}}\hskip -0.5cm \widehat{V} ((p-q)/N) \eta_{q} \big( b^*_p b^*_{-p} + b_p b_{-p} \big)
\end{split} \end{equation}
Observe that 
\begin{equation*}\begin{split}%\label{eq:quadr1}
 \Big| \langle \xi, N^3 \lambda_\ell  \sum_{p \in P_{H}} \widehat{\chi}_\ell (p) b_p b_{-p} \xi \rangle \Big| &\leq C \ell^{-3} \| (\cN_+ + 1)^{1/2} \xi \|  \sum_{p \in P_H} |p|^{-1} |\widehat{\chi}_\ell (p)| |p| \| b_p \xi \| \\ &\leq  C \ell^{-3 + \a} \| \widehat{\chi}_\ell \| \| (\cN_+ + 1)^{1/2} \xi \| \| \cK^{1/2} \xi \| 
 \\ &\leq C \ell^{\a -3/2} \| (\cN_+ + 1)^{1/2} \xi \| \| \cK^{1/2} \xi \| 
\end{split} \end{equation*}
Using $\widehat{\chi}_\ell * \eta = \eta$ (because $\chi_\ell (x) w_\ell (x) = w_\ell (x)$ in position space), we also find  
\[  \Big| \langle \xi, N^2 \lambda_\ell \sum_{p \in P_H, q \in \L^*} \widehat{\chi}_\ell (p-q) \eta_q (b_p^* b_{-p}^* + b_p b_{-p}) \xi \rangle \Big| \leq C N^{-1} \ell^{-3 + 3\a/2} \| (\cN_+ + 1)^{1/2} \xi \| \| \cK^{1/2} \xi \| \, . \]
Furthermore, we have  
\begin{equation}\begin{split}\label{eq:quadr3}
& \Big|   \langle \xi,  \frac 1 {2N} \sum_{\substack{p,q \in \L^*: \\ p\in P_H,\, q \in P_H^c}} \widehat{V} ((p-q)/N) \eta_{q} b_p b_{-p} \xi \rangle \Big| \\
& \leq  \frac{1}{2N} \bigg[ \sum_{\substack{p,q \in \L^*: \\ p\in P_H, q \in P_H^c}} \frac{1}{|q|^2} \frac{|\widehat{V} ((p-q)/N)|^2}{|p^2|} \bigg]^{1/2} \bigg[ \sum_{\substack{p,q \in \L^*: \\ p\in P_H,\, q \in P_H^c}}\frac{1}{|q|^2} |p|^2 \| b_p \xi \|^2 \bigg]^{1/2} \| (\cN_++1)^{1/2}\xi \| \\
& \leq  C \ell^{-\a} N^{-1/2} \| \cK^{1/2}\xi \| \| (\cN_+ + 1)^{1/2} \xi \|  
\end{split}\end{equation}
From (\ref{eq:proofGNellQ}), we conclude that 
\begin{equation} \label{eq:2lineG}
\pm \sum_{p \in P_{H}} \Big[p^2 \eta_p + \frac 12 \widehat{V} (p/N) + \frac 1 {2N} \sum_{\substack{r \in \L^* : \\ p+r \in P_H}}  \widehat{V} (r/N)  \eta_{p+r} \Big]  \big( b^*_p b^*_{-p} + b_p b_{-p} \big) 
\leq C \ell^{\a-3/2} (\cK + 1) \end{equation}
for $N$ large enough. As for the fifth line on the r.h.s. of \eqref{eq:proofGNell-1}, we can write it as 
\be \begin{split} \label{eq:quadr4}
 \frac{1}{2} \sum_{p \in P_H^c} \Big[ \widehat{V} (p/N) + \frac{1}{N} &\sum_{r \in \L^*:\; p+r \in P_H}\hskip -0.5cm \widehat{V} (r/N)  \eta_{p+r} \Big]\big( b_p b_{-p}+ b_{-p}^* b_p^*\big) \\
 & = \frac{1}{2}\sum_{p \in P_H^c}  (\widehat{V}(\cdot/N) * \widehat f_{N,\ell})_p \big( b_p b_{-p}+ b_{-p}^* b_p^*\big) + \cE_4 \end{split}\ee
where the error operator 
\[ \cE_4 =  \frac{1}{2N} \sum_{\substack{p,q \in \L^*: \\ p,\, q \in P_H^c}}  \widehat{V} ((p-q)/N) \eta_{q} \big( b_p b_{-p} + b_{-p}^* b_p^* \big) \]
can be bounded by $\pm \cE_4 \leq C N^{-1/2} \ell^{-\a} (\cK+1)$, similarly as in (\ref{eq:quadr3}). 

Combining (\ref{eq:proofGNell-1}) with (\ref{eq:1lineG}), (\ref{eq:2lineG}) and (\ref{eq:quadr4}), we conclude that  
\begin{equation*} \begin{split} % \label{eq:proofGNell-2}
\cG_{N,\ell} = \; & 4 \pi \frak{a}_0 (N-1) \left( \frac{N-\cN_+}{N} \right) + \left[ \widehat{V}(0) - 4\pi \frak{a}_0 \right] \cN_+ \left( \frac{N-\cN_+}{N} \right) \\
& +\cK +\sum_{p \in \Lambda^*_+}  \widehat{V} (p/N)  a^*_pa_p \frac{N-\cN_+}{N} 
  + \frac{1}{2}\sum_{p \in P_H^c}  (\widehat{V}(\cdot/N) * \widehat f_{N,\ell})_p \big( b_p b_{-p}+ b_{-p}^* b_p^*\big) \\
& + \frac{1}{\sqrt{N}} \sum_{p,q \in \L^*_+ : p + q \not = 0} \widehat{V} (p/N) \left[ b_{p+q}^* a_{-p}^* a_q  + \hc \right]  +\cV_N   + \cE_5
\end{split} \end{equation*} 
with
\[
\pm \cE_5 \leq  C \ell^{(\a-3)/2} \big( \cH_N +1 \big) + C \ell^{-\a}
\]
Observing that 
\[ \pm \sum_{p \in P_H} \widehat{V} (p/N) a_p^* a_p \leq C \ell^{2\a} (\cK+ 1) \, , \]
that $|\widehat{V} (p/N) - \widehat{V} (0)| \leq C |p| N^{-1}$, and that, by (\ref{eq:Vfa0}), 
\begin{equation}\label{eq:V-ao} \begin{split} |&(\widehat{V} (\cdot/N) * \widehat f_{N,\ell})_p  - 8 \pi \frak{a}_0 | \\ &\leq \int dx \, N^3 V(Nx) f_\ell (Nx) \big| e^{ip \cdot x} - 1 \big| + \left| \int  N^3 V(Nx) f_\ell (Nx) - 8\pi \frak{a}_0 \right|\leq C (|p| + 1) N^{-1} \end{split} \end{equation}
we arrive, with $\cG^\text{eff}_{N,\ell}$ defined as in (\ref{eq:GNeff}), at $\cG_{N,\ell} = \cG_\text{eff} + \cE_{N,\ell}$, with an error $\cE_{N,\ell}$ that satisfies 
\begin{equation}\label{eq:fin-cENell} \pm \cE_{N,\ell} \leq C \ell^{(\a-3)/2} \cH_N + C \ell^{-\a} \end{equation}
for all $N$ large enough. This completes the proof of (\ref{eq:GeffE}). The second bound in  (\ref{eq:errComm}) follows similarly, arguing as we did at the end of Prop. \ref{prop:K} (and noticing that the error term $\cE_3$ in (\ref{eq:1lineG}) which is responsible for the factor $\ell^{-\a}$ in (\ref{eq:fin-cENell}) actually commutes with $f (\cN_+/M)$).

Let us now prove (\ref{eq:theta-err}) and the first bound in (\ref{eq:errComm}). We have to control the off-diagonal quadratic term and the cubic term appearing in $\cG_{N,\ell}^\text{eff}$. We observe, first of all, that 
\begin{equation}\label{eq:theta10} \begin{split} \Big| 4\pi \frak{a}_0 \sum_{p \in P_H^c} \langle \xi , (b_p b_{-p} + b_{-p}^* b_p^*) \xi \rangle \Big| &\leq 4 \pi \frak{a}_0 \sum_{p \in P_H^c} \| (\cN_+ + 1)^{1/2} \xi \| \| b_p \xi \| \\ &\leq C \ell^{-\a/2} \| (\cN_+ + 1)^{1/2} \xi \| \| \cK^{1/2} \xi \| \end{split} \end{equation}
Using $[f_M, [f_M, b_p b_{-p}]] = (f (\cN_+/M) - f((\cN_++2)/M))^2 b_p b_{-p}$, and a similar identity for $[f_M, [f_M, b_p^* b_{-p}^*]]$, we also obtain 
\begin{equation}\label{eq:theta4}  \Big| 4 \pi \frak{a}_0 \sum_{p \in P_H^c} \langle \xi , \big[ f_M, \big[f_M, \big( b_p b_{-p} + b^*_p b^*_{-p} \big) \big] \big] \xi \rangle \Big| \leq C M^{-2} \ell^{-\a/2} \| f' \|_\infty^2  \| (\cN_+ + 1)^{1/2} \xi \| \| \cK^{1/2} \xi \| \end{equation}
It is possible to show an improved lower bound for the operator on the l.h.s. of (\ref{eq:theta10}), by noticing that, for an arbitrary $\delta > 0$, 
\[ \begin{split}
0 \leq\; &\sum_{p \in P_H^c} \left( \sqrt{\delta} |p| b_p^* + \frac{4\pi \frak{a}_0}{\sqrt{\delta} |p|} b_{-p} \right) \left( \sqrt{\delta} |p| b_p + \frac{4\pi \frak{a}_0}{\sqrt{\delta} |p|} b_{-p}^* \right) \\ 
= \; & \delta \sum_{p \in P_H^c} p^2 b_p^* b_p + \frac{(4\pi \frak{a}_0)^2}{\delta} \sum_{p \in P_H^c} \frac{1}{p^2} b_{-p} b_{-p}^* + 4\pi \frak{a}_0 \sum_{p \in P_H^c} (b_{-p} b_p + b_p^* b_{-p}^*) \end{split} \]
With (\ref{eq:comm-bp}), we commute 
\[ b_{-p} b_{-p}^* = b^*_{-p} b_{-p} + ( 1- \cN_+ / N) - N^{-1} a_{-p}^* a_{-p} \,. \]
Observing that 
\[ b_p^* b_p = a_p^* \frac{N-\cN_+}{N} a_p \leq a_p^* a_p \]
and that $\sum_{p \in P_H^c} |p|^{-2} \leq C \ell^{-\a}$, we conclude that there exists a constant $C > 0$, independent of $\ell \in (0;1/2)$ and of $N$, such that  
\begin{equation}\label{eq:off-diag} 4 \pi \frak{a}_0 \sum_{p \in P_H^c} (b_{-p} b_p + b_p^* b_{-p}^*)  \geq - \delta \cK - C \delta^{-1} \cN_+ - C \delta^{-1} \ell^{-\a} \end{equation}
for any $\delta > 0$. As for the cubic term on the r.h.s. of (\ref{eq:GNeff}), we have, switching 
to position space,  
\begin{equation}\label{eq:theta3} \begin{split} \Big| \frac{1}{\sqrt{N}} & \sum_{p,q \in \L_+^* : p+q \not = 0} \widehat{V} (p/N) \langle \xi , \big( b_{p+q}^* a_{-p}^* a_q + \hc \big) \xi \rangle \Big| \\ &\leq \int_{\L^2}  dx dy \, N^{5/2} V(N(x-y)) \| \check{a}_x \xi \| \| \check{a}_x \check{a}_y \xi \| \leq C \| (\cN_+ + 1)^{1/2} \xi \| \| \cV_N^{1/2} \xi \| \, . \end{split} \end{equation}
and analogously
\begin{equation}\label{eq:theta5} \begin{split}  &\Big| \frac{1}{\sqrt{N}} \sum_{p,q \in \L_+^* : p+q \not = 0} \widehat{V} (p/N) \langle \xi , \big[f_M, \big[ f_M, \big( b_{p+q}^* a_{-p}^* a_q + \hc \big) \big] \big] \xi \rangle \Big| \\ &\hspace{5cm}\leq C M^{-2} \| f' \|_\infty^2  \| (\cN_+ + 1)^{1/2} \xi \| \| \cV_N^{1/2} \xi \| \end{split} \end{equation}
Combining (\ref{eq:fin-cENell}) with (\ref{eq:theta10}) and (\ref{eq:theta3}), we obtain (\ref{eq:Gbd0}). From (\ref{eq:fin-cENell}),  (\ref{eq:off-diag}) and (\ref{eq:theta3}), we infer (\ref{eq:theta-err}). Combining instead the second bound in (\ref{eq:errComm}), with (\ref{eq:theta4}) and (\ref{eq:theta5}) we find the first bound in (\ref{eq:errComm}) (because all other contributions to $\cG_{N,\ell}^\text{eff}$ commute with $\cN_+$).

\section{Analysis of the excitation Hamiltonian $\cR_{N,\ell} $} \label{sec:RN}

The goal of this section is to prove Proposition \ref{prop:RNell}, which gives a lower bound on the excitation Hamiltonian $\cR_{N,\ell} = e^{-A} \cG_{N,\ell}^\text{eff} e^A$, with $\cG^\text{eff}_{N,\ell}$ as in (\ref{eq:GNeff}) and the cubic phase 
\begin{equation}\label{eq:Aell8} A = \frac1{\sqrt N} \sum_{r\in P_{H}, v \in P_{L}} \eta_r \big[b^*_{r+v}a^*_{-r}a_v - \text{h.c.}\big] \end{equation} 
introduced in (\ref{eq:Aell1}), with the high momentum set $P_H = \{ p \in \L_+^* : |p| \geq \ell^{-\a} \}$ and the low momentum set $P_L = \{ p \in \L_+^* : |p| \leq \ell^{-\b} \}$ for parameters $0 < \beta < \a$ and $\ell \in (0;1/2)$ (in the proof of Prop. \ref{prop:RNell}, we will assume $\a > 3$ and $\a/2 < \b < 2\a/3$). To study the properties of $\cR_{N,\ell}$, it is convenient to decompose 
\[ \cG_{N,\ell}^\text{eff} = \cD_{N} + \cK + \cQ_{N,\ell} + \cC_{N} + \cV_N \]
with $\cK$ and $\cV_N$ being the kinetic and the potential energy operators, as in (\ref{eq:KcVN}),  and
\begin{equation}\begin{split}\label{eq:wtGNell0to4} 
\cD_{N}  =&\;4\pi \frak{a}_0 (N-\cN_+) + \big[\widehat V(0)-4\pi \frak{a}_0\big]\cN_+ (1-\cN_+/N) , \\
\cQ_{N,\ell} =&\; \widehat V(0)\sum_{p\in P_H^c}   a^*_pa_p (1-N/\cN_+) + 4\pi \frak{a}_0\sum_{p\in P^c_H}  \big[ b^*_p b^*_{-p} + b_p b_{-p} \big] \\
\cC_{N} =&\;   \frac{1}{\sqrt N}\sum_{p,q\in\Lambda_+^*: p+q\neq 0} \widehat V(p/N)\big[ b^*_{p+q}a^*_{-p}a_q+ \emph{h.c.}\big] \, .		
\end{split}\end{equation} 
with $P_H^c = \Lambda^*_+ \backslash P_H$. To study the contributions of these operators to $\cR_{N,\ell}$ and to prove Proposition \ref{prop:RNell} we will need a-priori bounds controlling the growth of the expectation of the energy $\cH_N = \cK + \cV_N$ through cubic conjugation; these estimates are obtained 
in the next subsection. As we did in Section \ref{sec:GN}, also in this Section we will always assume that $V \in L^3 (\bR^3)$ is compactly supported, pointwise non-negative and spherically symmetric.  

\subsection{A priori bounds on the energy}\label{sec:aprioribnds}

Our first proposition controls the commutator of the cubic phase (\ref{eq:Aell8}) with the potential energy operator $\cV_N$. 
\begin{prop}\label{prop:commAcVN}
There exists a constant $C > 0$ such that  
\begin{equation}\label{eq:commAcVN} 
[\cV_N,A] = \frac{1}{N^{3/2}}\sum_{\substack{r\in\Lambda_+^*, v\in P_{L} \\ r\neq-v}}\big( \widehat V(\cdot/N)\ast \eta\big)(r)\big[b^*_{r+v}a^*_{-r} a_v +\emph{h.c.}\big] + \delta_{\cV_N}  \end{equation}
where 
  \begin{equation}\label{eq:commAcVNbnd}
        \begin{split}
        | \langle \xi,  \delta_{\cV_N}  \xi \rangle | \leq&\;  C \ell^{(\alpha- \beta)/2} \| \cV_N^{1/2}\xi \|\| \cK_L^{1/2}\xi\| + C\ell^{3(\alpha- \beta)/2}\| \cV_N^{1/2}\xi \| \| \cK^{1/2}\xi\|
        \end{split}\end{equation}
for all $\a > \b > 0$, $\ell \in (0;1/2)$ and $N$ large enough. Here $\cK_L = \sum_{p \in P_L} p^2 a_p^* a_p$ denotes the kinetic energy associated to momenta $p \in P_L  = \{ p \in \L_+^* : |p| \leq \ell^{-\b} \}$.        
\end{prop}

\begin{proof}
With 
\[\begin{split}
& [ a_{p+u}^* a_{q}^* a_{p}a_{q+u}, b^*_{r+v} a^*_{-r}a_{v}] \\
&\hspace{2cm}=[ a_{p+u}^* a_{q}^* a_{p}a_{q+u}, a^*_{r+v} ]\sqrt{1-(\cN_+/N)}a^*_{-r}a_{v}+b^*_{r+v} [ a_{p+u}^* a_{q}^* a_{p}a_{q+u}, a^*_{-r}a_{v}]\\
&\hspace{2cm} = b^*_{p+u} a_q^* a_{q+u} a^*_{-r} a_v \delta_{p, r+v} + b^*_{p+u} a_q^* a_p a^*_{-r} a_v  \delta_{q+u, r+v} \\ &\hspace{2.3cm} + b_{r+v}^* a_{p+u}^* a_{q}^* a_{p}a_{v} \delta_{-r, q+u}+ b_{r+v}^*  a_{p+u}^* a_{q}^* a_{q+u}a_{v}\delta_{-r, p}\\
	&\hspace{2.3cm} - b_{r+v}^*  a_{-r}^* a_{p+u}^* a_{p}a_{q+u}\delta_{q, v}- b_{r+v}^*  a_{-r}^* a_{q}^* a_{p}a_{q+u}\delta_{v, p+u} 
	%&\hspace{0.5cm} = b^*_{p+u} a_{q}^* a_{p} a^*_{-r}a_{v}\delta_{q+u,r+v}+b_{p+u}^* a_{q}^*a_{q+u}a^*_{-r}a_{v}\delta_{p,r+v}\\
	%&\hspace{1cm} +
	\end{split} \]
and normal ordering the first two terms, we obtain 
\[ [\cV_N, A] = \frac{1}{N^{3/2}}\sum^*_{u\in\Lambda^*, r \in P_H, v \in P_L} 
        \widehat{V}((u-r)/N)\eta_r b^*_{u+v} a_{-u}^*a_v + \Theta_2 + \Theta_3 + \Theta_4 + \hc \]
with 
\begin{equation}\label{eq:Theta1to4}
\begin{split}
% \Theta_1 &:=\;\frac{1}{N^{3/2}}\sum^*_{u\in\Lambda^*, r \in P^H, v \in P_L} 
%        \widehat{V}((u-r)/N)\eta_r b^*_{u+v} a_{-u}^*a_v   \\
         \Theta_2 &:=\;\frac{1}{N^{3/2}}\sum^*_{\substack{u\in \Lambda^*,p\in \Lambda_+^*,\\ r\in P_{H} , v\in P_{L} }}\widehat{V} (u/N) \eta_r b_{p+u}^* a_{r+v-u}^*a_{-r}^*a_{p}a_{v}   \\
         \Theta_3 &:=\;\frac{1}{N^{3/2}}\sum^*_{\substack{u\in \Lambda^*,p\in \Lambda_+^*,\\ r\in P_{H}, v\in P_{L} }}\widehat{V} (u/N) \eta_r b_{r+v}^* a_{p+u}^*a_{-r-u}^*a_{p}a_{v}   \\
          \Theta_4 &:=\;-\frac{1}{N^{3/2}}\sum^*_{\substack{u\in \Lambda^*,p\in \Lambda_+^*,\\ r\in P_{H} , v\in P_{L} }}\widehat{V} (u/N) \eta_r  b_{r+v}^* a_{-r}^*a_{p+u}^*a_{p}a_{v+u}  \\
         \end{split}
        \end{equation}     
The notation $\sum^*$ indicates that we exclude choices of momenta for which the argument of a creation or annihilation operator vanishes. Writing 
\[ \begin{split} &\frac{1}{N^{3/2}} \sum^*_{\substack{u \in\Lambda^*\\ r \in P_H, v \in P_L}} 
        \widehat{V}((u-r)/N)\eta_r b^*_{u+v} a_{-u}^* a_v \\ &= \frac{1}{N^{3/2}}\sum^*_{\substack{u, r \in\Lambda^*, \\ v \in P_L}}  \widehat{V}((u-r)/N)\eta_r b^*_{u+v} a_{-u}^* a_v - \frac{1}{N^{3/2}}\sum^*_{\substack{u \in\Lambda^*, v \in P_L, \\ r \in P_H^c \cup \{ 0 \}}}  \widehat{V}((u-r)/N)\eta_r b^*_{u+v} a_{-u}^* a_v \end{split}  \]
 and comparing with (\ref{eq:commAcVN}), we conclude that $\delta_{\cV_N} = \Theta_1 + \Theta_2 + \Theta_3 + \Theta_4 + \hc$, with 
 \[ \Theta_1 = - \frac{1}{N^{3/2}}\sum^*_{\substack{u \in\Lambda^*, v \in P_L, \\ r \in P_H^c \cup \{ 0 \}}}  \widehat{V}((u-r)/N)\eta_r b^*_{u+v} a_{-u}^* a_v \]
 and with $\Theta_2, \Theta_3,\Theta_4$ as defined in (\ref{eq:Theta1to4}). 
 
 To conclude the proof of the lemma, we show next that each error term $\Theta_j$, with $j=1,2,3,4$, satisfies (\ref{eq:commAcVNbnd}). 
We start with $\Theta_1$. For any $\xi\in \cF_+^{\leq N}$, switching (partly) to position space and applying Cauchy-Schwarz, we find 
\begin{equation}\label{eq:Theta1f} \begin{split}
	|\langle \xi, \Theta_1\xi \rangle |\leq &\; \frac{1}{\sqrt{N}} \bigg[ \int_{\Lambda^2}dxdy\; N^{2}V(N(x-y)) \sum_{ r\in \{0\}\cup P^c_{H}, v\in P_{L} }  |\eta_r| |v|^{-2} \|\check{b}_x \check{a}_y \xi \|^2 \bigg]^{1/2}\\
	&\;\hspace{1cm} \times \bigg[ \int_{\Lambda^2}dxdy\; N^{2}V(N(x-y)) \sum_{ r\in \{0\}\cup P_{H}^c, v\in P_{L} }  |\eta_r| |v|^2 \|a_v \xi \|^2 \bigg]^{1/2}\\
	\leq&\;  \frac{C\ell^{-\alpha- \beta/2}}{N}\| \cV_N^{1/2}\xi \|\| \cK_L^{1/2}\xi\|
	\end{split}\end{equation}
Denoting by $\check\eta_H\in L^2(\Lambda)$ the function with Fourier coefficients $\eta_H (p) = \eta_p \chi( p \in P_H)$ and using (\ref{eq:etaHL2}), we can bound the term $\Theta_2$ on the r.h.s. of (\ref{eq:Theta1to4}) by 
\[\begin{split}
	|\langle \xi, \Theta_2\xi \rangle |=&\; \bigg| \frac{1}{ N^{1/2} } \int_{\Lambda^2}dxdy \;N^2V(N(x-y))\sum_{  v\in P_{L} } e^{ivy} \langle \xi, \check{b}_{x}^*\check{a}_y^*a^*(\check{\eta}_{H,y}) \check{a}_x a_v\xi \rangle \bigg|  \\
	\leq &\; \frac{ \|\check{\eta}_{H}\|  }{N^{1/2}} \bigg [ \int_{\Lambda^2}dxdy\; N^2 V(N(x-y)) \sum_{  v\in P_{L} }|v|^{-2}  \| \cN_+^{1/2}\check{b}_x \check{a}_y  \xi \|^2 \bigg]^{1/2}\\
	&\;\hspace{3cm} \times \bigg[ \int_{\Lambda^2}dxdy\; N^2 V(N(x-y)) \sum_{  v\in P_{L} } |v|^{2} \|\check{a}_x a_v \xi \|^2 \bigg]^{1/2}\\	
	\leq&\;  C\ell^{(\alpha- \beta)/2}\| \cV_N^{1/2}\xi \|\| \cK_L^{1/2}\xi\|
	\end{split}\]
The remaining contributions $\Theta_3$ and $\Theta_4$ can be controlled similarly. We find
	\[\begin{split}
	|\langle \xi, \Theta_3\xi \rangle |=&\; \bigg| \frac{1}{ \sqrt{N} } \int_{\Lambda^2}dxdy \;N^2V(N(x-y))\sum_{r \in P_H, v\in P_L }  e^{-iry} \eta_r\langle \xi, b_{r+v}^*\check{a}_x^*\check{a}^*_y \check{a}_x a_v \xi \rangle\bigg|  \\
	\leq &\; \frac{1}{\sqrt{N}} \bigg[ \int_{\Lambda^2}dxdy\; N^2 V(N(x-y))   \sum_{ r \in P_H,  v\in P_{L} }|v|^{-2} \| b_{r+v}  \check{a}_x \check{a}_y  \xi \|^2 \bigg]^{1/2}\\
	&\;\hspace{1cm} \times \bigg[ \int_{\Lambda^2}dxdy\; N^2 V(N(x-y)) \sum_{r \in P_H,  v\in P_{L} } \eta_r^2 |v|^{2} \|\check{a}_x a_v \xi \|^2 \bigg]^{1/2}\\	
	\leq&\;  \frac{C \ell^{-\beta/2} \| \eta_H \|}{N}   \|  \cN_+^{1/2} \cV_N^{1/2} \xi \| \|  \cN_+^{1/2} \cK_L^{1/2} \xi \| \leq C \ell^{(\alpha- \beta)/2}\| \cV_N^{1/2}\xi \|\| \cK_L^{1/2}\xi\|
	\end{split}\]
as well as
	\[ \begin{split}
	|\langle \xi, \Theta_4\xi \rangle |=&\; \bigg| \frac{1}{ \sqrt{N} } \int_{\Lambda^2}dxdy \;N^2V(N(x-y))\sum_{r\in P_H, v\in P_L} \eta_r e^{-ivy} \langle \xi, b_{r+v}^*a^*_{-r}\check{a}^*_x \check{a}_x \check{a}_y \xi \rangle\bigg|  \\
	\leq &\; \frac{1}{\sqrt{N}} \bigg [ \int_{\Lambda^2} dx dy \; N^2 V(N(x-y)) \sum_{r\in P_H, v\in P_{L}} |r|^{-2}\eta_r^2\| \check{a}_x \check{a}_y  \xi \|^2 \bigg]^{1/2}\\
	&\;\hspace{1cm} \times \bigg [ \int_{\Lambda^2}dxdy \; N^{2}V(N(x-y)) \sum_{r\in P_H,  v\in P_{L}} 
	|r|^{2} \| b_{r+v}a_{-r}\check{a}_x  \xi \|^2 \bigg]^{1/2}\\	
	\leq&\;  C\ell^{3(\alpha- \beta)/2}\| \cV_N^{1/2}\xi \|\| \cK^{1/2}\xi\|
	\end{split} \]
Choosing $N > \ell^{-3\a/2}$ (to control the r.h.s. of (\ref{eq:Theta1f})), we obtain (\ref{eq:commAcVNbnd}).
\end{proof}

With the help of Prop. \ref{prop:commAcVN}, we can now control the growth of the expectation of the energy $\cH_N$ w.r.t. cubic conjugation.  
\begin{lemma} \label{lm:cHNpropagation}
There exists a constant $C > 0$ such that 
\begin{equation}\label{eq:expAHNexpA}
e^{-sA} \cH_N e^{sA} \leq C \cH_N + C \ell^{-\alpha} (\cN_+ +1)
\end{equation}
for all $\a > \b > 0$ with $\a > 4/3$, $s \in [0;1]$, $\ell \in (0;1/2)$ and $N \in \bN$ large enough.
\end{lemma}
\begin{proof}
We apply Gronwall's lemma. For a fixed $\xi \in \cF_+^{\leq N}$ and $s\in [0; 1]$, we define
		\begin{equation*}%\label{eq:deffxi} 
		f_\xi (s) := \langle \xi, e^{-sA} \cH_N e^{sA} \xi\rangle  \end{equation*}
Then 
\begin{equation}\label{eq:f'1} f'_\xi  (s) = \langle \xi, e^{-sA} [\cK, A] e^{sA} \xi\rangle + \langle \xi, e^{-sA} [\cV_N, A]  e^{sA} \xi\rangle  \end{equation}
Let us first consider the second term. From Prop. \ref{prop:commAcVN}, we find 
\[ [ \cV_N, A]  = \frac{1}{N^{3/2}} \sum_{\substack{r\in\Lambda_+^*, v\in P_{L}, r\neq-v}}\big( \widehat V(\cdot/N)\ast \eta\big)(r) \left[ b^*_{r+v} a^*_{-r} a_v + \hc \right] +  \delta_{\cV_N} \]
where the operator $\delta_{\cV_N}$ satisfies (\ref{eq:commAcVNbnd}). Switching to position space and applying Cauchy-Schwarz, we find 	
\begin{equation}\label{eq:estVN1}
		\begin{split}
		&\bigg |  \frac{1}{N^{3/2}}\sum_{\substack{r\in\Lambda_+^*, v\in P_{L}, r\neq-v}}\big( \widehat V(\cdot/N)\ast \eta\big)(r)\langle \xi,e^{-sA}b^*_{r+v}a^*_{-r} a_ve^{sA}\xi \rangle\bigg| \\
		&\hspace{0cm}= \bigg | \int_{\Lambda^2} dx dy\; N^{3/2} V(N(x-y))\check\eta(x-y) \sum_{ v\in P_{L} }e^{ivx}\langle \xi, e^{-sA} \check{a}^*_{x} \check{a}^*_{y} a_v e^{sA}\xi \rangle\bigg|\\
		&\hspace{0cm}\leq \frac{C\|\check\eta\|_\infty}{N}\| \cV_N^{1/2} e^{sA}\xi\|\bigg[ \int_{\Lambda^2}dxdy\; N^{3}V(N(x-y)) \Big\| \sum_{v \in P_{L} }e^{ivx}  a_v e^{s A}\xi \Big\|^2\bigg]^{1/2}\\	
	%& \leq C \| \cV_N^{1/2}\xi \|\bigg ( \int_{\Lambda}dx\; \sum_{v, v'\in P_{L} }e^{i(v-v')x}\langle   \xi, a^*_{v'}a_v\xi \rangle \bigg)^{1/2}\\
	&\leq C \| \cV_N^{1/2} e^{sA} \xi\| \| \cN_+^{1/2} e^{sA} \xi\|
		\end{split}
		\end{equation}	
because, by (\ref{eq:etax}), $\| \check{\eta} \|_\infty \leq C N$ and 
\[ \int_\L dx \, \Big\| \sum_{v \in P_{L} }e^{ivx}  a_v e^{s A}\xi \Big\|^2 = \sum_{v \in P_L} \langle e^{s A}\xi , a_v^* a_v e^{sA} \xi \rangle \leq \langle e^{sA} \xi, \cN_+ e^{sA} \xi \rangle  \]
Together with (\ref{eq:commAcVNbnd}), using $\alpha > \beta$, 
we conclude that 
\[ \Big| \langle \xi, e^{-sA} [\cV_N, A]  e^{sA} \xi \rangle \Big| \leq C \langle \xi, e^{-sA} \cH_N e^{sA} \xi \rangle \]
if $N$ is large enough. Let us consider the first term on the r.h.s. of (\ref{eq:f'1}). We compute 
		 \begin{equation}\label{eq:commAcK}
        		\begin{split}
       		[\cK,A] &= \frac{1}{\sqrt N}\sum_{r\in P_{H} , v\in P_{L} } 2r^2\eta_r \big[b^*_{r+v}a^*_{-r} a_v +\text{h.c.}\big]\\
        		&\hspace{3cm} +  \frac{2}{\sqrt N}\sum_{r\in P_{H}, v\in P_{L} }  r\cdot v \;\eta_r\big[b^*_{r+v}a^*_{-r} a_v +\text{h.c.}\big] \\ &=: \text{T}_1 + \text{T}_2 
        		\end{split}
        		\end{equation} 
We use (\ref{eq:eta-scat}) to rewrite the first term on the r.h.s. of (\ref{eq:commAcK}) as
		\begin{equation}\label{eq:commAcKterm1}\begin{split}
		\text{T}_1 = \; &-\frac{1}{\sqrt N}\sum_{\substack{r\in\Lambda_+^*, v\in P_{L},\\ r\neq-v} } (\widehat V(\cdot/N)\ast \widehat f_{N,\ell})(r) \big[b^*_{r+v}a^*_{-r} a_v +\text{h.c.}\big]\\
		&+\frac{1}{\sqrt N}\sum_{\substack{r\in P^c_{H}, v\in P_{L},\\ r\neq-v} } (\widehat V(\cdot/N)\ast \widehat f_{N,\ell})(r) \big[b^*_{r+v}a^*_{-r} a_v +\text{h.c.}\big]\\
		&+ \frac{1}{\sqrt N}\sum_{r\in P_{H}, v\in P_{L} } N^3\lambda_\ell (\widehat{\chi}_\ell * \widehat{f}_{N,\ell})(r)  \big[b^*_{r+v}a^*_{-r} a_v +\text{h.c.}\big]\\
		=: &\; \text{T}_{11} + \text{T}_{12} + \text{T}_{13} 
		\end{split}\end{equation}	
Since $\|f_{\ell}\|_\infty\leq 1$, the contribution of $\text{T}_{11}$ can be estimated as in (\ref{eq:estVN1}); we obtain
\begin{equation}\label{eq:T11} \big| \langle \xi, e^{-sA} \text{T}_{11} \, e^{sA} \xi \rangle \big| \leq C \|\cV_N^{1/2} e^{sA} \xi \| \| \cN_+^{1/2} e^{sA} \xi
\| \end{equation}
The second term in (\ref{eq:commAcKterm1}) can be controlled by
		\begin{equation*} %\label{eq:bndcKterm1}
		\begin{split}
		\big| \langle \xi, e^{-sA} \text{T}_{12} \, e^{sA} \xi \rangle \big| \leq &\; \frac{C}{\sqrt N}\bigg [\sum_{r\in P_{H}^c, v\in P_{L}, r\neq-v}  |r|^2  \| b_{r+v} a_{-r} e^{sA}\xi \|^2 \bigg]^{1/2} \\ &\hspace{3cm} \times \bigg[ \sum_{r\in P_{H}^c, v\in P_{L}, r\neq-v } |r|^{-2}\| a_{v} e^{sA}\xi \|^2 \bigg]^{1/2}\\	
		\leq \; &C \ell^{-\a/2} \| \cK^{1/2} e^{sA} \xi \| \| \cN_+^{1/2}  e^{sA}\xi \| 
				\end{split}\end{equation*}				
Finally, since $(\widehat{\chi}_\ell * \widehat{f}_{N,\ell}) (r)  = \widehat{\chi}_\ell (r) + N^{-1} \eta_r$, the explicit expression 
\[ \widehat{\chi_\ell}(r)= \frac{4\pi}{|r|^2}\bigg(\frac{\sin(\ell|r|)}{|r|} - \ell \cos(\ell|r|)\bigg) \]
and the bound (\ref{eq:modetap}) imply that $|(\widehat{\chi}_\ell * \widehat{f}_{N,\ell}) (r)| \leq C \ell |r|^{-2}$, for $N$ large enough. With Lemma~\ref{3.0.sceqlemma}, the third term on the r.h.s. of (\ref{eq:commAcKterm1}) can thus be estimated for $\a > 4/3$ by 
		\begin{equation}\label{eq:bndcKterm2}\begin{split}
		&\big |  \langle \xi, e^{-sA} \text{T}_{13} e^{sA} \xi \rangle\big| \\
		&\hspace{.5cm}\leq \frac{C\ell^{-2}}{\sqrt N}\bigg [\sum_{r\in P_{H} }   |r|^2|\|\cN_+^{1/2} a_{-r} e^{sA}\xi \|^2 \bigg]^{1/2}\bigg [ \sum_{r\in P_{H}, v\in P_{L} } |r|^{-6}\| a_{v} e^{sA}\xi \|^2 \bigg]^{1/2}\\	
		&\hspace{.5cm} \leq C\ell^{3\alpha/2-2} \| \cK^{1/2} e^{sA}\xi\|  \| \cN_+^{1/2} e^{sA} \xi \| \leq C \| \cK^{1/2} e^{sA}\xi\|  \| \cN_+^{1/2} e^{sA} \xi \|
		\end{split}\end{equation}
So far, we proved that
\begin{equation}\label{eq:T1-bd} |\langle \xi , \text{T}_1 \xi \rangle | \leq  C \ell^{-\a/2} \| \cH_N^{1/2} e^{sA}\xi\|  \| \cN_+^{1/2} e^{sA} \xi \| \end{equation}
for all $\xi \in \cF_+^{\leq N}$. Let us now consider the second term on the r.h.s. of (\ref{eq:commAcK}). We find
\begin{equation}\label{eq:bndcKterm3}
		\begin{split}
		&\big| \langle \xi, e^{-sA} \text{T}_2 e^{sA} \xi \rangle\big| \\
		&\hspace{.5cm} \leq \frac{C}{\sqrt N}\bigg [\sum_{r\in P_{H} }   |r|^2| \| \cN_+^{1/2} a_{-r} e^{sA}\xi \|^2 \bigg]^{1/2}\bigg [ \sum_{r\in P_{H}, v\in P_{L} } |v|^{2}\eta_r^2\| a_{v} e^{sA}\xi \|^2 \bigg]^{1/2}\\	
		&\hspace{.5cm} \leq C\ell^{\alpha/2} \| \cK^{1/2} e^{sA}\xi\| \| \cK_L^{1/2} e^{sA}\xi\|
		\end{split}
		\end{equation}
Together with (\ref{eq:T1-bd}), we conclude that
\begin{equation*} %\label{eq:comm1bnd} 
|\langle \xi, e^{-sA} [\cK, A]  e^{sA} \xi\rangle| \leq C \langle \xi, e^{-sA} \cH_N e^{sA}\xi \rangle + C \ell^{-\alpha} \langle \xi, e^{-sA} \cN_+ e^{-sA} \xi \rangle\end{equation*}
With Prop. \ref{prop:AellNgrow}, we obtain the differential inequality 
\[ | f'_\xi (s) | \leq C f_\xi (s) + C\ell^{-\a} \langle \xi, (\cN_+ + 1) \xi \rangle \]
By Gronwall's Lemma, we find (\ref{eq:expAHNexpA}). 
\end{proof}

The bound (\ref{eq:expAHNexpA}) is not yet ideal, because of the large constant proportional to $\ell^{-\a}$ multiplying the number of particles operator $\cN_+$. To improve it, it is useful to consider first the growth of the low-momentum part of the kinetic energy operator. For $\theta > 0$, we set 
\[ \cK_\theta = \sum_{p \in \L_+^* : |p| \leq \theta} p^2 a_p^* a_p \]
Comparing with the definition given in Prop. \ref{prop:commAcVN}, we have $\cK_L \equiv \cK_{\theta = \ell^{-\beta}}$. 
\begin{lemma} \label{lm:cKlowpropagation}
There exists a constant $C > 0$ such that  
\begin{equation}\label{eq:expAKlowexpA}
e^{-sA} \cK_{\theta} e^{sA} \leq C \cK_{\theta} + C\ell^{2(\alpha-\beta)}(\cH_N+1)
\end{equation}
for all $\a > \beta > 0$ with $\alpha > 4/3$, $\ell \in (0;1/2)$, $0< \theta < \ell^{-\alpha} -\ell^{-\beta}$, $s \in [0;1]$ and $N \in \bN$ large enough. 
\end{lemma}
\begin{proof} For a fixed $\xi\in \cF_+^{\leq N}$, we consider the function $g_\xi: [0; 1]\to \bR $, defined by $g_\xi(s) := \langle \xi, e^{-sA} \cK_{\theta} e^{sA} \xi \rangle$. For $r \in P_H$ and $v \in P_L$, we observe that $|r+v| \geq |r| - |v| \geq \ell^{-\a} - \ell^{-\beta} > \theta$. Hence, we obtain 
		\[\begin{split}
		[\cK_{\theta}, A] =&\;  \frac1{\sqrt N} \sum_{r\in P_{H},v\in P_{L}} \eta_r b^*_{r+v} a^*_{-r}[\cK_{\theta}, a_v ] + \text{h.c.}\\
		=&\;  -\frac1{\sqrt N} \sum_{r\in P_{H},v \in P_{L} : |v| \leq \theta} |v|^2 \eta_r \; b^*_{r+v} a^*_{-r} a_v +\text{h.c.}
		\end{split}\]	
We estimate	
\[\begin{split}
		&\bigg| \frac{1}{\sqrt N} \sum_{r\in P_{H},v \in P_{L}: |v| \leq \theta} |v|^2\eta_r \langle \xi,e^{-sA} b^*_{r+v}a^*_{-r} a_v e^{sA} \xi\rangle  \bigg|\\
		&\leq \frac1{\sqrt N} \sum_{r\in P_{H},v \in P_{L}: |v| \leq \theta} \frac{|v|}{|r+v|} |r+v| \| b_{r+v}a_{-r}e^{sA} \xi\| 
		\, | \eta_r | |v|  \| a_{v}e^{sA} \xi\| \\
		&\leq \frac{C \ell^{\alpha-\beta}}{\sqrt N} \bigg[ \sum_{r\in P_{H},v \in P_{L}: |v| \leq \theta}  |r+v|^2  \| b_{r+v}a_{-r}e^{sA}\xi\|^2 \bigg]^{1/2} \\ &\hspace{4cm} \times  \bigg[ \sum_{r\in P_{H}, v \in P_{L}: |v| \leq \theta} |\eta_r|^2  |v|^2 \| a_{v}e^{sA} \xi\|^2 \bigg]^{1/2}\\
		&\leq C \ell^{3\alpha/2-\beta}\| \cK^{1/2} e^{sA}\xi\|\| \cK_{\theta}^{1/2} e^{sA}\xi\|
		\end{split}\]
Hence, using $\cK \leq \cH_N$ and Lemma \ref{lm:cHNpropagation}, 
		\[ \begin{split}
		|(\partial_s g_\xi)(s)| & \leq C \ell^{3\alpha -2 \beta} \langle \xi, e^{-sA} \cH_N e^{sA} \xi \rangle + C g_\xi (s)  \leq C \ell^{2(\alpha - \beta)} \langle \xi , (\cH_N +1) \xi \rangle + C g_\xi (s) \end{split} \]
Gronwall's Lemma implies (\ref{eq:expAKlowexpA}). 
		\end{proof}

With Lemma \ref{lm:cKlowpropagation} we can now improve the estimate of Lemma \ref{lm:cHNpropagation} for the growth of the expectation of the potential energy $\cV_N$. 
\begin{cor}\label{cor:cVNuniformbnd} 
There exists a constant $C > 0$ such that  
\begin{equation}\label{eq:cVNuniformbnd}
e^{-sA} \cV_N e^{sA} \leq C(\cH_N+1)
\end{equation}
for all $\alpha > 4/3$ and $0< \beta < 2\a/3$, $\ell \in (0;1/2)$ small enough, $s \in [0;1]$ and $N \in \bN$ large enough. 
\end{cor}
\begin{proof} For $\xi \in \cF_+^{\leq N}$, consider the function $h_\xi:[0;1] \to\bR$ defined through $h_\xi (s) := \langle \xi, e^{-sA} \cV_N e^{sA} \xi \rangle$. By Prop. \ref{prop:commAcVN}, we have 
		\[\begin{split}
		h'_\xi (s) =&\; \frac{1}{N^{3/2}}\!\!\!\sum_{\substack{r\in\Lambda_+^*, v\in P_{L}, r\neq-v}} (\widehat V(\cdot/N)\ast \eta\big)(r)\langle \xi, e^{-sA}\big([b^*_{r+v}a^*_{-r} a_v +\text{h.c.}] \big) e^{sA}\xi \rangle\\
		&\;+\langle \xi, e^{-sA}\delta_{\cV_N} e^{sA}\xi \rangle
		\end{split}\]
where 
\[\begin{split}
        		\big| \langle \xi, e^{-sA} \delta_{\cV_N} e^{sA} \xi \rangle \big| \leq&\;  C \| \cV_N^{1/2} e^{sA}\xi\|\|\cK_L^{1/2} e^{sA} \xi \|  + C\ell^{3(\alpha- \beta)/2} \| \cV_N^{1/2} e^{sA}\xi\| \|\cK^{1/2} e^{sA}\xi\| 
       		\end{split}\]
The estimate (\ref{eq:estVN1}), in the proof of Lemma \ref{lm:cHNpropagation}, shows moreover that
		\[\begin{split}
		&\bigg |  \frac{1}{N^{3/2}}\sum_{\substack{r\in\Lambda_+^*, v\in P_{L}, r\neq-v}} \big( \widehat V(\cdot/N)\ast \eta\big)(r)\langle \xi,e^{-sA} b^*_{r+v}a^*_{-r} a_ve^{sA} \xi \rangle\bigg| \\
		&\hspace{7.5cm}\leq C \| \cV_N^{1/2} e^{sA} \xi\|\| \cN_+^{1/2} e^{sA} \xi\|
		\end{split}\]
With Proposition \ref{prop:AellNgrow}, Lemma \ref{lm:cHNpropagation} and Lemma \ref{lm:cKlowpropagation} (with $\theta=\ell^{-\beta}$), we deduce that 
		\[\begin{split}
		&| h'_\xi (s)| \leq C \|\cV_N^{1/2} e^{sA} \xi\|^2 +  C(1+ \ell^{2\alpha-3\beta} )\langle \xi, (\cH_N+1) \xi\rangle\leq Ch_\xi(s) + C \langle \xi, (\cH_N+1) \xi\rangle
		\end{split}\]
because $\beta < 2\a/3$. Notice that, for $\ell \in (0;1/2)$ small enough, we have $2\ell^{-\beta}< \ell^{-\alpha}$; thus, we may choose indeed $\theta=\ell^{-\beta}$ in Lemma \ref{lm:cHNpropagation}. Applying Gronwall's Lemma to the last bound concludes (\ref{eq:cVNuniformbnd}).
\end{proof}

Finally, we consider the growth of the kinetic energy operator; in this case, we do not get a bound uniform in 
$\ell$; still, we can improve the result of Lemma \ref{lm:cHNpropagation} and the estimate we obtain is sufficient for our purposes. 
\begin{cor}\label{cor:cKimproved} 
There exists a constant $C > 0$ such that  
\begin{equation}\label{eq:cKimprovedbnd}
e^{-sA} \cK e^{sA} \leq C\ell^{-(\alpha+\beta)/2} (\cH_N+1) 
\end{equation}
for all $\alpha > 4/3$ and $0< \beta < 2\a/3$, $s \in [0;1]$, $\ell \in (0;1/2)$ small enough and $N \in \bN$ large enough.
\end{cor}

\begin{proof} For a fixed $\xi \in \cF_+^{\leq N}$ define $j_\xi:[0;1]\to\bR$ by $j_\xi(s) := \langle \xi, e^{-sA} \cK e^{sA} \xi \rangle$.
From (\ref{eq:commAcK}) and (\ref{eq:commAcKterm1}), we infer that 
\[
[\cK, A ]  = \; \text{T}_{11} + \text{T}_{12} + \text{T}_{13} + \text{T}_2 
\]
with $\text{T}_{11}, \text{T}_{12}, \text{T}_{13}, \text{T}_2$ as in (\ref{eq:commAcK}) and  (\ref{eq:commAcKterm1}). Combining (\ref{eq:T11}) with Prop. \ref{prop:AellNgrow} and Corollary \ref{cor:cVNuniformbnd}, we find
\begin{equation}\label{eq:T11-2}     |\langle \xi, e^{-sA} \text{T}_{11} e^{sA} \xi \rangle | \leq C \| \cV_N^{1/2} e^{sA} \xi \| \| \cN_+^{1/2} e^{sA} \xi \|  \leq C \langle \xi, (\cH_N + 1) \xi \rangle \, . \end{equation}
From (\ref{eq:bndcKterm2}), Prop. \ref{prop:AellNgrow} and Lemma \ref{lm:cHNpropagation}, we obtain 
\begin{equation}\label{eq:T13-2} \begin{split}  | \langle \xi, e^{-sA} \text{T}_{13} e^{sA} \xi \rangle |  &\leq C \ell^{3\a/2 -2} \| \cK^{1/2} e^{sA} \xi \| \| \cN_+^{1/2} e^{sA} \xi \| \\ &\leq C \ell^{\a-2} \langle \xi, (\cH_N + 1) \xi \rangle  \leq C  \langle \xi, (\cH_N + 1) \xi \rangle \end{split} \end{equation}
Using (\ref{eq:bndcKterm3}), Lemma \ref{lm:cHNpropagation} and Lemma \ref{lm:cKlowpropagation}, we arrive at
\begin{equation}\label{eq:T2-2} |\langle \xi, e^{-sA} \text{T}_{2} e^{sA} \xi \rangle | \leq C \ell^{\a/2} \| \cK^{1/2} e^{sA} \xi \| \| \cK_L^{1/2} e^{sA} \xi \| \leq C  \langle \xi, (\cH_N + 1) \xi \rangle 
\end{equation}
Hence, to show (\ref{eq:cKimprovedbnd}), we only need to improve the bound on $\text{T}_{12}$. 
 To this end, we set $\theta = \ell^{-\a} - 5\ell^{-\beta}/4$ and we decompose
 \[ \begin{split} \text{T}_{12} = \; &\frac{1}{\sqrt N} \sum_{\substack{0< |r| \leq \theta \\ v\in P_{L}, r \neq-v} } (\widehat V(\cdot/N)\ast \widehat f_{N,\ell})(r) b^*_{r+v}a^*_{-r} a_v  
 \\ &+ \frac{1}{\sqrt N} \sum_{\substack{\theta < |r| \leq \ell^{-\a} ,\\ v \in P_{L},  r \neq -v} } (\widehat V(\cdot/N)\ast \widehat f_{N,\ell})(r)  b^*_{r+v}a^*_{-r} a_v \\ =: & \; \text{T}_{121} + \text{T}_{122} \end{split} \] 
With Prop. \ref{prop:AellNgrow}  and Lemma \ref{lm:cKlowpropagation}, we estimate 
 \begin{equation*}%\label{eq:cKimproved4}
 \begin{split}
\big| \langle \xi, e^{-sA} \text{T}_{121} e^{sA} \xi \rangle \big|  
&\leq \frac{C}{\sqrt N}\sum_{\substack{0<|r|\leq \theta,\\ v \in P_{L}, r\neq-v} } |r| \| a_{-r}b_{r+v} e^{sA} \xi\| \, |r|^{-1} \|a_v e^{sA}\xi\| \\
&\leq C\ell^{-\alpha/2} \| \cK^{1/2} _{\theta} e^{sA}\xi \| \| \cN_+^{1/2} e^{sA} \xi \| \leq C\ell^{-\alpha/2} \langle \xi, (\cH_N+1)\xi\rangle 
\end{split}\end{equation*}
On the other hand, since $\sum_{\theta < |r| < \ell^{-\alpha}} |r|^{-2} \leq C \ell^{-\beta}$, we find, by Prop. \ref{prop:AellNgrow} and Lemma~\ref{lm:cHNpropagation}, 
\begin{equation*}%\label{eq:cKimproved5}
\begin{split} 
\big| \langle \xi, e^{-sA} \text{T}_{121} e^{sA} \xi \rangle \big|  &\leq \frac{C}{\sqrt N}\sum_{\substack{\theta <|r| \leq \ell^{-\a},\\ v \in P_{L}, r\neq-v} } |r| \| a_{-r}b_{r+v} e^{sA} \xi\| \, |r|^{-1} \|a_v e^{sA}\xi\| \\
&\leq C\ell^{-\beta/2} \| \cK^{1/2} e^{sA}\xi \| \| \cN_+^{1/2} e^{sA} \xi \| 
\leq C\ell^{-(\alpha+\beta)/2} \langle \xi, (\cH_N+1)\xi\rangle 
\end{split}\end{equation*}
Combining the last two bounds with (\ref{eq:T11-2}), (\ref{eq:T13-2}), (\ref{eq:T2-2}), we obtain 	
\[|j'_\xi (s)|\leq C\ell^{-(\alpha+\beta)/2} \langle \xi, (\cH_N+1)\xi\rangle \]
for all $s \in [0;1]$. Integrating over $s$, we arrive at (\ref{eq:cKimprovedbnd}).
\end{proof}

\subsection{Analysis of $e^{-A} \cD_N e^{A}$}\label{sub:RNell0}

%%%%%%%%%%%%%%%%%%%%%%%%%%%%%%%%%%%%%%%%%%%%%%%%%%%
%%%%%%%%%%%%%%%%%%%%%%%%%%%%%%%%%%%%%%%%%%%%%%%%%%%

In this section we study the contribution to $\cR_{N,\ell}$ arising from the operator $\cD_N$, 
defined in (\ref{eq:wtGNell0to4}). To this end, it is convenient to use the following lemma. 
\begin{lemma} \label{lm:cNops}
There exists a constant $C > 0$ such that 
\begin{equation}\label{eq:cNopsbnd}
\begin{split}
&\Big|\sum_{p\in \Lambda_+^*} F_p  \langle \xi_1 , (e^{-A} a^*_p a_p  e^{A} -a^*_pa_p ) \xi_2\rangle\Big| \\
&\hspace{4cm}\leq  C\ell^{\alpha/2} \|  F\|_\infty \| (\cN_++1)^{1/2}\xi_1\| \| (\cN_++1)^{1/2}\xi_2\|
\end{split}
\end{equation}
for all $\alpha , \beta > 0$, $\xi_1, \xi_2 \in \cF_+^{\leq N}$, $F \in \ell^{\infty} (\Lambda^*_+)$, $\ell \in (0;1/2)$ and $N \in \bN$ large enough. 
\end{lemma}
\begin{proof} The lemma is a simple consequence of Proposition \ref{prop:AellNgrow}. We write
		\[\begin{split}
		\sum_{p\in \Lambda_+^*} F_p   (e^{-A} a^*_p a_p  e^{A} -a^*_p a_p ) = \int_0^1ds\;\sum_{p\in \Lambda_+^*} F_p e^{-sA}[a^*_p a_p, A] e^{sA}
		\end{split}\]
and compute 
		\[\sum_{p\in \Lambda_+^*} F_p [a^*_p a_p, A]  = \frac1{\sqrt N} \sum_{r\in P_{H}, v\in  P_{L} }(F_{r+v}+F_{-r} - F_v)\eta_r b^*_{r+v}a^*_{-r}a_v + \text{h.c.} \]
By Cauchy-Schwarz, we find with the help of Proposition \ref{prop:AellNgrow} that
		\[\begin{split}
		&\Big|\frac1{\sqrt N} \sum_{r\in P_{H},v\in P_{L}} (F_{r+v}+F_{-r} - F_v)\eta_r   \langle e^{sA}\xi_1 , b^*_{r+v}a^*_{-r}a_v e^{sA}\xi_2\rangle\Big| \\
		&\hspace{3.5cm}\leq  \frac {C\|  F \|_\infty}{\sqrt N} \sum_{r\in P_{H}, v\in P_{L}}|\eta_r |  \| a_ve^{sA}\xi_2\| \| a_{-r}b_{r+v}e^{sA} \xi_1\|\\
		&\hspace{3.5cm}\leq  C\ell^{\alpha/2} \|  F \|_\infty \| (\cN_++1)^{1/2}\xi_1\| \| (\cN_++1)^{1/2}\xi_2\|
		\end{split}\]
Since the bound is uniform in the integration variable $s\in [0;1]$, we obtain (\ref{eq:cNopsbnd}).
\end{proof}

\begin{prop}\label{prop:RNell0}
There exists a constant $C>0$ such that
\begin{equation*}
e^{-A} \cD_N e^A = 4\pi \frak{a}_0 (N-\cN_+) + \big[\widehat V(0)-4\pi \frak{a}_0\big] \cN_+ (1-\cN_+/N) + \delta_{\cD_N} 
\end{equation*}
where  		
\begin{equation*}%\label{eq:RNell0} 
 |\langle \xi, \delta_{\cD_N} \xi \rangle | \leq  C\ell^{\alpha/2} \langle \xi,  (\cN_+ +1)\xi \rangle\end{equation*}
for all $\alpha , \beta > 0$, $\xi\in \cF_+^{\leq N}$, $\ell \in (0;1/2)$ and $N \in \bN$ large enough. 
\end{prop}

\begin{proof}
Recall from (\ref{eq:wtGNell0to4}) that 
\begin{equation*}
\cD_N  = 4\pi \frak{a}_0 (N-\cN_+) + \big[\widehat V(0)-4\pi \frak{a}_0\big]\cN_+ (1-\cN_+/N)\end{equation*}
Lemma \ref{lm:cNops} implies that
\[\begin{split} 
\pm\bigg\{ e^{-A} &\left[ 4\pi \frak{a}_0 (N-\cN_+) + \big[\widehat V(0)-4\pi \frak{a}_0\big]\cN_+ \right] e^{A} \\ &\hspace{2cm} - \left[ 4\pi \frak{a}_0 (N-\cN_+) + \big[\widehat V(0)-4\pi \frak{a}_0\big]\cN_+ \right] \bigg\} \leq C \ell^{\alpha/2}(\cN_++1) \end{split} \]
As for the contribution quadratic in $\cN_+$, we can write  
\[ \begin{split} N^{-1} & \left\langle \xi, \left[ e^{-A} \cN_+^2 e^A - \cN_+^2 \right] \xi \right\rangle \\ &\hspace{2cm} = N^{-1}  \left\langle \xi_1, \left[ e^{-A} \cN_+ e^A - \cN_+ \right] \xi \right\rangle + N^{-1} \left\langle \xi, \left[e^{-A} \cN_+ e^A - \cN_+ \right] \xi_2 \right\rangle \end{split} \]
with $\xi_1 = e^{-A} \cN_+ e^A \xi$ and $\xi_2 = \cN_+ \xi$. Applying again Lemma \ref{lm:cNops}, we obtain 
\[\begin{split}  &\left| N^{-1} \left\langle \xi, \left[ e^{-A} \cN_+^2 e^A - \cN_+^2 \right] \xi \right\rangle \right| \\ &\hspace{2cm} \leq C N^{-1} \ell^{\a/2} \| (\cN_+ + 1)^{1/2} \xi \| \left[ \| (\cN_+ + 1)^{1/2} \xi_1 \| + \| (\cN_+ + 1)^{1/2} \xi_2 \| \right] \end{split} \]
Using (twice) Prop. \ref{prop:AellNgrow}, we find
\[ \| (\cN_+ + 1)^{1/2} \xi_1 \| = \| (\cN_+ +1)^{1/2} e^{-A} \cN_+ e^A \xi \| \leq C \| (\cN_+ +1)^{3/2} \xi \| \]
Hence,we conclude that 
\[ \begin{split}  &\left| N^{-1} \left\langle \xi, \left[ e^{-A} \cN_+^2 e^A - \cN_+^2 \right] \xi \right\rangle \right| \\ &\hspace{2cm} \leq C N^{-1} \ell^{\a/2} \| (\cN_+ + 1)^{1/2} \xi \| \| (\cN_++1)^{3/2} \xi \| \leq C \ell^{\a/2} \| (\cN_+ + 1)^{1/2} \xi \|^2 \end{split} \]
\end{proof}

\subsection{Contributions from $e^{-A} \cK e^{A}$}
\label{sec:RNellK}

%%%%%%%%%%%%%%%%%%%%%%%%%%%%%%%%%%%%%%%%%%%%%%%%%%%
%%%%%%%%%%%%%%%%%%%%%%%%%%%%%%%%%%%%%%%%%%%%%%%%%%%

In this subsection, we consider contributions to $\cR_{N,\ell}$ arising from conjugation of the kinetic energy operator $\cK = \sum_{p \in \L_+^*} p^2 a_p^* a_p$. In particular, in the next proposition, we establish properties of the commutator $[\cK, A]$.  
\begin{prop}\label{prop:RNell2K} 
There exists a constant $C>0$ such that
\begin{equation*} %\label{eq:commcKAimproved}
\begin{split}
[\cK, A] = &\; -\frac{1}{\sqrt N}\sum_{p\in \Lambda_+^*, q\in P_{L}, p \neq -q } (\widehat V(\cdot/N)\ast \widehat f_{N,\ell})(p) ( b^*_{p+q}a^*_{-p} a_q+ \emph{h.c.})\\
&\;+\frac{8\pi \mathfrak{a}_0}{\sqrt N} \sum_{p \in P_{H}^c , q \in P_{L}, p \neq -q } \big[ b^*_{p+q}a^*_{-p} a_q + \emph{h.c.}\big] + \delta_\cK
\end{split}
\end{equation*}
where 
\begin{equation}\label{eq:commcKAimprovedbnd}
\begin{split} |\langle \xi, \delta_\cK \xi \rangle | \leq \; &C (\ell^{3\alpha/2-2} + \ell^{\a/2}) \|\cK^{1/2}\xi\| \|(\cN_+ + \cK_L )^{1/2} \xi\|
\end{split} \end{equation}
for all $\a, \b > 0$, $\xi \in \cF_+^{\leq N}$, $\ell \in (0;1/2)$, $N \in \bN$ large enough. Moreover, we have 
\begin{equation} \label{eq:commCNrenA1}
\begin{split} 
\Big| \frac{8\pi \mathfrak{a}_0}{\sqrt N} &\sum_{p \in P^c_{H}, q \in P_{L}, p \neq -q} \langle \xi , \big[ b^*_{p+q}a^*_{-p} a_q , A \big] \xi \rangle \Big| \\ \leq \; &C\ell^{3(\alpha- \beta)/2} \|(\cN_++1)^{1/2}\xi\|\|\cK^{1/2}\xi \| +C\ell^{(\alpha- \beta)/2}\|\cK_L^{1/2}\xi \| \|\cN_+^{1/2}\xi\| \\ &+ C\ell^{\alpha} \|\cK^{1/2}\xi \|^2
\end{split}
\end{equation}
for all $\a, \b > 0$, $\xi \in \cF_+^{\leq N}$, $\ell \in (0;1/2)$ and $N \in \bN$ large enough. 
\end{prop}
\begin{proof} The bound (\ref{eq:commcKAimprovedbnd}) is a consequence of Eqs. (\ref{eq:commAcK}), (\ref{eq:commAcKterm1}), (\ref{eq:bndcKterm2}), (\ref{eq:bndcKterm3}) in the proof of Lemma \ref{lm:cHNpropagation}, and of the observation that, from the estimate (\ref{eq:V-ao}),  
\[\begin{split}
&\bigg| \frac{1}{\sqrt N}\sum_{p\in P^c_{H}, q \in P_{L}, p \neq-q} \big[(\widehat V(\cdot/N)\ast \widehat f_{N,\ell})(p)  -8\pi \mathfrak{a}_0\big]\langle \xi, b^*_{p+q}a^*_{-p} a_q \xi \rangle\bigg|\\
&\hspace{1cm} \leq C N^{-3/2} \sum_{p\in P^c_{H}, q \in P_{L}, p \neq-q} |p| \| b_{p+q} a_{-p} \xi \| \| a_q \xi \| \leq C N^{-1}\ell^{-3\alpha/2}\|\cK^{1/2}\xi\|\|\cN_+^{1/2}\xi\|
\end{split}\] 
which is bounded by the r.h.s. of (\ref{eq:commcKAimprovedbnd}) if $N$ is large enough. 
Let us now focus on (\ref{eq:commCNrenA1}). We have 
		\begin{equation}\label{eq:commCrenA}
		\begin{split}
		&\frac{8\pi \mathfrak{a}_0}{\sqrt N}\sum_{p\in P_{H}^c,  q\in P_{L}, p\neq-q} \big[b^*_{p+q}a^*_{-p} a_q , A \big]+ \text{h.c.}\\
		&= \frac{8\pi \mathfrak{a}_0}{N}\sum_{\substack{r\in P_{H}, p \in P^c_{H}, \\ q, v\in P_{L}, p \neq-q,  r\neq-v}} \eta_r \big[b^*_{p+q}a^*_{-p} a_q , b^*_{r+v}a^*_{-r} a_v - a_v^* a_{-r}b_{r+v} \big]+ \text{h.c.}
		\end{split}
		\end{equation}
We split the commutator into the four summands
\begin{equation}\label{eq:comm33} 
\begin{split}
[b^*_{p+q}a^*_{-p} a_q,  b^*_{r+v}a^*_{-r} a_v - a_v^* a_{-r}b_{r+v} ] =&\;  \big([b^*_{p+q},  b^*_{r+v}a^*_{-r} a_v ]+ [a_v^* a_{-r}b_{r+v}, b^*_{p+q} ]\big)a^*_{-p} a_q\\ 
&\; + b^*_{p+q}\big([a^*_{-p} a_q,  b^*_{r+v}a^*_{-r} a_v]+[a_v^* a_{-r}b_{r+v}, a^*_{-p} a_q  ]\big)
 \end{split}\end{equation}
We compute 
\begin{equation}\label{eq:commcubic1}
\begin{split}
[b^*_{p+q},  b^*_{r+v}a^*_{-r} a_v ] a^*_{-p} a_q = - b^*_{r+v}b^*_{-r} a^*_{-p} a_q \delta_{p+q, v} =  - b^*_{r+v}b^*_{-r} a^*_{q-v} a_q \delta_{p+q, v} 
\end{split}
\end{equation}
as well as 
\begin{equation}\label{eq:commcubic2}
\begin{split}
&[a_v^* a_{-r}b_{r+v}, b^*_{p+q}] a^*_{-p} a_q \\
&=(1-\cN_+/N)a^*_{v}a^*_{r+q} a_q a_{r+v} \delta_{p+q, -r}+(1-\cN_+/N)a^*_{v} a_v \delta_{p+q, -r} \delta_{r+v, -p} \\
&\hspace{0.5cm}+ (1-\cN_+/N)a_v^* a^*_{q-r-v} a_{-r}  a_q \delta_{p+q,r+v}+(1-\cN_+/N)a_v^*  a_v \delta_{p+q,r+v}\delta_{r,p} \\
&\;\hspace{0.5cm}- N^{-1} a_v^* a^*_{p+q}a^*_{-p}a_{-r} a_{r+v} a_q -N^{-1} a_v^*a^*_{q-r-v}a_{-r}   a_q \delta_{r+v,-p}-N^{-1} a_v^*a^*_{q+r}a_{r+v}   a_q \delta_{p,r}
\end{split}
\end{equation}
Similarly, we find
\begin{equation}\label{eq:commcubic3}
\begin{split}
b^*_{p+q}[a^*_{-p} a_q,  b^*_{r+v}a^*_{-r} a_v]=&\;b^*_{p+r+v} b^*_{-p} a^*_{-r} a_v\delta_{q, r+v} + 
b^*_{p-r}b^*_{r+v}a^*_{-p}  a_v\delta_{q,-r}\\
&\; - b^*_{q-v}b^*_{r+v}a^*_{-r} a_q\delta_{-p,v}
\end{split}
\end{equation}
and 
\begin{equation}\label{eq:commcubic4}
\begin{split}
b^*_{p+q}[a_v^* a_{-r}b_{r+v}, a^*_{-p} a_q] =&\; b^*_{q+r}a_v^*a_qb_{r+v}\delta_{r,p}  - b^*_{p+v}a^*_{-p} a_{-r}b_{r+v}\delta_{q,v} \\
&\;+ b^*_{q-r-v}a_v^* a_{-r} b_q \delta_{r+v,-p}
\end{split}
\end{equation}
Taking into account that $\delta_{r,p} = \delta_{q,-r} = \delta_{r+v,q} = 0$ for $r \in P_H, p \in P_H^c, q,v \in P_L$ we obtain, inserting these formulas into (\ref{eq:commCrenA}),  
\[ \frac{8\pi \mathfrak{a}_0}{\sqrt N}\sum_{p\in P_{H}^c,  q\in P_{L}, p\neq-q} \big[b^*_{p+q}a^*_{-p} a_q , A \big]+ \text{h.c.} = \sum_{j=1}^7 \Upsilon_j +\hc \]
where
\begin{equation}\label{eq:commCNrenA2}
\begin{split}
\Upsilon_1:= &\;  - \frac{16\pi \mathfrak{a}_0}{N}\sum_{\substack{r\in P_{H}; q, v\in P_{L} , \\ q\neq v,  r\neq-v} }\eta_r b^*_{r+v}b^*_{-r} a^*_{q-v} a_q ,\\
\Upsilon_2:= &\; \frac{8\pi \mathfrak{a}_0}{N}\sum_{\substack{r\in P_{H}; q, v\in P_{L}, \\ q+r P^c_{H}, r\neq-q,  r\neq-v} } \eta_r (1-\cN_+/N)a^*_{v}a^*_{r+q} a_q a_{r+v},\\
\Upsilon_3:= &\; \frac{8\pi \mathfrak{a}_0}{N}\sum_{\substack{r\in P_{H}, v\in P_{L},\\ r+v\in P^c_H}} \eta_r(1-\cN_+/N)a^*_{v} a_v,\\
\Upsilon_4:= &\; \frac{8\pi \mathfrak{a}_0}{N}\sum_{\substack{r\in P_{H}; q, v\in P_{L},\\ q-r-v\in P^c_{H}}} \eta_r(1-\cN_+/N)a_v^* a^*_{q-r-v} a_{-r}  a_q,\\
\Upsilon_5:= &\; -\frac{8\pi \mathfrak{a}_0}{N^2}\sum_{\substack{r\in P_{H}, p\in P^c_{H}, \\ q, v\in P_{L}, p \neq -q,  r \neq -v} } \eta_ra_v^* a^*_{p+q}a^*_{-p}a_{-r} a_{r+v} a_q,\\
\Upsilon_6:= &\; -\frac{8\pi \mathfrak{a}_0}{N^2}\sum_{\substack{r\in P_{H}; q, v\in P_{L}, \\ r+v\in P_{H}^c,  q\neq r+v} }\eta_r  a_v^*a^*_{q-r-v}a_{-r}   a_q,\\
\Upsilon_7:= &\; -\frac{8\pi \mathfrak{a}_0}{N}\sum_{\substack{r\in P_{H}, p\in P_{H}^c, \\ v\in P_{L};  p,r\neq-v} }\eta_r b^*_{p+v}a^*_{-p} a_{-r}b_{r+v} ,\\
\Upsilon_8:= &\; \frac{8\pi \mathfrak{a}_0}{N}\sum_{\substack{r\in P_{H}; q, v\in P_{L},\\ r+v\in P_{H}^c , q\neq r+v } }\eta_r b^*_{q-r-v}a_v^* a_{-r} b_q 	\\
\end{split}
\end{equation}	
In fact, $\Upsilon_1$ collects the contribution from (\ref{eq:commcubic1}) and the non-vanishing 
contribution from (\ref{eq:commcubic3}), $\Upsilon_2 - \Upsilon_6$ corresponds to the five non-vanishing terms on the r.h.s. of (\ref{eq:commcubic2}), $\Upsilon_7$ and $\Upsilon_8$ reflect the two non-vanishing terms on the r.h.s. of (\ref{eq:commcubic4}). 

To conclude the proof of Prop. \ref{prop:RNell2K}, we show that all operators in (\ref{eq:commCNrenA2}) satisfy (\ref{eq:commCNrenA1}). By Cauchy-Schwarz, we observe that 
\[\begin{split}
\big| \langle \xi, \Upsilon_1 \xi \rangle \big| &\leq  \frac{C\ell^{\alpha}}{N}\sum_{\substack{r\in P_{H}; q, v\in P_{L} , \\ q\neq v,  r\neq-v} } |\eta_r|  \|  a_q (\cN_++1)^{1/2}\xi \|  |r| \| a_{-r}a_{q-v} a_{r+v}  (\cN_++1)^{-1/2} \xi \| \\& \leq C\ell^{3(\alpha- \beta)/2} \|(\cN_++1)^{1/2}\xi\|\|\cK^{1/2}\xi \| \end{split}\]
The expectation of $\Upsilon_2$ is bounded by
\[\begin{split}
\big| \langle \xi, \Upsilon_2 \xi \rangle \big| &\leq  \frac{C}{N}\sum_{\substack{r\in P_{H}; q, v\in P_{L}, \\ q+r\in P^c_{H}, r \neq -q,  r \neq -v} }|\eta_r| |q|  \| a_q a_{r+v} \xi \|  |q|^{-1} \|a_{v}a_{r+q}  \xi \|  \\ & \leq C\ell^{(\alpha- \beta)/2}\|\cK_L^{1/2}\xi \| \|\cN_+^{1/2}\xi\|
\end{split}\]
where we recall the notation $\cK_L = \cK_{\ell^{-\beta}} = \sum_{|p| \leq \ell^{-\beta}} p^2 a_p^* a_p$ for the low-momenta kinetic energy. It is simple to see that $\pm \Upsilon_3\leq CN^{-1}\ell^{-\a} \cN_+$ and the expectations of the terms $\Upsilon_4$, $\Upsilon_6$ and $\Upsilon_8$ can all be estimated by the expectation
\[\begin{split}
\big| \langle \xi , (\Upsilon_4 + \Upsilon_6 + \Upsilon_8) \xi \rangle \big| &\leq \frac{C}{N}\sum_{\substack{r\in P_{H}; q, v\in P_{L},\\ |r| \leq (\ell^{-\alpha}+ 2\ell^{-\beta}), q-r-v\neq 0} } |\eta_r| |v|  \| a_v a_{q-r-v} \xi \|  |v|^{-1} \|a_{-r}  a_q  \xi \| \\ &\leq C\ell^{(\alpha- \beta)/2}\|\cK_L^{1/2}\xi \| \|\cN_+^{1/2}\xi\|
\end{split}\]
Finally, the expectations of $\Upsilon_5$ and $\Upsilon_7$ can be bounded by
\[\begin{split}
&\big| \langle \xi, \Upsilon_5 \xi \rangle \big|
\\ &\hspace{1cm}\leq  \frac{C\ell^{\alpha}}{N^2}\sum_{\substack{r\in P_{H}, p\in P^c_{H}, \\ q, v\in P_{L}, p \neq -q,  r \neq -v}} |\eta_r|  |p| \| a_{-p}a_v a_{p+q}\xi \|  |p|^{-1} |r|  \| a_{-r} a_{r+v} a_q \xi \| \leq C\ell^{\alpha }\|\cK^{1/2}\xi \|^2 
\end{split}\]
and by 
\[\begin{split}
&\big|\langle \xi , \Upsilon_7 \xi \rangle \big| \leq \frac{C\ell^{\alpha}}{N}\sum_{\substack{r\in P_{H}, p\in P^c_{H}, \\ v\in P_{L};  p,r \neq -v} } |\eta_r| |p| \| a_{-p} a_{p+v} \xi \|  |p|^{-1} |r| \| a_{-r}a_{r+v}  \xi \|  \leq C\ell^{\alpha}\|\cK^{1/2}\xi \|^2
		\end{split}\]
\end{proof}

\subsection{Analysis of $e^{-A} \cQ_{N,\ell} e^{A}$}
\label{sec:RNellV}

In this subsection, we consider contributions to $\cR_{N,\ell}$ arising from conjugation of $\cQ_{N,\ell}$, as defined in (\ref{eq:wtGNell0to4}). 
\begin{prop}\label{prop:RNell2V}
There exists a constant $C>0$ such that		
\begin{equation*}% \label{eq:RNell2V}
e^{-A} \cQ_{N,\ell} e^{A}  = \widehat V(0)\sum_{p\in P^c_H}  a^*_pa_p (1- \cN_+ /N)  + 4\pi \frak{a}_0 \sum_{p\in P_H^c} \big[ b^*_p b^*_{-p} + b_p b_{-p} \big] +\delta_{\cQ_{N,\ell}} 
		\end{equation*}
where 
\begin{equation}\label{eq:RNell2V1}  \pm \delta_{\cQ_{N,\ell}} \leq C\ell^{(\alpha-\beta)/2}  (\cH_N +1) 
\end{equation}
for all $\alpha >4/3$, $0 < \beta < 2\a/3$, $\ell \in (0;1/2)$ small enough and $N \in \bN$ large enough. 
\end{prop}

\begin{proof}
Proceeding as in the proof of Proposition \ref{prop:RNell0}, it follows from 
Lemma \ref{lm:cNops} that 
\begin{equation}\label{eq:RNell2V2}\begin{split} 
&\pm \bigg[ \widehat V(0)\sum_{p\in P_H^c}  e^{-A} a^*_pa_p (1-N/\cN_+)e^{A} - \widehat V(0)\sum_{p\in P_H^c}  a^*_pa_p (1-N/\cN_+) \bigg]  \\
&\hspace{10cm}\leq C\ell^{\alpha/2} (\cN_+ +1)
\end{split}\end{equation}
Let us thus focus on the remaining part of $\cR_{N,\ell}^{(2,V)}$. We expand
\begin{equation}\label{eq:RNell2V3}
\begin{split}
&4\pi \frak{a}_0 \sum_{p\in  P_H^c} \Big(e^{-A} \big[ b^*_p b^*_{-p} + b_p b_{-p} \big] e^{A} - \big[ b^*_p b^*_{-p} + b_p b_{-p} \big]\Big) \\
&\hspace{4cm} = 4\pi \frak{a}_0 \int_0^1ds\;  \sum_{p\in P^c_H} e^{-sA} 
\big[b^*_p b^*_{-p}, A\big] e^{sA}+ \text{h.c.} 
\end{split}
\end{equation}
We compute 
\[\begin{split} 
\big[b^*_p b^*_{-p},  b^*_{r+v}a^*_{-r}a_v - a^*_va_{-r}b_{r+v} ] = b^*_{r+v} \big[b^*_p b^*_{-p},a^*_{-r}a_v \big]+ \big[a^*_va_{-r}b_{r+v}, b^*_p b^*_{-p}  \big]
\end{split}\]
where 
\[b^*_{r+v}\big[b^*_p b^*_{-p},a^*_{-r}a_v \big] = -b^*_{r+v}b^*_{-v}b^*_{-r} (\delta_{-p, v}+\delta_{p, v})\]
and
\[\begin{split} 
\big[a^*_va_{-r}b_{r+v}, b^*_p b^*_{-p}  \big] =&\; b^*_{v}b^*_{r}b_{r+v} (\delta_{-r,p} + \delta_{r,p}) + (1-\cN_+/N)b^*_{-r-v}a^*_{v}a_{-r}(\delta_{r+v,p}+ \delta_{r+v,-p})\\
&\; - 2N^{-1} b^*_{v}a^*_{r}a_{r+v}(\delta_{p,-r}+ \delta_{r,p}) - 2N^{-1} b^*_{p}a^*_{-p}a^*_{v}a_{-r}a_{r+v}
\end{split}\]
Using the fact that $\delta_{p,-r}= \delta_{p,r}=0$ for $r\in P_{H}$ and $p\in P_H^c$, we find that $\sum_{p\in P_H^c} \big[b^*_p b^*_{-p}, A \big]+\text{h.c.} = \sum_{i=1}^3(\Phi_i+\text{h.c.})$, where 
\[
\begin{split}
\Phi_1:=&\; -\frac2{\sqrt N} \sum_{r\in P_{H}, v\in P_{L}}  \eta_r  b^*_{r+v}b^*_{-r}b^*_{-v},	\\
\Phi_2:=&\; \frac2{\sqrt N} \sum_{\substack{r\in P_{H}, v\in P_{L} : r+v \in P_{H}^c } }  \eta_r  (1-\cN_+/N)b^*_{-r-v}a^*_{v}a_{-r},	\\
\Phi_3:=&\; -\frac2{N^{3/2}} \sum_{\substack{r\in P_{H}, v\in P_{L} , p\in P^c_{H}} } \eta_r b^*_{p}a^*_{-p}a^*_va_{-r}a_{r+v}	
\end{split} \]
Let us now bound the expectation of the operators $\Phi_i, i=1,2,3,$. By Cauchy-Schwarz, we find that
\[\begin{split}
|\langle\xi, \Phi_1 \xi\rangle| &\leq\bigg|\frac2{\sqrt N}\sum_{r\in P_{H}, v\in P_{L}}  \eta_r   \langle \xi, b^*_{r+v}b^*_{-r}b^*_{-v}\xi\rangle \bigg|\\
& \leq \frac C{\sqrt N}\sum_{r\in  P_{H}, v\in P_{L}} |\eta_r |  |v|^{-1}  \|(\cN_++1)^{1/2}\xi\| \, |v| \| b_{-v}b_{r+v}b_{-r}(\cN_++1)^{-1/2}\xi\| \\
& \leq C\ell^{(\alpha-\beta)/2} \|(\cN_++1)^{1/2}\xi\| \|\cK_L^{1/2} \xi \|
\end{split}\]
as well as
\[\begin{split}
|\langle\xi, \Phi_2 \xi\rangle|&\leq\bigg|\frac2{\sqrt N} \sum_{r\in P_{H}, v\in P_{L}:  r+v \in P_{H}^c}  \eta_r  \langle \xi, (1-\cN_+/N)b^*_{-r-v}a^*_{v}a_{-r} \xi\rangle \bigg| \\
& \leq \frac C{\sqrt N}\sum_{r\in P_{H}, v\in P_{L}} |\eta_r| |v|^{-1}  \|(\cN_++1)^{1/2}\xi\|  |v|  \|a_{-v}b_{r+v}\xi\|\\ & \leq C\ell^{(\alpha-\beta)/2} \|(\cN_++1)^{1/2}\xi\| \|\cK_L^{1/2} \xi \| \end{split}\]
To bound $\Phi_3$ we notice that 	
\[\begin{split}
\big| \langle \xi, \Phi_3 \xi\rangle \big| &\leq \frac {C\ell^{\alpha}}{N^{3/2}} \sum_{r\in P_{H}, v \in P_{L}, p\in P_H^c} |\eta_r| |p|  \|a_{p} a_v (\cN_++1)^{1/2} \xi\| |p|^{-1} |r| \|a_{-r}a_{r+v}\xi\| \\
& \leq C\ell^{\alpha} \| \cK^{1/2} \xi\|^2  \end{split} \]
With (\ref{eq:RNell2V3}), we conclude that 
\[ \begin{split}
&\pm\bigg[ 4\pi \frak{a}_0 \sum_{p\in P^c_H} \Big( e^{-A} \big[ b^*_p b^*_{-p} + b_p b_{-p} \big] e^{A} - \big[ b^*_p b^*_{-p} + b_p b_{-p} \big]\Big)\bigg]\\
&\hspace{3.4cm}\leq C \int_0^1 ds\;  e^{-sA} \big[ \ell^{(\alpha-\beta)/2} (\cN_+ + \cK_L+1)  + \ell^{\a} \cK \big] e^{sA}
\end{split} \]
Finally, we apply Prop. \ref{prop:AellNgrow}, Lemma \ref{lm:cKlowpropagation} and Cor. \ref{cor:cKimproved} to conclude that 
		\[
		\begin{split}
		&\pm\bigg[4\pi \frak{a}_0 \sum_{p\in P_H^c} \Big(e^{-A}\big[ b^*_p b^*_{-p} + b_p b_{-p} \big] e^{A} - \big[ b^*_p b^*_{-p} + b_p b_{-p} \big]\Big)\bigg] \leq C  \ell^{(\alpha-\beta)/2} (\cH_N+1)
		\end{split}
		\]
Together with the estimate (\ref{eq:RNell2V2}), we arrive at (\ref{eq:RNell2V1}). 
\end{proof}

%%%%%%%%%%%%%%%%%%%%%%%%%%%%%%%%%%%%%%%%%%%%%%%%%%%
%%%%%%%%%%%%%%%%%%%%%%%%%%%%%%%%%%%%%%%%%%%%%%%%%%%

\subsection{Contributions from $e^{-A} \cC_N e^{A}$}\label{sec:RNell3}

In this subsection, we consider contributions to $\cR_{N,\ell}$ arising from conjugation of the cubic operator $\cC_N$ defined in (\ref{eq:wtGNell0to4}). In particular, in the next proposition, we establish properties of the commutator $[\cC_N,  A]$.  
\begin{prop}\label{prop:commAcCN}
There exists a constant $C>0$ such that
\begin{equation*}%\label{eq:commAGNell3} 
\begin{split}
\big[\cC_N, A \big]  =&\;  \frac{2}{N}\sum_{r\in P_H , v\in P_L}  \big[\widehat{V}(r/N)\eta_r+\widehat{V}((r+v)/N)\eta_r\big]a^*_va_v\frac{(N-\cN_+)}N  +\delta_{\cC_N} \\ 
\end{split}
\end{equation*}
where 
\begin{equation}\label{eq:ENell3bnd}
        \begin{split} 
         | \langle \xi, \delta_{\cC_N}   \xi \rangle |\leq &\;   C\ell^{3(\alpha-\beta)/2} \|(\cV_N + \cN_++1)^{1/2}\xi\| \|\cK^{1/2}\xi\| \\ &+C\ell^{(\alpha-\beta)/2} \| (\cK_L + \cV_N + \cN_+)^{1/2}\xi\|^2 
        \end{split}
        \end{equation}
for all $\a, \b > 0$, $\ell \in (0;1/2)$ and $N \in \bN$ large enough.         
\end{prop}

\begin{proof} 
We have
\[ \big[\cC_N, A \big] = \frac{1}{N} \sum_{\substack{p,q \in \L_+^* : p+q \not = 0 \\ r \in P_H, v \in P_L}} \widehat{V} (p/N) \eta_r \big[ b_{p+q}^* a_{-p}^* a_q , b_{r+v}^* a_{-r}^* a_v - a_v^* a_{-r} b_{r+v} \big] + \hc  
\]
From (\ref{eq:comm33}), (\ref{eq:commcubic1}), (\ref{eq:commcubic2}), (\ref{eq:commcubic3}) and (\ref{eq:commcubic4}) we arrive at
\begin{equation}\label{eq:dec-CA}  \big[\cC_N, A \big] = \frac{2}{N}\sum_{r\in P_H, v\in P_L}  \big[\widehat{V}(r/N)\eta_r+\widehat{V}((r+v)/N)\eta_r\big]a^*_va_v\frac{N-\cN_+}N + \sum_{j=1}^{12} ( \Xi_j + \hc ) \end{equation}
where 
\begin{equation*}%\label{eq:Xi1to6}
\begin{split}
 \Xi_1 :=&\; -\frac{1}{N}\sum_{\substack{ r\in P_H, v\in P_L, \\  p\in\Lambda_+^*: 
 p\neq v }}\widehat{V}(p/N) \eta_r b^*_{r+v}b^*_{-r} a^*_{-p} a_{v-p},\\ 
         \Xi_2 :=&\;  \frac{1}{N}\sum_{\substack{ r\in P_H, v \in P_L, \\  p \in\Lambda_+^*: r \neq -p }}\widehat{V}(p/N)\eta_r (1-\cN_+/N)a^*_{v}a^*_{-p} a_{-r-p} a_{r+v} ,\\ 
         \Xi_3 :=&\; \frac{1}{N}\sum_{\substack{ r\in P_H, v \in P_L, \\  p\in\Lambda_+^*: r+v\neq p     }} \widehat{V}(p/N)\eta_r (1-\cN_+/N)a_v^* a^*_{-p} a_{-r}  a_{r+v-p} ,\\ 
         \Xi_4 :=&\; -\frac{1}{N^2}\sum_{\substack{ r\in P_H, v \in P_L, \\  
         p,q\in\Lambda_+^*: p+q \neq 0}}\widehat{V}(p/N)\eta_r  a_v^* a^*_{p+q}a^*_{-p}a_{-r} a_{r+v} a_q,\\ 
          \Xi_5 :=&\; -\frac{1}{N^2}\sum_{\substack{ r\in P_H, v \in P_L, \\  q \in\Lambda_+^*: r+v \neq q     }} \widehat{V}((r+v)/N) \eta_r  	a_v^*a^*_{q-r-v}a_{-r}   a_q , \\
           \Xi_6 :=&\; -\frac{1}{N^2}\sum_{\substack{ r\in P_H,v \in P_L, \\  q \in\Lambda_+^*: r \neq -q     }}  \widehat{V}(r/N) \eta_r   a_v^*a^*_{q+r}a_{r+v}   a_q	\\  
              \Xi_7 :=&\; \frac{1}{N}\sum_{\substack{ r\in P_H,v \in P_L, \\  p \in\Lambda_+^*: r+v \neq -p     }}  \widehat{V}(p/N) \eta_r   b^*_{p+r+v} b^*_{-p} a^*_{-r} a_v ,\\ 
         \Xi_{8} :=&\;\frac{1}{N}\sum_{\substack{ r\in P_H, v \in P_L, \\  p\in\Lambda_+^*: r\neq -p     }}  \widehat{V}(p/N) \eta_r  b^*_{p-r}b^*_{r+v}a^*_{-p}  a_v ,\\ 
         \Xi_{9} :=&\; -\frac{1}{N}\sum_{\substack{ r \in P_H, v \in P_L, \\  q \in\Lambda_+^*: q\neq v     }} \widehat{V}(v/N) \eta_r	 b^*_{q-v}b^*_{r+v}a^*_{-r} a_q ,\\ 
         \Xi_{10} :=&\; \frac{1}{N}\sum_{\substack{ r\in P_H, v \in P_L, \\  q\in\Lambda_+^*: r\neq -q     }} \widehat{V}(r/N) \eta_r	 b^*_{q+r}a_v^*a_qb_{r+v} ,\\ 
   	\end{split}
	\end{equation*}
as well as 
	\begin{equation*}%\label{eq:Xi7to12}
        \begin{split}
      \Xi_{11} :=&\; -\frac{1}{N}\sum_{\substack{ r\in P_H, v \in P_L, \\  p\in\Lambda_+^*: p\neq -v     }} \widehat{V}(p/N) \eta_r	 b^*_{p+v}a^*_{-p} a_{-r}b_{r+v}	 ,\\ 
         \Xi_{12} :=&\; \frac{1}{N}\sum_{\substack{ r\in P_H, v \in P_L, \\  q\in\Lambda_+^*: q\neq r+v     }} \widehat{V}((r+v)/N) \eta_r	 b^*_{q-r-v}a_v^* a_{-r} b_q \\ 
        \end{split}
        \end{equation*}      
In fact, the first term on the r.h.s. of (\ref{eq:dec-CA}) arises from the second and fourth 
terms on the r.h.s. of (\ref{eq:commcubic2}), together with their Hermitean conjugates. 
The commutator (\ref{eq:commcubic1}) yields $\Xi_1$, the remaining terms from 
(\ref{eq:commcubic2}) produce the contributions $\Xi_2$ to $\Xi_6$, from (\ref{eq:commcubic3}) 
we find the operators $\Xi_7$ to $\Xi_9$ and from (\ref{eq:commcubic4}) we obtain $\Xi_{10}, \Xi_{11},\Xi_{12}$. 

To conclude the proof of the proposition, we have to show that all terms $\Xi_j$, $j=1,\dots , 12$, satisfy the bound (\ref{eq:ENell3bnd}). The expectation of $\Xi_1$ can be controlled with Cauchy-Schwarz by
\[\begin{split}
\big|   \langle \xi, \Xi_1 \xi \rangle \big| &\leq  \frac{C\ell^{\alpha}}{N}\sum_{\substack{ r\in P_H,v \in P_L, \\  p\in\Lambda_+^*:  p\neq v     }} |\eta_r|  \| (\cN_++1)^{1/2}a_{v-p} \xi\|  |r| 
\| a_{-r}a_{r+v}a_{-p} (\cN_++1)^{-1/2}\xi\| \\
&\leq C\ell^{3(\alpha-\beta)/2} \|(\cN_++1)^{1/2}\xi\| \|\cK^{1/2}\xi\|
\end{split}\]
The same bound applies (after relabeling) to $\Xi_9$; we find
\[ \big|   \langle \xi, \Xi_9 \xi \rangle \big| \leq C \ell^{3(\alpha-\beta)/2} \|(\cN_++1)^{1/2}\xi\| \|\cK^{1/2}\xi\| \]
Also the expectations of the terms $\Xi_2$, $\Xi_3$ and (again after relabeling) of the terms $\Xi_5$, $\Xi_6, \Xi_{10}$, $\Xi_{12}$ can be bounded similarly. We find
\[\begin{split}
&|  \langle \xi, \Xi_2 \xi \rangle   | + |  \langle \xi, \Xi_3 \xi \rangle   |+ |  \langle \xi, \Xi_5 \xi \rangle   | + |  \langle \xi, \Xi_6 \xi \rangle   |+ |  \langle \xi, \Xi_{10} \xi \rangle   | + |  \langle \xi, \Xi_{12} \xi \rangle   |\\
&\leq  \frac{C\ell^{\alpha}}{N}\sum_{r\in P_H, v \in P_L, p \in\Lambda_+^*}  \!\!\! \Big(  |\eta_r|  \| a_{v}a_{-p}  \xi\|  |r+v| \|  a_{r+v}a_{-r-p} \xi\| + |\eta_r|  \|  a_{-p} a_v \xi\|  |r| \| a_{-r}  a_{r+v-p}  \xi\| \\ &\hspace{4cm} +   |\eta_r|  \| a_v a_{p-r-v} \xi\|  |r| \| a_{-r}   a_p \xi\| +  |\eta_r|  \|  a_v a_{p+r}\xi\| |r+v|   \| a_{r+v}   a_p \xi\| \\ &\hspace{4cm} +   |\eta_r|  \| a_{p+r}a_v \xi\|  |r+v| \|a_{r+v}a_p  \xi\| +  |\eta_r|  \|  a_{p-r-v} a_v  \xi\| |r|  \| a_{-r} a_p \xi\|\Big)\\ &\leq C\ell^{3(\alpha-\beta)/2} \|(\cN_++1)^{1/2}\xi\| \|\cK^{1/2}\xi\| \end{split}\]
To control the remaining terms, we switch to position space and use the potential energy operator 
$\cV_N$. We start with $\Xi_4$. Applying Cauchy-Schwarz, we find
\[\begin{split}
|\langle \xi, \Xi_4 \xi\rangle | =&\; \bigg|\frac{1}{N}\int_{\Lambda^2}dxdy\; N^2 V(N(x-y)) \sum_{r \in P_H, v \in P_L} \eta_r  \langle \xi , \check{a}^*_{x}\check{a}^*_{y} a_v^* a_{-r} a_{r+v} \check{a}_x \xi\rangle  \bigg|\\ 
&\leq  \frac{1}{N}\int_{\Lambda^2}dxdy\; N^2 V(N(x-y)) \sum_{r \in P_H, v \in P_L} |\eta_r| \|a_v \check{a}_{x}\check{a}_{y} \xi \| \|a_{-r} a_{r+v} \check{a}_x\xi \| \\ &\leq C\ell^{\alpha/2} \|\cV_N^{1/2}\xi\| \|\cN_+^{1/2}\xi\|
\end{split}\]
Next, we rewrite $ \Xi_7$, $\Xi_8$ and $\Xi_{11}$ as 		
\[\begin{split}
\Xi_7 =&\; \int_{\Lambda^2}dxdy\; N^2 V(N(x-y)) \sum_{r \in P_H, v \in P_L}  e^{i(r+v)x}\eta_r   \check{b}^*_{x} \check{b}^*_{y} a^*_{-r} a_v,\\ 
\Xi_8 = &\;  \int_{\Lambda^2}dxdy\; N^2 V(N(x-y)) \sum_{r \in P_H, v \in P_L}   e^{-irx}\eta_r  \check{b}^*_{x} \check{b}^*_{y} a^*_{r+v}a_v,  \\
\Xi_{11} =&\;  -\int_{\Lambda^2}dxdy\; N^2 V(N(x-y)) \sum_{r \in P_H, v \in P_L}   e^{ivx}\eta_r  \check{b}^*_{x} \check{b}^*_{y} a_{-r}b_{r+v}
\end{split}\]
Thus, we obtain 
\[\begin{split}
|\langle \xi, \Xi_7 \xi \rangle | &\leq \int_{\Lambda^2}dxdy\; N^2 V(N(x-y))\!\! \sum_{r\in P_H} \!\! \| \check{a}_{x}\check{a}_{y} a_{-r} \xi \|  |\eta_r | \Big\|\sum_{v\in P_L} e^{ivx}a_{v}  \xi \Big \|  \\
		&\leq C\ell^{\alpha/2} \|\cV_N^{1/2}\xi\| \bigg[ \int_\Lambda dx \; \sum_{v, v'\in P_L} e^{i(v-v')x}\langle \xi, a^*_{v'}a_{v}\xi \rangle \bigg]^{1/2} \\ &\leq C\ell^{\alpha/2} \|\cV_N^{1/2}\xi\| \|\cN_+^{1/2}\xi\|
		\end{split}\]
as well as 
\[\begin{split}
&|\langle \xi, \Xi_8 \xi\rangle |+ |\langle \xi, \Xi_{11} \xi\rangle |\\
&\leq C \int_{\Lambda^2}dxdy\; N^2 V(N(x-y)) \\ &\hspace{1cm} \times \sum_{r\in P_H, v \in P_L} \Big(|v|^{-1} \|\check{a}_{x} \check{a}_{y} a_{r+v} \xi \| |\eta_r| |v|  \| a_v\xi \| + C\ell^{\alpha}|\eta_r|  \|\check{a}_{x} \check{a}_{y}  \xi \| |r| \| a_{-r}b_{r+v}\xi \| \Big)\\
		&\leq C\ell^{(\alpha-\beta)/2} \|\cV_N^{1/2}\xi\| \|\cK_{L}^{1/2}\xi\| + C\ell^{3(\alpha-\beta)/2} \|\cV_N^{1/2}\xi\| \|\cK^{1/2}\xi\|
		\end{split}\]
Collecting all the bounds above, we arrive at (\ref{eq:ENell3bnd}). 
\end{proof}

%%%%%%%%%%%%%%%%%%%%%%%%%%%%%%%%%%%%%%%%%%%%%%%%%%%
%%%%%%%%%%%%%%%%%%%%%%%%%%%%%%%%%%%%%%%%%%%%%%%%%%%

%\subsection{Analysis of $\cR_{N,\ell}^{(4)}$}\label{sec:RNell4}

%%%%%%%%%%%%%%%%%%%%%%%%%%%%%%%%%%%%%%%%%%%%%%%%%%%
%%%%%%%%%%%%%%%%%%%%%%%%%%%%%%%%%%%%%%%%%%%%%%%%%%%

\subsection{Proof of Proposition \ref{prop:RNell}}\label{sub:Rfin}

Let us now combine the results of Sections \ref{sec:aprioribnds}-\ref{sec:RNell3} to prove Proposition \ref{prop:RNell}. Here, we assume $\alpha > 3$ and $\alpha/2 < \beta < 2\a/3$. 

From Prop. \ref{prop:RNell0} and Prop. \ref{prop:RNell2V} we obtain that 
\begin{equation*}%\label{eq:proofpropRnell1}
\begin{split} 
\cR_{N,\ell} \geq &\; 4\pi \frak{a}_0 (N-\cN_+) + \big[\widehat V(0)-4\pi \frak{a}_0\big] \cN_+ ( 1- \cN_+/N ) \\
& +  \widehat V(0) \sum_{p\in P_H^c}   a^*_p a_p (1-\cN_+/N) + 4\pi \frak{a}_0 \sum_{p\in P^c_H}  \big[ b^*_p b^*_{-p} + b_p b_{-p} \big] \\
& + \cK  + \cC_N+ \cV_N + \int_0^1 ds\; e^{-sA}\big [ \cK  + \cC_N+ \cV_N, A \big]e^{sA} 
\\ &- C  \ell^{(\a-\beta)/2} (\cH_N + 1) 
\end{split}
\end{equation*}
with $\cC_N$ defined as in (\ref{eq:wtGNell0to4}). From Prop.  \ref{prop:commAcVN}, Prop. \ref{prop:RNell2K} and Prop.~\ref{prop:commAcCN}, we can write, for $N$ large enough, 
\begin{equation*}%\label{eq:proofpropRnell3}
\begin{split}  
&[ \cK  + \cC_N+ \cV_N, A \big] \\
&\geq  -\frac{1}{\sqrt N}\sum_{\substack{p\in\Lambda_+^*, q\in P_{L},\\ p \neq-q} } \!\!\!\widehat V(p/N) \big[ b^*_{p+q}a^*_{-p} a_q+ \text{h.c.}\big]  +\frac{8\pi \mathfrak{a}_0}{\sqrt N}\sum_{\substack{p \in P_{H}^c, q \in P_{L},\\ p \neq-q} } \big[ b^*_{p+q}a^*_{-p} a_q + \text{h.c.}\big] \\
&\hspace{0.5cm} + \frac{2}{N}\sum_{r\in P_H,v \in  P_L}  \big[\widehat{V}(r/N)\eta_r+\widehat{V}((r+v)/N)\eta_r\big] a^*_v a_v (1-\cN_+/N) \\ &\hspace{0.5cm} -  C (\ell^{\a-2} + \ell^{(\a-\beta)/4}) (\cN_+ + \cV_N + \cK_L) - C (\ell^{5 (\a-\b)/2} + \ell^{(3\a + \beta)/4} + \ell^{2\a-2}) \cK  
\end{split}
\end{equation*}
%where \[ \pm \delta_2 \leq C \left[ \ell^{(\a-\beta)/2} + N^{-1} \ll^{-\a-\beta/2} \right] (\cN_+ + \cV_N + 
%\cK_L) + C \ell^{3(\a - \beta)/2} (\cH_N + 1) \]
From Prop. \ref{prop:AellNgrow}, Lemma \ref{lm:cKlowpropagation}, Cor. \ref{cor:cVNuniformbnd} and Cor. \ref{cor:cKimproved} and recalling the definition (\ref{eq:wtGNell0to4}) of the operator $\cC_N$, 
we deduce that 
%\[ \pm e^{sA} \delta_2 e^{-sA} \leq C \left[ \ell^{(\a-\beta)/2} + \ell^{\a- 2\beta} \right] (\cH_N + 1) \]
%{\bf [use cond. on $\a,\beta$]} 
%Hence, we find 
\begin{equation}\label{eq:proofpropRnell4}
\begin{split}  
&\int_0^1 ds\;  e^{-sA}[ \cK  + \cC_N+ \cV_N, A \big] e^{sA} \\
&\geq \int_0^1 ds \;  e^{-sA} \Big[-\cC_N +\frac{8\pi \mathfrak{a}_0}{\sqrt N}\sum_{\substack{p\in P_{H}^c, q\in P_{L},\\ p\neq-q} } \big[ b^*_{p+q}a^*_{-p} a_q + \text{h.c.}\big] \\
&\hspace{2.5cm} + \frac{2}{N}\sum_{r\in P_H,v\in P_L}  \big[\widehat{V}(r/N)\eta_r+\widehat{V}((r+v)/N)\eta_r\big] a^*_v a_v\frac{(N-\cN_+)}{N}\Big]e^{sA} \\
&\hspace{0.5cm}+\frac{1}{\sqrt N}\int_0^1 ds\;   \sum_{\substack{p\in\Lambda_+^*, q\in P_{L}^c,\\ p\neq-q} } \widehat V(p/N) e^{-sA}\big[ b^*_{p+q}a^*_{-p} a_q+ \text{h.c.}\big] e^{sA} \\ &\hspace{0.5cm}-C (\ell^{(\a-\beta)/4} + \ell^{\a -2} + \ell^{2\a-3\b})  (\cH_N+1)	
\end{split}
\end{equation}
The expectation of the operator on the fourth line can be estimated after switching to position space with  Cor. \ref{cor:cVNuniformbnd} and Cor. \ref{cor:cKimproved}. We find 
\begin{equation}\label{eq:proofpropRnell5}
\begin{split}  
\bigg| \frac{1}{\sqrt N}\int_0^1 &ds\;   \sum_{\substack{p\in\Lambda_+^*, q \in P^c_{L}, \\ p\neq-q} } \widehat V(p/N) \langle \xi, e^{-sA}  b^*_{p+q}a^*_{-p} a_qe^{sA} \xi \rangle \bigg| \\
&\leq   \int_0^1 ds\;  \int_{\Lambda^2}dxdy\; N^{5/2}V(N(x-y)) \|\check{a}_x \check{a}_y  e^{sA} \xi \| \Big \|   \sum_{q\in P^c_{L}} e^{iqx} a_q e^{sA} \xi \Big\|\\
&\leq   C\int_0^1 ds\; \|\cV_N^{1/2}e^{sA}\xi\| \bigg[ \int_\Lambda dx\;  \sum_{q, q'\in P^c_{L}} 
e^{i(q-q')x } \langle  e^{sA} \xi, a^*_{q'} a_q e^{sA} \xi \rangle \bigg]^{1/2}\\
&\leq   C\ell^{\beta} \int_0^1 ds\; \|\cV_N^{1/2}e^{sA} \xi\| \|\cK^{1/2} e^{sA} \xi\| \leq C\ell^{(3\beta - \a)/4}  \|(\cH_N+1)^{1/2}\xi\|^2
\end{split}
\end{equation}
Next, we consider the term on the third line of (\ref{eq:proofpropRnell4}). With Lemma \ref{3.0.sceqlemma}, part ii), and since $\a > 1$, we have 
\[\begin{split}
&\bigg|\frac{1}{N}\sum_{r\in P_H}  \big[\widehat{V}(r/N)\eta_r+\widehat{V}((r+v)/N)\eta_r\big] -  \big [ 16 \pi \mathfrak{a}_0 - 2\widehat{V}(0)  \big] \bigg| \leq \frac{C \ell^{-\a} |v|}{N}  \end{split}\]
for every $v\in P_L$. With Lemma \ref{lm:cNops}, Prop. \ref{prop:AellNgrow} and Lemma \ref{lm:cKlowpropagation} we obtain, for $N \geq \ell^{-3\a}$, 
\begin{equation}\label{eq:proofpropRnell6}
\begin{split}  
&\pm\bigg[\frac{1}{N} \sum_{r\in P_H,v\in P_L} \big[\widehat{V}(r/N)\eta_r+\widehat{V}((r+v)/N)\eta_r\big]  e^{-sA}  a^*_v a_v \frac{(N-\cN_+)}{N}e^{sA}\\
& \hspace{5.2cm}- \big [ 16\pi \mathfrak{a}_0 - 2\widehat{V}(0)  \big]\sum_{v\in P_L}a^*_va_v\frac{(N-\cN_+)}{N} \bigg] \\
&\hspace{0.5cm}\leq  C (N^{-1} \ell^{-\beta} + \ell^{\alpha/2}) (\cH_N+1) \leq C \ell^{\a/2} (\cH_N + 1) 
\end{split}
\end{equation}
To handle the second term on the second line of (\ref{eq:proofpropRnell4}), we apply Prop.~\ref{prop:RNell2K} and then Prop. \ref{prop:AellNgrow}, Lemma \ref{lm:cKlowpropagation} and Cor. \ref{cor:cKimproved} to conclude, again for $N \geq \ell^{-3\a}$, 
\begin{equation*}%\label{eq:proofpropRnell7}
\begin{split}
&\pm\bigg(\frac{8\pi \mathfrak{a}_0}{\sqrt N}\int_0^1 ds\;  \sum_{\substack{p\in P^c_{H}, q\in P_{L}, \\ p\neq-q} } \Big[e^{-sA} b^*_{p+q}a^*_{-p} a_qe^{sA}- b^*_{p+q}a^*_{-p} a_q \Big]  + \text{h.c.}\bigg)\\
& \hspace{1cm}= \pm \bigg(\frac{8\pi \mathfrak{a}_0}{\sqrt N}\int_0^1 ds\; \int_0^s dt\;\sum_{\substack{p\in P_{H}^c, q\in P_{L},\\ p\neq-q} } e^{-tA} \Big[  b^*_{p+q}a^*_{-p} a_q , A \Big] e^{tA}\bigg) \\ 
& \hspace{1cm} \leq   C \big(\ell^{(2\a - 3\b)} + \ell^{(\alpha-\beta)/2} \big) (\cH_N+1)	
\end{split}
\end{equation*}
As for the first term on the second line of (\ref{eq:proofpropRnell4}), we use again Prop. \ref{prop:commAcCN}. Proceeding then as in (\ref{eq:proofpropRnell6}), we have 
		\begin{equation}\label{eq:proofpropRnell9}
		\begin{split} 
		 \int_0^1 ds\;  e^{-sA} \cC_Ne^{sA } &=  \cC_N + \int_0^1 ds\;  \int_0^{s}dt \;  e^{-tA } [\cC_N, A] e^{tA}\\
		 &\leq \cC_N + \big [ 16 \pi \mathfrak{a}_0 - 2\widehat{V}(0)  \big]\sum_{p\in P_L}a^*_pa_p\frac{(N-\cN_+)}{N} \\ &\hspace{.3cm} +C \big( \ell^{(\a-\b)/2} + \ell^{2\a-3\b} \big) (\cH_N + 1)		
\end{split}
\end{equation} 
Inserting the bounds (\ref{eq:proofpropRnell5})-(\ref{eq:proofpropRnell9}) into (\ref{eq:proofpropRnell4}) and using additionally the simple bounds 
\[  0\leq \sum_{p\in P_L^c \cap P_H }a^*_p a_p \leq \sum_{p\in P_L^c} a^*_p a_p \leq \ell^{2\beta} \cK \]
and
\[ \begin{split} 
\Big| \frac{8\pi \mathfrak{a}_0}{\sqrt N}\sum_{\substack{p\in P^c_{H}, q\in P^c_{L}, \\ p\neq-q} } \langle \xi , b^*_{p+q}a^*_{-p} a_q \xi \rangle \Big| \leq & \; \frac{C \ell^\b}{\sqrt{N}} \sum_{\substack{p\in P^c_{H}, q\in P^c_{L}, \\ p\neq-q} } |p| \| a_{-p} a_{p+q} \xi \| |p|^{-1} |q| \| a_q \xi \| \\ \leq \; &\frac{C \ell^{\b-\a/2}}{\sqrt{N}} \, \| \cK^{1/2} \cN_+^{1/2} \xi \| \bigg[ \sum_{q \in P_L^c} |q|^2 \| a_q \xi \|^2 \bigg]^{1/2} \\ \leq \; &C \ell^{\beta -\a/2} \| \cK^{1/2} \xi \|^2 \end{split} \]
we arrive at 
	\begin{equation}\label{eq:proofpropRnell-l}
		\begin{split} 
		\cR_{N,\ell}\geq &\;  4\pi \frak{a}_0 (N-\cN_+) + 4\pi \frak{a}_0 \,\cN_+\frac{(N-\cN_+)}{N} \\
		& + 8\pi \mathfrak{a}_0\sum_{p\in P_H^c}a^*_pa_p \frac{(N-\cN_+)}{N} + 4\pi \frak{a}_0\sum_{p\in P_H^c}  \big[ b^*_p b^*_{-p} + b_p b_{-p} \big]  \\
		&+\frac{8\pi \mathfrak{a}_0}{\sqrt N}\sum_{\substack{p\in P^c_{H}, q\in\Lambda_+^*: p\neq-q} } \big[ b^*_{p+q}a^*_{-p} a_q + \text{h.c.}\big] + \big(1 - C\ell^\kappa \big)(\cH_N + 1)
		\end{split}
		\end{equation}
with $\kappa = \min [ (\alpha-\beta)/4; \alpha-3; \beta-\alpha/2; 2\a -3\b] > 0$ under the assumptions $\alpha > 3$ and $\alpha/2 <  \beta < 2\alpha/3$.  
	
We define now the function $\nu_\ell \in L^\infty (\Lambda)$ by setting 
\[ \nu_\ell (x) := 8\pi \frak{a}_0  \sum_{p \in \{ 0 \} \cup P_H^c} e^{i p \cdot x} =8\pi \frak{a}_0  \sum_{p \in \L^* : |p| \leq \ell^{-\a}} e^{ip \cdot x} \]
In other words, $\nu_\ell$ is defined so that $\widehat{\nu}_\ell (p) = 8\pi \frak{a}_0$ for all $p \in \L^*$ with $|p| \leq \ell^{-\a}$ and $\widehat{\nu}_\ell (p) = 0$ otherwise. Observe, in particular, that $\widehat{\nu}_\ell (p) \geq 0$ for all $p \in \L^*$. Proceeding as in (\ref{eq:cLNj}), but now with $\widehat{V} (p/N)$ replaced by $\widehat{\nu}_\ell (p)$, we find that 
\[ \begin{split} U_N \left[ \frac{1}{N} \sum_{i<j}^N \nu_\ell (x_i - x_j) \right] U_N^* = \; &\frac{(N-1)}{N} 4 \pi \frak{a}_0 (N-\cN_+) + 4 \pi \frak{a}_0 \, \cN_+ \frac{(N-\cN_+)}{N} \\ &+ 8\pi \frak{a}_0 \sum_{p \in P_H^c}  a_p^* a_p \frac{(N-\cN_+)}{N} + 4\pi \frak{a}_0 \sum_{p \in P_H^c} (b_p^* b_{-p}^* + b_p b_{-p}) \\ &+ \frac{8\pi \frak{a}_0}{\sqrt{N}} \sum_{p \in P_H^c, q \in \L^*_+, p \not = -q} [b_{p+q}^* a_{-p}^* a_q + a_q^* a_{-p} b_{p+q} ] \\ &+ \frac{4\pi \frak{a}_0}{N}  \sum_{p,q \in \L_+^*, r \in P_H^c : r \not = -p, -q} a_{p+r}^* a_q^* a_p a_{q+r}
\end{split} \]
Comparing with (\ref{eq:proofpropRnell-l}) and noticing that 
\[ \begin{split}  \frac{4\pi \frak{a}_0}{N}  \sum_{\substack{p,q \in \L_+^*, r \in P_H^c : \\ r \not = -p, -q}} \langle \xi , a_{p+r}^* a_q^* a_p a_{q+r} \xi \rangle &\leq \frac{C}{N} \sum_{\substack{p,q \in \L_+^* , r \in P_H^c : \\ r \not = -p,-q}} \| a_{p+r} a_q \xi \| \| a_p a_{q+r} \xi \| \\ &\leq \frac{C \ell^{-3\a}}{N} \| \cN_+ \xi \|^2 \end{split} \]
we conclude that 
\begin{equation}\label{eq:RNfina}  \cR_{N,\ell}\geq U_N \left[ \frac{1}{N} \sum_{i<j}^N \nu_\ell (x_i - x_j) \right] U_N^* +  (1 - C \ell^\kappa ) \cH_N - C \ell^{-3\a} \cN_+^2 /N  -C \ell^\kappa \end{equation} 
Following standard arguments, for example from \cite[Lemma 1]{Sei}, we observe now that, since $\widehat{\nu}_\ell (p) \geq 0$ for all $p \in \L^*$, 
\[ \begin{split} 
0 &\leq \int_{\Lambda^2} dx dy \, \nu_\ell (x-y) \left[ \sum_{j=1}^N \delta (x - x_j) - N \right]  \left[ \sum_{i=1}^N \delta (y-x_i) - N \right] \\ &= \sum_{i,j =1}^N  \nu_\ell (x_i -x_j) - N^2 \widehat{\nu}_\ell (0) = 2 \sum_{i<j}^N \nu_\ell (x_i -x_j) + N \nu_\ell (0) - N^2 \widehat{\nu}_\ell (0) \end{split} \]
This implies that 
\[ \frac{1}{N} \sum_{i<j}^N \nu_\ell (x_i -x_j) \geq \frac{N}{2} \widehat{\nu}_\ell (0) - \nu_\ell (0) \geq 4\pi \frak{a}_0 N - C \ell^{-3\a} \]
From (\ref{eq:RNfina}), we finally obtain 
\[  \cR_{N,\ell}\geq 4\pi \frak{a}_0 N + (1-C \ell^\kappa) \cH_N - C \ell^{-3\a} \cN_+^2 / N - C \ell^{-3\a} \]
This completes the proof of Proposition \ref{prop:RNell}.

\appendix
\section{Properties of the scattering function}\label{appx:sceq}

In this appendix we give a proof of Lemma \ref{3.0.sceqlemma} containing the basic properties of the solution of the Neumann problem \eqref{eq:scatl}.

\begin{proof}[Proof of Lemma \ref{3.0.sceqlemma}]
Part i) and the bounds $0\leq f_\ell, w_\ell \leq 1$ in part ii) follow from  \cite[Lemma A.1]{ESY0}. 
We prove \eqref{eq:Vfa0}. We set $r=|x|$ and $m_\ell(r)=rf_\ell(r)$. We rewrite \eqref{eq:scatl} as
\begin{equation}\label{eq:mlambda}
-m_\ell''(r)+\frac{1}{2}V(r)m_\ell(r)=\l_\ell m_\ell(r)
\end{equation}
Let $R > 0$ be the radius of the support of $V$, so that $V(x) = 0$ for all $x \in \bR^3$ with $|x| > R$. For $r\in(R,N\ell]$ we can solve \eqref{eq:mlambda} explicitly; since the boundary conditions 
$f_\ell (N\ell) = 1$ and $(\partial_r f_\ell) (N\ell) = 0$ translate into $m_\ell (N\ell) = N\ell$ and 
$m'_\ell (N\ell) = 1$, we find
\begin{equation}\label{eq:Exp0} 
 m_\ell(r)=\l_\ell^{-1/2}\sin(\l_\ell^{1/2}(r-N\ell))+N\ell\cos(\l_\ell^{1/2}(r-N\ell))
\end{equation}
With the result of part i), we obtain 
\begin{equation}\label{eq:mExpanded}
 m_\ell(r)=r-\mathfrak{a}_0+\frac{3}{2}\frac{\frak{a}_0}{N\ell}r-\frac{1}{2}\frac{\frak{a}_0}{(N\ell)^3}r^3+\cO(\frak{a}_0^2(N\ell)^{-1})
\end{equation}
for all $r \in (R,N\ell]$ (the error is uniform in $r$). Using the scattering equation we can write
\begin{equation*}
\begin{split}
\int V(x) f_\ell (x) dx&=4\pi\int_0^{N\ell} dr\, rV(r)m_\ell(r)=8\pi\int_0^{N\ell}dr\, (rm_\ell''(r)+\l_\ell rm_\ell(r))\\
\end{split}
\end{equation*}
Integrating by parts, we observe that the first contribution on the r.h.s. vanishes (because $m_\ell (N\ell) = N\ell$, $m'_\ell (N\ell) = 1$ and $m_\ell (0) = 0$). With the result of part i) and with \eqref{eq:mExpanded}, we get 
\begin{equation*}
\begin{split}
8\pi\l_\ell\int_0^{N\ell}dr\,  rm_\ell(r)&= 8\pi\l_\ell\left(\frac{(N\ell)^3}{3}+\cO\big(\frak{a}_0(N\ell)^2\big)\right)= 8\pi \frak{a}_0 +\mathcal{O} \big(\frak{a}_0^2  / \ell N\big)
\end{split}
\end{equation*}
which proves \eqref{eq:Vfa0}.  
%Starting from the expansion \eqref{eq:mExpanded} and recalling that $w_\ell=1-f_\ell$, an easy calculation leads to \eqref{eq:intw}. 

We consider now part iii). Combining \eqref{eq:mExpanded} for $r \in (R,N\ell]$ with  $w_\ell (r) \leq 1$ for $r \leq R$, we obtain the first bound in (\ref{3.0.scbounds1}). To show the second bound in (\ref{3.0.scbounds1}), we observe that, for $r \in (R,N\ell]$, (\ref{eq:Exp0}) and the estimate 
in part i) imply that $|f'_\ell (r)| \leq C r^{-2}$, for a constant $C > 0$ independent of $N$ and $\ell$, provided $N \ell \geq 1$. For $r < R$ we write, integrating by parts, 
\begin{equation*}
\begin{split}
 f'_\ell(r)&=\frac{m_\ell'(r)r-m_\ell(r)}{r^2}=\frac{1}{r^2}\int_0^rds\, s\,m_\ell''(s)\\
\end{split}
\end{equation*}
With \eqref{eq:mlambda} and since $0 \leq f_\ell \leq 1$, we obtain
\begin{equation*}
\begin{split}
 |f'_\ell(r)|&=\Big|\frac{1}{r^2}\int_0^rds\, s\,\Big[\frac{1}{2}V(s)m_\ell(s)-\l_\ell m_\ell(s)\Big]\Big|\\
 &= \frac{1}{r^2}\Big[\frac{1}{8\pi}\int_{|x| < r} dx\, \,V(x)f_\ell(x)+\l_\ell \int_{|x| < r} dx \, f_\ell (x) 
 \Big] \leq C (\| V \|_3 + 1)
\end{split}
\end{equation*}
for a constant $C > 0$ independent of $N$ and $\ell$, if $N\ell \geq 1$ and for all $0< r < R$. This concludes the proof of the second bound  in (\ref{3.0.scbounds1}). 

To show part iv), we use \eqref{eq:wellp} and we observe that, by \eqref{eq:lambdaell}, \eqref{eq:Vfa0} and $f_\ell\leq1$, there exists a constant $C > 0$ such that
\begin{equation*}
\begin{split}
 |\widehat{w}_\ell(p/N)| &\leq \frac{N^2}{p^2} \left[  \big(\widehat V(./N)\ast \widehat{f}_{N,\ell}\big)(0)  + C \ell^{-3}  \big(\widehat{\chi}_\ell *\widehat{f}_{N,\ell}\big)(0)  \right]\\
 &\leq \frac{N^2}{p^2}\left[\int V(x)f_\ell(x) dx + C \ell^{-3} \int \chi_\ell (x) f_\ell (Nx) dx \right] \leq \frac{CN^2}{p^2}
\end{split}
\end{equation*}
for all $N \in \bN$ and $\ell > 0$, if $N \ell \geq 1$. 
\end{proof}

\section{Proof of Eq. (\ref{eq:BEC1})} 
\label{app:NRS}

In this section we outline the proof of (\ref{eq:BEC1}) from \cite{NRS}, adapting it to the translation invariant case. We follow the main steps summarized in \cite[Section 2]{NRS} and indicate some minor modifications due to the translation invariant setting. The proof follows very closely \cite{NRS} and we reproduce it here for the convenience of the reader only. Before we start, let us define the Gross-Pitaevskii functional 
$ \cE_{\text{GP}}: D_{\text{GP}}\to \bR$ by
		\begin{equation*}%\label{eq:defGPfunc}
		\cE_{\text{GP}}(u)  = \sum_{p\in\Lambda^*} \Big[ p^2 |\widehat u_p|^2 + 4\pi \mathfrak{a}_0 \big|\big(\widehat {|u|^2} \ast \widehat {|u|^2}\big)_p \big|^2 \Big] = \int_{\Lambda} \Big[|\nabla u(x)|^2 + 4\pi \mathfrak{a}_0 |u(x)|^4 \Big] \;dx\end{equation*}
on the domain
		\begin{equation}\label{eq:defGPfuncdom} D_{\text{GP}} = \Big\{ u\in L^2(\Lambda): \sum_{p\in \Lambda^*} p^2 |\widehat u(p)|^2 <\infty  \Big\}=\cQ(-\Delta) \end{equation}  
Here, $ \cQ(-\Delta)\subset L^2(\Lambda)$ denotes the form domain of the Laplacian $-\Delta$ with periodic boundary conditions. In particular, we have that the set of functions
		\[ \bigg\{ x\mapsto  \sum_{p\in \Lambda^*: |p|\leq M } \widehat u_p e^{ipx} : u \in \cQ(\Delta), M\in\NN  \bigg\} \subset C^\infty(\bT^3)  \]
is a form core for $ -\Delta$. Since we work with periodic boundary conditions, we identify in the following by slight abuse of notation the box $\Lambda  =[ 0; 1]^{3}$ with the unit torus $ \bT^3$ and denote by $ H^k(\Lambda)$, $k\in\NN_0$, the Sobolev spaces on $ \bT^3$ s.t. for example $ H^1(\Lambda) = \cQ(-\Delta)$.
		\begin{lemma}\label{lm:GPmin}
		The Gross-Pitaevskii functional $\cE_{\text{GP}}$ has a unique, positive minimizer in $ D_{\text{GP}}\cap\{u\in L^2(\Lambda):\|u\|_2=1\}$, given by the constant function $ \ph_0 =1_{|\Lambda}$. Moreover, any minimizer $\psi \in \cE_{\text{GP}}$ in $ D_{\text{GP}}\cap\{u\in L^2(\Lambda):\|u\|_2=1\}$ is given by $\psi =  c \ph_0$ for some constant $ c\in\mathbb{C}$ with $ |c|=1$.
		%, and satisfies the Euler-Lagrange equation
		%		\begin{equation}\label{eq:GPminEL} -\Delta \ph + 8\pi \mathfrak{a}_0 |\ph|^2 \ph = \mu\ph \end{equation}
		%where $ \mu = \cE_{\text{GP}} (\ph) + 4\pi \mathfrak{a}_0 \|\ph\|_4^4$.
		\end{lemma}
\begin{proof} 
Let $ \psi\in D_{\text{GP}}\cap \{u\in L^2(\Lambda):\|u\|_2=1\}$. Then we can bound
		\begin{equation}\label{eq:lmGPmin1}\begin{split}
		\cE_{\text{GP}}(\psi) \geq & \max_{p\in \Lambda^*}\big(p^2|\widehat{\psi}_p|^2\big) + 4\pi \mathfrak{a}_0 \|\psi\|_4^4  \geq  \max_{p\in \Lambda^*} \big(p^2|\widehat{\psi}_p|^2\big) +  4\pi \mathfrak{a}_0\geq 4\pi \mathfrak{a}_0 =\cE_{\text{GP}}(\ph_0)
		\end{split}\end{equation}
because, by H\"older's inequality, $\|\psi\|_4 = \|\ph_0 \|_{4/3} \|\psi\|_4  \geq  \|\psi\|_2 =1$. This shows that $\ph_0$ is a minimizer of $ \cE_{\text{GP}} $ in $ D_{\text{GP}}\cap\{u\in L^2(\Lambda):\|u\|_2=1\}$. Moreover, (\ref{eq:lmGPmin1}) is strict whenever $\max_{p\in \Lambda^*} \big(p^2|\widehat{u}_p|^2\big)>0 $. This implies that any minimizer $ \psi\in D_{\text{GP}}\cap \{f\in L^2(\Lambda):\|f\|_2=1\}$ of $\cE_{\text{GP}}$ is such that its Fourier transform $ (\widehat \psi_p)_{p\in \Lambda^*}$ satisfies $ \widehat \psi_p = 0 $ for all $ p\in \Lambda_+^* = \Lambda^*\setminus\{0\}$. Hence, $ \psi = c \ph_0 $ and from $ \|\psi\|_2=1$, it follows that $ |c| = 1$. 
%Finally, the Euler-Lagrange equation (\ref{eq:GPminEL}) follows from differentiating the map $ \eps \mapsto \| \psi + \eps\xi \|_2^{-1}( \psi + \eps \xi )$ at $ \eps=0$ for any $\xi \in D_{\text{GP}} $.
\end{proof}

%For $ \gamma\in (\frac12; \frac32)$, we let 
%		 \begin{equation}\label{eq:defThetaNUN} \Theta_N= \Theta(N^{\gamma}.), \hspace{0.5cm} U_N= N^{3\gamma}U(N^{\gamma}.)\end{equation} 
%We assume throughout the following that $ N$ is sufficiently large s.t. $0\leq N^{\gamma}\leq 1/2$. This ensures that $U_N$ has compact support in $\{x\in\bR^3:\frac12 N^{\gamma}\leq |x|\leq N^{\gamma}\}\subset \Lambda$.	

\medskip

\emph{Step 1. (Dyson's Lemma)}. In this step we prove a lower bound for $H_N$, defined in equation (\ref{eq:HN}), through a Hamiltonian with a less singular interaction potential. To reach this goal, we have to translate \cite[Lemma 4]{LSS} to the translation invariant setting. The adaptation is straightforward and we only recall the proof for the convenience of the reader.  
\begin{lemma}\label{lm:LSSlemmatorus} Let $ v \in L^3(\bR^3)$ be compactly supported in the ball of radius $ R_0 <1/2$ with scattering length $\mathfrak{a}(v)\geq 0$, let $R_0<R<1/2$ and denote by $\chi_R$ the characteristic function of the ball of radius $R$ centered at $x_0=0\in \Lambda$. Let $ \rho \in \ell^\infty (\Lambda^*)$ with $ 0\leq \rho_p\leq 1$, $ \rho_p =\rho_q$ whenever $ |p|=|q|$ and such that $h= (1-\rho)^{\vee} \in L^2(\Lambda)$ is bounded ($h$ is the function with Fourier coefficients $(1-\rho_p)$, for all $p \in \L^*$). Define $ g_R\in L^2(\Lambda)$ by
		\begin{equation} \label{eq:appNRS1}g_R(x) = \sup_{|y|\leq R} |h(x-y)-h(x) |\end{equation}
and $ j_R \in L^2(\Lambda)$ by 
		\begin{equation}\label{eq:appNRS2}j_R(x) = 16\pi \,\widehat{g_R}(0) g_R(x) = 16\pi g_R(x) \int_{\Lambda} g_R(y)\;dy\end{equation}
Then, for any positive, radial function $ u$ supported in $ \{x\in\Lambda: R_0\leq |x| \leq R\} $ with $ \widehat u(0) = 4\pi$ and for any $\eps>0$, we have in the sense of forms 
		\begin{equation*}%\label{eq:appNRS3}
		-\nabla \rho_p \,\chi_R(x)\, \rho_p \nabla + \frac12 v(x) \geq (1-\eps)\mathfrak{a}(v)u(x) -\frac {\mathfrak{a}(v)}\eps j_R(x)
		\end{equation*}
Here, $-\nabla \rho_p \chi_R(x) \rho_p \nabla$ localizes the Laplacian $-\Delta $ with periodic boundary conditions both in position space using $ \chi_R$ and in momentum space using $\rho$ ($\rho$ acts as multiplication with $\rho_p$ in Fourier space). 
\end{lemma}

\begin{proof} It is sufficient to show that for any smooth, periodic $ \psi \in C^\infty(\Lambda)$, we have for $\xi\in L^2(\Lambda)$, defined by $\widehat \xi_p =  \rho_p\widehat \psi_p $, that
		\begin{equation}\label{eq:appNRS4}
		\int_{|x|\leq R} \Big[ |\nabla \xi(x)|^2 + \frac12 v(x)|\psi(x)|^2\Big]\;dx \geq \int_{\Lambda}\Big[ (1-\eps)\mathfrak{a}(v)u(x)|\psi(x)|^2 -\frac {\mathfrak{a}(v)}\eps j_R(x)|\psi(x)|^2 \Big]\;dx
		\end{equation}
We prove (\ref{eq:appNRS4}) first in the special case where $ u(x) = R^{-2} \delta(|x|-R)$, $\delta$ denoting as usual the Dirac $\delta$-measure. The general case follows then by integrating over $B_R(0)$.

We denote by $ f_v$ the solution of the zero-energy scattering equation for $v$ in $\bR^3$, i.e.
		\begin{equation}\label{eq:appNRS5}\Big(-\Delta + \frac12 v\Big)f_v = 0\end{equation}
with $f_v(x)\to 1 $ as $|x|\to \infty$. Recall that $ f_v = 1- \mathfrak{a}(v)/|x| $ for $|x|\geq R_0$ and that 
		\begin{equation}\label{eq:appNRS6} \int_{\bR^3} \Big(|\nabla f_v(x)|^2 + \frac12v(x) |f_v(x)|^2\Big)\;dx = 4\pi \mathfrak{a}(v) \end{equation}
Denote by $ \nu$ a complex-valued function which is supported on the unit sphere $S^2 $ with $ \int_{S^2}|\nu|^2 =1$ and identify it with the map on $\bR^3$ taking values $  \nu(x/|x|)$. We define
		\[A = \int_{|x|\leq R}  \nu(x) \nabla\overline{\xi}(x) \cdot \nabla f_v(x)\;dx + \frac12 \int_{|x|\leq R} v(x) \overline{\psi}(x) f_v(x) \nu(x) \]
Applying Cauchy-Schwarz, performing an angular integration over $|\nu|^2$ and using (\ref{eq:appNRS6}), we arrive at
		\begin{equation}\label{eq:xi-A2} \int_{|x|\leq R} \Big[ |\nabla \xi(x)|^2 + \frac12 v(x)|\psi(x)|^2\Big]\;dx\geq \frac{|A|^2}{\mathfrak{a}(v)}\end{equation}
Hence, it is enough to prove a lower bound for $|A|^2$. By partial integration, we obtain 
		\[\begin{split}
		\int_{|x|\leq R}  \nu(x) \nabla\overline{\xi}(x) \cdot \nabla f_v(x)\;dx =&\; - \int_{|x|\leq R}\Big( \overline{\xi}(x) \nabla \nu(x) \cdot \nabla f_v(x) + \nu(x) \overline{\xi}(x) \Delta f_v(x)\Big)\;dx\\
		&\; + \frac{\mathfrak{a}(v)}{R^2} \int_{|x| = R} \overline{\xi}(x) \nu(x) \;d\omega_R\\
		=&\; - \int_{|x|\leq R} \nu(x) \overline{\xi}(x) \Delta f_v(x)\;dx + \frac{\mathfrak{a}(v)}{R^2} \int_{|x| = R} \overline{\xi}(x) \nu(x) \;d\omega_R
		\end{split}\]
where $ \omega_R$ denotes the surface measure for the surface of the ball of radius $R$ and where we used that $ \nabla \nu(x)\cdot \nabla f_v(x) = 0$ (because $ \nu $ is supported on the sphere and $ f_v $ is a radial function). With (\ref{eq:appNRS5}) and $\xi(x) =\psi(x) - (h\ast\psi)(x) =\psi(x) - \int_\Lambda h(x-y)\psi(y)\;dy$, the previous identity implies that 
		\[\begin{split}
		A = &\; \frac{\mathfrak{a}(v)}{R^2}\int_{|x|= R} \overline{\psi}(x)\nu(x) d\omega_R-\frac{\mathfrak{a}(v)}{R^2}\int_{|x|=R}(h\ast \overline{\psi})(x)\nu(x)d\omega_R \\
		&\;+\int_{|x|\leq R}\!\!(h\ast \overline{\psi})(x)\nu(x)\Delta f_v(x) dx\\
		=&\; \frac{\mathfrak{a}(v)}{R^2}\int_{|x|= R} \overline{\psi}(x)\nu(x) d\omega_R + \int_{\Lambda} \overline{\psi}(x) \bigg[\int_{|y|\leq R} h(x-y)\; d\mu(y)\bigg] dx
		\end{split}\]
where $ d\mu(y) = -\mathfrak{a}(v)R^{-2}\nu(y) \delta(|y|-R)\; dy + \nu(y)\Delta f_v(y)\;dy$ is supported in the ball of radius $R$.  Notice that we have used that $ h(x) = h(-x)$ for all $x\in\Lambda$, by defintion. We find that $\int_\Lambda d\mu(y) = 0$ and, by Cauchy-Schwarz, that 
		\[ \int_\Lambda d |\mu(y)| = \mathfrak{a}(v) \int_{S^2}|\nu| + \int_{S^2}|\nu|\int_{0}^R  \frac{d}{dr} \left[ r^2 f'_v (r) \right] dr = 2\mathfrak{a}(v) \int_{S^2}|\nu| \leq 4\sqrt{\pi} \mathfrak{a}(v)  \]
In particular, with (\ref{eq:appNRS1}) and (\ref{eq:appNRS2}), this implies that
		\[\begin{split}\bigg| \int_{\Lambda} \overline{\psi}(x) \bigg[\int_{|y|\leq R} h(x-y)\; d\mu(y)\bigg]\bigg|\leq&\; 4\sqrt{\pi}\mathfrak{a}(v)  \int_{\Lambda} |\psi(x)| g_R(x)\,dx\\
		\leq&\; \mathfrak{a}(v) \bigg(\int_{\Lambda} |\psi(x)|^2 j_R(x)\,dx\bigg)^{1/2} \end{split}\]
Finally, choosing $ \nu$ to be proportional to $\psi$ restricted to the sphere of radius $R$, that is, $ \nu(x) = (\int_{S^2} |\psi(Ry)|^2 dy)^{-1/2} \psi(Rx) $ for all $ x\in S^2$, we find
		\[A\geq \frac{\mathfrak{a}(v)}{R} \bigg(\int_{|x|= R} |\psi(x)|^2 d\omega_R \bigg)^{1/2} -\mathfrak{a}(v) \bigg(\int_{\Lambda} |\psi(x)|^2 j_R(x)\,dx\bigg)^{1/2} \]
and hence, by Cauchy-Schwarz, 
		\[\frac{|A|^2}{\mathfrak{a}(v)}\geq  \int_{|x|= R} \mathfrak{a}(v)(1-\eps)R^{-2}|\psi(x)|^2 -\frac{\mathfrak{a}(v)}{\eps} \int_{\Lambda} |\psi(x)|^2 j_R(x)\,dx \]
Together with (\ref{eq:xi-A2}), the last inequality implies (\ref{eq:appNRS4}) for $ u(x) = R^{-2}\delta(|x|-R)$. For a general $u$, we integrate the last inequality with $ R$ replaced by $ s\in [0;R]$ over $ u(s)s^2 ds$ with $ u(s) = u(x)$ for $|x|=s$ and use that $ \int_{0}^R s^2u(s)\,ds = 1 $, by assumption. Since $s\mapsto j_s(x) $ is monotonically increasing, this shows
		\[\int_{|x|\leq R} \Big[ |\nabla \xi(x)|^2 + \frac12 v(x)|\psi(x)|^2\Big]\;dx\geq \int_{\Lambda}\Big[ (1-\eps)\mathfrak{a}(v)u(x)|\psi(x)|^2 -\frac {\mathfrak{a}(v)}\eps j_R(x)|\psi(x)|^2 \Big]\;dx\]
\end{proof}

Following the notation of \cite{NRS}, we denote by $\Theta: \bR^3\to \bR$ a radial, smooth function s.t. $ 0\leq \Theta(x) \leq 1$ for all $x\in \bR^3$, $ \Theta(x)=0$ for $|x|\leq 1$ and $\Theta(x) =1$ 
for $|x|\geq 2$. Moreover, we denote by $ U:\bR^3\to \bR$ a non-negative, radially symmetric and smooth function supported in $\{x\in\bR^3:\frac12\leq |x|\leq 1\}$ with
		\begin{equation*} %\label{eq:hatU0} 
		\widehat U(0) = \int_{\bR^3} U(x)\;dx =4\pi \mathfrak{a}_0
		\end{equation*}
Here, $ \mathfrak{a}_0$ denotes the scattering length of the potential $V\in L^3(\bR^3)$, which we assume to be supported in the ball of radius $R_0 > 0$. 

\begin{lemma}\label{lm:genDyson} Let $ \gamma\in (\frac12, \frac32)$ and let $N$ be large enough s.t. $N^{-\gamma}>2R_0/N$. Then, for all $ s>0$, $0 < \eps < 1$ and for $H_N$ as defined in (\ref{eq:HN}), there exists a constant $ C>0$, independent 
		of $ N, s$ and $ \eps $, such that
				\begin{equation}\label{eq:genDyson} H_N\geq \sum_{i=1}^N p_i^2  \big(1- (1-\eps)\Theta(s^{-1} p_i)  \big) + \frac{(1-\eps)^2}{N} W_N -\frac{C N^{2-2\gamma}  s^5}{\eps}  \end{equation}
		Here, $ p_i^2$ corresponds to $ -\Delta_{x_i}$ in Fourier space and $W_N$ is defined by
				\begin{equation}\label{eq:defWN}W_N = \sum_{i\neq j} N^{3\gamma}U(N^{\gamma}(x_i-x_j)) \prod_{k\neq i,j } \Theta(2N^{\gamma}(x_k-x_j)) \end{equation}	
		\end{lemma}
\begin{proof} As explained in \cite{NRS}, the proof is an application of Lemma \ref{lm:LSSlemmatorus} with the choice $R= N^{-\gamma} $, $ v = N^2 V(N.)$, $ \mathfrak{a}(v) = \mathfrak{a}_0/N$ and $ \rho_p = \Theta(s^{-1}p)$ for $p\in \Lambda^*$, using \cite[Eq. (50) and (52)]{LSS}. Indeed, arguing as in \cite[Eq. (52)]{LSS}, the resulting function $ h = (1-\Theta(s^{-1}.))^{\vee}$ in Lemma \ref{lm:LSSlemmatorus} is such that $ h_s(.)=  s^{-3}h(./s)$ has bounded and integrable gradient, with the upper bounds independent of $s>0$. Observe that $ h$ has only finitely many non-zero Fourier coefficients so that for instance 
		\[|\nabla  h_s(.)|=| \nabla( s^{-3}h(./s))|\leq C s^{-4}\sum_{p\in\Lambda^*:|p|\leq 2s}|p|\leq C\]
By writing $h(.) = s^3 h_s(s.)$, it follows that $ j_R$, defined in (\ref{eq:appNRS2}), satisfies $ |j_R|\leq R^2 s^{5}  =N^{-2\gamma}s^5$. Then, for points $ y_j\in\Lambda$, $ j=1,\dots, N-1$, with $ \min_{j\neq k}|y_j-y_k|\geq 2N^{-\gamma}$, we have $ \sum_{j=1}^{N-1}\chi_{R}((x-y_j))\leq 1 $. Hence, Lemma \ref{lm:LSSlemmatorus} implies
		\[p^2 \Theta(s^{-1}p) + \frac12 \sum_{j=1}^{N-1}N^2V(N(x-y_j))\geq  \frac{(1-\eps)}{N}\sum_{j=1}^{N-1}N^{3\gamma}U(N^{\gamma}(x-y_j))- \frac{C \mathfrak{a}_0 N^{-2\gamma} s^5}{\eps}\]
Applying this bound in each coordinate $x_i$, multiplying both sides of the inequality by $(1-\eps)$ and using that $\prod_{k\neq i,j } \Theta(2N^{\gamma}(x_k-x_j)) \leq 1$ we obtain the claim.
\end{proof}

 \emph{Step 2. (Second Moment Estimate).} In the next step, we analyse the Hamiltonian
		\begin{equation*} %\label{eq:effHamNRS} 
		\wt{H}_N = \sum_{i=1}^N \wt{h}_{i} +  \frac{(1-\eps)^2}{N} W_N
		\end{equation*}
where we let $ \wt h= p^2   \big(1- (1-\eps)\Theta(s^{-1} p)  \big) +1 $ (defined as a multiplication operator in Fourier space) and where $ \wt h_{i}$ denotes the corresponding many-body operator acting on the $i$-th variable only. Comparing with the r.h.s. of (\ref{eq:genDyson}), we added here a constant to make sure that $ \wt h_{i} \geq 1$ for all $i =1, \dots , N$ (we will remove it when we will compare $H_N$ with $\wt{H}_N$). The next key step is bound the second moment of $ \wt{H}_N $ from below in terms of the second moment of $\sum_{i=1}^N \wt{h}_{i} $. To this end, we need the following lemma, which is the adaptation to the translation invariant setting of \cite[Lemma 3.2]{NRS} (similar results have been previously established in the study of the dynamics, for example in \cite[Lemma 6.4]{ESY2}).  
		\begin{lemma}\label{lm:NRSintbnds} Let $ 0\leq W\in L^1(\Lambda)\cap L^2(\Lambda)$ and consider the multiplication operator $ W(x-y)$ on $ L^2(\Lambda\times \Lambda)$. Then, 
		we have for all
		$ 0\leq \delta<\frac14$, $ 0<\eps<1$ and $s>0$ that
				\begin{equation}\label{eq:NRSintbnds}\begin{split}
				i)\hspace{0.2cm} &\; 0\leq W(x-y) \leq C\| W\|_{3/2}(-\Delta_x), \\
				ii)\hspace{0.2cm} &\; 	0\leq W(x-y)\leq C_\delta \|W\|_1 (1-\Delta_x)^{1-\delta}(1-\Delta_y)^{1-\delta}	\\
				iii)\hspace{0.2cm} &\; \wt h_x W(x-y) + W(x-y) \wt h_x \\ &\hspace{2cm} \geq -C (\|W\|_2 + (1+s^2) \|W\|_{3/2} ) (1-\Delta_x)(1-\Delta_y)
				\end{split}\end{equation}
		where $\wt h_x$ denotes the operator $\wt h$ acting only on the $x$-variable (recall that the parameter $s >0$ enters the definition of $\wt{h}$).
		\end{lemma}

\begin{proof}
The first two bounds i), ii) follow similarly as in \cite[Lemma 3.2]{NRS}, using H\"older's and Sobolev's inequalities on the torus (see e.g. \cite{BO} for a proof of Sobolev's inequality on the torus) and the fact that the discrete Fourier transform of the Green function of $ (1-\Delta)^{\delta-1}$ with Fourier coefficients $ (1+p^2)^{\delta -1}$, $p\in\Lambda^*$, is square summable in $\Lambda^*$ for any $ 0\leq \delta <\frac14 $.

Using partial integration on the torus, Cauchy-Schwarz and the bounds (\ref{eq:NRSintbnds}) i) and ii), 
we may proceed as in \cite[Lemma 3.2]{NRS} to deduce that 
		\begin{equation}\label{eq:NRSintbnds1} (-\Delta_x) W(x-y) + W(x-y) (-\Delta_x)\geq -C (\|W\|_{3/2} + \|W\|_2)(1-\Delta_x)(1-\Delta_y)
		\end{equation}
Indeed, to prove (\ref{eq:NRSintbnds1}), consider first smooth, periodic functions $ W\in C^{\infty}(\Lambda)$ and $ f \in C^{\infty}(\Lambda\times\Lambda)$. On such functions, $ -\Delta$ acts as the usual Laplacian in $\bR^3$. Hence, the fact that $ \nabla_x (W(x-y)) = -\nabla_y (W(x-y))$ and partial integration yield
		\begin{equation}\label{eq:NRSintbnds2}\begin{split}
		&\langle f, \big( (-\Delta_x) W(x-y) + W(x-y) (-\Delta_x)\big) f\rangle\\
		& = 2\int_{\Lambda\times\Lambda} |\nabla_x f(x,y)|^2 W(x-y)\;dxdy + 2 \operatorname {Re}\int_{\Lambda\times\Lambda}\nabla_x \overline{ f(x,y)} \nabla_x (W(x-y)) f(x,y) \;dxdy\\
		&\geq 2\operatorname{Re}\int_{\Lambda\times\Lambda}\big(\nabla_x\overline{f(x,y)} \nabla_y f(x,y)+\nabla_y(\nabla_x\overline{f(x,y)}) f(x,y) \big) W(x-y)\;dxdy
		\end{split}\end{equation}
Bounding the first term on the r.h.s. of (\ref{eq:NRSintbnds2}) by Cauchy-Schwarz and the estimate i) in (\ref{eq:NRSintbnds}) and the second term on the r.h.s. of (\ref{eq:NRSintbnds2}) by Cauchy-Schwarz and the estimate ii) in (\ref{eq:NRSintbnds}) (with $\delta=0$), we conclude (\ref{eq:NRSintbnds1}), for smooth, periodic $W\in C^{\infty}(\Lambda)$ and test functions $f \in C^{\infty}(\Lambda\times\Lambda)$. Since $C^{\infty}(\Lambda\times\Lambda)$ is dense in $ H^2(\Lambda\times \Lambda)$ (in fact, the set of smooth, periodic functions with only finitely many non-zero Fourier coeffficients is an operator core for $ -\Delta$ in $ \Lambda\times \Lambda$ with periodic boundary conditions), we obtain the operator bound (\ref{eq:NRSintbnds1}) on $ H^2(\Lambda\times\Lambda)$ for smooth, periodic $ W\in C^{\infty}(\Lambda)$. Finally, for a general $W\in L^1(\Lambda)\cap L^2(\Lambda)$ we can approximate it in $ L^2(\Lambda)$ by $\wt{W}\in C^{\infty}(\Lambda)$ and use the simple bound
		\[\begin{split} | \langle -\Delta_x f, (W-\wt W)f\rangle | &\leq \|f\|_{H^2(\Lambda)} \bigg( \int_{\Lambda\times\Lambda} |(W-\wt W)(x-y)|^2 |f(x,y)|^2\;dx dy\bigg)^{1/2}\\
		&\leq \|f\|_{H^2(\Lambda)} \|W-\wt W\|_{L^2(\Lambda)} \langle f, (1-\Delta_x)(1-\Delta_y)f\rangle^{1/2} \\ &\leq \|f\|^2_{H^2(\Lambda)} \|W-\wt W\|_{L^2(\Lambda)}   \end{split}\] 
by the estimate ii) in (\ref{eq:NRSintbnds}) (with $\delta =0$). This shows (\ref{eq:NRSintbnds1}). Finally, to prove the bound iii) in (\ref{eq:NRSintbnds}), we write
		\[\begin{split} 
		&\wt h_x W(x-y) + W(x-y) \wt h_x \\
		&\hspace{2.5cm}= 2W(x-y) + \eps \Big[ (-\Delta_x) W(x-y) + W(x-y) (-\Delta_x)\Big] \\
		&\hspace{3cm}+ (1-\eps) \Big[ p_x^2 (1-\Theta(s^{-1} p_x)  ) W(x-y) +  W(x-y) p_x^2   (1-\Theta(s^{-1} p_x))  \Big]
		\end{split}\]	
To bound the first line on the r.h.s. of the last equation, we drop the term $ 2W(x-y) \geq 0$ and apply (\ref{eq:NRSintbnds1}). To control the second line, on the other hand, we use that $0\leq  (1-\Theta(s^{-1}p))\leq \chi(|p|\leq 2s)$ for any $s>0$ and we proceed as in \cite[Eq. (3.9) to (3.10)]{NRS}.  
\end{proof}

Lemma \ref{lm:NRSintbnds} is used to deduce the following crucial result (similar estimates have been previously used in the analysis of the time-evolution, for example in \cite[Prop. 5.1]{EY} in the mean field setting, or in \cite[Prop. 3.1]{ESY2} in the Gross-Pitaevskii regime).
\begin{lemma}\label{lm:NRSsecondmoment} For every $0 < \eps < 1$, $s>0$ and $\gamma \in (\frac12; \frac32 )$, s.t. $N^{-\gamma}\gg N^{-2/3} $ as $N\to \infty$, we have, for sufficiently large $N$, 
				\begin{equation}\label{eq:NRSsecondmoment} \big (\wt H_N\big)^2\geq \frac13 \bigg (\sum_{i=1}^N \wt h_i\bigg)^{2}
				\end{equation}	
		\end{lemma} 
\begin{proof}
We proceed exactly as in the proof of \cite[Lemma 3.1]{NRS}, which is based on Cauchy-Schwarz estimates, the operator bounds from Lemma \ref{lm:NRSintbnds} and considering several cases to analyse the different contributions to $W_N$, defined in (\ref{eq:defWN}). We can apply the same analysis in our setting and conclude (\ref{eq:NRSsecondmoment}). 
\end{proof}

\emph{Step 3. (Three-Body estimate).}
In this step, we bound $\wt H_N$ from below by a mean field Hamiltonian,  up to errors which are given in terms of powers of $\wt H_N$. We observe, first of all, that the operator $W_N$ defined in (\ref{eq:defWN}) is such that 
\[W_N \geq \sum_{1\leq i<j \leq N}2 N^{3\gamma}U(N^{\gamma}(x_i-x_j)) - \sum_{i\neq j}N^{3\gamma}U(N^{\gamma}(x_i-x_j))\sum_{k\neq i, j} (1- \Theta(2N^\gamma(x_k-x_j))\]	
The second term on the r.h.s. vanishes if $|x_k - x_j| \geq N^{-\g}$ for all $k \not = i,j$; it gives instead an important contribution when there is at least one additional particle close to $i$ and $j$. The next lemma allows us to control this three-body term. 
%Notice that, by definition of $ U$, we have $2 \int_{\Lambda} N^{3\gamma}U(N^{\gamma}(x))= 8\pi 
%\mathfrak{a}_0$. 
\begin{lemma}\label{lm:NRSthreebody}	For every $0 < \eps < 1$, $s>0$ and $\gamma \in (\frac12; \frac32 )$, s.t. $N^{-\gamma}\gg N^{-2/3} $ as $N\to \infty$, we have, for sufficiently large $N$, 
\begin{equation*}%\label{eq:NRSthreebody1}
\sum_{i\neq j}N^{3\gamma}U(N^{\gamma}(x_i-x_j))\sum_{k\neq i, j} (1- \Theta(2N^\gamma(x_k-x_j))\leq C_{\eps, s} N^{-2\gamma-1} \big(\wt H_N\big)^4
\end{equation*}
for some constant $ C(\eps, s)>0$, which depends on $\eps, s$ but is independent of $N$. In particular, 
\begin{equation}\label{eq:NRSthreebody2}\wt H_N \geq \sum_{i=1}^N \wt h_i +\frac{(1-\eps)^2}{N} \sum_{1\leq i<j\leq N}2N^{3\gamma}U(N^{\gamma}(x_i-x_j)) -  
C_{\eps, s} N^{-2\gamma-2} \big(\wt H_N\big)^4\end{equation}
\end{lemma}
\begin{proof}
We proceed as in the proof of \cite[Lemma 3.4]{NRS}, which is based on the bounds from Lemma \ref{lm:NRSintbnds}. 
\end{proof}

\emph{Step 4. (Convergence of Ground State Energy).} Using Lemma \ref{lm:genDyson}, Lemma \ref{lm:NRSsecondmoment} and Lemma \ref{lm:NRSthreebody}, we are now able to show the convergence of the ground state energy per particle of the Hamiltonian $H_N$ to the minimum of the Gross-Pitaevskii functional $ \cE_{\text{GP}}$ in the limit $N\to \infty$. The proof follows from the same arguments as in \cite{NRS}. We recall the main steps for completeness only. Since the minimizer of the Gross-Pitaevskii functional $ \cE_{\text{GP}}$ is unique and since we do not include magnetic fields in our analysis, some steps of the analysis of \cite{NRS} can be slightly simplified. The proof relies crucially on the Quantum de Finetti Theorem which we state as in \cite{NRS}.
		\begin{theorem}[Quantum de Finetti] \label{thm:qdftheorem} Let $ \mathfrak{H}$ be a separable Hilbert space and assume that $ (\psi_N)_{N\in\NN}$ is a sequence with $ \psi_N \in \otimes_{sym}^N\mathfrak{H}$ and $ \|\psi_N\|_{\mathfrak{H}}=1$ for each $N\in\NN$. For $n \in \bN$, let $\gamma_N^{(n)} = \tr_{n+1, \dots, N} |\psi_N \rangle \langle \psi_N|$ denote the $n$-particle reduced density matrix associated with $\psi_N$. Assume that $\gamma^{(1)}_N$ converges, as $N\to \infty$, in trace class norm topology. Then, up to a subsequence, there exists a (unique) Borel probability measure $\mu$ on the unit sphere $S(\mathfrak{H})$ in $\mathfrak{H}$, invariant under the action of $S^1$, such that,  for every $n\in\NN$,
\begin{equation*}%\label{eq:qdftheorem} 
\lim_{N\to\infty} \Big| \gamma_N^{(n)} -\int_{S(\mathfrak{H})} |u^{\otimes n} \rangle \langle u^{\otimes n} | \; d\mu(u)\Big|=0 \end{equation*}
		\end{theorem}
Before we start to prove the energy convergence (\ref{eq:en-ti}), let us define the energy functionals $ \cE_{\text{NL}}^{\eps, s}$ for $0 < \eps < 1$ and $s>0$ by 
		\begin{equation*}% \label{eq:defNLfunctionals} 
		\cE_{\text{NL}}^{\eps, s} (u):= \langle u, \wt h u\rangle + (1-\eps)^2 4\pi \mathfrak{a}_0\int_{\Lambda}   |u(x)|^4  \;dx
		\end{equation*}
on the domain $ \cQ(-\Delta)=H^1(\Lambda)$, defined in (\ref{eq:defGPfuncdom}). Recalling that 
		\[ \wt h = (1-\eps) p^2(1 -\Theta(s^{-1}p)) + \eps p^2 +1 \]
we may argue as in the proof of Lemma \ref{lm:GPmin} to show that $\cE_{\text{NL}}^{\eps, s} $ has a unique, positive minimizer in  $ \cQ(-\Delta)\cap \{u\in L^2(\Lambda): \|u\| =1\}$ given by the constant function $ \ph_0=1_{|\Lambda}$, for any fixed $ 1>\eps>0 $ and $s>0$ . The minimum of $ \cE_{\text{NL}}^{\eps, s} $ in $ \cQ(-\Delta)\cap \{u\in L^2(\Lambda): \|u\| =1\}$ is therefore
		\begin{equation}\label{eq:mincENL}\cE_{\text{NL}}^{\eps, s}(\ph_0) = (1-\eps)^2 4\pi \mathfrak{a}_0 +1\end{equation}
and any other minimizer of $  \cE_{\text{NL}}^{\eps, s}$ is given by $  c\ph_0$ for some $ c\in\mathbb{C}$ with $ |c|=1$. 
		\begin{prop}\label{prop:mfappr} Let $ 0< \eps < 1$, $ s>0$ and $\gamma\in (\frac12;\frac32)$. Then
\begin{equation*} %\label{eq:mfappr} 
\lim_{N\to\infty} \frac{\inf \sigma(\wt H_N)}{N} = (1-\eps)^2 4\pi \mathfrak{a}_0 +1
\end{equation*}
\end{prop}
\begin{proof} The upper bound follows easily by testing $ \wt H_N$ with $ \ph_0^{\otimes N}$, so that we only need to prove the lower bound. Following the notation from \cite{NRS}, we denote by $ \wt \psi_N$ a ground state vector for $ \wt H_N$. Such a vector exists, because $ W_N\geq 0 $ and because $ \wt h$ has compact resolvent. Lemma \ref{lm:NRSsecondmoment} and the ground state equation imply that
		\[\langle \wt \psi_N, \wt h_1 \wt h_2 \wt \psi_N \rangle \leq C_{\eps, s}\]
for some $ C_{\eps, s}$ independent of $N$. Denoting by $ \wt \gamma_N^{(k)}  $ the $k$-particle reduced density matrices associated to $ \wt \psi_N$, equation (\ref{eq:NRSthreebody2}) implies that
		\begin{equation}\label{eq:mfappr1}  \lim_{N\to\infty} \frac{\inf \sigma(\wt H_N)}{N} \geq \liminf_{N\to\infty} \Big[ \operatorname{tr} \big( \wt h \wt \gamma_N^{(1)}  \big) + (1-\eps)^2 \operatorname{tr} \big( N^{3\gamma}U(N^{\gamma}(x-y))\wt \gamma_N^{(2)} \big) \Big]\end{equation}
Since $ \wt h$ has compact resolvent and since the previous bound shows that $ \operatorname{tr} \big( \wt h \wt \gamma_N^{(1)}  \big) $ is uniformly bounded in $ N$, standard arguments (see for instance the argument before \cite[Theorem 2]{LS2}) imply that, up to a subsequence, $ \wt \gamma_N^{(1)}$ converges to a limit in trace class norm. By Theorem \ref{thm:qdftheorem}, this shows that there exists a probability measure $ \wt\mu$ on the unit sphere $S(L^2(\Lambda))$, which is invariant under the action of $ S^1 $, such that for every $k\in\NN$
				\begin{equation}\label{eq:mfappr2}\lim_{N\to\infty} \Big| \wt \gamma_N^{(k)} -\int_{S(L^2(\Lambda))} |u^{\otimes k} \rangle \langle u^{\otimes k} | \; d\wt\mu(u)\Big|=0 \end{equation}
In particular, by $ \wt h\geq 0$ and Fatou's Lemma, we find that
		\begin{equation} \label{eq:mfappr3} \liminf_{N\to\infty} \operatorname{tr} \big( \wt h \wt \gamma_N^{(1)}  \big) \geq \int_{S(L^2(\Lambda))} \langle u, \wt h u\rangle \;d\wt\mu(u) \end{equation}
The last bound implies in particular that any $ u\in L^2(\Lambda)$ in the support of $ \wt\mu$ lies in the form domain $\cQ(\wt h) $, which is equal to $\cQ(\wt h) = \cQ(-\Delta)=H^1(\Lambda) $, by definition of $\wt h$. 

To deal with the interaction term on the r.h.s. of (\ref{eq:mfappr1}), we cannot apply Fatou's Lemma directly; we use a localization argument instead. Denote by $ \chi(\wt h\leq \zeta)$ the spectral projection of $\wt h$ onto $(-\infty; \zeta) $. Since $ \wt h$ has compact resolvent, $ \chi(\wt h\leq \zeta)$ is a finite rank operator for every $ \zeta >0$. We let $ P_\zeta = \chi(\wt h\leq \zeta)^{\otimes 2}$ and $ Q_\zeta = 1-P_\zeta$. Since $ N^{3\gamma }U(N^\gamma.)$ is pointwise non-negative, the Cauchy-Schwarz inequality 
yields the operator bound 
\[\begin{split}N^{3\gamma }U(N^\gamma(x-y)) &= (P_\zeta + Q_\zeta) N^{3\gamma }U(N^\gamma(x-y))(P_\zeta +Q_\zeta)\\ & \geq (1-\delta) P_\zeta N^{3\gamma }U(N^\gamma(x-y)) P_\zeta -\delta^{-1} Q_\zeta N^{3\gamma }U(N^\gamma(x-y)) Q_\zeta \end{split}\]	
Using the bound ii) in (\ref{eq:NRSintbnds}) and the fact that $-\Delta \leq C_{\eps, s} \wt h $, one arrives at
\[N^{3\gamma }U(N^\gamma(x-y))- P_\zeta N^{3\gamma }U(N^\gamma(x-y)) P_\zeta \geq -C_{\eps, s} (\delta ^{-1} \zeta^{-1/5} +\delta) \wt h_1 \wt h_2\]
Taking the trace against $ \wt \gamma_N^{(2)}$ and using that $ \langle \wt \psi_N, \wt h_1 \wt h_2 \wt \psi_N \rangle \leq C_{\eps, s}$, the last bound implies together with the choice $ \delta = \zeta^{-1/10}$ that
		\begin{equation*}%\label{eq:mfappr4}
		\operatorname{tr} \big( N^{3\gamma}U(N^{\gamma}(x-y))\wt \gamma_N^{(2)} \big) - \operatorname{tr} \big(P_\zeta N^{3\gamma}U(N^{\gamma}(x-y))P_\zeta\wt \gamma_N^{(2)} \big)\geq -C_{\eps, s}\zeta^{-1/10}
		\end{equation*}
But then, since the operator norm of $ P_\zeta N^{3\gamma }U(N^\gamma(x-y)) P_\zeta$ is bounded uniformly in $N$ by the bound ii) in (\ref{eq:NRSintbnds}) and by the definition of $P_\zeta$, the convergence (\ref{eq:mfappr2}) implies 
		\begin{equation*}%\label{eq:mfappr5}
		\begin{split}
		&\liminf_{N\to \infty} \operatorname{tr} \big( N^{3\gamma}U(N^{\gamma}(x-y))\wt \gamma_N^{(2)} \big) \\
		&\geq \liminf_{N\to\infty } \int_{S(L^2(\Lambda))} \langle P_\zeta u^{\otimes 2}, N^{3\gamma}U(N^{\gamma}(x-y)) P_\zeta u^{\otimes 2}\rangle \;d\wt\mu(u)  - C_{\eps, s} \zeta^{-1/10}\\
		& = 4\pi \frak{a}_0  \int_{S(L^2(\Lambda))} \|\chi(\wt h\leq \zeta)u \|_4^4 \;d\wt\mu(u)- C_{\eps, s} \zeta^{-1/10}
		\end{split}\end{equation*} 
Here, we have used in the last step that $P_\zeta $ is a finite rank projector and that \linebreak $ \lim_{N\to \infty} \langle u^{\otimes 2}, N^{3\gamma}U(N^{\gamma}(x-y)) u^{\otimes 2}\rangle_2 = 4\pi \frak{a}_0 \|u\|_4^4$ for every $ u\in H^1(\Lambda)$. Letting $ \zeta\to \infty$, using Fatou's Lemma and recalling (\ref{eq:mfappr3}) and (\ref{eq:mincENL}), we obtain 
\[\lim_{N\to\infty} \frac{\inf \sigma(\wt H_N)}{N} \geq \int_{S(L^2(\Lambda))} \Big[ \langle u, \wt h u\rangle  + 4\pi \frak{a}_0 (1-\eps)^2 \|u\|_4^4\Big] \;d\wt\mu(u) \geq  4\pi \mathfrak{a}_0 (1-\eps)^2 +1 \]
This proves the claim.
\end{proof}
		\begin{cor}\label{cor:gsenergyconv}
		Let $E_N$ denote the ground state energy of $H_N$, defined as in (\ref{eq:HN}). Then 
				\[\lim_{N\to\infty} \frac{E_N}{N} = 4\pi \mathfrak{a}_0\]
		\end{cor}
\begin{proof} It is enough to prove the lower bound, the upper bound follows from Prop. \ref{prop:GNell} by testing $ \cG_{N,\ell}$ with the vacuum in $ \cF_+^{\leq N}$. By equation (\ref{eq:genDyson}) and Proposition \ref{prop:mfappr}, we have for every fixed $0< \eps < 1$, $s>0$ that
		\[\liminf_{N\to\infty } \frac{E_N}{N} \geq \liminf_{N\to \infty} \frac{\inf \sigma(\wt H_N) }{N} -1 = (1-\eps)^2 4\pi \mathfrak{a}_0 \] 
Since $ \eps>0$ is arbitrary, this proves the claim.
\end{proof}

\emph{Step 5. (Convergence of Ground States).}
In this last step, we conclude the proof of (\ref{eq:BEC1}). We summarize the main steps from the proof of \cite{NRS}. 

The proof is based on a Feynman-Hellmann principle. For $ v\in L^2(\Lambda)$ and $k\in\NN$, let 
		\[S_{v,k } = \frac{k!}{N^{k-1} }\sum_{1\leq i_1< i_1< \dots <i_k\leq N} |v^{\otimes k}\rangle\langle v^{\otimes k} |_{i_1, \dots, i_k}\]
		\begin{lemma}\label{lm:FeynHell} Let $H_N$ be defined as in (\ref{eq:HN}). Then, for every $ v\in L^2(\Lambda)$ and $k\in\NN$, we have that
				\begin{equation*} %\label{eq:FeynHell} 
				\liminf_{N\to \infty} \frac{\inf \sigma(H_N - S_{v,k})}{N}\geq \inf_{u\in H^1(\Lambda), \|u\|_2=1} \big(\cE_{\emph{GP}}(u) - |\langle v, u\rangle|^{2k} \big)
				\end{equation*}
		\end{lemma}
\begin{proof} The Lemma is obtained along the same lines as Proposition \ref{prop:mfappr} and Corollary \ref{cor:gsenergyconv}, we refer to the proof of \cite[Lemma 4.3]{NRS} for the details. We remark that, compared to the proof of Proposition \ref{prop:mfappr}, one needs to argue in addition that 
		\[\begin{split}
		&\lim_{\eps\to 0 }\lim_{s\to \infty} \bigg( \inf_{u\in H^1(\Lambda), \|u\|_2=1}\big( \langle u, (\wt h-1)u\rangle +(1-\eps)^2 4\pi \mathfrak{a}_0 \|u\|_4^4  - |\langle v, u\rangle|^{2k} \big)\bigg)\\
		 &\hspace{8cm}= \inf_{u\in H^1(\Lambda), \|u\|_2=1} \big(\cE_{\emph{GP}}(u) - |\langle v, u\rangle|^{2k} \big)
		\end{split}\]
This follows from a standard compactness argument from \cite{LS2}, using (in our setting on the torus) that $ -\Delta +1$ has compact resolvent. For the details of the argument, we refer to \cite[Section 4B, Step 1]{NRS}.
\end{proof}

		\begin{prop}\label{prop:convstates} Let $H_N $ be defined as in (\ref{eq:HN}) and let $ (\psi_N)_{N\in\NN}$ be a normalized sequence in $ L^2_s(\Lambda^N)$ such that
				\[ \lim_{N\to\infty} \frac{\langle \psi_N , H_N \psi_N\rangle}{N} = 4\pi \mathfrak{a}_0\]
		Then, denoting by $ \ph_0$ the constant function $\ph_0=1_{|\Lambda}$ and by 
		$ (\gamma_N^{(k)})_{N\in\NN}$ the $k$-particle reduced density matrices associated to $ (\psi_N)_{N\in\NN}$, we have that
				\begin{equation}\label{eq:convstates}
				\lim_{N\to \infty } \operatorname{tr} \Big| \gamma_N^{(k)} - |\ph_0^{\otimes k}\rangle \langle \ph_0^{\otimes k}|\Big|=0
				\end{equation}
		\end{prop}
\begin{proof}
The assumption on $ (\psi_N)_{N\in\NN}$ and Lemma \ref{lm:FeynHell} imply that
		\[\limsup_{N\to\infty} \operatorname{tr} \big(|v^{\otimes k}\rangle\langle v^{\otimes k}| \gamma_N^{(k)}  \big)\leq 4\pi \mathfrak{a}_0 - \inf_{u\in H^1(\Lambda),\|u\|_2=1}\big(\cE_{\text{GP}}(u) - |\langle v, u\rangle|^{2k} \big)\]
for any $ v\in L^2(\Lambda)$ and $ k\in\NN$. Replacing $ v$ by $ \lambda^{1/(2k)}v$ in the previous bound shows that
		\[\limsup_{N\to\infty} \operatorname{tr} \big(|v^{\otimes k}\rangle\langle v^{\otimes k} |\gamma_N^{(k)}  \big)\leq \frac{1}{\lambda}\Big[4\pi \mathfrak{a}_0 - \inf_{u\in H^1(\Lambda),\|u\|_2=1}\big(\cE_{\text{GP}}(u) - \lambda|\langle v, u\rangle|^{2k} \big) \Big]\]
Now, denote by $ u_\lambda $ a normalized minimizer of $ u\mapsto \cE_{\text{GP}}(u) - \lambda|\langle v, u\rangle|^{2k}$. Then $ \langle u_\lambda, -\Delta u_\lambda\rangle$ is uniformly bounded in $\lambda$ so that, choosing a sequence $ \lambda_j\to 0$ as $j\to \infty$, the sequence $ (u_{\lambda_j })_{j\in \NN}$ has a weakly convergent subsequence in $H^1(\Lambda)$. Since $-\Delta+1$ has compact resolvent, we find that $ u_{\lambda_j} \to u_0 $ in $ L^2(\Lambda)$ and pointwise a.e. in $\Lambda$ as $j\to \infty$, choosing possibly a further subsequence. By Fatou's lemma, we conclude that $ u_0$ must be a minimizer of $ \cE_{\text{GP}}$ so that 
		\[\limsup_{j\to \infty } \frac{1}{\lambda_j } \Big[ 4\pi \mathfrak{a}_0 - \inf_{u\in H^1(\Lambda),\|u\|_2=1}\big(\cE_{\text{GP}}(u) - \lambda_{j}|\langle v, u\rangle|^{2k} \big)\Big]\leq |\langle v, \ph_0\rangle |^{2k}  \]
Here, we used the uniqueness of the minimizer of $\cE_{\text{GP}}$, by Lemma \ref{lm:GPmin}. In particular, the last bound implies that
		\begin{equation}\label{eq:convstate1} \limsup_{N\to\infty} \operatorname{tr} \big(|v^{\otimes k}\rangle\langle v^{\otimes k} |\gamma_N^{(k)}  \big)\leq  |\langle v, \ph_0\rangle |^{2k}  \end{equation}		
for any $ v\in L^2(\Lambda) $ and any $ k\in\NN$.

Arguing next as in the proof of Proposition \ref{prop:mfappr}, the Quantum de Finetti Theorem \ref{thm:qdftheorem} implies that, up to a subsequence, there exists a probability measure $ \mu$ on the unit sphere $S(L^2(\Lambda))$, which is invariant under the action of $ S^1 $, such that for every $k\in\NN$
			\begin{equation}\label{eq:convstate2}\lim_{N\to\infty} \Big| \gamma_N^{(k)} -\int_{S(L^2(\Lambda))} |u^{\otimes k} \rangle \langle u^{\otimes k} | \; d\mu(u)\Big|=0 \end{equation}
To conclude the proposition, we use the bound (\ref{eq:convstate1}) to show that the measure $\mu$ is supported on the set of minimizers of $ \cE_{\text{GP}}$, i.e. on $ \{ c\ph_0 \in L^2(\Lambda): c\in \mathbb{C}, |c|=1 \}$. Once this is proved, we immediately conclude (\ref{eq:convstates}) from (\ref{eq:convstate2}).

To show that $\mu$ has support in  $ \{ c\ph_0 \in L^2(\Lambda): c\in \mathbb{C}, |c|=1 \}$, assume by contradiction that there exists $v_0 \in L^2(\Lambda)$ in the support of $\mu$ s.t. $ v_0 $ is not a minimizer of $\cE_{\text{GP}}$. Denote by $ B_\delta$ the set of points in the support of $\mu$ s.t. $ \| v-v_0\|_2\leq \delta$. Then, there must exist some $\delta\in (0;\frac12)$ s.t.
		\begin{equation}\label{eq:convstate5} | \langle v, \ph_0\rangle| \leq 1 - 3\delta^2\end{equation}
for all $ v\in B_\delta$. If this was not the case, we would find a sequence $ (v_j)_{j\in\NN} $ in the support of $\mu $ converging in $ L^2(\Lambda)$ to $v_0$ as well as to $ \ph_0$. But this contracticts our assumption that $v_0$ is not a minimizer of $ \cE_{\text{GP}}$. Hence, pick such a $\delta \in (0;\frac12)$ s.t. (\ref{eq:convstate5}) holds true. By the triangle inequality, we also have that $  | \langle v, u\rangle| \geq 1- 2\delta ^2$ for all $ u, v \in B_\delta$. But then (\ref{eq:convstate1}) and (\ref{eq:convstate2}) imply that
		\[\begin{split} 
		\mu(B_\delta)^2 (1-2\delta^2)^{2k}\leq&\; \int_{B_\delta} \int_{B_\delta} | \langle v, u\rangle|^{2k} \; d\mu(u) d\mu (v) \\
		\leq &\; \int_{B_\delta} | \langle v, \ph_0\rangle|^{2k} \; d\mu(v) \leq \mu (B_\delta) (1-3\delta^2)^{2k}
		\end{split}\]
In particular, by letting $k\to \infty$ in the previous bound, we find that $ \mu (B_\delta)=0$, which is a contradiction to the fact that $ v_0\in \mu (B_\delta)$ is in the support of $\mu$ and that $\mu$ is a Borel measure. This concludes the proof.
\end{proof}
Proposition \ref{prop:convstates} completes the proof of (\ref{eq:BEC1}).

\end{document}